\documentclass[12pt,openright,a4paper,twoside,DIV12,BCOR20mm,headings=small,
headsepline,listof=totoc, bibliography=totoc]{scrbook}
\usepackage{amsmath,amsfonts,amssymb,amsthm} \allowdisplaybreaks
\usepackage[english]{babel}
\usepackage{url}
\usepackage{fncychap}
\usepackage{tikz}
\usepackage{tikz-3dplot}
\usepackage{enumerate}
\usepackage[bf,small]{caption}
\usepackage[automark]{scrpage2}
\newcommand{\defeq}{:=}
\newcommand{\ZZ}{\mathbb{Z}}
\newcommand{\RR}{\mathbb{R}}
\newcommand{\CC}{\mathbb{C}}
\newcommand{\NN}{\mathbb{N}}
\newcommand{\EE}{\mathbb{E}}
\newcommand{\cc}{\mathfrak{C}}
\newcommand{\PP}{\mathbb{P}}
\newcommand{\Pro}{\ensuremath{p}}
\newcommand{\Inc}{\ensuremath{\iota}}
\newcommand{\euler}{\mathrm{e}}
\newcommand{\norm}[1]{\left\Vert#1\right\Vert}
\newcommand{\cN}{\mathcal{N}}
\newcommand{\BV}{\mathrm{BV}}
\newcommand{\cD}{\mathcal{D}}
\DeclareMathOperator{\supp}{\operatorname{supp}}
\DeclareMathOperator{\re}{Re}
\DeclareMathOperator{\im}{Im}
\DeclareMathOperator{\diam}{diam}
\DeclareMathOperator{\dist}{dist}
\DeclareMathOperator{\Tr}{Tr}
\DeclareMathOperator{\cov}{cov}
\DeclareMathOperator{\ran}{Ran}
\newcommand{\sprod}[2]{\left\langle#1,#2 \right \rangle}
\newcommand{\abs}[1]{\left\vert#1\right\vert}
\newcommand{\drm}{\ensuremath{\mathrm{d}}}
\renewcommand{\i}{\protect{\mathrm{i}}}
\newcommand{\BIGOP}[1]{\mathop{\mathchoice%
{\raise-0.22em\hbox{\huge $#1$}}%
{\raise-0.05em\hbox{\Large $#1$}}{\hbox{\large $#1$}}{#1}}}
\newcommand{\BIGboxplus}{\mathop{\mathchoice%
{\raise-0.35em\hbox{\huge $\boxplus$}}%
{\raise-0.15em\hbox{\Large $\boxplus$}}{\hbox{\large $\boxplus$}}{\boxplus}}}
\newcommand{\bigtimes}{\BIGOP{\times}}
\newtheorem{theorem}{Theorem}[chapter]
\newtheorem{proposition}[theorem]{Proposition}
\newtheorem{lemma}[theorem]{Lemma}
\theoremstyle{definition}                       
\newtheorem{definition}[theorem]{Definition}
\newtheorem{assumption}{Assumption}

\newtheorem{example}[theorem]{Example}
\theoremstyle{remark}                           
\newtheorem{remark}[theorem]{Remark}
\newcommand{\hm}[1]{\textbf{*}\leavevmode{\marginpar{\tiny%
$\hbox to 0mm{\hspace*{-0.5mm}$\leftarrow$\hss}%
\vcenter{\vrule depth 0.1mm height 0.1mm width \the\marginparwidth}%
\hbox to 0mm{\hss$\rightarrow$\hspace*{-0.5mm}}$\\\relax \raggedright #1}}} 
\usepackage{microtype}
\begin{document}
\frontmatter
\begin{titlepage}
\begin{center}
{\Huge \textbf{\sffamily Localization for alloy-type \\ models with non-monotone \\[1ex] potentials}} \\[10ex] 
{\sffamily by} \\[10ex]
\textbf{\sffamily Martin Tautenhahn} \\[17ex]
\tdplotsetmaincoords{70}{150}
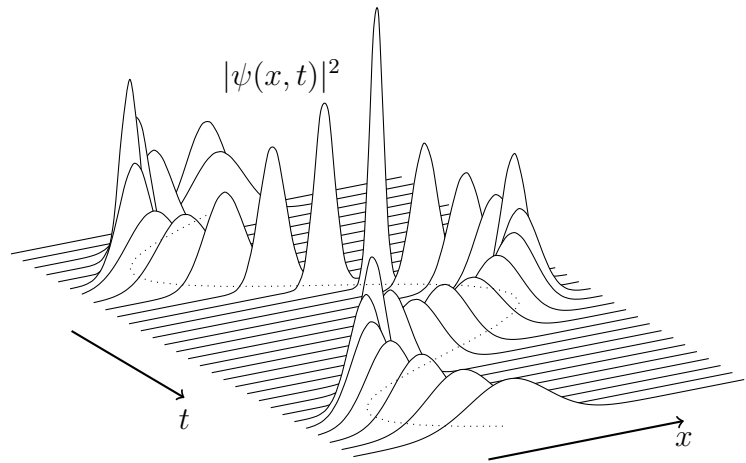
\begin{tikzpicture}[scale=1,tdplot_main_coords]
%

\draw[thick,->] (1.5,10.1,0) -- (-1.5,10.1,0) node[anchor=north]{$x$};
\draw[thick,->] (3.8,3,0) -- (3.8,6,0) node[anchor=north]{$t$};
\draw[thick] (0,2,2.6) node{$\lvert \psi (x,t) \rvert^2$};
\def\a{1};
\filldraw[fill=white, draw=black,smooth,domain=-3:3, samples=50] plot (\x,0*0.314,{(\a/(0.3*2.5))*exp(-0.5*((\x-0)/0.3)*((\x-0)/0.3))});
\filldraw[fill=white, draw=black,smooth,domain=-3:3, samples=50] plot (\x,1*0.314,{(\a/(0.4*2.5))*exp(-0.5*((\x-0)/0.4)*((\x-0)/0.4))});
\filldraw[fill=white, draw=black,smooth,domain=-3:3, samples=50] plot (\x,2*3.14/10,{(\a/(0.3*2.5))*exp(-0.5*((\x-1.176)/0.3)*((\x-1.176)/0.3))});
\filldraw[fill=white, draw=black,smooth,domain=-3:3, samples=50] plot (\x,3*3.14/10,{(\a/(0.2*2.5))*exp(-0.5*((\x-1.618)/0.2)*((\x-1.618)/0.2))});
\filldraw[fill=white, draw=black,smooth,domain=-3:3, samples=50] plot (\x,4*3.14/10,{(\a/(0.15*2.5))*exp(-0.5*((\x-1.902)/0.15)*((\x-1.902)/0.15))});
\filldraw[fill=white, draw=black,smooth,domain=-3:3, samples=50] plot (\x,5*3.14/10,{(\a/(0.25*2.5))*exp(-0.5*((\x-2.000)/0.25)*((\x-2.000)/0.25))});
\filldraw[fill=white, draw=black,smooth,domain=-3:3, samples=50] plot (\x,6*3.14/10,{(\a/(0.4*2.5))*exp(-0.5*((\x-1.902)/0.4)*((\x-1.902)/0.4))});
\filldraw[fill=white, draw=black,smooth,domain=-3:3, samples=50] plot (\x,7*3.14/10,{(\a/(0.4*2.5))*exp(-0.5*((\x-1.618)/0.4)*((\x-1.618)/0.4))});
\filldraw[fill=white, draw=black,smooth,domain=-3:3, samples=50] plot (\x,8*3.14/10,{(\a/(0.3*2.5))*exp(-0.5*((\x-1.176)/0.3)*((\x-1.176)/0.3))});
\filldraw[fill=white, draw=black,smooth,domain=-3:3, samples=50] plot (\x,9*3.14/10,{(\a/(0.2*2.5))*exp(-0.5*((\x-0.618)/0.2)*((\x-0.618)/0.2))});
\filldraw[fill=white, draw=black,smooth,domain=-3:3, samples=50] plot (\x,10*3.14/10,{(\a/(0.15*2.5))*exp(-0.5*((\x-0.000)/0.15)*((\x-0.000)/0.15))});
\filldraw[fill=white, draw=black,smooth,domain=-3:3, samples=80] plot (\x,11*3.14/10,{(\a/(0.1*2.5))*exp(-0.5*((\x+0.618)/0.1)*((\x+0.618)/0.1))});
\filldraw[fill=white, draw=black,smooth,domain=-3:3, samples=50] plot (\x,12*3.14/10,{(\a/(0.2*2.5))*exp(-0.5*((\x+1.176)/0.2)*((\x+1.176)/0.2))});
\filldraw[fill=white, draw=black,smooth,domain=-3:3, samples=50] plot (\x,13*3.14/10,{(\a/(0.25*2.5))*exp(-0.5*((\x+1.618)/0.25)*((\x+1.618)/0.25))});
\filldraw[fill=white, draw=black,smooth,domain=-3:3, samples=50] plot (\x,14*3.14/10,{(\a/(0.3*2.5))*exp(-0.5*((\x+1.902)/0.3)*((\x+1.902)/0.3))});
\filldraw[fill=white, draw=black,smooth,domain=-3:3, samples=50] plot (\x,15*3.14/10,{(\a/(0.2*2.5))*exp(-0.5*((\x+2.000)/0.2)*((\x+2)/0.2))});
\filldraw[fill=white, draw=black,smooth,domain=-3:3, samples=50] plot (\x,16*3.14/10,{(\a/(0.3*2.5))*exp(-0.5*((\x+1.902)/0.3)*((\x+1.902)/0.3))});
\filldraw[fill=white, draw=black,smooth,domain=-3:3, samples=50] plot (\x,17*3.14/10,{(\a/(0.35*2.5))*exp(-0.5*((\x+1.618)/0.35)*((\x+1.618)/0.35))});
\filldraw[fill=white, draw=black,smooth,domain=-3:3, samples=50] plot (\x,18*3.14/10,{(\a/(0.4*2.5))*exp(-0.5*((\x+1.176)/0.4)*((\x+1.176)/0.4))});
\filldraw[fill=white, draw=black,smooth,domain=-3:3, samples=50] plot (\x,19*3.14/10,{(\a/(0.45*2.5))*exp(-0.5*((\x+0.618)/0.45)*((\x+0.618)/0.45))});
\filldraw[fill=white, draw=black,smooth,domain=-3:3, samples=50] plot (\x,20*3.14/10,{(\a/(0.4*2.5))*exp(-0.5*((\x+0)/0.4)*((\x+0)/0.4))});
\filldraw[fill=white, draw=black,smooth,domain=-3:3, samples=50] plot (\x,21*3.14/10,{(\a/(0.4*2.5))*exp(-0.5*((\x-0.618)/0.4)*((\x-0.618)/0.4))});
\filldraw[fill=white, draw=black,smooth,domain=-3:3, samples=50] plot (\x,22*3.14/10,{(\a/(0.3*2.5))*exp(-0.5*((\x-1.176)/0.3)*((\x-1.176)/0.3))});
\filldraw[fill=white, draw=black,smooth,domain=-3:3, samples=50] plot (\x,23*3.14/10,{(\a/(0.2*2.5))*exp(-0.5*((\x-1.618)/0.2)*((\x-1.618)/0.2))});
\filldraw[fill=white, draw=black,smooth,domain=-3:3, samples=50] plot (\x,24*3.14/10,{(\a/(0.25*2.5))*exp(-0.5*((\x-1.902)/0.25)*((\x-1.902)/0.25))});
\filldraw[fill=white, draw=black,smooth,domain=-3:3, samples=50] plot (\x,25*3.14/10,{(\a/(0.3*2.5))*exp(-0.5*((\x-2.000)/0.3)*((\x-2.000)/0.3))});
\filldraw[fill=white, draw=black,smooth,domain=-3:3, samples=50] plot (\x,26*3.14/10,{(\a/(0.35*2.5))*exp(-0.5*((\x-1.902)/0.35)*((\x-1.902)/0.35))});
\filldraw[fill=white, draw=black,smooth,domain=-3:3, samples=50] plot (\x,27*3.14/10,{(\a/(0.4*2.5))*exp(-0.5*((\x-1.618)/0.4)*((\x-1.618)/0.4))});
\filldraw[fill=white, draw=black,smooth,domain=-3:3, samples=50] plot (\x,28*3.14/10,{(\a/(0.5*2.5))*exp(-0.5*((\x-1.176)/0.5)*((\x-1.176)/0.5))});
\filldraw[fill=white, draw=black,smooth,domain=-3:3, samples=50] plot (\x,29*3.14/10,{(\a/(0.6*2.5))*exp(-0.5*((\x-0.618)/0.6)*((\x-0.618)/0.6))});
\draw[dotted,
color=black, domain=0:9.106, samples=30, color=black] plot ({2*sin(\x r)},\x,0);
%
\end{tikzpicture} \\[17ex]
{\sffamily 2012}
\end{center}
\newpage
\thispagestyle{empty}
\noindent
The present manuscript is a slightly modified version of my doctoral thesis under the advise of Prof.\, Ivan Veseli\'c. The thesis was submitted on 31/01/2012 and defended on 21/06/2012.
\vfill
\noindent
\textbf{Prof.~Dr.~rer.~nat.~habil.~Ivan Veseli\'c}\\
Technische Universit\"at Chemnitz \\
Fakult\"at f\"ur Mathematik \\
D-09107 Chemnitz \\[2ex]
\textbf{Dr.~rer.~nat.~Martin Tautenhahn}\\
Technische Universit\"at Chemnitz \\
Fakult\"at f\"ur Mathematik \\
D-09107 Chemnitz
\end{titlepage}
\phantom{sd}\thispagestyle{empty}
\vspace{.25\textheight}
\begin{flushright}
Dedicated to my father Ulrich  Tautenhahn
\end{flushright}
\newpage {\phantom{sd}\thispagestyle{empty}  } 
\newpage
 \begin{center}
  \textsc{\bf  Acknowledgement} 
 \end{center}
\thispagestyle{plain}
The first thank belongs beyond doubt to my adviser Prof.\ Ivan Veseli\'c. His guidance, knowledge, endurance, inspiration and friendship during the last years are invaluable for me and my research. Thank you Ivan! Furthermore, let me thank the research group working on mathematical physics in Chemnitz, namely Prof.\ Peter Stollmann, Marcel Hansmann, Karsten Leonhardt, Reza Samavat, Carsten Schubert, Christoph Schumacher, Fabian Schwarzenberger, Christian Seifert and Daniel Wingert. A lot of discussions either in passing or in a more formal framework like in our joint research seminar contributed to my research. I appreciate this friendly and fertile atmosphere in Chemnitz. Special thanks go to my office-mate and friend Fabian Schwarzenberger for increasing the fun and decreasing the confusion on math, for going the extra mile with the kind of feedback that put flesh on the bones.
\\[2ex]
As Prof.\ Daniel Lenz left Chemnitz and started a position in Jena, the alliance Chemnitz-Jena never broke down. I would like to acknowledge and extend my heartfelt gratitude to the group around Prof.\ Daniel Lenz for their support and for many stimulating discussions. I always enjoyed the visits in Jena!
\\[2ex]
It is a honor for me to thank my co-authors Alexander Elgart, Helge Kr\"uger, Ivan Veseli\'c and Norbert Peyerimhoff. The cooperation was always friendly and fruitful. I learned a big deal from all of them and thank for the endurance they have with me, never accepting less than my best efforts.
\\[2ex]
Brooding on the past, I offer my regards and blessings to all of those who supported me in any respect during the completion of this thesis. Not least I wish to express my appreciation common to all who have been available for scientific conversations for years. 
\\[2ex]
My love and gratitude is dedicated to my parents Angela and Ulrich Tautenhahn, and my sister Susanne Tautenhahn. In particular, I am deeply grateful to my father Ulrich. Without his support, especially with his mathematical insight and intuition, this work would not have been possible.
\newpage {\phantom{sd}\thispagestyle{empty}} 
\newpage
 \begin{center}
  \textsc{\bf  Abstract} \thispagestyle{plain}
 \end{center}
We consider a family of self-adjoint operators 
\[
H_\omega = - \Delta + \lambda V_\omega, \quad \omega \in \Omega = \bigtimes_{k \in \ZZ^d} \RR , 
\]
on the Hilbert space $\ell^2 (\ZZ^d)$ or $L^2 (\RR^d)$. Here $\Delta$ denotes the Laplace operator (discrete or continuous), $V_\omega$ is a multiplication operator given by the function
$$V_\omega (x) = \sum_{k \in \ZZ^d} \omega_k u(x-k) \ \text{on $\ZZ^d$, or}\quad  V_\omega (x) = \sum_{k \in \ZZ^d} \omega_k U(x-k) \ \text{on $\RR^d$},$$ and $\lambda > 0$ is a real parameter modeling the strength of the disorder present in the model. The functions $u:\ZZ^d \to \RR$ and $U:\RR^d \to \RR$ are called single-site potential. Moreover, there is a probability measure $\PP$ on $\Omega$ modeling the distribution of the individual configurations $\omega \in \Omega$. The measure $\PP = \prod_{k \in \ZZ^d} \mu$ is a product measure where $\mu$ is some probability measure on $\RR$ satisfying certain regularity assumptions. The operator on $L^2 (\RR^d)$ is called alloy-type model, and its analogue on $\ell^2 (\ZZ^d)$ discrete alloy-type model.
\par
In the pioneer work \cite{Anderson-58}, Anderson argued that the solution of the Schr\"o\-din\-ger equation according to the operator $H_\omega$, the so called wave function, becomes localized in space if the disorder is large enough. This phenomenon of localization manifests itself in the sense that there are intervals $I\subset \RR$ such that the continuous spectrum of $H_\omega$ in $I$ is empty for $\PP$-almost all $\omega \in \Omega$ and the corresponding eigenfunctions decay exponentially, called exponential localization or Anderson localization.
\par
There are two methods to prove exponential localization in multidimensional space, the multiscale analysis and the fractional moment method. However, both methods strongly rely on the property that the operator $H_\omega$ depends (in the sense of quadratic forms) monotonically on the random parameters $\omega_k$, $k \in \ZZ^d$. This is for instance the case if the single-site potential is non-negative. 
\par
This thesis refines the methods in the case where the single-site potential is allowed to change its sign. In particular, we develop the fractional moment method and prove exponential localization for the discrete alloy-type model in the case where the support of $u$ is finite and $u$ has fixed sign at the boundary of its support. We also prove a Wegner estimate for the discrete alloy-type model in the case of exponentially decaying but not necessarily finitely supported single-site potentials. This Wegner estimate is applicable for a proof of localization via multiscale analysis. In an appendix we prove a Wegner estimate for the alloy-type model if the single-site potential is a so-called generalized step-function. Moreover, we show for the alloy-type model that the typical fractional moment decay implies localization under minimal assumptions on the measure $\PP$ in the case where $U$ has bounded support.
\newpage {\phantom{sd}\thispagestyle{empty}  } 
\newpage

 \begin{center}
  \textsc{\bf  Publication notice} \thispagestyle{plain}
 \end{center}
\makeatletter
\newenvironment{mybibliography}[1]
     {\@mkboth{\MakeUppercase\bibname}{\MakeUppercase\bibname}%
      \list{\@biblabel{\@arabic\c@enumiv}}%
           {\settowidth\labelwidth{\@biblabel{#1}}%
            \leftmargin\labelwidth
            \advance\leftmargin\labelsep
            \@openbib@code
            \usecounter{enumiv}%
            \let\p@enumiv\@empty
            \renewcommand\theenumiv{\@arabic\c@enumiv}}%
      \sloppy
      \clubpenalty4000
      \@clubpenalty \clubpenalty
      \widowpenalty4000%
      \sfcode`\.\@m}
     {\def\@noitemerr
       {\@latex@warning{Empty `mybibliography' environment}}%
      \endlist}
\makeatother
The major part of this thesis concerns joint works with Alexander Elgart, Norbert Peyerimhoff and Ivan Veseli\'c. The relevant references are:

\begin{mybibliography}{EKTV11}
\bibitem[ETV10]{0_ElgartTV-10}
A.~Elgart, M.~Tautenhahn, and I.~Veseli\'c, \emph{Localization via fractional
  moments for models on $\mathbb{Z}$ with single-site potentials of finite
  support}, J. Phys. A: Math. Theor. \textbf{43} (2010), no.~47, 474021.

\bibitem[TV10a]{0_TautenhahnV-10b}
M.~Tautenhahn and I.~Veseli\'c, \emph{A note on regularity for discrete
  alloy-type models}, Technische Universit\"at Chemnitz, Preprintreihe der
  Fakult\"at f\"ur Mathematik, Preprint 2010-6, ISSN 1614-8835, 2010.

\bibitem[ETV11]{0_ElgartTV-11}
A.~Elgart, M.~Tautenhahn, and I.~Veseli\'c, \emph{Anderson localization for a class of models with a
  sign-indefinite single-site potential via fractional moment method}, Ann.
  Henri Poincar\'e \textbf{12} (2011), no.~8, 1571--1599.

\bibitem[PTV11]{0_PeyerimhoffTV-11}
N.~Peyerimhoff, M.~Tautenhahn, and I.~Veseli\'c, \emph{Wegner estimate for
  alloy-type models with sign-changing and exponentially decaying single-site
  potentials}, Technische Universit\"at Chemnitz, Preprintreihe der Fakult\"at
  f\"ur Mathematik, Preprint 2011-9, ISSN 1614-8835, 2011.
\end{mybibliography}
Related publications or conference proceedings I (co-)authored are:
\begin{mybibliography}{EKTV11}

\bibitem[TV10b]{0_TautenhahnV-10}
M.~Tautenhahn and I.~Veseli\'c, \emph{Spectral properties of discrete alloy-type models}, XVIth
  International Congress On Mathematical Physics (P.~Exner, ed.), 2010,
  pp.~551--555.

\bibitem[Tau11]{0_Tautenhahn-11}
M.~Tautenhahn, \emph{Localization criteria for anderson models on locally
  finite graphs}, J. Stat. Phys. \textbf{144} (2011), no.~1, 60--77.

\bibitem[EKTV11]{0_ElgartKTV-11}
A.\ Elgart, H.\ Kr{\"u}ger, M.\ Tautenhahn, and I.\ Veseli\'c, \emph{Discrete
  {S}chr\"odinger operators with random alloy-type potential},
  arXiv:1107.2800v1 [math-ph]  (2011), to appear in Proceedings of the Spectral Days
  2010, Pontificia Universidad Cat\'olica de Chile, Santiago.

\end{mybibliography}
\newpage {\phantom{sd}\thispagestyle{empty}  } 
\newpage
\tableofcontents
\listoffigures
%
%
\mainmatter
\chapter{Introduction}
In this chapter we lead the reader to the topic of localization theory for random Schr\"o\-ding\-er operators. 
In a first section we introduce very basic mathematical concepts of quantum mechanics developed in the 1920's, which is used to model the time-evolution of atomic particles like electrons. The time-evolution of an electron moving under the influence of a static electric potential is governed by a self-adjoint operator in some Hilbert space and the associated time-dependent Schr\"o\-ding\-er equation, introduced in 1926 by E.~Schr\"o\-ding\-er \cite{Schroedinger-26}. 
The spectrum of this operator gives insights into the asymptotic behavior of the time-evolution of the atomic particle as discussed in Section~\ref{sec:foundations}. 
\par
In Section \ref{sec:rand_operators} we will introduce the concept and examples of ergodic random operators. They are for instance used to model the time-evolution of an electron under the influence of a random electric potential. The application one should have in mind is a disordered solid; each configuration of the randomness corresponds to a possible realization of the random medium.
\par
In the pioneer work \cite{Anderson-58}, Anderson argued that the solution of the Schr\"o\-ding\-er equation (called wave function) becomes localized in space in certain randomness/energy regions if one considers random operators. The medium has lost its transport properties when compared to ideal crystals. This phenomenon of localization for random operators manifests itself in the fact that either the spectrum of the considered random operator is only of pure point type in some energy region (\textit{spectral localization}), or that the wave function (corresponding to some energy interval) stays trapped in a finite region of space for all time (\textit{dynamical localization}). Localization for random operators will be defined in Section \ref{sec:phenomenon_loc} and the main matter of this thesis is to investigate when localization occurs.
\par
Section \ref{sec:methods_msa_fmm} is devoted to a discussion of the existing methods to prove localization for random operators. Beside some methods only available in the one-dimensional setting, there are exactly two methods to prove localization for random operators: the \textit{multiscale analysis} and the \textit{fractional moment method}. 
At the end of Section \ref{sec:methods_msa_fmm} we explain the structure of this thesis.
\section{Some mathematical foundations of quantum mechanics} \label{sec:foundations}
%
In quantum mechanics (Schr\"o\-ding\-er picture), the state of an electron moving in $d$-dimensional space $\RR^d$ is described by a complex valued function $\psi : \RR^d \times \RR \to \CC$, $\psi (x,t) = \psi (x_1,\ldots , x_d , t)$, called the \emph{wave function}\index{wave function}, where $x \in \RR^d$ corresponds to a point in space and $t$ corresponds to the time variable. The quantity $\lvert \psi (\cdot , t) \rvert^2$ is interpreted as the probability density of the particles location at time $t$. For this reason it is reasonable to assume that $\psi (\cdot , t)$ is an element of $L^2 (\RR^d)$ with $\lVert \psi (\cdot , t) \rVert_{L^2} = 1$ for all $t \in \RR$. For measurable sets $A \subset \RR^d$, the number
\[
 \int_A \lvert \psi (x , t) \rvert^2 \drm x
\]
is interpreted as the probability of finding the particle in $A$ at time $t$. Note that in contrast to classical mechanics, the particle is not localized at a certain point in space, rather it is spread in space according to the probability density $\lvert \psi (\cdot , t) \rvert^2$. 
\par
Given an initial state $\psi_0 \in L^2 (\RR^d)$ with $\lVert \psi_0 \rVert_{L^2} = 1$, the time evolution of the wave function is governed by the time-dependent Schr\"o\-ding\-er equation\index{Schr\"o\-ding\-er equation}
\begin{equation} \label{eq:schroedinger}
 \i \hbar \frac{\partial}{\partial t} \psi (\cdot , t) = - \frac{\hbar^2}{2m} \Delta \psi (\cdot , t) + V \psi (\cdot , t), \quad \psi (\cdot , 0) = \psi_0,
\end{equation}
where $m$ denotes the mass of the particle, $\hbar = h/(2\pi)$, $\Delta$ is the Laplace operator on $L^2 (\RR^d)$ and $V$ is the multiplication operator on $L^2 (\RR^d)$ by the (classical) potential energy. Here $h$ denotes the Planck constant. The operator $H = -(\hbar^2 / 2m) \Delta + V$ from Eq.~\eqref{eq:schroedinger} is called \emph{Schr\"o\-ding\-er operator}. For the analysis of the Schr\"o\-ding\-er equation the constants $m$ and $\hbar$ are irrelevant. On this account we set $\hbar/(2m)$ to one for the rest of this thesis. Under some mild regularity properties on the potential $V$ the Schr\"o\-ding\-er operator $H$ is self-adjoint on a certain dense domain of $L^2 (\RR^d)$, see e.g.\ \cite{ReedS-78b}, which we always assume. Hence, by the spectral theorem, the problem from Eq.~\eqref{eq:schroedinger} has the unique solution
\begin{equation} \label{eq:wave_function}
 \psi (\cdot , t) = \euler^{-\i t H} \psi_0 .
\end{equation}
Figure~\ref{abb_evolution} shows a thinkable example for the time evolution of the squared wave function, i.e.\ of the probability density of the position of the electron.
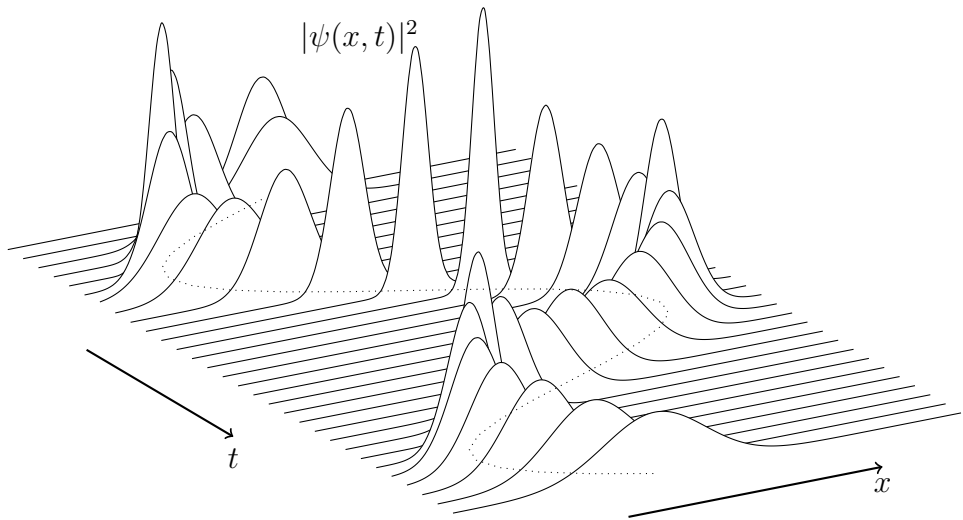
\begin{figure}\centering
\tdplotsetmaincoords{70}{150}
\begin{tikzpicture}[scale=1.3,tdplot_main_coords]
%

\draw[thick,->] (1.5,10.1,0) -- (-1.5,10.1,0) node[anchor=north]{$x$};
\draw[thick,->] (3.8,3,0) -- (3.8,6,0) node[anchor=north]{$t$};
\draw[thick] (0,2,2.4) node{$\lvert \psi (x,t) \rvert^2$};
\def\a{1};
\def\s{100};
\filldraw[fill=white, draw=black,smooth,domain=-3:3, samples=\s] plot (\x,0*0.314,{(\a/(0.3*2.5))*exp(-0.5*((\x-0)/0.3)*((\x-0)/0.3))});
\filldraw[fill=white, draw=black,smooth,domain=-3:3, samples=\s] plot (\x,1*0.314,{(\a/(0.4*2.5))*exp(-0.5*((\x-0)/0.4)*((\x-0)/0.4))});
\filldraw[fill=white, draw=black,smooth,domain=-3:3, samples=\s] plot (\x,2*3.14/10,{(\a/(0.3*2.5))*exp(-0.5*((\x-1.176)/0.3)*((\x-1.176)/0.3))});
\filldraw[fill=white, draw=black,smooth,domain=-3:3, samples=\s] plot (\x,3*3.14/10,{(\a/(0.2*2.5))*exp(-0.5*((\x-1.618)/0.2)*((\x-1.618)/0.2))});
\filldraw[fill=white, draw=black,smooth,domain=-3:3, samples=\s] plot (\x,4*3.14/10,{(\a/(0.15*2.5))*exp(-0.5*((\x-1.902)/0.15)*((\x-1.902)/0.15))});
\filldraw[fill=white, draw=black,smooth,domain=-3:3, samples=\s] plot (\x,5*3.14/10,{(\a/(0.25*2.5))*exp(-0.5*((\x-2.000)/0.25)*((\x-2.000)/0.25))});
\filldraw[fill=white, draw=black,smooth,domain=-3:3, samples=\s] plot (\x,6*3.14/10,{(\a/(0.4*2.5))*exp(-0.5*((\x-1.902)/0.4)*((\x-1.902)/0.4))});
\filldraw[fill=white, draw=black,smooth,domain=-3:3, samples=\s] plot (\x,7*3.14/10,{(\a/(0.4*2.5))*exp(-0.5*((\x-1.618)/0.4)*((\x-1.618)/0.4))});
\filldraw[fill=white, draw=black,smooth,domain=-3:3, samples=\s] plot (\x,8*3.14/10,{(\a/(0.3*2.5))*exp(-0.5*((\x-1.176)/0.3)*((\x-1.176)/0.3))});
\filldraw[fill=white, draw=black,smooth,domain=-3:3, samples=\s] plot (\x,9*3.14/10,{(\a/(0.2*2.5))*exp(-0.5*((\x-0.618)/0.2)*((\x-0.618)/0.2))});
\filldraw[fill=white, draw=black,smooth,domain=-3:3, samples=\s] plot (\x,10*3.14/10,{(\a/(0.15*2.5))*exp(-0.5*((\x-0.000)/0.15)*((\x-0.000)/0.15))});
\filldraw[fill=white, draw=black,smooth,domain=-3:3, samples=\s] plot (\x,11*3.14/10,{(\a/(0.13*2.5))*exp(-0.5*((\x+0.618)/0.13)*((\x+0.618)/0.13))});
\filldraw[fill=white, draw=black,smooth,domain=-3:3, samples=\s] plot (\x,12*3.14/10,{(\a/(0.2*2.5))*exp(-0.5*((\x+1.176)/0.2)*((\x+1.176)/0.2))});
\filldraw[fill=white, draw=black,smooth,domain=-3:3, samples=\s] plot (\x,13*3.14/10,{(\a/(0.25*2.5))*exp(-0.5*((\x+1.618)/0.25)*((\x+1.618)/0.25))});
\filldraw[fill=white, draw=black,smooth,domain=-3:3, samples=\s] plot (\x,14*3.14/10,{(\a/(0.3*2.5))*exp(-0.5*((\x+1.902)/0.3)*((\x+1.902)/0.3))});
\filldraw[fill=white, draw=black,smooth,domain=-3:3, samples=\s] plot (\x,15*3.14/10,{(\a/(0.2*2.5))*exp(-0.5*((\x+2.000)/0.2)*((\x+2)/0.2))});
\filldraw[fill=white, draw=black,smooth,domain=-3:3, samples=\s] plot (\x,16*3.14/10,{(\a/(0.3*2.5))*exp(-0.5*((\x+1.902)/0.3)*((\x+1.902)/0.3))});
\filldraw[fill=white, draw=black,smooth,domain=-3:3, samples=\s] plot (\x,17*3.14/10,{(\a/(0.35*2.5))*exp(-0.5*((\x+1.618)/0.35)*((\x+1.618)/0.35))});
\filldraw[fill=white, draw=black,smooth,domain=-3:3, samples=\s] plot (\x,18*3.14/10,{(\a/(0.4*2.5))*exp(-0.5*((\x+1.176)/0.4)*((\x+1.176)/0.4))});
\filldraw[fill=white, draw=black,smooth,domain=-3:3, samples=\s] plot (\x,19*3.14/10,{(\a/(0.45*2.5))*exp(-0.5*((\x+0.618)/0.45)*((\x+0.618)/0.45))});
\filldraw[fill=white, draw=black,smooth,domain=-3:3, samples=\s] plot (\x,20*3.14/10,{(\a/(0.4*2.5))*exp(-0.5*((\x+0)/0.4)*((\x+0)/0.4))});
\filldraw[fill=white, draw=black,smooth,domain=-3:3, samples=\s] plot (\x,21*3.14/10,{(\a/(0.4*2.5))*exp(-0.5*((\x-0.618)/0.4)*((\x-0.618)/0.4))});
\filldraw[fill=white, draw=black,smooth,domain=-3:3, samples=\s] plot (\x,22*3.14/10,{(\a/(0.3*2.5))*exp(-0.5*((\x-1.176)/0.3)*((\x-1.176)/0.3))});
\filldraw[fill=white, draw=black,smooth,domain=-3:3, samples=\s] plot (\x,23*3.14/10,{(\a/(0.2*2.5))*exp(-0.5*((\x-1.618)/0.2)*((\x-1.618)/0.2))});
\filldraw[fill=white, draw=black,smooth,domain=-3:3, samples=\s] plot (\x,24*3.14/10,{(\a/(0.25*2.5))*exp(-0.5*((\x-1.902)/0.25)*((\x-1.902)/0.25))});
\filldraw[fill=white, draw=black,smooth,domain=-3:3, samples=\s] plot (\x,25*3.14/10,{(\a/(0.3*2.5))*exp(-0.5*((\x-2.000)/0.3)*((\x-2.000)/0.3))});
\filldraw[fill=white, draw=black,smooth,domain=-3:3, samples=\s] plot (\x,26*3.14/10,{(\a/(0.35*2.5))*exp(-0.5*((\x-1.902)/0.35)*((\x-1.902)/0.35))});
\filldraw[fill=white, draw=black,smooth,domain=-3:3, samples=\s] plot (\x,27*3.14/10,{(\a/(0.4*2.5))*exp(-0.5*((\x-1.618)/0.4)*((\x-1.618)/0.4))});
\filldraw[fill=white, draw=black,smooth,domain=-3:3, samples=\s] plot (\x,28*3.14/10,{(\a/(0.5*2.5))*exp(-0.5*((\x-1.176)/0.5)*((\x-1.176)/0.5))});
\filldraw[fill=white, draw=black,smooth,domain=-3:3, samples=\s] plot (\x,29*3.14/10,{(\a/(0.6*2.5))*exp(-0.5*((\x-0.618)/0.6)*((\x-0.618)/0.6))});
\draw[dotted,
color=black, domain=0:9.106, samples=30, color=black] plot ({2*sin(\x r)},\x,0);
%
\end{tikzpicture}
\caption{Illustration of the time evolution of the wave function} \label{abb_evolution}
\end{figure}
Given a self-adjoint operator $H$, the solution of the Schr\"o\-ding\-er equation has a very complicated structure. However, it turns out that the (time)-asymptotic behaviour of the wave function has something to do with spectral properties of the operator $H$. This is formulated in the so-called RAGE-theorem, see Theorem~\ref{theorem:rage} below.
\par
Let $H$ be a self-adjoint operator on some Hilbert space $\mathcal{H}$. For $\psi \in \mathcal{H}$ we denote by $\mu_\psi$ the spectral measure and define the \emph{absolutely continuous}, \emph{singular continuous} and the \emph{pure point} subspace of $\mathcal{H}$ by
\begin{align*}
 \mathcal{H}_{\rm ac} &=  \mathcal{H}_{\rm ac} (H) = \{\psi \in \mathcal{H} \mid \mu_\psi \ \text{is absolutely continuous}\}, \\
\mathcal{H}_{\rm sc} &=\mathcal{H}_{\rm sc} (H) =  \{\psi \in \mathcal{H} \mid \mu_\psi \ \text{is singularly continuous}\}, \\
\mathcal{H}_{\rm pp} &= \mathcal{H}_{\rm pp} (H) = \{\psi \in \mathcal{H} \mid \mu_\psi \ \text{is pure point}\},
\end{align*} 
see e.g.\ \cite{Teschl-09} for more details. We also define the \emph{continuous} subspace by $\mathcal{H}_{\rm c} =  \mathcal{H}_{\rm ac} \oplus  \mathcal{H}_{\rm sc}$. The subspaces $\mathcal{H}_\bullet$, $\bullet \in \{\rm ac, sc, pp\}$ reduce $H$ \cite{Weidmann-00,Weidmann-03}, and the \emph{absolutely continuous}, \emph{singular continuous} and \emph{pure point spectrum}\index{spectrum}\index{spectrum!absolutely continuous}\index{spectrum!singular continuous}\index{spectrum!pure point} of $H$ are defined by
\[
 \sigma_{\rm ac} (H) = \sigma (H|_{\mathcal{H}_{\rm ac}}), \quad
 \sigma_{\rm sc} (H) = \sigma (H|_{\mathcal{H}_{\rm sc}}) \quad \text{and} \quad
 \sigma_{\rm pp} (H) = \sigma (H|_{\mathcal{H}_{\rm pp}}).
\]
The \emph{continuous spectrum}\index{spectrum!continuous} of $H$ is defined by $\sigma_{\rm c} (H) = \sigma_{\rm ac} (H) \cup \sigma_{\rm sc} (H)$. It follows from these definitions that $$\mathcal{H} = \mathcal{H}_{\rm ac} \oplus \mathcal{H}_{\rm sc} \oplus \mathcal{H}_{\rm pp}$$ and hence $\sigma (H) = \sigma_{\rm ac} (H) \cup \sigma_{\rm sc} (H) \cup \sigma_{\rm pp} (H)$. Note that $\sigma_{\rm pp} (H)$ is not the set of all eigenvalues of $H$, but its closure. The introduced decomposition of the Hilbert space and the spectrum has a physical interpretation. Roughly speaking, if the initial state $\psi_0$ is an element of $\mathcal{H}_{\rm pp}$, then the wave function given by Eq.~\eqref{eq:wave_function} (and so the particle) will stay in a compact region of space for all time. On the other hand, if $\psi_0 \in \mathcal{H}_{\rm c}$, then the wave function will leave any compact set in space in the average of time. This is formulated precisely in the so-called RAGE-theorem, named after D.~Ruelle \cite{Ruelle-69}, W.~Amrein and V.~Georgescu \cite{AmreinG-73}, and V.~Enss \cite{Enss-78}. For a proof of the RAGE-Theorem we refer the reader to the books \cite{Teschl-09,Weidmann-03}.
\begin{theorem}[RAGE] \label{theorem:rage}
Let $H$ be a self-adjoint operator in some Hilbert space $\mathcal{H}$. Suppose $K_n$, $n \in \NN$, is a sequence of bounded linear operators in $\mathcal{H}$ which converges strongly to the identity and assume that for each $n$, $K_n$ is relatively compact with respect to $H$. Then
\begin{align}
 \mathcal{H}_{\rm c} &= \Bigl\{ \psi \in \mathcal{H} \mid \lim_{n \to \infty} \lim_{T \to \infty} \frac{1}{T} \int_{0}^T \lVert K_n \euler^{-\i t H} \psi \rVert \drm t = 0 \Bigr \}, \label{eq:time_mean} \\[1ex]
\mathcal{H}_{\rm pp} &= \Bigl\{ \psi \in \mathcal{H} \mid \lim_{n \to \infty} \sup_{t \geq 0}  \lVert (1-K_n) \euler^{-\i t H} \psi \rVert = 0 \Bigr \}  \nonumber.
\end{align}
Moreover, if $\psi \in \mathcal{H}_{\rm ac}$, then for any $n \in \NN$ we have
\[
 \lim_{t \to \infty} \lVert K_n \euler^{-\i t H} \psi \rVert = 0.
\]
\end{theorem}
Recall, an operator $K$ on $\mathcal H$ is called relatively compact with respect to an operator $H$ on $\mathcal H$ if $K (H - z)^{-1}$ is compact for one $z \in \rho (H)$.
\par
Let us again consider our quantum mechanical particle (electron) moving in $\RR^d$ under the influence of the potential $V$. The time evolution of the particle is governed by the Schr\"o\-ding\-er equation with the operator $H = -\Delta + V$ acting in $\mathcal{H} = L^2 (\RR^d)$. We assume that $V$ is relatively bounded with respect to $\Delta$ with relative bound smaller than one. As a consequence we have that $\chi_K (H - z)^{-1}$ is compact for any compact $K \subset \RR^d$, see e.g. \cite{Weidmann-03}. The RAGE-Theorem tells us, if $\psi_0 \in \mathcal{H}_{\rm pp}$, then for any $\epsilon > 0$ there exists a compact set $K_\epsilon \subset \RR^d$, such that
\begin{equation} \label{eq:localized}
 \int_{K_\epsilon} \bigl\lvert (\euler^{-\i t H} \psi_0)(x) \bigr\rvert^2 \drm x \geq 1-\epsilon \quad \forall t \geq 0.
\end{equation}
This means that the probability of finding the particle in $K_\epsilon$ stays larger or equal to $1-\epsilon$ for all $t$; the particle is \emph{localized}! In the special case $\sigma_{\rm c} (H) = \emptyset$ we have $\mathcal{H} = \mathcal{H}_{\rm pp}$, and thus the assertion \eqref{eq:localized} holds true for all initial states $\psi_0 \in \mathcal{H}$. In contrast to that, if $\psi_0 \in \mathcal{H}_{\rm ac}$, we have for all compact sets $K \subset \RR^d$ that
\begin{equation} \label{eq:delocalized}
\int_{K} \bigl\lvert (\euler^{-\i t H} \psi_0)(x) \bigr\rvert^2 \drm x  \to 0 \quad \text{as} \ t \to \infty.
\end{equation}
The probability of finding the particle in $K$ tends to zero as time tends to infinity; the particle is \emph{delocalized}! For $\psi_0 \in \mathcal{H}_{\rm c}$ the particle leaves any compact set in time mean as Eq.~\eqref{eq:time_mean} shows. For this reason, elements of the set $\mathcal{H}_{\rm pp}$ are called \emph{bound states} and elements of the set $\mathcal{H}_{\rm c}$ are called \emph{scattering states in time mean}, and the set of initial states $\psi_0 \in \mathcal{H}$ for which Eq.~\eqref{eq:delocalized} holds are called \emph{scattering states}. Hence, $\mathcal{H}_{\rm ac}$ is a subset of the scattering states. Let us end this section with the classical example, the hydrogen atom.
\begin{example}[Hydrogen atom]
The hydrogen atom consists of one proton and one electron in $\RR^3$. For simplicity we assume that the proton sits at the origin and the electron moves under the influence of the Coulomb potential. The corresponding Schr\"o\-ding\-er operator is $H = -\Delta + V$ on $L^2 (\RR^3)$, where $\Delta$ is the Laplace operator and $V$ is the multiplication operator by the function $V(x) = -C/\lvert x \rvert$ with some constant $C$. The spectrum of $H$ consists of eigenvalues (belonging to the pure point part) $E_n = -c/n^2$, $n \in \NN$, with some constant $c$, and absolutely continuous spectrum in the interval $[0,\infty)$. The negative eigenvalues correspond to bound states. Electrons in these states will stay in a finite region around the proton for all time, the so-called orbitals. Electrons in states belonging to the absolutely continuous subspace correspond to scattering states, they are called free electrons and will leave any finite region in space if time tends to infinity.
\end{example}
\section{Random operators} \label{sec:rand_operators}
The amount of literature on random operators is huge, see \cite{CarmonaL-90,PasturF-92} and the references therein. In this section we introduce only some idea of random operators, and in doing so we sometimes follow the line of reasoning of the introductions of \cite{Klein-08,Hislop-08}.
\par
A strong form of idealization in solid state physics is to deal with ideal crystals in the so called one-electron approximation. In order to model the electronic properties of such a crystal one considers one electron moving in a periodic lattice of atoms. The potential of the corresponding Schr\"o\-ding\-er operator $H_{\rm per}$ would be a periodic function. Under some mild regularity assumption on the periodic potential, it is known that $H_{\rm per}$ has only absolutely continuous spectrum, see e.g.\ \cite{ReedS-78d}. By the RAGE-theorem, the underlying Hilbert space consists only of scattering states, and the crystal may have good transport properties.
\par
However, real materials will have distortions (e.g.\ dislocations, vacancies, presence of impurity atoms) which may be assumed to be randomly distributed through the material. Their modeling leads to the study of a family of self-adjoint operators $H_\omega$, $\omega \in \Omega$, where each configuration $\omega$ corresponds to one individual realization of the (random) medium. One assumes additionally that these configurations are distributed according to a measure $\PP$ on the space $\Omega$ of all possible configurations. It turns out that random operators modeling disordered systems obey a different behavior, namely the phenomenon of localization, than periodic operators modeling periodic systems. This will be discussed in the next section. Here we consider some examples of random operators.
\par
The simplest model that describes the time-evolution of a single electron in a random environment is the Anderson model \cite{Anderson-58}, named after P.~W.~Anderson who won (together with S.~N.~F.~Mott and J.~van~Vleck) the Nobel Prize for his investigations of the electronic structure of magnetic and disordered systems. First we define the probability space $\Omega = \times_{k \in \ZZ^d} \RR$ equipped with the $\sigma$-algebra generated by the cylinder sets and the product measure $\PP = \prod_{k \in \ZZ^d} \nu$, where $\nu$ is some probability measure on $\RR$ with bounded support. Hence, the projections $\Omega \ni \omega = (\omega_k)_{k \in \ZZ^d} \mapsto \omega_j$, $j \in \ZZ^d$, give rise to a collection of independent identically distributed bounded real random variables. The Anderson model is given by the family of self-adjoint operators $H_\omega^{\rm A}$, $\omega \in \Omega$, on $\ell^2 (\ZZ^d)$ defined by
\[
 H_\omega^{\rm A} = -\Delta + \lambda V_\omega^{\rm A}
\]
with $\lambda \geq 0$. Here $\Delta : \ell^2 (\ZZ^d) \to \ell^2 (\ZZ^d)$ denotes the discrete Laplacian and $V_\omega:\ell^2 (\ZZ^d) \to \ell^2 (\ZZ^d)$ is a multiplication operator. They are defined by
\[
 (\Delta \psi)(x) = \sum_{\lvert y-x \rvert_{1} = 1} \psi (y) \quad \text{and} \quad (V_\omega^{\rm A} \psi)(x) = \omega_x \psi (x) .
\]
The parameter $\lambda$ measures the strength of the interaction and hence is a measure of the disorder present in the model. Notice that the Anderson model is defined on the Hilbert space $\ell^2 (\ZZ^d)$. However, we introduced the basic concept of quantum mechanics in $L^2 (\RR^d)$ in Section \ref{sec:foundations}. But all the considerations of Section \ref{sec:foundations} hold true also for the discrete setting, with the wave function replaced by a function $\psi : \ZZ^d \times \RR \to \CC$ and consequently some integrals replaced by sums. The Anderson model has been studied, e.g., in \cite{FroehlichS-83,FroehlichMSS-85,DreifusK-89,AizenmanM-93,Aizenman-94,Graf-94,Hundertmark-00,AizenmanFSH-01}.
\par
A second example for a random operator modeling one electron in a random environment is the alloy-type model. This is the family of self-adjoint operators $H_\omega^{\rm B}$, $\omega \in \Omega$, on $L^2 (\RR^d)$ defined by $H_\omega^{\rm B} = -\Delta + V_{0} +  V_\omega^{\rm B}$, where $\Delta$ denotes the Laplace operator, $V_{0}$ is some $\ZZ^d$-periodic potential and $V_\omega^{\rm B}$ is the multiplication operator by the function
\[
 V_\omega^{\rm B} (x) = \sum_{k \in \ZZ^d} \omega_k U( x-k ) .
\]
The function $U : \RR^d \to \RR$ is called single-site potential. Physically one can think of a lattice of atoms sitting at the lattice sites of $\ZZ^d$, each atom at $k \in \ZZ^d$ producing the electric potential $U(\cdot - k)$ in space. One then assumes that the electron at $x \in \RR^d$ couples differently strong to the single-site potentials of different atoms. The alloy-type model is defined precisely in Section \ref{sec:cont_model}. The case where the single-site potential is non-negative has been studied in a number of articles, e.g.\ \cite{CombesH-94,KirschSS-98b,GerminetK-01, DamanikS-01,AizenmanENSS-06}.
\par
As a third example we refer to a discrete analogue of the alloy-type model, the discrete alloy-type model. This model will be in the focus of this thesis and we will introduce it in Section \ref{sec:model}. Let us emphasize that in the discrete alloy-type model, as well as in the alloy-type model, one distinguishes the case where the single-site potential is non-negative (monotone case) and the case where it is sign-indefinite (non-monotone case). In the monotone case, the quadratic form corresponding to the operator depends in a monotonic way on the random parameters. The existing proofs of localization for such random operators strongly rely on this fact. The non-monotone case requires new methods and has for example been studied in \cite{Klopp-95a,HislopK-02,Klopp-02,Veselic-02a,KostrykinV-06,KloppN-09,Veselic-10a,Veselic-10b,ElgartTV-10,ElgartTV-11,PeyerimhoffTV-11,Krueger-11}. The focus of this thesis is to develop the fractional moment method and a proof for the Wegner estimate for the discrete alloy-type model with sign-changing single-site potential.
\par
A fourth example for a random Schr\"o\-ding\-er operator is the so-called random displacement model. This is the Schr\"o\-ding\-er operator
\[
 H_\omega^{\rm D} = -\Delta + V_\omega^{\rm D}
\]
on $L^2 (\RR^d)$ where the random potential is of the form 
\[
 V_\omega^{\rm D} (x) = \sum_{k \in \ZZ^d} U(x-k-\omega_k). 
\]
The random variables $\omega_k$, $k \in \ZZ^d$ are assumed to be independent and identically distributed random variables with values in $\RR^d$ and $U : \RR^d \to \RR$ is some single-site potential. This potential models a random perturbation from the periodic potential $\sum_{k\in \ZZ^d} U(\cdot -k)$ where the atoms are displaced randomly. Therefore one also speaks of structural disorder. Similarly to the alloy-type model (discrete and continuous) with sign-changing single-site potential, the random displacement model does not have any obvious monotonicity properties with respect to the random parameters. But it is just such a monotonicity property which is frequently used in the proofs of localization! The random displacement model has been studied, e.g., in \cite{BakerLS-08,KloppLNS-11,KloppLNS-11b}.
\par
There are more models modeling properties of disordered media. We refer to \cite{Klein-08} where the most prominent models are introduced, and the references therein for a deeper insight into these models. 
\par
As the random operators introduced above model the quantum mechanical behavior of a single electron in a disordered solid, we are not interested in properties of the electron for one single configuration of the randomness. Rather we are interested in properties which hold for almost all configurations $\omega \in \Omega$. One fundamental example for such a property is the spectrum. To be precise, let us note that all models introduced above share the property that they are ergodic random operators in the following sense.
\begin{definition}
An ergodic random operator is a $\ZZ^d$-ergodic measurable map $\omega \mapsto A_\omega$ from a probability space $(\Omega , \mathcal{F} , \PP)$ to the set of all self-adjoint operators on either $L^2 (\RR^d)$ or $\ell^2 (\ZZ^d)$.
\end{definition}
To explain the notion of measurability of such operator valued functions let $E_\omega (\lambda)$ be the corresponding resolution of identity. One calls a map $\omega \mapsto A_\omega$ from a probability space $(\Omega , \mathcal{F} , \PP)$ to the set of all self-adjoint operators on either $L^2 (\RR^d)$ or $\ell^2 (\ZZ^d)$ measurable if the function $\omega \mapsto E_\omega (\lambda)$ is weakly measurable for each $\lambda \in \RR$. Such a map is called $\ZZ^d$-ergodic if there is a family $\{T_i\}_{i \in \ZZ^d}$ of measure-preserving transformations and unitary operators $\{U_i\}_{i \in \ZZ^d}$ such that $$A_{T_i \omega} = U_i A_\omega U_i^*$$ for all $i \in \ZZ^d$. For details we refer the reader to \cite{KirschM-82,CarmonaL-90}. 
A consequence of this ergodicity property is that the spectrum is a deterministic set. 
\begin{theorem}
 Let $(A_\omega)_{\omega \in \Omega}$ be an ergodic random operator. Then there are sets $\Sigma, \Sigma_{\rm pp} , \Sigma_{\rm c}, \Sigma_{\rm ac}, \Sigma_{\rm sc} \subset \RR$, such that for almost all $\omega \in \Omega$ we have
\[
 \sigma (A_\omega) = \Sigma , \quad \text{and} \quad \sigma_\bullet (A_\omega) = \Sigma_\bullet , \quad \bullet \in \{\rm pp , \rm c , \rm ac , \rm sc \} .
\]
The set $\Sigma$ is called almost sure or deterministic spectrum of $(A_\omega)_{\omega\in \Omega}$.
\end{theorem}
The proof of this theorem goes back to \cite{Pastur-80,KunzS-80} where operators on $\ell^2 (\ZZ^d)$ were considered, and to \cite{KirschM-82} where the extension to each part of the spectrum and general random self-adjoint operators was achieved.
\section{Phenomenon of localization} \label{sec:phenomenon_loc}
Once we know that ergodic random operators have almost surely a deterministic spectrum (and also non-random components), we can ask for the spectral types. As already discussed, periodic Schr\"o\-ding\-er operators modeling ideal crystals have purely absolutely continuous spectrum. By the RAGE theorem the wave function leaves every compact region in space if time tends to infinity. The wave function is delocalized and the crystal may have good transport properties. 
\par
Random operators behave differently. There are energy intervals $I$ such that for almost all configurations of the randomness, the wave function corresponding to the energy interval $I$ stays trapped in a finite region of space for all time. It is localized and one can think of materials not having good transport properties. This manifests in the sense that the spectrum in $I$ is almost surely only of pure point type and the phenomenon behind this is known as Anderson localization, spectral localization or exponential localization. 
\begin{definition} \label{def_spectral_loc}
Let $I \subset \RR$. A \index{self-adjoint}self-adjoint operator $H$ on either $L^2 (\RR^d)$ or $\ell^2 (\ZZ^d)$ is said to exhibit \emph{spectral localization in $I$}, if the spectrum of $H$ in $I$ is only of pure point type, i.e.\ $\sigma_{\rm c} (H) \cap I = \emptyset$. If, additionally, the eigenfunctions of $H$ corresponding to the eigenvalues in $I$ decay exponentially we say that $H$ exhibits \emph{exponential localization in $I$}. If $I=\RR$, we say that
$H$ exhibits \emph{spectral localization} or \emph{exponential localization}, depending whether the eigenfunctions decay exponentially or not.
\end{definition}
Beside the spectral interpretation, there are also interpretations of localization from the dynamical point of view, called dynamical localization. Dynamical localization for the Anderson model on $\ZZ^d$ was first shown in \cite{Aizenman-94}.
\begin{definition}
 Let $I \subset \RR$ and $\mathcal H$ be either $L^2 (\RR^d)$ or $\ell^2 (\ZZ^d)$. A self-adjoint operator $H$ on $\mathcal H$ is said to exhibit \emph{dynamical localization in $I$}, if for every $\psi_0 \in \mathcal{H}$ with compact support and all $p \geq 0$ we have
\[
 \sup_{t \in \RR} \bigl\lVert \lvert x \rvert^p \euler^{- \i H t} \chi_I (H) \psi_0 \bigr\rVert < \infty ,
\]
where $\chi_I:\RR \to \{0,1\}$ denotes the characteristic function of the set $I$. If $I=\RR$, we say that $H$ exhibits \emph{dynamical localization}.
\end{definition}
%
Let us note that dynamical localization implies spectral localization by the RAGE-theorem, see e.g.\ \cite{Stolz-10}, but not vice versa as examples in \cite{RioJLS1996} show. There are various notions of dynamical localization (for example strong dynamical localization), see e.g.\ \cite{Klein-08}.
\par
Let us discuss the picture of localization using the example of the Anderson model $H_\omega^{\rm A} : \ell^2 (\ZZ^d) \to \ell^2 (\ZZ^d)$, with $\nu$ being the uniform distribution on $[-1,1]$, introduced in the previous section. If the disorder parameter $\lambda$ is zero, then $\sigma (H_\omega^{\rm A}) = \sigma (-\Delta) = [-2d,2d]$ and of absolutely continuous type, since the operator $\Delta$ is seen to be unitary equivalent to multiplication by the function $2 \sum_{j=1}^d \cos (2\pi k_j)$, $k_j \in [0,1]$, by the Fourier transform. More generally, it is known that the almost sure spectrum of $H_\omega^{\rm A}$ is given by 
\[
\Sigma = [-2d , 2d] + \lambda \supp \nu ,
\]
see \cite{KunzS-80}. What may be expected about the spectral type? The physical picture is the following. If the disorder $\lambda$ increases, then there are intervals near the band edges, where, for almost all $\omega \in \Omega$, $H_\omega^{\rm A}$ exhibits spectral and dynamical localization, called \emph{localization near the band edges}, while in the center of the band there is still absolutely continuous spectrum. If the disorder $\lambda$ is sufficiently large, then, for almost all $\omega \in \Omega$, $H_\omega^{\rm A}$ exhibits spectral and dynamical localization. To be more precise, we follow the line of reasoning in \cite{Hislop-08} and introduce the following conjectures, see also Fig.~\ref{fig:localization} for an illustration of the disorder/energy regimes of localization and delocalization. 
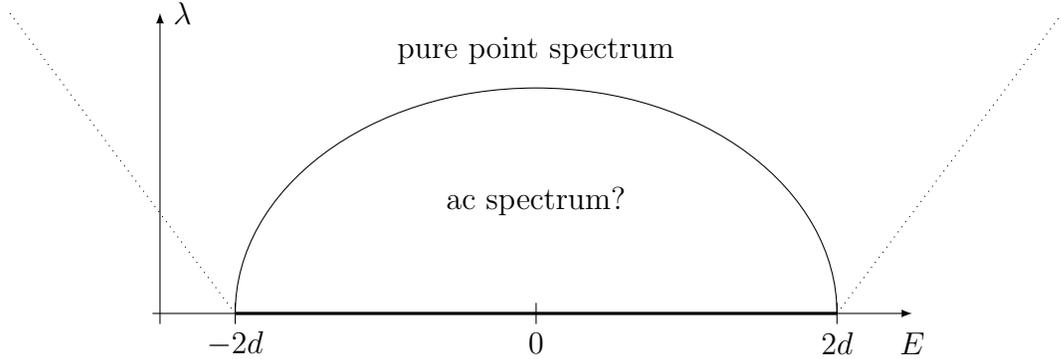
\begin{figure}\centering
 \begin{tikzpicture}
  \draw[-latex] (-5.1,0)--(5,0);
  \draw (5,-0.4) node {$E$};
  \draw[very thick] (-4,0)--(4,0);
  \draw[dotted] (-4,0)--(-7,4);
  \draw[dotted] (4,0)--(7,4);
   \draw (4,0) arc (0:180:4cm and 3cm);
  \draw[-latex] (-5,-4pt) -- (-5,4);
  \draw (-4.7,4) node {$\lambda$};
  \draw (4,4pt)--(4,-4pt);
  \draw (-4,4pt)--(-4,-4pt);
\draw (0,4pt)--(0,-4pt);
  \draw (-4,-0.4) node {$-2d$};
\draw (4,-0.4) node {$2d$};
\draw (0,-0.4) node {$0$};
\draw (0,3.5) node {pure point spectrum};
\draw (0,1.5) node {ac spectrum?};
 \end{tikzpicture}
\caption[Illustration of the energy/disorder regimes of localization and delocalization]{Illustration of the energy/disorder regimes of localization and delocalization.}
\label{fig:localization}
\end{figure}
\begin{enumerate}[(1)]
\item Fix $\lambda > 0$. Then there is $E_0 \in (2d,2d+\lambda)$ such that, for almost all $\omega \in \Omega$, the operator $H_\omega^{\rm A}$ exhibits spectral localization in $\{(-\infty , -E_0] \cup [E_0 , \infty)\}$, i.e.\ $\sigma_{\rm c} (H_\omega^{\rm A}) \cap \{(-\infty , -E_0] \cup [E_0 , \infty)\} = \emptyset$.
\item There is $\lambda_0 > 0$ such that for all $\lambda > \lambda_0$ the operator $H_\omega^{\rm A}$ exhibits spectral localization for almost all $\omega\in\Omega$.
\item Let $d = 1$ and $\lambda > 0$. Then, for almost all $\omega \in \Omega$, $H_\omega^{\rm A}$ exhibits spectral localization.

\item Let $d \geq 3$ and $\lambda>0$ sufficiently small. Then there is $E_{\rm m} = E_{\rm m} (\lambda) \leq E_0$ such that, for almost all $\omega \in \Omega$, the spectrum of $H_\omega^{\rm A}$ in $[-E_{\rm m} , E_{\rm m}]$ is absolutely continuous. 
\end{enumerate}
Statements (1), (2) and (3) are known. Localization at large disorder was first proven in \cite{GoldsteinMP-77} for a continuum random Schr\"o\-ding\-er operator in one space dimension. The groundbreaking articles towards statements (1), (2) and (3) for the Anderson model in arbitrary dimension are \cite{FroehlichS-83,FroehlichMSS-85,SimonW-86,DreifusK-89,AizenmanM-93}.
\par
Statement (4) is not yet proven. However, there are associated results for related models. On the one hand the existence of absolutely continuous spectrum was proven in \cite{Klein-98} for an Anderson model not on $\ZZ^d$, but on the Bethe lattice. Another proof of this result was given in \cite{FroeseHS-07}. Related results for the stability of absolutely continuous spectrum for Anderson models on trees were also given in \cite{AizenmanSW-06,KellerLW-11}. On the other hand the existence of absolutely continuous spectrum at small disorder has been studied for an Anderson model on $\ZZ^d$ by imposing certain decay conditions on the random potential, see e.g.\ \cite{KirschKO-00,Krishna-11}.
\section{Fractional moment method, multiscale analysis and outline of the thesis} \label{sec:methods_msa_fmm}
In Section~\ref{sec:phenomenon_loc} we discussed the phenomenon of localization, specifically for the Anderson model on $\ZZ^d$. There are exactly two methods to prove statements (1) and (2) from Section~\ref{sec:phenomenon_loc} for the multidimensional Anderson model; the \emph{multiscale analysis} and the \emph{fractional moment method}.
Let us note that there are also methods only suitable in the one-dimensional situation which are different to multiscale analysis and fractional moment method, see e.g.\ \cite{BuschmanS-00,Stolz-02}.
\par
The multiscale analysis was invented by \cite{FroehlichS-83} and further developed, e.g., in \cite{FroehlichMSS-85,DreifusK-89}. Originally the multiscale analysis yields exponential localization, later it was shown that multiscale analysis also implies dynamical localization, see e.g. \cite{GerminetB-98,DamanikS-01,GerminetK-01}. The other method, the fractional moment method, was introduced for the Anderson model on the lattice in \cite{AizenmanM-93,Aizenman-94,AizenmanFSH-01}, and extended to the continuum in \cite{AizenmanENSS-06,BoutetNSS-06}. It is believed that the fractional moment method ``requires that the conditional expectation of certain random variables have bounded densities'' \cite{Klein-08}. However, in this thesis we develop the fractional moment method for the discrete alloy-type model which does not satisfy this property in general as shown in Section~\ref{sec:regularity}.
\par
Both multiscale analysis and fractional moment method were first developed for the Anderson model $H_\omega^{\rm A}$ on $\ZZ^d$ \cite{FroehlichS-83,AizenmanM-93}. The methods were then subsequently adapted and generalized to other random operators on $\ell^2 (\ZZ^d)$ or $L^2 (\RR^d)$. Let us name a few of them.
\begin{itemize}
\item The multiscale analysis as well as the fractional moment method are developed for the Anderson model on $\ZZ^d$ where the potential values at different lattice sites are not assumed to be stochastically independent \cite{DreifusK-91,AizenmanM-93,AizenmanFSH-01}. As the discrete alloy-type model is an Anderson model on the lattice with correlated potential values, one can ask the question whether the results of \cite{DreifusK-91,AizenmanM-93,AizenmanFSH-01} apply to the discrete alloy-type model or not. However, they impose certain stringent conditions on the joint distribution of the potential values which are not satisfied for the discrete alloy-type model. This will be discussed in Section~\ref{sec:regularity} in detail.
\item The multiscale analysis was adapted to the alloy-type model $H_\omega^{\rm B}$ on $L^2 (\RR^d)$ with non-negative (monotone case) and compactly supported single-site potentials. This goes back to \cite{HoldenM-84}, see \cite{Stollmann-01,Kirsch-08} and the references therein for further advances of the multiscale analysis in this research direction. The generalization to non-compactly supported but still non-negative single-site potentials was done in \cite{Klopp-95a,KirschSS-98b}. 
\par
The fractional moment method was likewise developed for the alloy-type model with non-negative and compactly supported single-site potentials, see \cite{AizenmanENSS-06,BoutetNSS-06} for the multidimensional case and \cite{HamzaSS-10} for the one-dimensional model where localization is shown at all energies independent of the disorder strength as stated in statement (3) in Section~\ref{sec:phenomenon_loc}.
\item In 2001 Germinet and Klein \cite{GerminetK-01} developed the so-called bootstrap multiscale analysis, which yields stronger results than the original multiscale analysis and is applicable for a large class of random operators including the alloy-type model with non-negative and compactly supported single-site potential as an example. 
%
\item Both, the multiscale analysis and the fractional moment method were used to prove localization for multi-particle random Schr\"o\-ding\-er operators with sign-definite single-site potentials and sign-definite electron-electron interaction. This was done for operators on $\ell^2 (\ZZ^d)$ in \cite{ChulaevskyS-08,Kirsch-08b,ChulaevskyS-09a,ChulaevskyS-09b} via the multiscale analysis and in \cite{AizenmanW-09,AizenmanW-10} based on the fractional moment method. 
\par
For multi-particle Schr\"o\-ding\-er operators in $L^2 (\RR^d)$ with an external alloy-type potential, localization has been studied according to the multiscale analysis in \cite{BoutetCSY-10,BoutetCY-11}. The fractional moment method has not been applied to multi-particle Schr\"o\-ding\-er operators in $L^2 (\RR^d)$. Currently, Michael Fauser and Simone Warzel pursue this research direction.
\end{itemize}
Localization proofs via multiscale analysis or the fractional moment method strongly rely on the property that the random operator depends, in the sense of quadratic forms, monotonically on the random parameters. There is no physically compelling reason for a random Schr\"o\-ding\-er operator to have such a monotonicity property, and one can ask the natural question whether localization can be established if one relinquishes the monotonicity property. 
Examples for a random operator which lack monotonicity are the alloy-type model (discrete or continuous) with sign-changing single-site potential and the random displacement model. 
\par
Since the existing methods rely strongly on a monotonicity property, the phenomenon of localization is much less understood for non-monotone models than for monotone ones. In the last two decades, a lot of research has been done to close the gap between the existing results for non-monotone and monotone random operators. 
Next we list some results for non-monotone random operators which have been achieved via multiscale analysis.
\begin{itemize}
\item
The paper \cite{Klopp-95a} establishes exponential localization in $(-\infty,E_0)$, $E_0>0$, for the alloy-type model on $L^2 (\RR^d)$ in the case where $\Sigma=\RR$ and where the single-site potential decays exponentially. The key new feature is a proof of a Wegner estimate which works without sign assumptions on the single-site potential. Related results concerning localization at band edges (in the weak disorder regime) have been obtained for the alloy-type model with sign-changing single-site potential in \cite{Klopp-02,HislopK-02,KloppN-09} and for the so-called generalized alloy-type model in \cite{KloppN-10}.
\item
The papers \cite{Veselic-02a,KostrykinV-06,Veselic-10b,PeyerimhoffTV-11} establish a Wegner estimate (which can be used for a localization proof via multiscale analysis) for a class of alloy-type models with a sign-indefinite single-site potential of a so-called generalized step function form. The obtained Wegner estimates are valid on the whole energy axis. We present our results of \cite{PeyerimhoffTV-11} in Section \ref{sec:wegner_cont} of this thesis.  
\item
For the random displacement model, also a model with no obvious monotonicity property with respect to the random variables, spectral properties have been studied according to the multiscale analysis, e.g., in \cite{BakerLS-08,KloppLNS-11,KloppLNS-11b}. 
\item
The discrete alloy-type model with sign-changing single-site potential has been studied according to the multiscale analysis in \cite{Veselic-10a,Krueger-11,PeyerimhoffTV-11,CaoE-11}. Veseli\'c, Peyerimhoff and Tautenhahn \cite{Veselic-10a,PeyerimhoffTV-11} establish a Wegner estimate at all energies, see Chapter \ref{chap:wegner}. Kr\"uger \cite{Krueger-11} develops the multiscale analysis not relying on a (classical) Wegner estimate and obtains dynamical and spectral localization in the strong disorder regime for a class of models including the discrete alloy-type model as an example. Kr\"uger uses ideas of \cite{Bourgain-09} where a Wegner estimate is established for the matrix valued Anderson model, which also lacks monotonicity. Recently, Cao and Elgart show in \cite{CaoE-11} localization for the discrete alloy-type model in the Lifshitz tails regime, i.e.\ at small disorder and for low energies.
\end{itemize}
The fractional moment method is far less developed for non-monotone random Schr\"o\-ding\-er operators. Our paper \cite{ElgartTV-10} develops the fractional moment method and gives a proof of localization for the one-dimensional discrete alloy-type model with sign-changing and compactly supported single-site potential, cf.\ Section \ref{sec:d=1}. This result was generalized to the multidimensional case in \cite{ElgartTV-11}, where the single-site potential is assumed to be compactly supported and to have fixed sign at the boundary of its support, cf.\ Chapter \ref{chap:fmm}. We also refer the reader to our survey paper \cite{ElgartKTV-11} where recent results for the discrete alloy-type model are presented. Let us stress that it is generally acknowledged that even for models where the multiscale analysis is well established, the implementation of the fractional moment method gives new insights and slightly stronger
results. For instance, the paper \cite{AizenmanENSS-06} concerns models for which the multiscale analysis was developed much earlier.
\par
The aim of this thesis is to expand the methods for proving localization for non-monotone random operators. In particular, we 
\begin{itemize}
\item develop the fractional moment method for the discrete alloy-type model with sign-changing single-site potential according to our papers \cite{ElgartTV-10,ElgartTV-11}, see Chapter~\ref{chap:fmm} and Appendix~\ref{chap:non_local}.
\item establish a new variant for concluding exponential localization from fractional moment bounds for the discrete alloy-type model following our results in \cite{ElgartTV-10,ElgartTV-11}, see Section~\ref{sec:loc}. An adaptation of this result to the alloy-type model on $L^2 (\RR^d)$ with sign-changing single-site potential is presented in Appendix~\ref{sec:loc_cont}.
\item prove a Wegner estimate at all energies for the discrete and continuous alloy-type model with sign-changing single-site potentials of a generalized step-function type which has already been published in \cite{PeyerimhoffTV-11}, see Chapter~\ref{chap:wegner} and Appendix~\ref{sec:wegner_cont}. 
\end{itemize}

\chapter{Discrete alloy-type model, main results and regularity properties}
%
In Section~\ref{sec:model} we will introduce the discrete alloy-type model, which is a generalization of the classical Anderson model in the sense that the potential values at different lattice sites are not independent. Moreover, since the single-site potential may change its sign, certain properties of the discrete alloy-type model depend in a non-monotonic way on the random parameters. Both, the independence and the monotonicity properties distinguish the classical Anderson model from the discrete alloy-type model. The existing methods for proving localization strongly rely on the independence and/or the monotonicity property. For this reason, new methods are needed to overcome the problems arising from non-monotonicity and non-independence. The discrete alloy-type model has been studied, for instance, in \cite{Veselic-10a,ElgartTV-10,ElgartTV-11,PeyerimhoffTV-11,Krueger-11,CaoE-11}.
\par
In Section~\ref{sec:results} we state our main results on localization via the fractional moment method and a Wegner estimate for the discrete alloy-type model. These results have already been published in \cite{ElgartTV-10,ElgartTV-11,PeyerimhoffTV-11} and will be proven in Chapter~\ref{chap:fmm} and Chapter~\ref{chap:wegner}. 
\par
Anderson models where the potential values at different lattice sites are not independent have already been studied, see e.g.\ \cite{AizenmanM-93,DreifusK-91,Hundertmark-00,AizenmanFSH-01}, by imposing certain regularity assumptions on a conditional distribution of the potential values. In Section \ref{sec:regularity} we show that these regularity assumptions are generally not satisfied for the discrete alloy-type model.
%
%
%
%
\section{The discrete alloy-type model} \label{sec:model}
Let $d \geq 1$. For $x \in \ZZ^d$ we recall the standard norms
\[
 \lvert x \rvert_1 \defeq \sum_{i=1}^d \lvert x_i \rvert \quad \text{and} \quad
 \lvert x \rvert_\infty \defeq \max \{ \lvert x_1 \rvert , \ldots , \lvert x_d \rvert \} .
\]
For any set $\Gamma \subset \ZZ^d$ we introduce the Hilbert space\index{Hilbert space}
\[
\ell^2 (\Gamma) = \Bigl\{\psi : \Gamma \to \CC : \sum_{k \in \Gamma} \lvert \psi (k) \rvert^2 < \infty\Bigr\}
\]
which is equipped with the inner product\index{inner product} $\sprod{\phi}{\psi} = \sum_{k \in \Gamma} \overline{\phi(k)} \psi (k)$. On the Hilbert space $\ell^2 (\ZZ^d)$ we consider the family of discrete Schr\"o{}dinger operators\index{Schr\"o{}dinger operator}
\begin{equation} \label{eq:hamiltonian}
 H_\omega \defeq -\Delta + \lambda V_\omega , \quad \omega \in \Omega , \quad \lambda > 0 .
\end{equation}
Here, $\omega$ is an element of the probability space\index{probability space} specified
below, $\Delta: \ell^2 \left(\ZZ^d\right) \to \ell^2
\left(\ZZ^d\right)$ denotes the discrete Laplace operator\index{discrete Laplace operator} and
$V_\omega : \ell^2 \left(\ZZ^d\right) \to \ell^2 \left(\ZZ^d\right)$
is a random multiplication operator. They are defined by
\begin{equation*}
\left(\Delta \psi \right) (x) \defeq \sum_{\abs{e}_1 = 1} \psi (x+e) \quad \mbox{and} \quad
\left( V_\omega \psi \right) (x) \defeq  V_\omega (x) \psi (x)
\end{equation*}
and represent the kinetic energy\index{kinetic energy} and the random potential energy\index{potential energy}, respectively. Note that we have suppressed the diagonal term $-2d \psi (x)$ in the definition of the discrete Laplacian, since this corresponds just to a spectral shift and we are mostly interested in spectral types. The parameter $\lambda$ models the strength of the disorder\index{disorder}. We assume that the probability space has a product structure $\Omega := \bigtimes_{k \in \ZZ^d} \RR$ and is equipped with the $\sigma$-algebra generated by the cylinder-sets and the probability measure 
\[
\PP (\drm \omega) := \prod_{k \in \ZZ^d} \nu(\drm \omega_k)
\]
where $\nu$ is a probability measure on $\RR$ with compact support. We define $R := \max\{\lvert \inf \supp \nu \rvert , \lvert \sup \supp \nu \rvert\}$.
By definition, the projections $\Omega \ni \omega \mapsto \omega_k$ give rise to a sequence
$\omega_k$, $k \in \ZZ^d$, of independent identically distributed\index{independent identically distributed} (i.i.d.) random variables, each distributed according to the probability measure $\nu$. The symbol
$\mathbb{E}$ denotes the expectation\index{expectation} with respect to the probability
measure, i.e.\ $\mathbb{E} (\cdot) := \int_\Omega (\cdot)
\PP (\drm \omega)$. For a set $\Gamma \subset \ZZ^d$, $\mathbb{E}_\Gamma$
denotes the expectation with respect to $\omega_k$, $k \in \Gamma$. That is,
$\mathbb{E}_{\Gamma} (\cdot) := \int_{\Omega_\Gamma} (\cdot)
\prod_{k \in \Gamma} \nu(\drm \omega_k)$ where $\Omega_\Gamma
:= \bigtimes_{k \in \Gamma} \RR$. 
\par
Let the \emph{single-site potential}\index{single-site potential}
$u : \ZZ^d \to \RR$ be a function in $\ell^1 (\ZZ^d ; \RR) = \{u :\ZZ^d \to \RR \colon \sum_{k \in \ZZ^d} \lvert u(k) \rvert < \infty\}$. We assume that the random
potential $V_\omega $ has an alloy-type\index{alloy-type\addcontentsline{toc}{chapter}{Index}} structure, i.e.\ the potential value
\begin{equation*}
V_\omega (x) := \sum_{k \in \ZZ^d} \omega_k u (x-k)
\end{equation*}
at a lattice site $x \in \ZZ^d$ is a linear combination of the i.i.d.\ random
variables $\omega_k$, $k\in\ZZ^d$, with coefficients provided by the single-site
potential. We call the Hamiltonian \eqref{eq:hamiltonian} a \emph{discrete alloy-type model}\index{discrete alloy-type model}. The function $u(\cdot - k)$ may be interpreted as a potential generated by the atom sitting at the lattice site $k\in\ZZ^d$. We assume (without loss of generality) that $0 \in \Theta := \supp u$.
\par
For the operator $H_\omega$ in \eqref{eq:hamiltonian} and $z \in \CC \setminus \sigma
(H_\omega)$ we define the corresponding \emph{resolvent}\index{resolvent} by $G_\omega (z)
:= (H_\omega - z)^{-1}$. For the \emph{Green function}\index{Green function}, which assigns
to each  $(x,y) \in \ZZ^d \times \ZZ^d$ the corresponding matrix element of the
resolvent, we use the notation
\begin{equation*} \label{eq:greens}
G_\omega (z;x,y) := \sprod{\delta_x}{(H_\omega - z)^{-1}\delta_y}.
\end{equation*}
For $\Gamma \subset \ZZ^d$ and $k \in \Gamma$, $\delta_k \in \ell^2 (\Gamma)$ denotes the
Dirac function\index{Dirac function} given by $\delta_k (k) = 1$ and
$\delta_k (j) = 0$ for $j \in \Gamma \setminus \{k\}$.
Let $\Gamma_1 \subset \Gamma_2 \subset \ZZ^d$. We define the canonical restriction $\Pro_{\Gamma_1}^{\Gamma_2} : \ell^2 (\Gamma_2) \to \ell^2 (\Gamma_1)$ by
\[
 \Pro_{\Gamma_1}^{\Gamma_2} \psi := \sum_{k \in \Gamma_1} \psi (k) \delta_k ,
\]
where the Dirac function has to be understood as an element of $\ell^2 (\Gamma_1)$.
Note that the corresponding embedding $\Inc_{\Gamma_1}^{\Gamma_2} := (\Pro_{\Gamma_1}^{\Gamma_2})^* : \ell^2 (\Gamma_1) \to \ell^2 (\Gamma_2)$ is given by
\[
\Inc_{\Gamma_1}^{\Gamma_2} \phi := \sum_{k \in \Gamma_1} \phi (k) \delta_k ,
\]
where here the Dirac function has to be understood as an element of $\ell^2 (\Gamma_2)$.
If $\Gamma_2 = \ZZ^d$ we will drop the upper index and write $\Pro_{\Gamma_1}$ and $\Inc_{\Gamma_1}$ instead of $\Pro_{\Gamma_1}^{\ZZ^d}$ and $\Inc_{\Gamma_1}^{\ZZ^d}$.
For an arbitrary set $\Gamma \subset \ZZ^d$ we define the restricted operators $\Delta_\Gamma, V_\Gamma, H_\Gamma:\ell^2 (\Gamma) \to \ell^2 (\Gamma)$ by $\Delta_\Gamma := \Pro_\Gamma \Delta \Inc_\Gamma$, $V_\Gamma := \Pro_\Gamma V_\omega \Inc_\Gamma$ and
\[
 H_\Gamma := \Pro_\Gamma H_\omega \Inc_\Gamma = -\Delta_\Gamma + \lambda V_\Gamma .
\]
Furthermore, we define $G_\Gamma (z) := (H_\Gamma - z)^{-1}$ and $G_\Gamma (z;x,y) := \bigl\langle \delta_x, G_\Gamma (z) \delta_y \bigr\rangle$ for $z \in \CC \setminus \sigma (H_\Gamma)$ and $x,y \in \Gamma$. Pay attention that we suppress the dependence on $\omega \in \Omega$ for the restricted operators as well as for the restricted Green function. 
\par
If $\Lambda\subset\ZZ^d$ is finite, $H_\Lambda$ can be interpreted as a finite dimensional random matrix. For example, if $d = 1$ and $\Lambda = \{-2,-1,0,1,2\}$ the matrix representation of $H_\Lambda$ with respect to the canonical basis $\{\delta_k\}_{k \in \Lambda}$ is given by
 \[ 
 H_\Lambda = \begin{pmatrix}
 v_{-2}    & -1        &   0          & 0        &  0       \\ 
 -1        & v_{-1}    &-1            & 0        &  0       \\
  0        &	-1     & v_{0}        &    -1    &  0       \\
  0        &	0      &   -1         & v_{1}    &   -1     \\
  0        &	0      &   0          &     -1   & v_{2}
          
            \end{pmatrix},
\quad v_x = \sum\limits_{k \in \ZZ} \omega_k u(x-k), \quad x \in \{-2,\ldots , 2\} .
\]
The diagonal elements $v_x$, $x \in \{-2,-1,\ldots , 2\}$, correspond to the multiplication operator $V_\Lambda$ and the minus ones on the upper and lower diagonal come up from the negative discrete Laplacian $-\Delta_\Lambda$. 
\par
If we consider operators on $\ell^2 (\Lambda)$ for $\Lambda \subset \ZZ^d$ finite, we will always use the same symbol for the operator as well as for the corresponding matrix representation with respect to the canonical basis.
\par
Let us finally introduce some further assumptions on the model which may hold or not hold. First we introduce some notation. For $\Lambda \subset \ZZ^d$ finite, $\lvert \Lambda \rvert$ will denote the number of elements in $\Lambda$. For $\Lambda \subset \ZZ^d$ we denote by $\partial^{\rm i} \Lambda = \{k \in \Lambda \colon \lvert \{ j \in \Lambda \colon \lvert k-j \rvert_1 = 1  \} \rvert < 2d \}$ the interior boundary\index{boundary!interior} of $\Lambda$ and by $\partial^{\rm o} \Lambda = \partial^{\rm i} (\Lambda^{\rm c})$ the exterior boundary\index{boundary!exterior} of $\Lambda$. Here, $\Lambda^{\rm c} = \ZZ^d \setminus \Lambda$ denotes the complement\index{complement} of $\Lambda$. Recall that $\Theta = \supp u$.
\begin{assumption} \label{ass:d=1}
Assume $d = 1$, the measure $\nu$ has a probability density $\rho \in L^\infty (\RR)$ and $\Theta = \{0,1,\ldots , n-1\}$ for some $n \in \NN$.
\end{assumption}

\begin{assumption} \label{ass:d=1_2}
Assume $d = 1$, the measure $\nu$ has a probability density $\rho \in L^\infty (\RR)$ and $\Theta$ is finite with $\min \Theta = 0$ and $\max \Theta = n-1$ for some $n \in \NN$.
\end{assumption}

\begin{assumption} \label{ass:monotone} Assume that $\Theta$ is a finite set, the measure $\nu$ has a probability density $\rho \in L^\infty (\RR)$, and that the function $u$ satisfies $u (k) > 0$ for all $k \in \partial^{\rm i} \Theta$.
\end{assumption}

\begin{assumption} \label{ass:exponential} The measure $\nu$ has a probability density $\rho \in \BV (\RR)$ and there are constants $C,\alpha>0$ such that for all $k \in \ZZ^d$ we have $\lvert u(k) \rvert \leq C \euler^{-\alpha \lVert k \rVert_1}$.
\end{assumption}
Here $\BV (\RR)$ denotes the space of functions of finite total variation\index{total variation}. A precise definition of this function space is given in Section~\ref{sec:abstract_wegner}.
\begin{assumption} \label{ass:finite}
Let $\Theta$ be a finite set.
\end{assumption}
\begin{assumption} \label{ass:ubar}
The measure $\nu$ has a probability density in the Sobolev space $W^{1,1} (\RR)$, $\Theta$ is finite and the single-site potential satisfies $\bar u := \sum_{k \in \ZZ^d} u(k) \not = 0$.
\end{assumption}
If we say that a measure $\nu$ on $\RR$ has a probability density $\rho$, we mean that the measure $\nu$ is absolutely continuous with respect to the Lebesgue measure with the corresponding density function $\rho$, i.e.\ we have $\nu (A) = \int_A\rho (x) \drm x$ for all measurable sets $A\subset \RR$. 
%
%
%
\section{Main results} \label{sec:results}

For $x \in \ZZ^d$ and $L>0$, we denote by $\Lambda_{L,x} = \{ k \in \ZZ^d : \lvert x-k  \rvert_\infty \leq L \}$ the cube of side length $2L$. Furthermore, we set $\Lambda_L = \Lambda_{L,0}$.
\par
As we are interested in the spectral type of the almost sure spectrum of $H_\omega$, more precisely that there is no continuous spectrum in certain energy/disorder regimes, let us first note that the almost sure spectrum of $H_\omega$ is an interval.
\begin{theorem}[The almost sure spectrum]
Let $\supp \nu$ be a bounded interval. Then, for almost all $\omega \in \Omega$, the spectrum of $H_{\omega}$ is an interval.
\end{theorem}
This result is based on a discussion of Helge Kr\"uger and Ivan Veseli\'c. For a proof we refer to \cite{ElgartKTV-11}.
\par
Our first result is a so-called finite volume criterion. It can be used to prove exponential decay of an averaged fractional power of the Green function at typical perturbative regimes.
\begin{theorem}[Finite volume criterion] \label{theorem:finite_volume}
Suppose that Assumption \ref{ass:monotone} is satisfied, let $\Gamma \subset \ZZ^d$, $z
\in \CC \setminus \RR$ with $\lvert z \rvert \leq m$ and $s \in (0,1/3)$. Then there
exists a constant $B_s$ which depends only on $d$, $\rho$, $u$, $m$ and $s$,
such that if the condition
\begin{equation*}
b_s(\lambda, L,\Lambda): = \frac{B_s L^{3(d-1)} \Xi_s (\lambda)}{\lambda^{2s/(2\lvert
\Theta \rvert)}}\, \sum_{w\in\partial^{\rm o} W_x}\mathbb{E}
\bigl(\lvert G_{\Lambda\setminus W_x} (z;x,w)\rvert^{s/(2\lvert
\Theta \rvert)}\bigr)< b
\end{equation*}
is satisfied for some $b \in (0,1)$, arbitrary $\Lambda \subset \Gamma$, and all $x\in
\Lambda$, then for all $x,y \in \Gamma$
\begin{equation*} 
\mathbb{E} \bigl(\lvert G_\Gamma (z;x,y)\rvert^{s/(2\lvert \Theta
\rvert)} \bigr)\leq A \euler^{-\mu|x-y|_\infty} .
\end{equation*}
Here
\[
A=\frac{C_s \Xi_s (\lambda)}{b}, \quad 
\Xi_s (\lambda) = \max \{ \lambda^{- s/ 2 \lvert \Theta \rvert} , \lambda^{-2s} \}, \quad 
\mu=\frac{\lvert \ln b \rvert}{L+\diam \Theta + 2} ,
\] 
with $C_s$, depending only on $s$, inherited from the a priori bound of Lemma \ref{lemma:bounded}. The set $W_x$ is a certain annulus around $x$, defined precisely in Eq.~\eqref{eq:Wx} and the text below, and $L \geq \diam \Theta + 2$ is some fixed number determining the size of the annulus $W_x$.
\end{theorem}
This criterion works in the strong disorder regime using an a priori bound (see Section~\ref{sec:boundedness}), and Theorem~\ref{theorem:exp_decay} follows. Theorem~\ref{theorem:exp_decay} is the typical output of the fractional moment method, i.e.\ the exponential decay of an averaged fractional power of the Green function. It applies to arbitrary finite $\Theta \subset \ZZ^d$ assuming that the single-site potential $u$ has fixed sign on the interior vertex boundary of $\Theta$. 
\begin{theorem}[Fractional moment decay] \label{theorem:exp_decay}
Let $\Gamma \subset \ZZ^d$, $s \in (0,1/3)$ and suppose that Assumption \ref{ass:monotone} is satisfied.
Then for a sufficiently large $\lambda$ there are constants $\mu,A \in (0,\infty)$,
depending only on $d$, $\rho$, $u$,  $s$ and $\lambda$,
such that for all $z \in \CC \setminus \RR$ and all $x,y \in \Gamma$
\begin{equation*} 
\mathbb{E} \bigl(\lvert G_\Gamma (z;x,y)\rvert^{s/(2\lvert \Theta \rvert)}\bigr)\leq A \euler^{-\mu|x-y|_\infty} .
\end{equation*}
\end{theorem}
The next theorem states that the exponential decay of an averaged fractional power of the Green function implies exponential localization. Notice that this implication is well known for the (discrete and continuous) alloy-type model with sign-definite single-site potential. If the single-site potential may change its sign, the existing methods do not apply. We provide a new variant for proving this implication which applies to the sign-indefinite case. Let us also emphasize that this result does not rely on Assumption \ref{ass:monotone}. What is needed is that $\Theta$ is a finite set and $\nu$ is allowed to be an arbitrary probability measure with compact support. See Remark~\ref{remark:no_finite} below for a discussion on the assumption that $\Theta$ is finite.
\begin{theorem} \label{theorem:exp_decay_loc}
Let Assumption \ref{ass:finite} be satisfied, $s \in (0,1)$, $C,\mu, \in (0,\infty)$, and $I \subset \RR$ be an interval. Assume that
\[
\EE \bigl( \lvert G_{\Lambda_{L,k}} (E + \i \epsilon;x,y) \rvert^{s} \bigr) \leq C \euler^{-\mu \lvert x-y \rvert_\infty}
\]
for all $k \in \ZZ^d$, $L \in \NN$, $x,y \in \Lambda_{L,k}$, $E \in I$ and all $\epsilon \in (0,1]$.
Then, for almost all $\omega \in \Omega$, $H_\omega$ exhibits exponential localization in $I$.
\end{theorem}
\begin{remark} \label{remark:no_finite}
Theorem~\ref{theorem:exp_decay_loc} requires Assumption~\ref{ass:finite}. It seems that the proof and so the result can be extended to the case where Assumption~\ref{ass:finite} is replaced by $\lvert u(x) \rvert \leq C \lvert x \rvert_1^{-m}$ for some constant $m>4d$ and $\lvert x \rvert_1$ sufficiently large, but $u$ not necessarily of finite support. This may be realized by using ideas from \cite{KirschSS-98b} relying on a so-called uniform Wegner estimate to provide the independence of two finite-volume Hamiltonians even though the single-site potential is not of finite support.
\end{remark}
Putting together Theorem~\ref{theorem:exp_decay} and Theorem~\ref{theorem:exp_decay_loc}, we obtain exponential localization in the case of sufficiently large disorder.
\begin{theorem}[Exponential localization] \label{theorem:localization}
Let Assumption \ref{ass:monotone} be satisfied and $\lambda$ sufficiently large.
Then, for almost all $\omega \in \Omega$, $H_\omega$ exhibits exponential localization.
\end{theorem}
All the previous theorems concern a proof of localization according to the fractional moment method. If one pursues a proof via multiscale analysis, there is a need for a Wegner estimate as an input for the multiscale analysis. The next theorem states a Wegner estimate for the discrete alloy-type model with sign-indefinite and exponentially decaying single-site potential. Note that the single-site potential is not assumed to have compact support.
\begin{theorem}[Wegner estimate] \label{theorem:wegner}
Let Assumption \ref{ass:exponential} be satisfied. Then there exists $C(u)>0$ and $I_0 \in \NN_0^d$ both depending only on $u$ such that for any $l\in \NN$ and any bounded interval $I \subset \RR$
\[
 \EE \bigl (\Tr \chi_I (H_{\Lambda_{l}}) \bigr)\le
\lambda^{-1} C(u) \lVert \rho \rVert_{\rm Var} \lvert I \rvert (2l+1)^{2d + \lvert I_0 \rvert} .
\]
A precise definition of $I_0 \in \NN_0^d$ is given in Eq.~\eqref{eq:cF}.

\end{theorem}
\begin{remark}
Theorem \ref{theorem:wegner} is related to the results established in \cite{Veselic-10a}. There several Wegner bounds have been proven for $u \in \ell^1 (\ZZ^d ; \RR)$ (not necessarily of exponential decay) and by considering the assumptions
\begin{enumerate}[(i)]
 \item $u$ is finitely supported, or
 \item $\bar u := \sum_{k\in\ZZ^d} u(k) \neq 0$, or
 \item the space dimension satisfies $d=1$ and $u$ decays exponentially,
\end{enumerate}
which may hold or not hold.
If one of the above conditions is satisfied, the Wegner bound of \cite{Veselic-10a} is linear in the length of the energy interval and polynomially in the volume of the cube. A particularly important case in \cite{Veselic-10a} is the one when both conditions (i) and (ii) hold. In this situation the exponent of the length scale can be chosen to be equal to the space dimension $d$, and yields the Lipschitz continuity of the integrated density of states, and consequently its derivative, the density of states, exists for almost all $E \in \RR$, see \cite{Veselic-08} for a detailed discussion.
Concerning the case where assumption (iii) holds, one can say that Theorem \ref{theorem:wegner} is  a multidimensional generalization of the Wegner estimate from \cite{Veselic-10a}, though the improved proof presented here allows more explicit control on the volume dependence. However, if one goes back to the case where the single site potential is non-negative, the volume dependence in Theorem \ref{theorem:wegner} is quadratic while the one from the results in \cite{Veselic-10a} is linear.
\end{remark}
\begin{remark} \label{remark:wegner_loc}
The Wegner estimate from Theorem \ref{theorem:wegner} is linear in the energy-interval length and polynomial in the volume of the cube. Hence, our Wegner bound can be used for a localization proof via multiscale analysis in any energy region where an initial length scale estimate holds, see e.g.\ \cite{FroehlichS-83,FroehlichMSS-85,DreifusK-89,Kirsch-08}. If the single-site potential does not have compact support, one has to use an enhanced version of the multiscale analysis and so-called uniform Wegner estimates to prove localization, see \cite{KirschSS-98b}. 
\end{remark}
The proofs of the main results are organized as follows. In Chapter~\ref{chap:fmm} we prove all the results concerning the fractional moment method. In Section \ref{sec:d=1} we restrict ourselves to the special case $d=1$, since the important steps of the proof of localization are particularly transparent and the restriction to the one-dimensional case allows an explicit control over the constants. In Section~\ref{sec:boundedness} we prove an a priori bound which is used in Section~\ref{sec:exp_decay} to prove Theorem~\ref{theorem:finite_volume} and Theorem~\ref{theorem:exp_decay}. Section~\ref{sec:loc} is devoted to the proof of Theorems~\ref{theorem:exp_decay_loc} and \ref{theorem:localization}. In Chapter~\ref{chap:wegner} we prove Theorem~\ref{theorem:wegner}, a Wegner estimate for the discrete alloy-type model. 
\par
In an appendix, we present several results related to our main results. In particular, we prove an alternative a priori bound to the one in Section~\ref{sec:boundedness} in Appendix~\ref{chap:non_local}. Appendix~\ref{sec:cont_results} concerns some results on the continuous counterpart of the discrete alloy-type model, the alloy-type model. We prove a Wegner estimate for the alloy-type model with sign-changing single-site potential of a so-called generalized step-function form. Moreover, we show for the alloy-type model with sign-changing single-site potential an analogue of Theorem~\ref{theorem:exp_decay_loc}, i.e.\ that the exponential decay of fractional moments implies exponential localization.
\begin{remark}
The results proven according to the fractional moment method concern an sign-indefinite discrete alloy-type model on $\ell^2 (\ZZ^d)$. It is an interesting question whether the results can be generalized to a discrete alloy-type model defined on a locally finite graph. More precisely, let $G = (V,E)$ be an infinite, locally finite, connected graph without loops or multiple edges where $V = V (G)$ is the set of vertices and $E = E(G)$ is the set of edges. We use the notation $x \sim y$ to indicate that an edge connects the vertices $x$ and $y$, and $m(x) := \lvert \{y \in V \colon y \sim x\} \rvert$ for the number of vertices connected by an edge to $x$. On $\ell^2 (V)$ we consider the operator
\[
 H_\omega^{\rm G} = -\Delta^{\rm G} + \lambda V_\omega^{\rm G}
\]
where $(\Delta^{\rm G} \psi )(x) := -m(x) \psi (x) + \sum_{y \sim x} \psi (y)$, $V_\omega^{\rm G} (x) = \sum_{k \in V} \omega_k u(x-k)$, $u\in\ell^1 (V;\RR)$ and $\omega_k$, $k \in V$, a sequence of bounded and i.i.d.\ random variables whose distribution has a probability density $\rho \in L^{\infty} (\RR)$. Since the graph $G$ may have no uniform bound on the vertex degree, the Laplacian may be an unbounded operator. Let us note that $\Delta^{\rm G}$ is essentially self-adjoint on the dense subset of functions $\psi : V \to \CC$ with finite support, see e.g.\ \cite{Wojciechowski2007,Weber2010,Jorgensen2008}, or \cite{KellerL2011} for a proof in a more general framework.
\par
In the case where $V_\omega^{\rm G} (x) = \omega_x$ and $\lambda$ is sufficiently large, exponential and dynamical localization has been proven in \cite{Tautenhahn-11} for a certain class of locally finite graphs, including all graphs which have a uniform bound on the vertex degree and also some graphs with no uniform bound on the vertex degree.
\par
If one combines the methods from \cite{Tautenhahn-11} and this thesis (\cite{ElgartTV-10,ElgartTV-11}), one may find new criteria on the graph $G$ and on the sign-changing single-site potential $u$, such that exponential or dynamical localization can be proven for the operator $H_\omega^{\rm G}$ almost surely in the case of sufficiently large disorder.
\end{remark}

%
%
\section{Regularity properties\index{regularity properties} for the discrete alloy-type model} \label{sec:regularity}
Anderson models where the potential values at different lattice sites are not independent have been studied previously in the literature according to the multiscale analysis and the fractional moment method, see e.\,g. \cite{AizenmanM-93,DreifusK-91,AizenmanG-98,Hundertmark-00,AizenmanFSH-01,Hundertmark-08}. Among others, they prove localization as long as the potential values satisfy certain regularity conditions. More precisely, they require regularity of the distribution of the potential at $x \in \ZZ^d$ conditioned on arbitrary fixed potential values elsewhere. One question is, whether the regularity conditions of the above mentioned papers are satisfied for the discrete alloy-type potential or not, in other words, whether the theorems in \cite{AizenmanM-93,DreifusK-91,AizenmanG-98,Hundertmark-00,AizenmanFSH-01,Hundertmark-08} apply to our model or not. To be specific, let us formulate the regularity condition from \cite{AizenmanFSH-01}.
\begin{definition}\label{def:conditional}
Let $X$ be a countable set, $\{v_x\}_{x\in X}$ be a collection of real valued random variables and $\rho_x(\cdot\mid v_x^\perp)$ the probability distribution of $v_x$ conditioned on the random variables $v_x^\perp = \{v_j\}_{j \in X \setminus \{x\}}$. The collection $\rho_x(\cdot \mid v_x^\perp)$, $x\in X$,
is said to be (uniformly) \emph{$\tau$-H\"older continuous}\index{uniformly $\tau$-H\"older continuous} for $\tau \in(0,1]$ if there is a constant $C$ such that
\[
\sup_{x\in X} \sup_{v_x^\perp} \rho_x([a,b]\mid v_x^\perp)
\le C(b-a)^\tau \quad \text{ for all } [a,b] \subset \RR.
\]
The second supremum is taken over all possible values of $v_x^\perp$
in $\times_{\ZZ^d\setminus \{x\}}\RR$.
\end{definition}
In the following we want to study in which cases the regularity condition of Definition~\ref{def:conditional} is satisfied for the alloy-type potential and for which not. Therefore, we assume that the measure $\mu$ has a probability density $\rho \in L^\infty (\RR)$. The established results have already been published in \cite{TautenhahnV-10b}. First we give a result in the negative direction. 
\begin{lemma} \label{lemma:negexample}
 Let $d = 1$, $\Theta = \{0,1,\dots,n-1\}$ for some $n \in \NN$, $\inf \supp \rho = 0$ and $\sup \supp \rho = 1$. Then there are constants $c,m,s^+ \in (-\infty,\infty)$, depending only on $u$, such that for all $\delta > 0$ and $\delta \geq \delta' > 0$
\[
 \PP \bigl(\bigl\{ V_\omega (0) \in [m-c\delta , m + c \delta] \bigr\} \mid \bigl\{ V_\omega(-1),V_\omega (n-1) \in [s^+ - \delta' , s^+]  \bigr\}\bigr) = 1 .
\]
The values of the constants $c$, $m$ and $s^+$ can be inferred from the proof.
\end{lemma}
Notice that, under the assumptions of Lemma~\ref{lemma:negexample}, $V_\omega (-1)$ and $V_\omega (n-1)$ are stochastically independent and $\PP (\{ V_\omega(-1),V_\omega (n-1) \in [s^+ - \delta' , s^+] \}) > 0$, where $s^+$ is defined in the proof of Lemma~\ref{lemma:negexample}.
\begin{proof}[Proof of Lemma \ref{lemma:negexample}]
Let $\Theta^{\rm p} := \{k \in \ZZ : u (k) > 0\}$ and $\Theta^{\rm n} := \{k \in \ZZ : u (k) < 0\}$. Further let $u_{\rm max} = \max_{k \in \Theta} \lvert u(k) \rvert$, $u_{\rm min} = \min_{k \in \Theta} \lvert u(k)\rvert$ and $s^+ = \sum_{k \in \Theta^{\rm p}} u(k)$. Let us introduce two further subsets of $\Theta$ which are important in our study. The first one is
\[
\Theta_1 =
\begin{cases}
\Theta^{\rm p} + 1 & \text{if $n-1 \not \in \Theta^{\rm p}$}, \\[1ex]
\left ( (\Theta^{\rm p} + 1)\cap \Theta \right ) \cup \{0\} & \text{if $n-1 \in \Theta^{\rm p}$},
\end{cases}
\]
with $\Theta^{\rm p} + 1 = \{k\in \NN : (k-1)\in \Theta^{\rm p} \}$. The second subset is the complement $\Theta_0 = \Theta \setminus \Theta_1$. To end the proof we show the following interval arithmetic result:
\par
\textit{Let $\delta \geq \delta' > 0$ and $V_\omega (-1),V_\omega(n-1) \in [s^+ - \delta' , s^+]$. Then
\begin{equation} \label{eq:proof0}
V_\omega (0) \in [m-c\delta' , m + c \delta'] \subset [m-c\delta , m + c \delta]
\end{equation}
with $c = n u_{\rm max} / u_{\rm min}$ and $m = \sum_{k \in \Theta_1} u (k)$.}
\par
We divide the proof of \eqref{eq:proof0} into three parts. The first step is to argue that
\begin{equation} \label{eq:proof1}
\omega_{-1-k} \in 
\begin{cases}
\bigl[ 1 - \frac{\delta'}{u_{\rm min}} , 1 \bigr]  & \text{for $k \in \Theta^{\rm p}$}, \\[1ex]
\bigl[ 0, \frac{\delta'}{u_{\rm min}} \bigr]   & \text{for $k \in \Theta^{\rm n}$}.
\end{cases} 
\end{equation}
For the proof of the first part of \eqref{eq:proof1} we use the assumption $V_\omega(-1) \geq s^+ - \delta'$ and obtain
\[
s^+ - \delta' \leq V_\omega (-1) = \sum_{k \in \Theta} u(k) \omega_{-1-k} \leq \sum_{k \in \Theta^{\rm p}} u(k) \omega_{-1-k} , 
\]
and hence $\sum_{k \in \Theta^{\rm p}} u(k) (1- \omega_{-1-k}) \leq \delta'$. We conclude that for all $k \in \Theta^{\rm p}$ we have $u(k) (1-\omega_{-1-k}) \leq \delta'$ which gives the first part of \eqref{eq:proof1}. For the proof of the second part of \eqref{eq:proof1} we use again the assumption $V_\omega (-1) \geq s^+ - \delta'$ and obtain
\begin{equation*}
\sum_{k \in \Theta^{\rm p}} u(k) \omega_{-1-k} - \delta'  \leq s^+ - \delta'  \leq  V_\omega (-1) = \sum_{k \in \Theta^{\rm p}} u(k) \omega_{-1-k} + \sum_{k \in \Theta^{\rm n}} u(k) \omega_{-1-k}
\end{equation*}
which gives $-\delta' \leq \sum_{k \in \Theta^{\rm n}} u(k) \omega_{-1-k}$. Thus, for all $k \in \Theta^{\rm n}$ we have $\omega_{-1-k} \leq -\delta' / u(k) = \delta' / \lvert u(k) \rvert$ which gives the second part of \eqref{eq:proof1}. In a second step we argue that 
\begin{equation} \label{eq:proof2}
\omega_{-k+n-1} \in 
\begin{cases}
\bigl[ 1 - \frac{\delta'}{u_{\rm min}} , 1 \bigr]  & \text{for $k \in \Theta^{\rm p}$}, \\[1ex]
\bigl[ 0, \frac{\delta'}{u_{\rm min}} \bigr]   & \text{for $k \in \Theta^{\rm n}$} . 
\end{cases} 
\end{equation}
The proof of \eqref{eq:proof2} can be done in analogy to the proof of \eqref{eq:proof1}, but using the assumption $V_\omega (n-1) \geq s^+ - \delta'$. 
In a third step we ask the question for which $k \in \Theta$ we have $\omega_{-k} \in [1-\delta' / u_{\rm min} , 1]$. Using the definition of the set $\Theta_1$ we find with \eqref{eq:proof1} and \eqref{eq:proof2} that
\begin{equation} \label{eq:proof3}
\omega_{-k} \in 
\begin{cases}
\bigl[ 1 - \frac{\delta'}{u_{\rm min}} , 1 \bigr]  & \text{for $k \in \Theta_1$}, \\[1ex]
\bigl[ 0, \frac{\delta'}{u_{\rm min}} \bigr]   & \text{for $k \in \Theta_0$} . 
\end{cases} 
\end{equation}
Now, the desired result \eqref{eq:proof0} follows from \eqref{eq:proof3} and the decomposition
\[
V_\omega (0) = \sum_{k \in \Theta} u(k) \omega_{-k} = \sum_{k \in \Theta_1} u(k) \omega_{-k} + \sum_{k \in \Theta_0} u(k) \omega_{-k} .
\]
Hence, the proof is complete.
\end{proof}
\begin{remark}
 The assumption $\inf \supp \rho = 0$ and $\sup \supp \rho = 1$ in Lemma \ref{lemma:negexample} is not crucial. What matters is that $\supp \rho$ is a bounded set.
\end{remark}
\begin{remark}
Lemma \ref{lemma:negexample} implies that the collection of random variables $V_\omega (k)$, $k \in \ZZ^d$, is not uniformly $\tau$-H\"o{}lder continuous. Hence, the results in \cite{AizenmanFSH-01} do not apply to the discrete alloy-type model in general. 
\end{remark}
Now, we consider the case $d = 1$, $\Theta = \{-1,0\}$, $u(0) = 1$ and where $\rho$ is a Gaussian density function with mean zero and variance $\sigma^2$. Note that $\rho$ has unbounded support, although we assumed that the measure $\nu$ has compact support. However, let us consider for the sake of this discussion a generalization of the discrete alloy-type model with a probability measure $\nu$ of unbounded support. In this situation it turns out that the regularity assumption from \cite{AizenmanFSH-01} is satisfied as long as $\lvert u(-1) \rvert \not = 1$. We also refer the reader to \cite{DreifusK-91} where a simiar condition to Definition \ref{def:conditional} is studied in the Gaussian case. Our study is based on the following classical result which may be found in \cite{Port-94}.
\begin{proposition} \label{prop:cond0}
Let $X$ be normally distributed on $\RR^d$, $Y = a \cdot X$ where $a \in \RR^{d}$, and $W = B X$ where $B \in \RR^{m \times d}$. Assume $W$ has a non-singular distribution. Then the distribution of $Y$ conditioned on $W = v \in \RR^m$ is the Gaussian distribution having mean
\[
 \mathbf{E} ( Y ) + \cov (Y,W) \cov(W,W)^{-1} [v - \mathbf{E} (W)]
\]
and variance
\[
 \cov (Y,Y) - \cov (Y,W) \cov(W,W)^{-1} \operatorname{cov} (W,Y) .
\]
\end{proposition}
For $l \in \NN$ let $A_{l} \in \RR^{l \times l+1}$ be the matrix with coefficients in the canonical basis given by $A_{l} (i,i) = 1$, $A_l (i,i+1) = u(-1)$ for $i \in \{1,\dots,l\}$, and zero otherwise, namely
\[
A_l= \begin{pmatrix}
1 & u(-1) & & &\\
& \ddots & \ddots & &\\
 & & \ddots & u(-1) &\\
  & & & 1 & u(-1)
\end{pmatrix} \in \RR^{l \times l+1} .
\]
Notice, if we apply $A_l$ on the vector $\omega_{[x,x+l]} = (\omega_{x+k-1})_{k=1}^{l+1}$, we obtain a vector containing the potential values $V_\omega (k)$, $k \in \{x,x+1,\ldots,x+l\}$. Moreover, the vector $(V_\omega (x+k-1))_{k=1}^l = A_l \omega_{[x,x+l]}$ is normally distributed with mean zero and covariance $\sigma^2 A_l A_l^{\rm T}$.
The matrix $A_l A_l^T$ has the form
\[
A_l A_l^T = \begin{pmatrix}
1 + u^2(-1) 	& u(-1) 	& 		& 		\\
u(-1) 	& 1+u^2(-1) 	& \ddots 	& 		\\
 		& \ddots 	& \ddots 	& u(-1)	\\
  		& 		& u(-1) 	& 1 + u^2(-1)
\end{pmatrix} \in \RR^{l \times l} .
\]
By induction we find that the determinant of $A_l A_l^{\rm T}$ is given by
\[
 \det (A_l A_l^{\rm T}) = s_l > 0 \quad \text{where} \quad s_l := \sum_{i=1}^{l} \bigl(u(-1)\bigr)^{2i} .
\]
Since the minor $M_{11}$ and $M_{ll}$ of $A_l A_l^{\rm T}$ equals $A_{l-1} A_{l-1}^{\rm T}$ we obtain by Cramers rule for the elements $(1,1)$ and $(l,l)$ of the inverse of $A_{l-1} A_{l-1}^{\rm T}$
\begin{equation} \label{eq:inverseelement}
 (A_{l} A_{l}^{\rm T})^{-1} (1,1) = (A_{l} A_{l}^{\rm T})^{-1} (l,l) = \frac{s_{l-1}}{s_l} .
\end{equation}
\begin{lemma} \label{lemma:gaussian1}
Let $d = 1$, $l,m \geq 1$, $\Theta = \{-1,0\}$, $u(0)=1$ and $\rho$ be the Gaussian density with mean zero and variance $\sigma^2$. Let further $v^+ \in \RR^{l}$ and $v^- \in \RR^{m}$. Then the distribution of $V_\omega (0)$ conditioned on $(V_\omega (k))_{k=1}^{l} = v^+$ and $(V_\omega (-m+k-1))_{k=1}^{m} = v^-$ is Gaussian with variance
\begin{equation*} 
 \gamma = \sigma^2 \Bigl( u(-1)^2 - 1 + \frac{1}{s_m} + \frac{1}{s_{l}} \Bigr)
\end{equation*}
and mean
\[
 n = u(-1) \left( \sum_{i=1}^{m} (A_{m} A_{m}^{\rm T})^{-1} (m,i) \, v^-_i + \sum_{i=1}^{l} (A_{l} A_{l}^{\rm T})^{-1} (1,i) \, v^+_i \right) .
\]
\end{lemma}
\begin{proof}
 Let $X := (\omega_{-m-1+k})_{k=1}^{l+m+2} \in \RR^{l+m+2}$, $a = (a_{i})_{i=1}^{l+m+2} \in \RR^{l+m+2}$ the vector with coefficients $a_{m+1}=1$, $a_{m+2} = u(-1)$ and zero otherwise. Let us further define the block-matrix
\[
 B = \begin{pmatrix}
      A_{m} & 0 \\
      0   & A_{l}
     \end{pmatrix} \in \RR^{(m+l) \times (m+l+2) }.
\]
Notice that $Y := a \cdot X = V_\omega (0)$,
\[
 A_{m} \omega_{[-m,0]} = (V_\omega (-m+k-1))_{k=1}^{m}, \quad \text{and} \quad 
 A_{l} \omega_{[1,l+1]}= (V_\omega (k))_{k=1}^{l} , 
\]
where $\omega_{[-m,0]} = (\omega_{-m+k-1})_{k=1}^{m+1}$ and $\omega_{[1,l+1]} = (\omega_{k})_{k=1}^{l+1}$.
Hence $W := B X$ is the $(m+l)$-dimensional vector containing the potentials $V_\omega (k)$, $k \in \{-m, \dots, l\} \setminus \{0\}$. Notice that $Y$ and $W$ have mean zero, since $X$ has mean zero.
We apply Proposition~\ref{prop:cond0} with these choices of $X$, $Y$ and $W$, and obtain that the distribution of $V_\omega (0)$ conditioned on $(V_\omega (-m+k-1))_{k=1}^{m} = v^-$ and $(V_\omega (k))_{k=1}^{l} = v^+$ is Gaussian with mean 
\begin{equation*}
n = \cov (Y,W) \cov(W,W)^{-1} v 
\end{equation*}
and variance 
\[
\gamma = \cov (Y,Y) - \cov (Y,W) \cov(w,w)^{-1} \operatorname{cov} (W,Y), 
\]
where $v = (v^- , v^+)^{\rm T}$. It is straightforward to calculate $\cov (Y,Y) = \sigma^2 (1+ u(-1)^2)$ and $\cov (W,Y) =  z = (z^- , z^+)^{\rm T}$, where $z^- = (0,\dots,0,\sigma^2 u(-1))^{\rm T} \in \RR^{m}$ and $z^+ = (\sigma^2 u(-1) ,0,\dots,0)^{\rm T} \in \RR^{l}$. We also have
\[
 \cov (W,W) =   \sigma^2 \begin{pmatrix}
                         A_{m} A_{m}^{\rm T} & 0 \\
                        0  & A_{l} A_{l}^{\rm T}
                       \end{pmatrix} .
\]
Hence by Eq.~\eqref{eq:inverseelement}
\begin{align*}
\gamma &= \sigma^2 (1 + u(-1)^2) - \sigma^{-2}  z^{\rm T} \begin{pmatrix}
                        ( A_{m} A_{m}^{\rm T})^{-1} & 0 \\
                        0  & (A_{l} A_{l}^{\rm T})^{-1}
                       \end{pmatrix}^{-1} z  \nonumber \\
%
%
&= \sigma^2 (1 + u(-1)^2) - \sigma^{-2} \biggl[ \sigma^4 u^2(-1) \frac{s_{m-1}}{s_m} +  \sigma^4 u^2 (-1) \frac{s_{l-1}}{s_l} \biggr] \\
&= \sigma^2 (1 + u(-1)^2) - \sigma^{2} \biggl(  1-\frac{1}{s_m} \biggr)- \sigma^{2} \biggl(  1-\frac{1}{s_m} \biggr) ,
\end{align*}
and
\begin{equation*}
 n =  \bigl[{z^-}^{\rm T}(\sigma^2 A_{m} A_{m}^{\rm T})^{-1} {v^-} + {z^+}^{\rm T}(\sigma^2 A_{l} A_{l}^{\rm T})^{-1} {v^+}\bigr] .
\end{equation*}
This proves the statement of the lemma.
\end{proof}
The case of Lemma \ref{lemma:gaussian1} where either $m$ or $l$ equals zero can be proven analogously and is indeed contained in the statement of Lemma \ref{lemma:gaussian1} in the sense that $s_0 = 1$. However, to avoid confusion let us reformulate the case $m = 0$.
\begin{lemma} \label{lemma:gaussian2}
 Let $d = 1$, $l \geq 1$, $\Theta = \{-1,0\}$, $u(0) = 1$, $\rho$ be the Gaussian density with mean zero and variance $\sigma^2$ and $v \in \RR^{l}$. Then the distribution of $V_\omega (0)$ conditioned on $(V_\omega (k))_{k=1}^{l} = v$ is Gaussian with variance
\begin{equation*}
 \gamma = \sigma^2 \Bigl( u(-1)^2 + \frac{1}{s_l} \Bigr) \quad \text{and mean} \quad   m = u(-1) \sum_{i=1}^{l} (A_{l} A_{l}^{\rm T})^{-1} (1,i) \, v_i .
\end{equation*}
\end{lemma}
Note that the variance from Lemma \ref{lemma:gaussian2} is bounded from below uniformly even if $u(-1)^2 = 1$, while the variance from Lemma  \ref{lemma:gaussian1} tends to zero if $u(-1)^2 = 1$. 
\begin{remark}
We want to discuss the validity of the regularity assumption from \cite{AizenmanFSH-01} in the case $d = 1$, $\Theta = \{-1,0\}$, $u(0) = 1$ and $\rho$ the Gaussian density function with mean zero and variance $\sigma^2$. Notice that the Gaussian distribution is $\tau$-H\"o{}lder continuous with a constant $C$ independent on the mean but depending on the variance, and the property that $C \to \infty$ if the variance tends to zero.
\par
Let $l,m \geq 1$. If $\lvert u(-1) \rvert \not = 1$, Lemma \ref{lemma:gaussian1} and Lemma \ref{lemma:gaussian2} give that the distribution of $V_\omega (0)$ conditioned on fixed potential values $V_\omega (k)$, $k \in \{-m,\dots,l\} \setminus \{0\}$, is again Gaussian with variance bounded from below by $\sigma^2 \lvert u^2 (-1) - 1 \rvert$. 
As a consequence, the random field $V_\omega (k)$, $k \in \{-m+1,\dots,n-1\}$ is uniformly $\tau$-H\"o{}lder continuous and the constant $C$ from Definition \ref{def:conditional} may be chosen independently from $m,l \in \NN$. Hence the method from \cite{AizenmanFSH-01} applies and gives localization.
\par
If $\lvert u(-1) \rvert = 1$, the situation is somehow different. In this case Lemma \ref{lemma:gaussian1} and Lemma \ref{lemma:gaussian2} give that the random field $V_\omega (k)$, $k \in \Lambda_L = \{-L,\dots,L\}$, $L \in \NN$, satisfies
\[
 \sup_{x \in \Lambda_L} \sup_{v \in \RR^{2L}} \PP(\{V_\omega (x) \in [a,b]\} \mid \{V_\omega (k) = v_k, k \in \Lambda_L \setminus \{x\}\}) \le C_L (b-a)^\tau 
\]
but the constant $C_L$ cannot be chosen uniformly in $L \in \NN$. In particular, $C_L \to \infty$ if $L \to \infty$. As a consequence, the method of \cite{AizenmanFSH-01} will give a bound on the expectation of $\lvert G_{\Lambda_L} (z;i,j) \rvert^s$ which depends on the volume of $\Lambda_L$, and hence does not immediately yield localization. If one considers finite volume restriction $H_{\Lambda_L}$, an analogue condition to Definition \ref{def:conditional} which is sufficient for localization would be the following. There is a $\tau \in(0,1]$ and a constant $C$ such that
\[
\sup_{L \in \NN} \sup_{x \in \Lambda_L} \sup_{v \in \RR^{2L}} \PP(\{V_\omega (x) \in [a,b]\} \mid \{V_\omega (k) = v_k, k \in \Lambda_L \setminus \{x\}\}) \le C(b-a)^\tau 
\]
for all $[a,b] \subset \RR$. This condition is obviously not satisfied if $\lvert u(-1) \rvert = 1$ by Lemma~\ref{lemma:gaussian1} and Lemma~\ref{lemma:gaussian2}.
\end{remark}
\chapter{Fractional moment method for discrete alloy-type models} \label{chap:fmm}
In this chapter we show exponential localization for the discrete alloy-type model in the strong disorder regime under Assumption~\ref{ass:monotone}. In a first section we provide certain tools needed for a proof of localization according to the fractional moment method. In Section~\ref{sec:d=1} we restrict ourselves to the special case $d=1$, since the restriction to the one-dimensional case allows an elegant and short proof of localization. Moreover, we have an explicit control over the constants, for instance on the assumption on the disorder to prove localization. The results presented in Section \ref{sec:d=1} concern a joint work with Alexander Elgart and Ivan Veseli\'c and have already been published in \cite{ElgartTV-10}. The reader who is interested in the theory for arbitrary $d\in\NN$ can directly jump to Section~\ref{sec:boundedness}. 
\par
In Section~\ref{sec:boundedness} we start to develop the theory for arbitrary dimension and prove the boundedness of an averaged fractional power of the Green function. Via a decoupling argument we show in Section~\ref{sec:exp_decay} a so-called finite volume criterion. Together with the fractional moment bound from Section~\ref{sec:boundedness} this gives exponential decay of fractional moments if the disorder is sufficiently large. In Section~\ref{sec:loc} we show that the exponential decay of fractional moments implies exponential localization. The results from Sections~\ref{sec:boundedness} to \ref{sec:loc} concern a joint work with Alexander Elgart and Ivan Veseli\'c and have already been published in \cite{ElgartTV-11}.
\section{Spectral averaging and Schur complement} \label{sec:spectral_averaging}
One important step in a proof of localization according to the fractional moment method is the boundedness of an averaged fractional power of the Green function. In order to prove this, several estimates on averages of resolvents and determinants are useful. 
\par
Let us first motivate the problems arising from the sign-indefiniteness of the single-site potential. To this end, we first consider the sign-definite case where $u(0) = 1$ and $u (k) = 0$ for $k \in \ZZ^d \setminus \{0\}$, and where the measure $\nu$ has a probability density $\rho \in L^{\infty} (\RR)$. By rank one perturbation one can show for all $x \in \ZZ^d$ and $z \in \CC \setminus \RR$ that
\begin{equation} \label{eq:represent1}
 G_\omega (z; x,x) = \frac{\lambda^{-1}}{\omega_x + \hat G_\omega (z;x,x)^{-1}},
\end{equation}
where $\hat G_\omega (z;x,x) = \langle \delta_x , (\hat H_\omega - z)^{-1} \delta_y  \rangle$ and where $\hat H_\omega$ is obtained from $H_\omega$ by setting the random variable $\omega_x$ to zero. If we take the absolute value to the power $s \in (0,1)$ of this identity, and average with respect to $\omega_x$, we obtain
\begin{equation*} \label{eq:rank_one}
 \int_\RR \lvert G_\omega (z; x,x) \rvert^s \rho (\omega_x) \drm \omega_x = \lambda^{-s} \int_\RR \frac{1}{\lvert \omega_x + \hat G_\omega (z;x,x)^{-1} \rvert^s} \rho (\omega_x) \drm \omega_x .
\end{equation*}
Since $\beta = \hat G_\omega (z;x,x)$ is independent of $\omega_x$, one can estimate the integral in Eq.~\eqref{eq:rank_one} by a constant uniform in $\omega_k$, $k \in \ZZ^d \setminus \{x\}$. Roughly speaking, the pole at $-\beta$ is integrable since $s<1$ and $\rho$ is bounded, and the integrand is integrable at infinity since $\rho \in L^1 (\RR)$. More precisely, let $g:\RR \to \RR$ be non-negative with $g \in L^{\infty}(\RR) \cap L^1(\RR)$ and $s \in (0,1)$. Then we have for all $\beta \in \CC$
\begin{equation} \label{eq:graf}
\int_\RR \abs{\xi - \beta}^{-s} g(\xi) \drm \xi \leq \norm{g}_\infty^s \norm{g}_{L^1}^{1-s} \frac{2^s s^{-s}}{1-s} ,
\end{equation}
see e.g.\ \cite{Graf-94} for a proof. Hence, 
\[
 \EE_{\{x\}} \bigl(\lvert G_\omega (z; x,x) \rvert^s \bigr) \leq \lambda^{-s} \norm{\rho}_\infty^s \frac{2^s s^{-s}}{1-s} .
\]
Now assume $d=1$ and that $\Theta = \supp u = \{0,1\}$. Using the Schur complement formula, one can show an analogue of Eq.~\eqref{eq:represent1}, namely
\begin{equation} \label{eq:represent2}
 \begin{pmatrix}
  G_\omega (z; x,x) & G_\omega (z; x,x+1) \\
  G_\omega (z; x+1,x) & G_\omega (z; x+1,x+1)
 \end{pmatrix} = \left[A + \lambda \omega_x \begin{pmatrix} u(0) & 0 \\ 0 & u(1) \end{pmatrix}  \right]^{-1} ,
\end{equation}
where $A$ is some $2\times 2$ matrix which is independent of $\omega_x$. If $u(0)>0$ and $u(1)>0$ there is a multidimensional analogue of Eq.~\eqref{eq:graf}, i.e.
\begin{align*} 
 \int_{\RR} \biggl\lVert \left[A + \lambda \omega_x \begin{pmatrix} u(0) & 0 \\ 0 & u(1) \end{pmatrix}  \right]^{-1} \biggr\rVert^s \rho (\omega_k) \drm \omega_k \leq \left(\frac{2 C_{\rm W} \lVert \rho \rVert_\infty}{\min \{u(0),u(1)\} \lambda}\right)^s \frac{1}{1-s} ,
\end{align*}
where $C_{\rm W}$ is a constant independent of $A$, $u$ and $\lambda$, see Lemma~\ref{lemma:monotone2} below. This implies
\[
 \EE_{\{x\}} \bigl(\lvert G_\omega (z; x,x) \rvert^s \bigr) \leq \left(\frac{2 C_{\rm W} \lVert \rho \rVert_\infty}{\min \{u(0),u(1)\} \lambda}\right)^s \frac{1}{1-s} .
\]
These two examples show how to prove the boundedness of an averaged fractional power of some Green's function elements if the single-site potential is non-negative. If the single-site potential changes its sign, the diagonal matrix on the right hand side of Eq.~\eqref{eq:represent2} is no longer positive definite, and as a consequence Lemma~\ref{lemma:monotone2} is not applicable. For this reason one has to develop new techniques for non-monotone spectral averaging in order to prove the boundedness of Green's function. 
\par
In the one-dimensional situation it turns out that certain matrix elements of the resolvent may be represented as an inverse of a determinant. To estimate the expectation of a fractional power of these Green's function elements, we will use the following averaging lemma for determinants. More precisely, if $\Theta = \{0,1,\ldots , n-1\}$ we will apply Lemma~\ref{lemma:det}, while for arbitrary finite $\Theta$ we shall need a slight extension formulated in Lemma \ref{lemma:detgen}.
\begin{lemma} \label{lemma:det}
 Let $n \in \NN$ and $A, V \in \CC^{n \times n}$ be two matrices and assume that $V$ is invertible. Let further $0 \leq \rho \in L^1(\RR) \cap L^\infty (\RR)$ and $s \in (0,1)$. Then we have for all $\lambda > 0$ the bound
\begin{align}
 \int_{\RR} \abs{\det (A + rV)}^{-s/n} \rho (r) \drm r
&\leq \abs{\det V}^{-s/n} \Vert \rho \Vert_{L^1}^{1-s} \Vert\rho\Vert_{\infty}^{s} \frac{2^{s} s^{-s}}{1-s} \label{eq:det1} \\[1ex]
&\leq \abs{\det V}^{-s/n}\Bigl( \lambda^{-s} \Vert\rho\Vert_{L^1} + \frac{2 \lambda^{1-s}}{1-s} \Vert\rho\Vert_\infty  \Bigr) \label{eq:det2} .
\end{align}
\end{lemma}
\begin{proof}
 Since $V$ is invertible, the function $r \mapsto \det (A + rV)$ is a polynomial of order $n$ and thus the set $\{r \in \mathbb{R}  \colon A + rV \text{ is singular}\}$ is a discrete subset of $\mathbb{R}$ with Lebesgue measure zero. We denote the roots of the polynomial by $z_1,\dots , z_n \in \CC$. By multilinearity of the determinant we have
\[
 \abs{\det (A + rV)} = \abs{\det V} \prod_{j=1}^n \lvert r - z_j\rvert \geq
 \abs{\det V} \prod_{j=1}^n \lvert r - \re{z_j}\rvert .
\]
The H\"older inequality implies for $s \in (0,1)$ that
\begin{equation*}
 \int_\RR \abs{\det (A + rV)}^{-s/n} \rho (r) \drm r \leq \abs{\det V}^{-s/n} \prod_{j = 1}^n \left( \int_\RR \lvert r - \re z_j\rvert^{-s} \rho (r) \drm r  \right)^{1/n} .
\end{equation*}
For arbitrary $\lambda > 0$ and all $z \in \RR$ we have
\begin{align*}
\int_\RR \frac{1}{\abs{r - z}^{s}} \rho(r)\drm r &=  \int\limits_{\abs{r - z} \geq \lambda} \frac{1}{\abs{r - z}^{s}} \rho(r)\drm r + \int\limits_{\abs{r - z} \leq \lambda} \frac{1}{\abs{r - z}^{s}} \rho(r)\drm r \\[1ex]
& \leq \lambda^{-s} \Vert\rho\Vert_{L^1} + \Vert\rho\Vert_\infty \frac{2 \lambda^{1-s}}{1-s} ,
\end{align*}
which gives Ineq.~\eqref{eq:det2}. We now choose $\lambda = s \Vert \rho \Vert_{L^1} / (2 \Vert \rho \Vert_\infty)$ (which minimizes the right hand side of Ineq.~\eqref{eq:det2}) and obtain Ineq.~\eqref{eq:det1}.
\end{proof}
\begin{lemma} \label{lemma:detgen}
Let $N,n \in \NN$ and $A, V_0, V_1, \dots,V_N \in \CC^{n \times n}$ be matrices.
Let $(\alpha_k)_{k=0}^N \in \RR^{N+1}$ with $\alpha_0 \not = 0$. Assume that $\sum_{k=0}^N \alpha_k V_k$ is invertible. Let further $0 \leq \rho \in L^1(\RR) \cap L^\infty (\RR)$ with $\supp \rho \subset [-R,R]$, $R>0$, $\Vert\rho\Vert_{L^1} = 1$ and $t \in (0,1)$.
Then,
\begin{align*}
I &= \int_{\RR^{N+1}} \Bigl|\det \Bigl(A + \sum_{i=0}^N r_iV_i\Bigr)\Bigl|^{-t/n} \prod_{i=0}^N\rho (r_i) \drm r_i  \\
&\leq \Bigl| \det \Bigl( \sum_{k=0}^N \alpha_k V_k \Bigr)\Bigr|^{-t/n} \lvert \alpha_0\rvert^t\Bigl(1+ \max_{i\in \{1,\dots,N\}} \frac{\lvert\alpha_i\rvert}{\lvert \alpha_0\rvert} \Bigr)^{Nt}
\,\, \frac{2^t t^{-t}}{1-t} (2R)^{Nt} \Vert\rho\Vert_\infty^{(N+1)t} .
\end{align*}
\end{lemma}
\begin{proof}
 Substituting
\[
 \begin{pmatrix}
  r_0 \\[0.8ex]
  r_1 \\
  \vdots \\
  \vdots \\
  r_N
 \end{pmatrix} =
T \begin{pmatrix}
  x_0 \\[.2ex]
  x_1 \\[.2ex]
  \vdots \\[.2ex]
  \vdots \\[.2ex]
  x_N
 \end{pmatrix}
=
 \begin{pmatrix}
  \alpha_0 & 0        & \cdots  &\cdots   &  0     \\
  \alpha_1 & \alpha_0 &0        &         & \vdots \\
  \alpha_2 & 0        & \alpha_0&\ddots   & \vdots \\
  \vdots   & \vdots   & \ddots  &\alpha_0 & 0      \\
  \alpha_N & 0        & \hdots  & 0       & \alpha_0
 \end{pmatrix}
 \begin{pmatrix}
  x_0 \\[.2ex]
  x_1 \\[.2ex]
  \vdots \\[.2ex]
  \vdots \\[.2ex]
  x_N
 \end{pmatrix}
=
\begin{pmatrix}
  \alpha_0 x_0   \\[.2ex]
  \alpha_1 x_0 + \alpha_0 x_1 \\[.2ex]
  \alpha_2 x_0 + \alpha_0 x_2 \\[.5ex]
  \vdots     \\[.5ex]
  \alpha_N x_0 + \alpha_0 x_N
 \end{pmatrix}
\]
we obtain
\[
 I = \int_{\RR^{N}} \left(  \int_\RR \Bigl|\det \Bigl(\tilde A +  x_0\sum_{i=0}^N \alpha_i V_i\Bigr)\Bigr|^{-t/n} g(x_0,\dots,x_N) \drm x_0 \right) \lvert \alpha_0\rvert^{N+1} \drm x_1  \dots  \drm x_N
\]
where $\tilde A = A + \alpha_0 \sum_{i=1}^N x_i V_i$ and
$g(x_0,\dots,x_N) = \rho(\alpha_0 x_0)\prod_{i=1}^N \rho(\alpha_i x_0 + \alpha_0 x_i)$.
Since $x_0 \mapsto g(x_0,$ $\dots,x_N)$ is an element of $L^1(\RR) \cap L^\infty (\RR)$ we may apply Lemma \ref{lemma:det} and obtain for all $\lambda > 0$
\begin{align*}
 I &\leq \Bigl|\det \Bigl(\sum_{i=0}^N \alpha_i V_i\Bigr)\Bigr|^{-t / n}
  \int_{\RR^N} \Bigl(  \int_{\RR} \frac{g(x_0,x)}{\lambda^{t}} \drm x_0 + \frac{2 \lambda^{1-t}}{1-t} \sup_{x_0 \in \RR} g(x_0,x) \Bigr) \vert\alpha_0\rvert^{N+1}  \drm x \\[1ex]
&= \Bigl|\det \Bigl(\sum_{i=0}^N \alpha_i V_i\Bigr)\Bigr|^{-t/n} \Bigl(
   \lambda^{-t}  + \frac{2 \lambda^{1-t}}{1-t} \int_{\RR^N} \sup_{x_0 \in \RR} g(x_0,x) \lvert\alpha_0\rvert^{N+1}  \drm x \Bigr)
\end{align*}
with $\drm x = \drm x_1  \dots  \drm x_N$ and $g(x_0,x) = g(x_0,x_1,\ldots,x_N)$. We use $\supp \rho \subset [-R,R]$ and see that if $\lvert x_j\rvert > R \,\Vert T^{-1}\Vert_\infty$ holds 
for some $j = 0,\dots,N$, then we have $g(x_0,\dots,x_N) = 0$. Thus it is sufficient to integrate over the cube $[-R\Vert T^{-1}\Vert_\infty,R\Vert T^{-1}\Vert_\infty]^N$.
We estimate $\sup_{x_0 \in \RR} g(x_0,\dots,x_N) \leq \Vert \rho \Vert_\infty^{N+1}$ and choose $\lambda = t/(2$ $\Vert\rho\Vert_\infty^{N+1} \lvert\alpha_0^{N+1}\rvert (2R$ $\Vert T^{-1}_\infty \Vert)^N)$.
The row-sum norm of $T^{-1}$ equals $\Vert T^{-1} \Vert_\infty = \max_{i \in \{1,\dots,N\}}$ $(\abs{\alpha_0}^{-1}$ $+ \lvert\alpha_i/\alpha_0^2\rvert)$.
\end{proof}
Lemma~\ref{lemma:det} can be used to obtain the following averaging result for resolvents, which will be an important tool for the multi-dimensional case.
\begin{lemma}\label{lemma:averagenorm}
 Let $n \in \mathbb{N}$, $A \in \mathbb{C}^{n \times n}$ an arbitrary matrix, $V \in \mathbb{C}^{n \times n}$ an invertible matrix and $s \in (0,1)$. Let further $0 \leq \rho \in L^1(\RR) \cap L^\infty (\RR)$ with $\supp \rho \subset [-R,R]$ for some $R>0$. Then we have the bounds
\begin{equation} \label{eq:norm_estimate}
\Vert V^{-1} \Vert \leq \frac{\Vert V \Vert^{n-1}}{\abs{\det V}}
\end{equation}
and
\begin{equation} \label{eq:average_norm}
\int_{-R}^R \bigl\Vert (A+rV)^{-1} \bigr\Vert^{s/n} \rho(r) {\rm d}r \leq \frac{\lVert\rho\rVert_{L^1}^{1-s} \lVert\rho\rVert_\infty^s (\Vert A \Vert + R \Vert V \Vert)^{s(n-1)/n}}{s^s 2^{-s} (1-s) \lvert \det V \rvert^{s/n}} .
\end{equation}
\end{lemma}
\begin{proof}
To prove Ineq.~\eqref{eq:norm_estimate} let $0< s_1 \leq s_2 \leq \ldots \leq s_n$ be the singular values of $V$. Then we have $\prod_{i=1}^n s_i \leq s_1 s_n^{n-1}$, that is,
\begin{equation}\label{eq:2}
 \frac{1}{s_1} \leq \frac{s_n^{n-1}}{\prod_{i=1}^n s_i} .
\end{equation}
For the norm we have $\Vert V^{-1} \Vert = 1/s_1$ and $\Vert V \Vert = s_n$. For the determinant of $V$ we have $\lvert \det V \rvert = \prod_{i=1}^n s_i$. Hence, Ineq.~\eqref{eq:norm_estimate} follows from Ineq.~\eqref{eq:2}. To prove Ineq.~\eqref{eq:average_norm} recall that, since $V$ is invertible, the set $\{r \in \RR \colon \text{$A+rV$ is singular}\}$ is a discrete set. Thus, for almost all $r \in [-R,R]$ we may apply Ineq.~\eqref{eq:norm_estimate} to the matrix $A+rV$ and obtain
\[
 \bigl\Vert (A+rV)^{-1} \bigr\Vert^{s/n} \leq \frac{(\Vert A \Vert + R \Vert V \Vert)^{s(n-1)/n}}{\lvert \det (A+rV) \rvert^{s/n}} .
\]
Inequality \eqref{eq:average_norm} now follows from Lemma \ref{lemma:det}.
\end{proof}
The assumption that the single-site potential $u$ is monotone at the boundary (cf. Assumption \ref{ass:monotone}) allows us to use monotone spectral averaging at some stage. For this purpose we cite a special case of \cite[Proposition 3.1]{AizenmanENSS-06}. Recall, a densely defined operator $T$ on some Hilbert space $\mathcal{H}$ with inner product $\langle \cdot , \cdot \rangle_{\mathcal{H}}$ is called \emph{dissipative} if $\im \langle x,Tx \rangle_{\mathcal{H}} \geq 0$ for all $x \in D(T)$.
\begin{lemma} \label{lemma:monotone}
Let $A \in \CC^{n \times n}$ be a dissipative matrix, $V \in \RR^{n \times n}$ diagonal and strictly positive definite and $M_1 , M_2 \in \CC^{n \times n}$ be arbitrary matrices. Then there exists a constant $C_{\rm W}$ (independent of $n$, $A$, $V$, $M_1$ and $M_2$), such that
\[
\mathcal{L}  \bigl(\bigl\{ r \in \RR : \lVert M_1 (A + r V)^{-1} M_2 \rVert_{\rm HS} > t \bigr\}\bigr) \leq C_{\rm W} \lVert M_1 V^{-1/2} \rVert_{\rm HS} \lVert M_2 V^{-1/2} \rVert_{\rm HS} \frac{1}{t} .
\]
Here, $\mathcal{L}$ denotes the Lebesgue-measure and $\lVert \cdot \rVert_{\rm HS}$ the Hilbert Schmidt norm.
\end{lemma}
As a corollary we have
\begin{lemma} \label{lemma:monotone2}
Let $A \in \CC^{n \times n}$ be a dissipative matrix, $V \in \RR^{n \times n}$ diagonal and strictly positive definite, $M_1 , M_2 \in \CC^{n \times n}$ be arbitrary matrices and $\rho \in L^\infty (\RR) \cap L^1 (\RR)$ non-negative with $\lVert \rho \rVert_{L^1} = 1$. Then there exists a constant $C_{\rm W}$ (independent of $A$, $V$, $M_1$ and $M_2$), such that
\[
\int_\RR \lVert M_1 (A + r V)^{-1} M_2 \rVert^s \rho (r) \drm r \leq
\frac{(n C_{\rm W} \lVert M_1 V^{-1/2} \rVert \lVert M_2 V^{-1/2} \rVert \lVert \rho \rVert_\infty )^s}{1-s} .
\]
\end{lemma}
\begin{proof}
 First note that for a matrix $T \in \CC^{n \times n}$ we have $\lVert T \rVert \leq \lVert T \rVert_{\rm HS} \leq \sqrt{n} \lVert T \rVert$. With the use of the layer cake representation, 
see e.g.\ \cite[p.\ 26]{LiebL2001}, and Lemma \ref{lemma:monotone} we obtain for all $\kappa > 0$
\begin{align*}
I &= \int_\RR \lVert M_1 (A + r V)^{-1} M_2 \rVert^s \rho (r) \drm r
= \int_0^\infty \int_\RR \mathbf{1}_{\{\lVert M_1 (A + r V)^{-1} M_2 \rVert^s > t \}} \rho (r) \drm r \drm t \\[1ex]
& \leq \kappa + \int_\kappa^\infty \lVert \rho \rVert_\infty n C_{\rm W} \lVert M_1 V^{-1/2} \rVert \lVert M_2 V^{-1/2} \rVert \frac{1}{t^{1/s}} \drm t \\[1ex]
& = \kappa + \lVert \rho \rVert_\infty n C_{\rm W} \lVert M_1 V^{-1/2} \rVert \lVert M_2 V^{-1/2} \rVert \frac{s}{1-s} \kappa^{(s-1)/s} .
\end{align*}
If we choose $\kappa = (\lVert \rho \rVert_\infty n C_{\rm W} \lVert M_1 V^{-1/2} \rVert \lVert M_2 V^{-1/2} \rVert)^s$ we obtain the statement of the lemma.
\end{proof}
The above estimates on spectral averaging concern finite-dimensional matrices only. In order to use these estimates for our infinite-dimensional operator $G_\omega (z)$, we will use a variant of the Schur complement formula (also known as the Feshbach formula or Grushin problem), cf.\ \cite[Appendix]{BellissardHS-07}.
\begin{lemma} \label{lemma:schur1}
Let $\Lambda \subset \Gamma \subset \ZZ^d$ and $\Lambda$ finite. Then we have for all $z \in \CC\setminus\RR$ the identity
\begin{equation*}
 \Pro_{\Lambda}^\Gamma(H_\Gamma - z)^{-1} \Inc_{\Lambda}^\Gamma = \bigl(H_{\Lambda} - B_\Gamma^\Lambda - z \bigr)^{-1} ,
\end{equation*}
where $B_\Gamma^\Lambda : \ell^2 (\Lambda) \to (\Lambda)$ is specified in Eq.~\eqref{eq:bij}. Moreover, the operator $B_\Gamma^\Lambda$ is independent of $V_\omega (k)$, $k \in \Lambda$.
\end{lemma}
\begin{proof}
We consider $H_\Gamma - z$ as the block operator matrix 
\[
H_\Gamma - z =
 \begin{pmatrix}
  H_\Lambda - z  & -\Pro_{\Lambda}^\Gamma\Delta_\Gamma \Inc_{\Gamma \setminus \Lambda}^\Gamma \\
  -\Pro_{\Gamma \setminus \Lambda}^\Gamma \Delta_\Gamma \Inc_{\Lambda}^\Gamma  & H_{\Gamma\setminus \Gamma} - z
 \end{pmatrix} .
\]
Since $\Lambda$ is finite, $H_\Lambda$ is bounded and the Schur complement formula gives
\begin{equation*}
 \Pro_\Lambda^\Gamma(H_\Gamma - z)^{-1} \Inc_\Lambda^\Gamma =
 \Bigl[ H_\Lambda-z - \Pro_\Lambda^\Gamma \Delta_\Gamma \Inc_{\Gamma\setminus\Lambda}^\Gamma \bigl(H_{\Gamma \setminus \Lambda} - z \bigr)^{-1} \Pro_{\Gamma\setminus\Lambda}^\Gamma \Delta_\Gamma \Inc_{\Lambda}^\Gamma \Bigr]^{-1} ,
\end{equation*}
compare, e.g., \cite[Appendix]{BellissardHS-07}. For $\Lambda \subset \Gamma \subset \ZZ^d$ we define
\begin{subequations} \label{eq:bij}
\begin{equation} \label{eq:bij1}
 B_\Gamma^\Lambda := \Pro_\Lambda^\Gamma \Delta_\Gamma \Inc_{\Gamma\setminus\Lambda}^\Gamma \bigl(H_{\Gamma \setminus \Lambda} - z \bigr)^{-1} \Pro_{\Gamma\setminus\Lambda}^\Gamma \Delta_\Gamma \Inc_{\Lambda}^\Gamma .
\end{equation}
For the matrix elements of $B_\Gamma^\Lambda$ one calculates
\begin{equation} \label{eq:bij2}
\sprod{\delta_x}{B_\Gamma^\Lambda \delta_y} =
\begin{cases}
   \quad \!\!0  & \text{if $x \not \in \partial^{\rm i} \Lambda \vee y \not \in \partial^{\rm i} \Lambda$,} \\[1ex]
   \sum\limits_{\genfrac{}{}{0pt}{2}{k \in \Gamma\setminus\Lambda :}{\abs{k-x} = 1} }
   \sum\limits_{\genfrac{}{}{0pt}{2}{l \in \Gamma\setminus\Lambda :}{\abs{l-y} = 1} } G_{\Gamma\setminus\Lambda}(z;k,l) & \text{if $x \in \partial^{\rm i} \Lambda \wedge y \in \partial^{\rm i} \Lambda$.}
\end{cases}
\end{equation}
\end{subequations}
$G_{\Gamma\setminus\Lambda}$ is independent of $V_\omega (k)$, $k\in \Lambda$. Thus it is also $B_\Gamma^\Lambda$.
\end{proof}
\begin{lemma}\label{lemma:schur2}
 Let $\Gamma \subset \ZZ^d$ and $\Lambda_1 \subset \Lambda_2 \subset \Gamma$. We assume that $\Lambda_1$ and $\Lambda_2$ are finite sets and that $(\partial^{\rm i} \Lambda_2) \cap \Lambda_1 = \emptyset$. Then we have for all $z \in \CC\setminus\RR$ the identity
\begin{multline*}
\Pro_{\Lambda_1}^\Gamma(H_\Gamma - z)^{-1} \Inc_{\Lambda_1}^\Gamma
= \Bigl[ H_{\Lambda_1} - z \\ - \Pro_{\Lambda_1} \Delta \Inc_{\Lambda_2 \setminus \Lambda_1}
\Bigl(H_{\Lambda_2 \setminus \Lambda_1} - z - \Pro_{\Lambda_2 \setminus \Lambda_1}^{\Lambda_2} B_\Gamma^{\Lambda_2} \Inc_{\Lambda_2 \setminus \Lambda_1}^{\Lambda_2}\Bigr)^{-1}
\Pro_{\Lambda_2 \setminus \Lambda_1} \Delta \Inc_{\Lambda_1}
\Bigr]^{-1}.
\end{multline*}
\end{lemma}
\begin{proof}
We decompose $\Lambda_2 = \Lambda_1 \cup (\Lambda_2 \setminus \Lambda_1)$ and notice that $\langle \delta_x , B_\Gamma^{\Lambda_2} \delta_y \rangle = 0$ if $x \in \Lambda_1$ or $y \in \Lambda_1$ by Eq.~\eqref{eq:bij2} and our hypothesis $(\partial^{\rm i} \Lambda_2) \cap \Lambda_1 = \emptyset$. Due to this decomposition we write $H_{\Lambda_2} - z - B_{\Gamma}^{\Lambda_2}$ as the block operator matrix
\[
H_{\Lambda_2} - z - B_{\Gamma}^{\Lambda_2} =
\begin{pmatrix}
  H_{\Lambda_1} - z & -\Pro_{\Lambda_1} \Delta \Inc_{\Lambda_2 \setminus \Lambda_1} \\[2ex]
  -\Pro_{\Lambda_2 \setminus \Lambda_1} \Delta \Inc_{\Lambda_1} & H_{\Lambda_2 \setminus \Lambda_1} - z - \Pro_{\Lambda_2 \setminus \Lambda_1}^{\Lambda_2} B_\Gamma^{\Lambda_2} \Inc_{\Lambda_2 \setminus \Lambda_1}^{\Lambda_2}
\end{pmatrix} .
\]
The Schur complement formula gives $\Pro_{\Lambda_1}^{\Lambda_2}(H_{\Lambda_2} - z - B_{\Gamma}^{\Lambda_2})^{-1} \Inc_{\Lambda_1}^{\Lambda_2} = S^{-1}$ where $S$ equals
\[
 H_{\Lambda_1} - z - \Pro_{\Lambda_1} \Delta \Inc_{\Lambda_2 \setminus \Lambda_1}
\bigl(H_{\Lambda_2 \setminus \Lambda_1} - z - \Pro_{\Lambda_2 \setminus \Lambda_1}^{\Lambda_2} B_\Gamma^{\Lambda_2} \Inc_{\Lambda_2 \setminus \Lambda_1}^{\Lambda_2} \bigr)^{-1}
\Pro_{\Lambda_2 \setminus \Lambda_1} \Delta \Inc_{\Lambda_1} .
\]
Since $\Pro_{\Lambda_1}^{\Lambda_2}(H_{\Lambda_2} - z - B_{\Gamma}^{\Lambda_2})^{-1} \Inc_{\Lambda_1}^{\Lambda_2} = \Pro_{\Lambda_1}^\Gamma (H_\Gamma - z)^{-1} \Inc_{\Lambda_1}^\Gamma$ by Lemma~\ref{lemma:schur1}, we obtain the statement of the lemma.
\end{proof}
\section{The one-dimensional case}\label{sec:d=1}
In this section we prove exponential localization for the discrete alloy-type model under Assumptions \ref{ass:d=1} and \ref{ass:d=1_2}. Although Assumption \ref{ass:d=1} is more restrictive than Assumption \ref{ass:d=1_2}, we first prove the exponential decay of fractional moments under Assumption \ref{ass:d=1} in Section \ref{sec:boundedness_d=1} and \ref{sec:exp_decay_d=1}, since the spectral averaging step is more transparent and the constants are more explicit. In Section~\ref{sec:general_single_site} we generalize this result to general finitely supported single-site potentials.
\par
In Section \ref{sec:loc_d=1} we conclude exponential localization in the strong disorder regime by using the results from Section~\ref{sec:loc}, where the implication from exponential decay of fractional moments to exponential localization is proven for arbitrary dimension $d \in \NN$.
\par
In the one-dimensional situation it is natural to ask whether it is possible to extract a positive part from the random potential in such a way, that the original methods (with monotone spectral averaging) for deriving fractional moment bounds apply. It turns out that this is not possible in general (even in one space dimension) but that the corresponding class of single-site potentials can be characterized in the following way.
If the polynomial $p_u(x): = \sum_{k=0}^{n-1} u(k) \, x^k$ does not vanish on $[0, \infty)$ it is possible to extract from $V_\omega$ a positive single-site potential with certain additional properties. In this situation the method of \cite{AizenmanENSS-06} applies and gives exponential decay of fractional moments of the Green function. This is worked out in detail in Section~\ref{sec:red_monotone}.
\par
Let us emphasize that our proof of exponential decay of fractional moments in the one-dimensional case nowhere uses monotonicity in the sense of a monotone spectral averaging. The main tool which allows us to get along without the use of any monotonicity property is an averaging result for determinants formulated in Lemma~\ref{lemma:det}.
\subsection{Boundedness of fractional moments} \label{sec:boundedness_d=1}
A very useful information in the one-dimensional case is that certain matrix elements of the resolvent are given by the inverse of a determinant, which makes Lemma~\ref{lemma:det} and Lemma~\ref{lemma:detgen} applicable. In order to work out this specific structure we will apply the Schur complement formula as given in Lemma~\ref{lemma:schur1}.
\par
A set $\Gamma \subset \ZZ$ is called connected if $\partial^{\rm i} \Gamma \subset \{\inf \Gamma , \sup \Gamma\}$. In particular, $\ZZ$ is a connected set.
\begin{lemma} \label{lemma:finitness1}
Let Assumption \ref{ass:d=1} be satisfied, $s \in (0,1)$, and $\Gamma \subset \ZZ$ be connected. Then,
\begin{itemize}
\item[(i)] for every pair $x,x+n-1 \in \Gamma$ and all $z \in \CC \setminus \RR$ we have
\begin{equation} \label{eq:finitness1}
 \mathbb{E}_{\{x\}} \bigl ( \lvert G_\Gamma (z;x,x+n-1)\rvert^{s/n}  \bigr ) \leq  \frac{C_u C_\rho}{\lambda^s} =: C_{\lambda , u , \rho} .
\end{equation}
\item[(ii)] if $1 \leq \abs{\Gamma} \leq n$, we have for  all $z \in \CC \setminus \RR$ the bound
\begin{equation} \label{eq:finitness2}
 \mathbb{E}_{\{\gamma_0\}} \bigl \{ \lvert G_\Gamma (z;\gamma_0,\gamma_1)\rvert^{s/n}  \bigr \} \leq C_u^+ C_\rho^+ \max \{ \lambda^{-s} , \lambda^{-s/n} \} =: C_{\lambda, u,\rho}^+
\end{equation}
where $\gamma_0 = \min \Gamma$ and $\gamma_1 = \max \Gamma$.
\end{itemize}
The constants $C_u$, $C_\rho$, $C_u^+$ and $C_\rho^+$ are given in Eq.~\eqref{eq:cucrho} and \eqref{eq:cucrhoplus}.
\end{lemma}
\begin{proof}
We start with the first statement of the lemma. By assumption $x,x+n-1 \in \Gamma$. We apply Lemma \ref{lemma:schur1} with $\Lambda = \{x, x+1, \dots, x+n-1\}\subset \Gamma$ (since $\Gamma$ is connected) and obtain for all $x,y \in \Lambda$
\[
 G_\Gamma (z;x,y) = \bigl \langle \delta_x , (H_{\Lambda} - B_\Gamma^\Lambda - z)^{-1} \delta_{y} \bigr \rangle ,
\]
where the operator $B_\Gamma^\Lambda$ is given by Eq.~\eqref{eq:bij}.
Set $D = H_{\Lambda} - B_\Gamma^\Lambda - z$ and notice that
 \[
 D = \begin{pmatrix}
V_\omega (x)    & -1      &         &         \\
    -1          &\ddots   & \ddots  &      \\
	        & \ddots  & \ddots  & -1     \\
	        &         & -1      & V_\omega (x+n-1)   \\
     \end{pmatrix} - 
\begin{pmatrix}
b_1 & \phantom{b_1} & \phantom{b_1} & \phantom{b_1} \\
& & &  \\[1ex]
& & &  \\[1ex]
& & &  b_2 \\
\end{pmatrix}
 - z ,
\]
where $b_1 = \langle \delta_x , B_\Gamma^\Lambda \delta_x \rangle$ and $b_2 = \langle\delta_{x+n-1} , B_\Gamma^\Lambda \delta_{x+n-1}\rangle$, and where only the non-zero matrix elements are plotted.
By Cramer's rule we have $G_\Gamma (z;x,y) =  \det C_{y,x} / \det D$. Here, ${C}_{i,j} = (-1)^{i+j} M_{i,j}$ and $M_{i,j}$ is obtained from the tridiagonal matrix $D$ by deleting row $i$ and column $j$. Thus $C_{x+n-1,x}$ is a lower triangular matrix with determinant $\pm 1$. Hence,
\[
 \abs{G_\Gamma (z;x,x+n-1)} = \frac{1}{\abs{\det D}}.
\]
Since $\Theta = \supp u = \{0,\dots,n-1 \}$, every potential value $V_\omega (k)$, $k \in \Lambda$, depends on the random variable $\omega_{x}$, while the operator $B_\Gamma^\Lambda $ is independent of $\omega_x$. Thus we may write $D$ as a sum of two matrices
\[
 D = A + \omega_{x} \lambda V,
\]
where $V \in \RR^{n \times n}$ is diagonal with the elements $u(k-x)$, $k=x,\dots,x+n-1$, and $A := D - \omega_x \lambda V$. Since $A$ is independent of $\omega_x$ we may apply Lemma \ref{lemma:det} and obtain for all $s \in (0,1)$ the estimate \eqref{eq:finitness1} with
\begin{equation} \label{eq:cucrho}
 C_u = \Bigl | \prod_{k \in \Theta} u(k) \Bigr |^{-s/n} \quad \text{and} \quad C_\rho =
 \Vert\rho\Vert_{\infty}^{s} \frac{2^{s} s^{-s}}{1-s} .
\end{equation}
The proof of Ineq.~\eqref{eq:finitness2} is similar but does not require Lemma \ref{lemma:schur1}. We have the decomposition  $H_\Gamma - z = \tilde A + \omega_{\gamma_0} \lambda \tilde V$, where
$m = \gamma_1-\gamma_0$,  $\tilde{V}\in \RR^{(m+1) \times (m+1)}$ is diagonal with elements $u (k - \gamma_0)$, $k = \gamma_0,\dots,\gamma_1$, and $\tilde A :=H_\Gamma - z - \omega_{\gamma_0} \lambda \tilde V$ is independent of $\omega_{\gamma_0}$. By Cramer's rule and Lemma \ref{lemma:det} we obtain for all $t \in (0,1)$
\begin{equation*}
 \mathbb{E}_{\{\gamma_0\}} \bigl ( \lvert G_\Gamma (z;\gamma_0,\gamma_1)\rvert^{t/(m+1)}  \bigr ) \leq
\Bigl | \prod_{k=0}^{m} u(k) \Bigr |^{-t/(m+1)} \Vert\rho\Vert_{\infty}^{t} \lambda^{-t} \frac{2^{t} t^{-t}}{1-t} .
\end{equation*}
We choose $t = s \frac{m+ 1}{n} \leq s$ and obtain Ineq.~\eqref{eq:finitness2} with the constants
\begin{equation} \label{eq:cucrhoplus}
 C_u^+ = \max_{i \in \Theta} \Bigl | \prod_{k=0}^{i} u(k) \Bigr |^{-s/n} \quad \text{and} \quad C_\rho^+ = \max \bigl \{ \Vert\rho\Vert_{\infty}^{s} , \Vert\rho\Vert_{\infty}^{s/n} \bigr \} \frac{2^{s} s^{-s}}{1-s} .
\end{equation}
In the final step we have used $s \geq t$ and the monotonicity of the function $(0,1) \ni x \mapsto 2^x x^{-x}/(1-x)$.
\end{proof}
\subsection{Exponential decay of fractional moments} \label{sec:exp_decay_d=1}
To conclude exponential decay from the boundedness estimates of Lemma \ref{lemma:finitness1}, we use so-called ``depleted'' Hamiltonians to formulate a geometric resolvent formula. Such Hamiltonians are obtained by setting to zero the ``hopping terms'' of the Laplacian along a collection of bonds. More precisely, let $\Lambda \subset \Gamma \subset \ZZ$ be arbitrary sets. We define the depleted Laplace operator $\Delta_\Gamma^\Lambda :\ell^2 (\Gamma) \to \ell^2 (\Gamma)$ by
\begin{equation*} \label{eq:de1}
 \sprod{\delta_x}{\Delta_\Gamma^\Lambda \delta_y} :=
\begin{cases}
  0 & \text{if $x \in \Lambda$, $y \in \Gamma \setminus \Lambda$ or $y \in \Lambda$, $x \in \Gamma \setminus \Lambda$} , \\
  \bigl \langle \delta_x , \Delta_\Gamma \delta_y \bigr \rangle & \text{else} .
\end{cases}
\end{equation*}
In other words, the hopping terms which connect $\Lambda$ with $\Gamma \setminus \Lambda$ or vice versa are deleted. The depleted Hamiltonian $H_\Gamma^\Lambda : \ell^2 (\Gamma) \to \ell^2 (\Gamma)$ is then defined by
\[
 H_\Gamma^\Lambda := -\Delta_\Gamma^\Lambda + V_\Gamma .
\]
Let further $T_\Gamma^\Lambda := \Delta_\Gamma - \Delta_\Gamma^\Lambda$ be the difference between the the ``full'' Laplace operator and the depleted Laplace operator.
Analogously to Eq.~\eqref{eq:greens} we use the notation $G_\Gamma^\Lambda (z) := (H_\Gamma^\Lambda - z)^{-1}$ and $G_\Gamma^\Lambda (z;x,y) := \bigl \langle \delta_x, G_\Gamma^\Lambda(z) \delta_y \bigr \rangle$. The second resolvent identity yields for arbitrary sets $\Lambda \subset \Gamma \subset \ZZ$
\begin{align}
 G_\Gamma (z)   & = G_\Gamma^\Lambda (z) + G_\Gamma (z) T_\Gamma^\Lambda G_\Gamma^\Lambda (z)             \label{eq:resolvent} \\[1ex]
        & = G_\Gamma^\Lambda (z) + G_\Gamma^\Lambda (z)T_\Gamma^\Lambda G_\Gamma (z) .
          \label{eq:resolvent2}
\end{align}
In the following we will use that $G_\Gamma^\Lambda (z;x,y) = G_\Lambda (z;x,y)$ for all $x,y \in \Lambda$ and that $G_\Gamma^\Lambda (z;x,y) = 0$ if $x \in \Lambda$ and $y \not \in \Lambda$ or vice versa. Thes two properties follow from the fact that $H_\Gamma^\Lambda$ is block-diagonal,
\begin{lemma} \label{lemma:iteration1}
 Let Assumption \ref{ass:d=1} be satisfied, $\Gamma \subset \ZZ$ be connected, and $s \in (0,1)$.
 Then we have for all $x,y \in \Gamma$ with $y-x \geq n$, $\Lambda = \{x+n,x+n+1,\dots\}\cap \Gamma$ and all $z \in \CC \setminus \RR$ the bound
\begin{equation*}
 \mathbb{E}_{\{x\}}\bigl(\lvert G_\Gamma (z;x,y)\rvert^{s/n}\bigr) \leq C_{\lambda , u,\rho} \cdot \lvert G_\Lambda (z;x+n,y)\rvert^{s/n} .
\end{equation*}
In particular,
\begin{equation}\label{eq:iteration1}
 \mathbb{E} \bigl(\lvert G_\Gamma (z;x,y)\rvert^{s/n}\bigr) \leq C_{\lambda , u,\rho} \cdot \mathbb{E} \bigl(\lvert G_\Lambda (z;x+n,y)\rvert^{s/n}\bigr) .
\end{equation}

\end{lemma}
\begin{proof}
Our starting point is Eq.~\eqref{eq:resolvent}. Taking the matrix element $(x,y)$ yields
\begin{equation*}
 G_\Gamma (z;x,y) = G_\Gamma^{\Lambda} (z;x,y) + \bigl \langle \delta_x , G_\Gamma (z) T_\Gamma^\Lambda G_\Gamma^\Lambda (z) \delta_y \bigr \rangle .
\end{equation*}
Since $x \not \in \Lambda$ and $y \in \Lambda$, the first summand on the right vanishes as the depleted Green function $G_\Gamma^{\Lambda} (z;x,y)$ decouples $x$ and $y$. For the second summand we calculate
\begin{align}
 G_\Gamma (z;x,y) &= G_\Gamma (z;x,x+n-1) G_\Gamma^\Lambda (z;x+n,y) \nonumber \\
&= G_\Gamma (z;x,x+n-1) G_\Lambda (z;x+n,y) \label{e:geometric-resolvent}  .
\end{align}
The second factor is independent of $\omega_x$. Thus, taking expectation with respect to $\omega_x$ bounds the first factor using Ineq.\ \eqref{eq:finitness1} and the proof is complete.
\end{proof}
\begin{lemma} \label{lemma:iteration2}
 Let Assumption \ref{ass:d=1} be satisfied, $\Gamma = \{x,x+1,...\}$, $y \in \Gamma$ with $n \leq y - x < 2n$, and $s \in (0,1)$. Then we have for all $z \in \CC \setminus \RR$ the bound
\begin{equation} \label{eq:iteration2}
 \mathbb{E}_{\{y-n+1,x\}}\bigl(\lvert G_\Gamma (z;x,y)\rvert^{s/n}\bigr) \leq C_{\lambda , u,\rho}^+ C_{\lambda , u,\rho} .
\end{equation}
\end{lemma}
\begin{proof}
The starting point is Eq.~\eqref{eq:resolvent2}. Choosing $\Lambda = \{x,\dots,y-n\}$ gives analoguously to Eq.~\eqref{e:geometric-resolvent}
\[
 G_\Gamma (z;x,y) = G_\Lambda (z;x,y-n) G_\Gamma (z;y-n+1,y) .
\]
Since $G_\Lambda (z;x,y-n)$ depends only on the potential values at lattice sites in $\Lambda$,
it is independent of $\omega_{y-n+1}$. We take expectation with respect to $\omega_{y-n+1}$ to bound the second factor of the above identity using Ineq.~\eqref{eq:finitness1}.
Since $1 \leq \abs{\Lambda} \leq n$ by assumption, we may apply Ineq.~\eqref{eq:finitness2} to $G_\Lambda (z;x,y-n)$ which finishes the proof.
\end{proof}
\begin{theorem} \label{theorem:exp_d=1}
Let Assumption \ref{ass:d=1} be satisfied, $\Gamma \subset \ZZ$ connected and $s \in (0,1)$. Assume
\begin{equation} \label{eq:disorder}
 \frac{\Vert\rho\Vert_\infty}{\lambda} < \frac{(1-s)^{1/s}}{2s^{-1}} \Bigl | \prod_{k=0}^{n-1} u(k) \Bigr|^{1/n} .
\end{equation}
Then $\mu = - \ln C_{\lambda , u,\rho}$ is strictly positive and we have for all $x,y \in \Gamma$ with $\lvert x-y \rvert \geq 2n$ and all $z \in \CC \setminus \RR$ the bound
\[
 \mathbb{E} \bigl( \lvert G_\Gamma (z;x,y)\rvert^{s/n} \bigr) \leq C_{\lambda , u,\rho}^+ \exp \Biggl\{-\mu \Biggl \lfloor \frac{\lvert x-y\rvert}{n} \Biggl\rfloor\Biggr\} .
\]
Here, $\lfloor \cdot \rfloor$ is for $m \in \RR$ defined by $\lfloor m \rfloor := \max\{k\in \ZZ \mid k\leq m\}$.
\end{theorem}
\begin{remark}
The statement of Theorem \ref{theorem:exp_d=1} holds also true if the density function $\rho$ is not compactly supported since Lemma \ref{lemma:det} holds true for non-compactly supported functions $\rho$. This is formulated in \cite{ElgartTV-11}. Notice, if $\rho$ (and hence the measure $\nu$) is not compactly supported, then one has to be careful about the domain of the family of self-adjoint operators $H_\omega$. In particular, $H_\omega$ is essentially self-adjoint on the compactly supported sequences, see \cite{Kirsch-08}.
\end{remark}

\begin{proof}[Proof of Theorem~\ref{theorem:exp_d=1}]
The constant $\mu$ is larger than zero since $C_{\lambda , u,\rho} < 1$ by assumption.
By symmetry we assume without loss of generality $y-x \geq 2n$.
In order to estimate $\mathbb{E} (\lvert G_\Gamma (z;x,y)\rvert^{s/n})$,
we iterate Eq.~\eqref{eq:iteration1} of Lemma \ref{lemma:iteration1}
and finally use Eq.~\eqref{eq:iteration2} of Lemma \ref{lemma:iteration2} for the last step.
Figure \ref{fig:iteration} shows this procedure schematically.
\begin{figure}[ht]
\centering
\begin{tikzpicture}[scale=1]
 \draw[very thick] (9.1,0) -- (11.1,0);
 \draw (-0.5,0)--(5.5,0);\draw[dotted] (5.5,0)--(6.5,0);\draw (6.5,0)--(13.5,0);
 \draw[thick] (0 cm,2.5pt) -- (0 cm,-2.5pt) node[anchor=north]  {\footnotesize $\phantom{I}x\phantom{I}$};
 \draw[thin] (13 cm,2.5pt) -- (13 cm,-2.5pt) node[anchor=north]{\footnotesize $\phantom{I}y\phantom{I}$};
 \draw[thick] (2 cm,2.5pt) -- (2 cm,-2.5pt) node[anchor=north]  {\footnotesize $\phantom{I}x+n\phantom{I}$};
 \draw[thick] (4 cm,2.5pt) -- (4 cm,-2.5pt) node[anchor=north]  {\footnotesize $\phantom{I}x+2n\phantom{I}$};
 \draw[thin] (11 cm,2.5pt) -- (11 cm,-2.5pt);
 \draw (11.25,-0.07) node[anchor=north]{\footnotesize $\phantom{I}y-n\phantom{I}$};
 \draw[thin] (9 cm,2.5pt) -- (9 cm,-2.5pt) node[anchor=north]  {\footnotesize $\phantom{I}y-2n\phantom{I}$};
 \draw[thick] (10 cm,2.5pt) -- (10 cm,-2.5pt);
 \draw (10.15,-0.07) node[anchor=north]{\footnotesize $\phantom{I}x+pn\phantom{I}$};
 \draw[thick] (8 cm,2.5pt) -- (8 cm,-2.5pt);
\draw (7.5,-0.07) node[anchor=north]{\footnotesize $\phantom{I}x+(p-1)n\phantom{I}$};
 \draw[-latex] (0,0) .. controls  (0.7,0.5) and (1.3,0.5) .. (2,0);  
 \draw[-latex] (2,0) .. controls  (2.7,0.5) and (3.3,0.5) .. (4,0);
 \draw         (4,0) .. controls  (4.25,0.2)  .. (4.5,0.3);
\draw[-latex] (8,0) .. controls  (8.7,0.5) and (9.3,0.5) .. (10,0);
 \draw[dotted] (4.5,0.3) .. controls  (5,0.4)  .. (5,0.4);
 \draw[dotted] (7,0.4) .. controls (7,0.4) .. (7.5,0.3);
 \draw[-latex]        (7.5,0.3) .. controls  (7.75,0.20)  .. (8,0);
\draw[-latex] (10,0) .. controls  (11,0.5) and (12,0.5)  .. (13,0);
\draw (4.5,0.8) node {\footnotesize Lemma \ref{lemma:iteration1}};
\draw (11.5,0.8) node {\footnotesize Lemma \ref{lemma:iteration2}};
\end{tikzpicture}
\caption{Illustration to the proof of Theorem \ref{theorem:exp_d=1}}
\label{fig:iteration}
\end{figure}
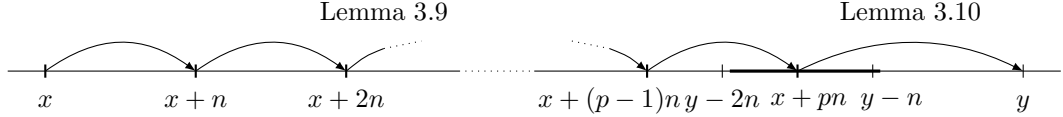
We choose $p = \lfloor (y-x)/n \rfloor - 1 \in \NN$ such that $y-2n < x+pn \leq y-n$.
We iterate Eq.~\eqref{eq:iteration1} exactly $p$ times and obtain
\[
 \mathbb{E} \bigl(\lvert G_\Gamma (z;x,y)\rvert^{s/n}\bigr) \leq C_{\lambda , u,\rho}^{p} \cdot \mathbb{E}\bigl(\lvert G_{\Lambda_{p}} (z;x+pn,y)\rvert^{s/n}\bigr )
\]
where $\Lambda_p = \{ x+pn, x+pn+1, \dots \}$. Now the first $p$ jumps of Fig.~\ref{fig:iteration} are done
and it remains to estimate $\mathbb{E}\bigl(\lvert G_{\Lambda_{p}} (z;x+pn,y)\rvert^{s/n}\bigr)$.
Since $n \leq y-(x+pn) < 2n$ and $\Lambda_p = \{ x+pn, x+pn+1, \dots \}$ we may apply Lemma \ref{lemma:iteration2} and get
\[
 \mathbb{E}\bigl(\lvert G_\Gamma (z;x,y)\rvert^{s/n}\bigr) \leq C_{\lambda , u,\rho}^{p+1} C_{\lambda , u,\rho}^+
= C_{\lambda , u,\rho}^+ \,\, {\rm e}^{(p+1) \ln C_{\lambda , u,\rho}} . \qedhere
\]
\end{proof}
\subsection{General single-site potentials} \label{sec:general_single_site}
We now want to get rid of the assumption that $\Theta = \{0,1,\ldots,n-1\}$, i.e.\ we want to consider the case where $\Theta$ is still finite but not necessarily connected, and prove an analogue of Theorem \ref{theorem:exp_d=1} in this case. By translation, we assume without loss of generality that $\min \Theta = 0$ and $\max \Theta = n-1$ for some $n \in \NN$. Furthermore, we define
\begin{equation} \label{eq:r}
 r := \max\big\{ n + 1 \in \NN \colon \exists \, a \in \{1,2,\ldots,n-1\} \colon \{a , \ldots , a+n\} \cap \Theta = \emptyset \big\} .
\end{equation}
Here we use the convention that $\max \emptyset = 0$. Thus $r$ is the number of elements of the largest gap in $\Theta$, and $r = 0$ if $\Theta$ is the connected set $\{0 , 1 , \ldots , n-1\}$.
\par
To illustrate the difficulties arising for non-connected supports $\Theta$ we consider an example.
Suppose $\Theta = \{0,2,3,\dots , n-1\}$ so that $r = 1$. If we set $\Lambda = \{0,\dots,n-1\}$ there is no decomposition $H_\Lambda - B_\Gamma^\Lambda = A + \omega_0 \lambda V$ with an invertible $V$. If we set $\Lambda = \{0,\dots,n-1+r\} = \{0,\dots,n\}$ we observe that every diagonal element of $H_\Lambda$ depends at least on one of the variables $\omega_0$ and $\omega_1 = \omega_r$, while the elements of $B_\Gamma^\Lambda$ (which appear after applying Lemma \ref{lemma:schur1}) are independent of $\omega_k$, $k \in \{0,\dots,r\} = \{0,1\}$. Thus we have a decomposition $H_\Lambda - B_\Gamma^\Lambda = A + \omega_0 \lambda V_0 + \omega_1 \lambda V_1$, where $A$ is independent of $\omega_k$, $k \in \{0,1\}$, and for all $i \in \Lambda$ either $V_0 (i)$ or $V_1 (i)$ is not zero.
As a consequence there is an $\alpha \in \RR$ such that $V_0 + \alpha V_1$ is invertible on $\ell^2 (\Lambda )$. This makes Lemma~\ref{lemma:detgen} applicable and we can prove the following analogues of Lemma~\ref{lemma:finitness1} and Theorem~\ref{theorem:exp_d=1}.
\begin{lemma} \label{lemma:finitness2}
 Let Assumption \ref{ass:d=1_2} be satisfied and $\Gamma \subset \ZZ$ be connected. Let further $r$ be as in Eq.~\eqref{eq:r} and $s \in (0,1)$.
Then there exists a constant $D = D(\lambda , u,\rho , s,r)$ such that for all $x,x+n-1+r \in \Gamma$ and $z \in \CC \setminus \RR$
\begin{equation} \label{eq:leqD}
\mathbb{E}_{\{x,\dots,x+r\}} \bigl ( \lvert G_\Gamma (z;x,x+n-1+r)\rvert^{s/(n+r)}  \bigr ) \leq D\, .
\end{equation}
 The constant $D$ is characterized in Eq.~\eqref{eq:defD} and estimated in Ineq.~\eqref{eq:D}. If $1 \leq \abs{\Gamma} \leq n+r$ with $\gamma_0 = \min \Gamma$ and $\gamma_1 = \max \Gamma$ there exists a constant $D^+ = D^+ (\lambda , u,\rho ,s,r)$ such that for all $z \in \CC \setminus \RR$
\begin{equation} \label{eq:leqDplus}
 \mathbb{E}_{\{\gamma_0,\dots,\gamma_0+r\}} \bigl ( \lvert G_\Gamma (z;\gamma_0,\gamma_1) \rvert^{s/(n+r)}  \bigr ) \leq D^+.
\end{equation}
The constant $D^+$ is characterized in Eq.~\eqref{eq:defDplus} and estimated in Ineq.~\eqref{eq:Dplus}.
\end{lemma}
\begin{proof}
 The proof is similar to the proof of Lemma \ref{lemma:finitness1}. Apply Lemma \ref{lemma:schur1} with $\Lambda = \{x,x+1,\dots,x+n-1+r\}$ and Cramer's rule to get $\lvert G_\Gamma (z;x,x+n-1+r)\rvert = 1/\abs{\det B}$ where $B = H_\Lambda - B_\Gamma^\Lambda - z$. Note that $B_\Gamma^\Lambda$ is independent of $\omega_k$, $k \in \{x,\dots,x+r\}$. We have the decomposition $B = A + \lambda \sum_{k=0}^{r} \omega_{x+k} V_k$ where the elements of the diagonal matrices $V_k \in \RR^{(n+r) \times (n+r)}$, $k = 0,\dots,r$, are given by $V_k(i) = u(i-k)$, $i=0,\dots,n-1+r$, and $A = B - \lambda \sum_{k=0}^r \omega_{x+k} V_k$ is independent of $\omega_k$, $k \in \{x,\dots,x+r\}$.
We apply Lemma \ref{lemma:detgen} and obtain for all $\alpha = (\alpha_k)_{k=0}^r \in M = \{\alpha \in \RR^{r+1} : \alpha_0 \not = 0,  \text{ $\sum_{k=0}^r \alpha_k V_k$ is invertible}\}$ the bound  $\mathbb{E}_{\{x,\dots,x+r\}} \bigl ( \lvert G_\Gamma (z;x,x+n-1+r)\rvert^{s/(n+r)} \bigr ) \leq D_\alpha$ where
\begin{equation*}
 D_\alpha  = \Vert\rho\Vert_\infty^{(r+1)s} (2R)^{rs} \frac{2^s s^{-s}}{1-s} \lvert \alpha_0\rvert^s \Biggl( 1+\max_{i\in \{1,\dots,r\}} \frac{\lvert\alpha_i\rvert}{\lvert\alpha_0\rvert} \Biggr)^{rs} \, \prod_{i=0}^{n-1+r} \Bigl| \sum_{k=0}^r \alpha_k \lambda u(i-k) \Bigr|^{-\frac{s}{n+r}} .
\end{equation*}
The set $M$ is non-empty and equal to the set $\{\alpha \in \RR^{r+1} : \alpha_0 \not = 0, D_\alpha \text{ is finite}\}$.
Thus Ineq.~\eqref{eq:leqD} holds with the constant
\begin{equation} \label{eq:defD}
 D := \inf_{\alpha \in M} D_\alpha .
\end{equation}
In the following we establish an upper bound for $D$. Using a volume comparison criterion we can find a vector $\alpha'=(\alpha_k')_{k=0}^r \in [0,1]^{r+1}$ which has to each hyperplane $\sum_{k=0}^r \alpha_k u(i-k) = 0$, $i=0,\dots,n-1+r$, at least the Euclidean distance $(2 (n+r) (r+1)^{r/2})^{-1}$, as outlined in Fig.~\ref{fig:volume}.
\begin{figure}[t]
\centering
\begin{tikzpicture}[scale=4]
\draw[dotted] (0,1)--(1,1);\draw[dotted] (1,1)--(1,0);\draw[dotted] (1,0)--(0,0);\draw[dotted] (0,0)--(0,1);
\draw[<->] (1.1207,0.9793)--(0.9793,1.1207); \draw (1.1,1.1) node {\small $\epsilon$};
\draw (0,0)--(0,1);\draw[dashed] (0.1,0)--(0.1,1);
\draw (0,0)--(1,1); \draw[dashed] (0,0.14)--(0.86,1); \draw[dashed] (0.14,0)--(1,0.86);
\draw (0,0)--(1,0.4); \draw[dashed] (0,0.108)--(1,0.508); \draw[dashed] (0.27,0)--(1,0.292);
\draw (0,1.1) node {\small $H_0^\epsilon$};
\draw (1.1,0.8) node {\small $H_1^\epsilon$};
\draw (1.185,0.4) node {\small $H_{n-1+r}^\epsilon$};
\filldraw (0.4,0.95) circle (0.4pt); \draw (0.4,0.9)  node {\small $\alpha'$};
\draw (2.2,0.8) node {\small $\operatorname{Vol}(W) = 1$};
\draw (2.2,0.6) node {\small $\operatorname{Vol}(\cup_i H_i^\epsilon) \leq (n+r) (r+1)^{r/2} \epsilon$};
\draw (2.2,0.4) node {\small $\operatorname{Vol}(W \setminus \cup_i H_i^\epsilon) \geq 1 - (n+r) (r+1)^{r/2} \epsilon$};
\end{tikzpicture}
\caption[\protect{Sketch of the existence of a vector $\alpha' \in W = [0,1]^{r+1}$} with the desired properties]{Sketch of the existence of a vector $\alpha' \in W = [0,1]^{r+1}$
with the desired properties: Let  $H_i^\epsilon$ denote the $\epsilon$-neighborhood of the hyperplane $H_i=\{ \alpha \in W \mid \sum_{k=0}^r \alpha_k u(i-k) = 0\}$
for  $i \in \{0,\dots,n-1+r\}$. Since the volume of $W \setminus \cup_i H_i^\epsilon$ is positive if $\epsilon$ is smaller  than $(n+r)^{-1} (r+1)^{-r/2} = d_0$, we conclude (using continuity) that there is a vector $\alpha'$ whose distance to each hyperplane $H_i$, $i \in \{0,\dots,n-1+r\}$, is at least $d_0/2$.}
\label{fig:volume}
\end{figure}
This implies $\alpha_0' \geq ( 2 (n+r) (r+1)^{r/2})^{-1}$ since the hyperplane for $i=0$ is $\alpha_0 = 0$. With this choice of $\alpha$ and the notation $u_i = (u(i-k))_{k=0}^r$, $i \in \{0,\dots,n-1+r\}$, we have
\begin{align} 
 \prod_{i=0}^{n-1+r} \Bigl| \sum_{k=0}^r \alpha_k' u(i-k) \Bigr|^{-\frac{s}{n+r}}  &=\prod_{i=0}^{n-1+r} \Bigl| \Vert u_i \Vert \sprod{\alpha'}{u_i/ \Vert u_i\Vert}_2 \Bigr|^{-\frac{s}{n+r}} \nonumber \\
&\leq
\frac{\bigl[ 2 (n+r) (r+1)^{r/2} \bigr]^s}{\Bigl| \prod_{i=0}^{n-1+r} \Bigl( \sum_{k=0}^r u(i-k)^2 \Bigr) \Bigr|^{\frac{s}{2(n+r)}}} \label{eq:volume}
\end{align}
where $\sprod{\cdot}{\cdot}_2$ denotes the standard Euclidean scalar product. Now we choose $\alpha = \alpha'$ and obtain
\begin{equation} \label{eq:D}
 D \leq  \Vert\rho\Vert_\infty^{(r+1)s} (2R)^{rs} \frac{2^s s^{-s}\bigl( 1+2 (n+r) (r+1)^{r/2} \bigr)^{rs} \bigl[ 2 (n+r) (r+1)^{r/2} \bigr]^s}{(1-s)\Bigl| \prod_{i=0}^{n-1+r} \Bigl( \sum_{k=0}^r u(i-k)^2 \Bigr) \Bigr|^{\frac{s}{2(n+r)}} \lambda^s} .
\end{equation}
\par
The proof of the second statement is similar but without use of Lemma~\ref{lemma:schur1}. By Cramer's rule we get $\lvert G_\Gamma (z;\gamma_0,\gamma_1)\rvert = 1/\lvert\det (H_\Gamma - z)\rvert$. Set $l = \gamma_1 - \gamma_0$. We have the decomposition $H_\Gamma - z = \tilde A + \lambda \sum_{k=0}^{r} \omega_{\gamma_0+k} \tilde V_k$, where the elements of the diagonal matrices $\tilde V_k \in \RR^{(l + 1) \times (l + 1)}$, $k = 0,\dots,r$, are given by $\tilde V_k(i) = u(i-k)$, $i  \in \{0,\dots,l\}$, and $\tilde A = H_\Gamma - z - \lambda \sum_{k=0}^r \omega_{\gamma_0+k} \tilde V_k$ is independent of $\omega_k$, $k \in \{x,\dots,x+r\}$.
We apply Lemma \ref{lemma:detgen} with $t = s \frac{l+1}{n+r}$ and obtain (using $s \geq t$) for all
$\alpha = (\alpha_k)_{k=0}^r \in \tilde M = \{\alpha \in \RR^{r+1} : \alpha_0 \not = 0,~ \text{$\sum_{k=0}^r \alpha_k \tilde V_k$ is invertible}\}$ that $\mathbb{E}_{\{\gamma_0,\dots,\gamma_0+r\}} \bigl ( \lvert G_\Gamma (z;\gamma_0,$ $\gamma_1) \rvert^{s/(n+r)} \bigr ) \leq D_\alpha^+ (l)$ where
\begin{multline*}
 D_\alpha^+ (l)  =  \Vert\rho\Vert_\infty^{s \frac{(r+1)(l+1)}{n+r}} (2R)^{s \frac{r(l+1)}{n+r}} \frac{2^s s^{-s}}{1-s} \lvert\alpha_0\rvert^{s \frac{l+1}{n+r}} \\
\cdot \Biggl( 1+\max_{i\in \{1,\dots,r\}} \frac{\lvert\alpha_i\rvert}{\lvert\alpha_0\rvert} \Biggr)^{s r} \prod_{i=0}^{l} \Bigl| \sum_{k=0}^r \alpha_k \lambda u(i-k) \Bigr|^{\frac{s}{n+r}} .
\end{multline*}
Since $\tilde M \supset M$ for each $l \in {0,\dots n-1+r}$ the set $\tilde M$ is non-empty. Thus Ineq.~\eqref{eq:leqDplus} holds with the constant
\begin{equation} \label{eq:defDplus}
 D^+ := \max_{l \in \{0,\dots,n-1+r\}\phantom{\tilde M}} \inf_{\alpha \in \tilde M} D_\alpha^+ (l) .
\end{equation}
We again choose $\alpha = \alpha'$ as in Fig.~\ref{fig:volume}, use $\alpha_k' \in [0,1]$ and $\alpha_0' \geq (2 (n+r) (r+1)^{r/2})^{-1}$, estimate $D_{\alpha'}^+ (l)$ similar to Ineq.~\eqref{eq:volume}, and obtain
\begin{multline} \label{eq:Dplus}
D^+\leq  \\ \max_{l \in \{0,\dots,n-1+r\}} \left\lbrace 
  \frac{\Vert\rho\Vert_\infty^{s\frac{(r+1)(l+1)}{n+r}} \bigl[1+2 (l+1) (r+1)^{r/2} \bigr]^{sr} \bigl[2(l+1)(r+1)^{r/2} \bigr]^s}{(2R)^{-s\frac{r(l+1)}{n+r}} 2^{-s} s^s (1-s) \Bigl| \prod_{i=0}^l \sum_{k=0}^r \lambda u(i-k) \Bigr|^{\frac{s}{2(n+r)}}}
\right\rbrace ,
\end{multline}
which ends the proof.
\end{proof}
\begin{theorem} \label{theorem:exp_d=1_2}
Let Assumption \ref{ass:d=1_2} be satisfied, $\Gamma \subset \ZZ$ connected, $s \in (0,1)$, $r$ as in Eq.~\eqref{eq:r}, and $D$ and $D^+$ the constants from Lemma \ref{lemma:finitness2}. Assume $D < 1$.
Then $m = - \ln D$ is strictly positive and we have the bound
\[
 \mathbb{E} \bigl( \lvert G_\Gamma (z;x,y)\rvert^{s/(n+r)} \bigr) \leq 
 D^+ \exp \Biggl\{-m \Biggl \lfloor \frac{\lvert x-y\rvert}{n+r} \Biggl\rfloor\Biggr\}
\]
for all $x,y \in \ZZ$ with $\lvert x-y\rvert\geq 2(n+r)$ and all $z \in \CC \setminus \RR$.
\end{theorem}
\begin{remark}
 The assumption $D<1$ can be achieved by choosing $\lambda$ sufficiently large, see Ineq.~\eqref{eq:D}.
\end{remark}
\begin{remark}
The statement of Lemma~\ref{lemma:detgen} and so also Theorem~\ref{theorem:exp_d=1_2} can be proven if $\rho$ is not compactly supported, but in the Sobolev space $W^{1,1} (\RR)$. This is formulated in \cite{ElgartTV-10}.
\end{remark}

\begin{proof}[Proof of Theorem~\ref{theorem:exp_d=1_2}]
 The proof is similar to the proof of Theorem~\ref{theorem:exp_d=1}. We again assume $y>x$. Let $\Gamma_1 \subset \ZZ$ be connected. Using Eq.~\eqref{eq:resolvent} with $\Lambda = \{x+n+r,\dots\} \cap \Gamma_1$ and Lemma \ref{lemma:finitness2} we have for all pairs $x,y \in \Gamma_1$ with $y-x \geq n+r$
 \begin{equation} \label{eq:iteration3}
  \mathbb{E} \bigl( \lvert G_{\Gamma_1} (z;x,y)\rvert^{s/(n+r)}\bigr) \leq D \,\,  \mathbb{E} \bigl( \lvert G_\Lambda (z;x+n+r,y)\rvert^{s/(n+r)}\bigr )
\end{equation}
which is the analogue to Lemma \ref{lemma:iteration1}. Now, let $\Gamma_2 = \{x,x+1,\dots\}$ and $y \in \Gamma_2$ with $n+r \leq y-x < 2(n+r)$. By Eq.~\eqref{eq:resolvent2} with $\Lambda = \{x,\dots,y-(n+r)\}$ and Lemma \ref{lemma:finitness2} we have
\begin{equation} \label{eq:iteration4}
 \mathbb{E} \bigl( \lvert G_{\Gamma_2} (z;x,y)\rvert^{s/(n+r)}\bigr) \leq D D^+
\end{equation}
    which is the analogue of Lemma \ref{lemma:iteration2}. Iterating Eq.~\eqref{eq:iteration3} exactly $p=\lfloor(y-x)/(n+r)\rfloor - 1$ times, starting with $\Gamma_1 = \Gamma$, and finally using Eq.~\eqref{eq:iteration4} once gives the statement of the theorem.
\end{proof}
\subsection{A priori bound and exponential localization} \label{sec:loc_d=1}
The statements of Theorem \ref{theorem:exp_d=1} and \ref{theorem:exp_d=1_2} concern only off-diagonal elements of the Green function. The next theorem shows a global uniform bound on $(x,y) \mapsto \EE (\lvert G_\Gamma (z;x,y) \rvert^s)$ for $s>0$ sufficiently small. We use the notation $u_j (x) = u(x-j)$, $j,x \in \ZZ$, for the translated function as well as for the corresponding multiplication operator.
\begin{theorem}\label{t:a-priori_d=1}
Let Assumption \ref{ass:d=1_2} be satisfied, $\Gamma \subset \ZZ$ be connected and $s \in (0,1)$.
Then there is a positive constant $C = C(\lambda , u,\rho , s)$ such that for all $x,y \in \Gamma$ and all $z \in \CC \setminus \RR$ we have
\[
\mathbb{E} \bigl ( \lvert G_\Gamma (z;x,y) \rvert^{s/(4n)} \bigr ) \leq C .
\]
\end{theorem}
\begin{proof}
To avoid notation we assume $\Gamma = \ZZ$. Since $\supp \rho \subset [-R,R]$, $H_\omega$ is a bounded operator. Set $m = \lVert H_\omega \rVert + 1$. If $\lvert z\rvert \geq m$, we use $\lVert G_\omega (z) \rVert = \sup_{\lambda \in \sigma (H_\omega)} \lvert \lambda - z \rvert^{-1} \leq 1$ and obtain the statement of the theorem. Thus it is sufficient to consider $\lvert z\rvert \leq m$.
If $\lvert x-y\rvert \geq 4n$ Theorem \ref{theorem:exp_d=1_2} applies, since $r \leq n$. We thus only consider the case $\lvert x-y\rvert \leq 4n-1$. By translation we assume $x = 0$ and by symmetry $y \geq 0$. Set $\Lambda_+ = \{-1,\dots,4n\}$ and $\Lambda = \{0,\dots,4n-1\}$. Lemma \ref{lemma:schur1} with $\Lambda_1 = \Lambda$ and $\Lambda_2 = \Lambda_+$ gives
\begin{equation} \label{eq:2xschur}
 \Pro_{\Lambda} G_\omega (z) \Inc_{\Lambda} = \Bigl(H_{\Lambda} - z - \Pro_{\Lambda} \Delta \Inc_{\partial^{\rm i}\Lambda_+} \, (K-z)^{-1} \, \Pro_{\partial^{\rm i}\Lambda_+} \Delta \Inc_{\Lambda} \Bigr)^{-1}
\end{equation}
where
\[
 K = H_{\partial^{\rm i}\Lambda_+} - \Pro_{\partial^{\rm i} \Lambda_+}^{\Lambda_+} \, B_\ZZ^{\Lambda_+} \, \Inc_{\partial^{\rm i} \Lambda_+}^{\Lambda_+} 
\]
and $\langle\delta_x , B_\Gamma^{\Lambda_+} \delta_y \rangle = \sum_{k \in \Gamma \setminus \Lambda_+ , \abs{k-x} = 1} \langle \delta_k , (H_{\Gamma \setminus \Lambda_+} - z)^{-1}\delta_k \rangle$ if $x=y$ and $x \in \partial \Lambda_+ = \{-1,4n\}$, and zero else.
Note that $B_\ZZ^{\Lambda_+}$ is independent of $\omega_k$, $k \in \{-1,\dots,3n+1\}$,
and $K$ is independent of $\omega_k$, $k \in \{0,\dots,3n\}$. Thus, in matrix representation with respect to the canonical basis, the operator $K:\ell^2 (\partial \Lambda_+) \to \ell^2 (\partial \Lambda_+)$ may be decomposed as
\begin{equation*}
K = \begin{pmatrix}
    \omega_{-1} \lambda u(0) & 0 \\
    0     & \omega_{3n+1} \lambda u(n-1)
   \end{pmatrix} -
\begin{pmatrix}
    f_1 & 0 \\
    0     & f_2
   \end{pmatrix}
\end{equation*}
where $f_1 := \sum_{k \in \ZZ \setminus \{-1\}} \omega_k \lambda u(-1-k) - \langle \delta_{-1}, B_\ZZ^{\Lambda_+} \delta_{-1} \rangle$
and $f_2 := \sum_{k \in \ZZ \setminus \{3n+1\}} \omega_k \lambda$ $u(4n-k) - \langle \delta_{4n},$ $B_\ZZ^{\Lambda_+} \delta_{4n} \rangle$
are independent of $\omega_{-1}$ and $\omega_{3n+1}$.
Standard spectral averaging or Lemma \ref{lemma:det} gives for all $t \in (0,1)$
\begin{equation} \label{eq:averageK}
 \mathbb{E}_{\{-1,3n+1\}} \Bigl(\bigl\Vert(K-z)^{-1}\bigr\Vert^t \Bigr) \leq \lambda^{-t}\bigl(\lvert u(0)\rvert^{-t}+\lvert u(n-1)\rvert^{-t}\bigr) \norm{\rho}_\infty^t \frac{2^t t^{-t}}{1-t} .
\end{equation}
Now, the operator $H_{\Lambda}$ can be decomposed as $H_{\Lambda} = A + \sum_{k=0}^{3n} \omega_k \lambda u_k$
where $A := H_{\Lambda} - \sum_{k=0}^{3n} \omega_k \lambda u_k$ is
independent of $\omega_k$, $k \in \{0,\dots,3n\}$. Let $\alpha := (\alpha_k)_{k=0}^{3n} \in [0,1]^{3n+1}$
with $\alpha_0 \not = 0$. Similarly to the proof of Lemma \ref{lemma:detgen}, we use the substitution $\omega_0 = \alpha_0 \zeta_0$ and $\omega_i = \alpha_i \zeta_0 + \alpha_0 \zeta_i$ for $i \in \{1,\dots,3n\}$ and obtain from Eq.~\eqref{eq:2xschur}
\begin{align*}
 E &:= \mathbb{E}_{\{0,\dots,3n\}} \Bigl( \bigl\Vert  \Pro_{\Lambda} G_\omega (z) \Inc_{\Lambda} \bigr\Vert^{s/(4n)} \Bigr) \\[1ex]
  &\leq \norm{\rho}_\infty^{3n+1} \!\!\!\!\!\!\!\! \int\limits_{[-R,R]^{3n+1}} \!\!\!\!\!\!  \Bigl\Vert  \Bigl(H_\Lambda - z - \Pro_{\Lambda} \Delta \Pro_{\partial\Lambda_+}^* \, (K-z)^{-1} \, \Pro_{\partial\Lambda_+}^\Lambda \Delta \Pro_{\Lambda}^* \Bigr)^{-1} \Bigr\Vert^{s/(4n)} \drm \omega_0 \dots \drm\omega_{3n} \\[1ex]
  &\leq \norm{\rho}_\infty^{3n+1} \int_{[-S, S]^{3n+1}}  \Bigl\Vert  \Bigl(A' + \zeta_0 \sum_{k=0}^{3n} \alpha_k \lambda u_k \Bigr)^{-1} \Bigr\Vert^{s/(4n)} \lvert \alpha_0\rvert^{3n+r} \drm \zeta_0 \dots \drm\zeta_{3n}
\end{align*}
where $S = R (1+\max_{i \in \{1,\dots,3n\}}$ $\lvert\alpha_i / \alpha_0 \rvert)/\lvert \alpha_0\rvert$ and $A' = A + \alpha_0 \sum_{k=1}^{3n} \zeta_k \lambda u_k - z - \Pro_{\Lambda} \Delta \Pro_{\partial\Lambda_+}^*$ $(K-z)^{-1} \Pro_{\partial\Lambda_+} \Delta \Pro_{\Lambda}^*$. Since $\bigcup_{i=0}^{3n} \supp u_i = \Lambda$, there exists an $\alpha \in [0,1]^{3n+1}$ such that $\sum_{k=0}^{3n} \alpha_k u_k$ is invertible on $\ell^2 (\Lambda)$, compare the proof of Lemma \ref{lemma:finitness2} and Figure \ref{fig:volume}. Thus we may apply Lemma \ref{lemma:averagenorm} and obtain
\begin{equation} \label{eq:Eint}
 E \leq \norm{\rho}_\infty^{3n+1} \int\limits_{[-S,S]^{3n}}  \frac{2 \Bigl(\norm{A'} + S \norm{\sum_{k=0}^{3n} \alpha_k \lambda u_k}\Bigr)^{s(4n-1)/(4n)}}{s^{s} S^{s-1}(1-s)\abs{\det \bigl( \sum_{k=0}^{3n} \alpha_k \lambda u_k \bigr)}^{s/(4n)}} \drm \zeta_1 \dots \drm \zeta_{3n} .
\end{equation}
Using $\zeta_k \in [-S,S]$ for $k \in \{1,\dots,3n\}$, $\omega_k \in [-R,R]$ for $k \in \ZZ \setminus \{0,\dots,3n\}$ and $\alpha_k \in [0,1]$ for $k \in \{0,\dots,3n\}$, the norm of $A'$ can be estimated as
\begin{align} \label{eq:A'}
 \norm{A'} &= \Bigl\Vert H_{\Lambda} - \sum_{k=0}^{3n} \omega_k \lambda u_k + \alpha_0 \sum_{k=1}^{3n} \zeta_k \lambda u_k - z - \Pro_{\Lambda} \Delta \Pro_{\partial\Lambda_+}^* \, (K-z)^{-1} \, \Pro_{\partial\Lambda_+} \Delta \Pro_{\Lambda}^* \Bigr\Vert \nonumber \\[1ex]
&\leq 2 + (n-1) R \lambda \lVert u\rVert_{\infty} + 3S n \lambda \lVert u\lVert_{\infty} + m + 4  \bigl\Vert(K-z)^{-1} \bigr\Vert \, .
\end{align}
All terms in the sum \eqref{eq:A'} are independent of $\zeta_k$, $k \in \{0,\dots,3n\}$.
Using $(\sum \lvert a_i\rvert)^t \leq \sum \lvert a_i\rvert^t$ for $t<1$ we see from Ineq.~\eqref{eq:Eint} and \eqref{eq:A'} that there are constants $C_1$ and $C_2$ such that $E \leq C_1 + C_2\Vert(K-z)^{-1} \Vert^{s(4n-1)/(4n)}$.
If we average over $\omega_{-1}$ and $\omega_{3n+1}$, Ineq.~\eqref{eq:averageK} gives the desired result.
\end{proof}
We can now conclude exponential localization for the one-dimensional case.
\begin{theorem}
 Let one of the following assumptions be satisfied:
\begin{enumerate}[(i)]
 \item Let Assumption \ref{ass:d=1} be satisfied and 
\begin{equation*}
 \frac{\Vert\rho\Vert_\infty}{\lambda} < \frac{(1-s)^{1/s}}{2s^{-1}} \Bigl | \prod_{k=0}^{n-1} u(k) \Bigr|^{1/n} .
\end{equation*}
\item Let Assumption \ref{ass:d=1_2} be satisfied and $\lambda$ be sufficiently large.
\end{enumerate}
Then, for almost all $\omega \in \Omega$, $H_\omega$ exhibits exponential localization.
\end{theorem}
\begin{proof}
By our assumptions and Theorem \ref{theorem:exp_d=1}, \ref{theorem:exp_d=1_2} and \ref{t:a-priori_d=1}, we conclude that the hypothesis of Theorem \ref{theorem:exp_decay_loc} is satisfied. This gives the result.
\end{proof}
\subsection{Reduction to the monotone case} \label{sec:red_monotone}
In this subsection we discuss whether a special transformation of random variables allows us to extract a positive part of the potential, which makes monotone spectral averaging applicable. First we present a criterion which ensures that an appropriate one-parameter family of positive
potentials can be extracted from the random potential $V_\omega$.
\begin{lemma} \label{l:equivalence}
Let Assumption~\ref{ass:d=1_2} be satisfied, such that 
\[
u= \sum_{k=0}^{n-1}u(k) \delta_k\colon \ZZ\to \RR .
\]
Then the following statements are equivalent.
\begin{enumerate}[(a)]
\item
There exists an $N \in \NN$ and real $\alpha_0, \dots, \alpha_N$
such that $w := u \ast \alpha := \alpha_0 u_0 + \dots + \alpha_N u_N$ is a non-negative function and
$w(0)>0 $, $w(N+n-1)>0 $ hold.
\item
There exists an $M \in \NN$ and real $\gamma_0, \dots, \gamma_M$
such that $v := u \ast \gamma := \gamma_0 u_0 + \dots + \gamma_M u_M$ is a non-negative function and
$\supp v = \{0, \dots, M+n-1\}$ holds.
\item
The polynomial $\CC \ni z \mapsto p_u(z) := \sum_{k=0}^{n-1}u(k) z^k$ has no roots in $[0,\infty)$.
\end{enumerate}
\end{lemma}
Note, if $u(0) \not = 0$ and $u(n-1) \not = 0$, then $\{0, \dots, M+n-1\}$ is the union of the supports of $u_0, \dots, u_M$.
If (a) or (b) hold
we may assume that $\lvert \alpha_0\rvert $ or $ \lvert\gamma_0\rvert $ equals one.
\begin{proof}
If (a) holds, one may choose $v (x) = \sum_{j=0}^{N+n-2} w(x-j)$ to conclude (b). Thus it is sufficient to show (b)$\Leftrightarrow$(c). Using Fourier transform and the identity theorem for holomorphic functions one sees that (b) is equivalent to
\begin{enumerate}[(a)]
\item[(d)] There exists an $M \in \mathbb{N}$ and real $\gamma_0,\dots,\gamma_M$ such that all coefficients of the polynomial $p_u(z) \cdot \sum_{j=0}^M \gamma_j z^j$ are strictly positive.
\end{enumerate}
If (d) holds, $p_u(x) \cdot \sum_{j=0}^M \gamma_j x^j$ is strictly positive for $x \in [0,\infty)$. Thus its divisor $p_u$ has no root in $[0,\infty)$ and one concludes (c). Assuming (c), one infers from Corollary 2.7 of \cite{MotzkinS-69} that there exists a polynomial $p$ such that $p_u \cdot p$ has strictly positive coefficients. Choosing $M = \deg (p)$ and $\gamma_0,\dots,\gamma_M$ to be the coefficients of $p$ leads to (d).
\end{proof}
If the random potential $V_\omega$ contains a positive building block
$w$ as in (a) of the previous lemma, one
obtains exponential decay of fractional moments with the methods from \cite{AizenmanENSS-06}, as we outline now.
The crucial tool is Proposition 3.2 of \cite{AizenmanENSS-06}. Here are two direct consequences of the latter.
\begin{lemma}\label{l:AENSS}
Let $H$ be bounded and self-adjoint on $\ell^2(\ZZ)$,
$\phi,\psi\colon\ZZ \to [0,\infty)$ bounded, $z \in \CC$ with $\im z >0$, and $t,S\in (0,\infty)$. Then there is a universal constant $C_{\rm W}\in (0,\infty)$ such that 
\begin{enumerate}
\item[(i)] for all $x, y\in \ZZ$
\begin{multline*}
\sqrt{\phi(x)\psi(y)} \,\mathcal{L} \bigl\{ v_1,v_2 \in [-S,S] : \lvert \langle\delta_x , (H+z - v_1 \phi - v_2 \psi)^{-1}\delta_y \rangle \rvert >t  \bigr \}
\\ \leq 4 C_{\rm W} \frac{S}{t} ,
\end{multline*}
where $\mathcal{L}$ denotes Lebesgue measure.
\item[(ii)] If $\phi(x)\psi(y)\neq0$ and $s \in (0,1)$, we have for all $x, y\in \ZZ$
\begin{equation*}
\int_{[-S,S]^2}
\lvert\langle\delta_x , (H+z - v_1 \phi - v_2 \psi)^{-1}\delta_y \rangle \rvert^s \drm v_1 \drm  v_2
 \leq
\frac{4}{1-s} \left(\frac{C_{\rm W}}{\sqrt{\phi(x)\psi(y)}}\right)^s S^{2-s} .
\end{equation*}
\end{enumerate}
\end{lemma}
To obtain statement (ii) from (i) use the layer cake representation
\[
\int_{[-S,S]^2}  \lvert f(v_1,v_2)\rvert^s \drm v_1 \drm  v_2
= \int_0^\infty  \mathcal{L}\{\lvert v_1\rvert,\lvert v_2\rvert \le S : \lvert f(v_1,v_2)\rvert^s >t\} \drm t
\]
and decompose the integration domain into $[0, \kappa]$ and $[\kappa, \infty)$ with the choice $\kappa = (C_{\rm W} / S \sqrt{\phi(x)\psi(y)} )^s$.

\begin{proposition} \label{p:bounded}
Let $\Gamma \subset \NN$ be connected and Assumption~\ref{ass:d=1_2} be satisfied. 
Assume that  $u$ satisfies condition (a) in Lemma \ref{l:equivalence}. Let $N \in \NN$ be the constant from condition (a) and set $\Lambda^x = \{x,\dots,x+N\}$ and $\Lambda^j=\{j-n+1-N,\dots,j-n+1\}$.
Then we have for all $x,j \in \Gamma$ with $\lvert j-x\rvert \geq 2(N+n) -1$
and all $z \in \CC$ with $\im z>0$
\[
 \mathbb{E}_\Lambda \bigl( \lvert G_\Gamma (z;x,j) \rvert^s  \bigr) \leq C
\]
where $C$ is defined in Eq.~\eqref{e:K} and $\Lambda= \Lambda^x \cup \Lambda^j$.
\end{proposition}
\begin{proof}
Without loss of generality we assume $j-x \geq 2(N+n) -1$ and $\alpha_0=1$.
By assumption $\Gamma \supset \{x,x+1,\dots,j\}$. Note that the operator
$A' = H_\Gamma - z -\sum_{k \in \Lambda^x} \omega_k u_k - \sum_{k \in \Lambda^j} \omega_k u_k$ is independent of $\omega_k$, $k \in \Lambda$.
To estimate the expectation
\begin{align*}
 E &= \mathbb{E}_{\Lambda} \Bigl( \bigl| G_\Gamma (z;x,j) \bigr|^s  \Bigr) \\ &= \int_{[-R,R]^{\lvert\Lambda\rvert}} \Bigl| \Bigl\langle\delta_x , \Bigl(A' + \lambda \sum_{k \in \Lambda^x} \omega_k u_k + \lambda \sum_{k \in \Lambda^j} \omega_k u_k \Bigr)^{-1} \delta_j\Bigr\rangle \Bigr|^s \prod_{k \in \Lambda} \rho(\omega_k) \drm \omega_k
\end{align*}
we use the substitutions
\[
 \begin{pmatrix}
  \omega_x \\[0.8ex]
  \omega_{x+1} \\
  \vdots \\
  \vdots \\
  \omega_{x+N}
 \end{pmatrix}
=
 T
\begin{pmatrix}
  \zeta_x \\[.2ex]
  \zeta_{x+1} \\[.2ex]
  \vdots \\[.2ex]
  \vdots \\[.2ex]
  \zeta_{x+N}
 \end{pmatrix}
\quad \text{ and } \quad
 \begin{pmatrix}
  \omega_{j-n+1-N} \\[0.8ex]
  \omega_{j-n+2-N} \\
  \vdots \\
  \vdots \\
  \omega_{j-n+1}
 \end{pmatrix}
=
 T
 \begin{pmatrix}
  \zeta_{j-n+1-N} \\[.2ex]
  \zeta_{j-n+2-N} \\[.2ex]
  \vdots \\[.2ex]
  \vdots \\[.2ex]
  \zeta_{j-n+1}
 \end{pmatrix}
\]
where the matrix $T$ is the same as in Lemma \ref{lemma:detgen}.
This gives the bound
\begin{multline*}
  E \leq \norm{\rho}_\infty^{\lvert\Lambda\rvert}  \int_{[-S,S]^{\lvert\Lambda\rvert}} 
\Bigl| \Bigl\langle \delta_x ,
    \Bigl(A + \zeta_x \lambda \sum_{k \in \Lambda^x} \alpha_{k-x} u_k \\ +  \zeta_{j-n+1-N} \lambda \sum_{k \in \Lambda^j} \alpha_{k-(j-n+1-N)} u_k \Bigr)^{-1}
\delta_j \Bigr\rangle \Bigr|^s
\drm \zeta_\Lambda .
\end{multline*}
where $\drm \zeta_\Lambda = \prod_{k \in \Lambda} \drm \zeta_k$,
$S = R (1+\max_{i\in \{1,\dots,N\}} \lvert\alpha_i\rvert)$, and
\[
A = A' + \lambda \sum_{k \in \Lambda^x \setminus\{x\}}  \zeta_k u_k + \lambda \!\!\!\!\!\!\!\!\! \sum_{k \in \Lambda^j \setminus\{j-n+1-N\}} \!\!\!\!\!\!\!\!\! \zeta_k u_k
\]
 is independent of $\zeta_x$ and $\zeta_{j-n+1-N}$.
By assumption the functions
\[ 
\phi= \sum_{k \in \Lambda^x} \alpha_{k-x} u_k \quad \text{and} \quad \psi= \sum_{k \in \Lambda^j} \alpha_{k-(j-n+1-N)} u_k
\]
are bounded and  non-negative, with $\phi(x)=  u(0) > 0$ and $\psi(j) = \alpha_N u(n-1) > 0$.
Using Lemma \ref{l:AENSS} we obtain
\begin{align*}
 E' &= \int_{[-S,S]^{2}} \Bigl| \Bigl\langle \delta_x , \bigl(A + \zeta_x \lambda \phi + \zeta_{j-n+1-N} \lambda \psi \bigr)^{-1} \delta_j \Bigr\rangle \Bigr|^s \drm \zeta_x \drm \zeta_{j-n+1-N} \\[1ex]
&\leq
\frac{4}{1-s} \left(\frac{C_W}{\lambda \sqrt{\phi(x)\psi(j)}}\right)^s S^{2-s} .
\end{align*}
Thus the original integral is estimated by
\begin{align}
  E &\leq
\norm{\rho}_\infty^{\lvert\Lambda\rvert}  (2S)^{\lvert\Lambda\rvert-2}
\frac{4}{1-s} \left(\frac{C_W}{\lambda \sqrt{\phi(x)\psi(j)}}\right)^s S^{2-s} \nonumber
 \\[1ex]
 &=
\frac{4}{1-s} \left(\frac{C_W}{\lambda \sqrt{u(0) \alpha_N u(n-1) }}\right)^s
\big(2S\norm{\rho}_\infty \big)^{2(N+1)} \frac{1}{S^{s}} =: C . \label{e:K} ,
\end{align}
which ends the proof.
\end{proof}
The last proposition and a formula analogous to \eqref{e:geometric-resolvent} now give for $j=x + 2(N+n)-1$ and
$x + 2(N+n) \le y$
\begin{align*}
\EE_\Lambda & \Bigl( \bigl| G_{\Gamma_0} (z;x,y) \bigr|^s  \Bigr)
\\ &=\EE_\Lambda \Bigl( \bigl| G_{\Gamma_0} (z;x,x + 2(N+n)-1) \bigr|^s  \Bigr)
 \bigl| G_{\Gamma_1} (z;x + 2(N+n),y) \bigr|^s \\
& \leq
C \bigl| G_{\Gamma_1} (z;x + 2(N+n),y) \bigr|^s
\end{align*}
where $\Gamma_0 = \ZZ$, $\Gamma_1 = \{x + 2(N+n), x + 2(N+n)+1, \dots\}$ and $\Lambda$ as in Lemma \ref{l:AENSS}. In an appropriate large disorder regime, where the constant $C$ in \eqref{e:K}
is smaller than one, exponential decay of fractional moments for the off-diagonal matrix elements now follows by iteration similarly as in Theorem \ref{theorem:exp_d=1} and Theorem~\ref{theorem:exp_d=1_2}. 
\par
To conclude exponential localization we also need the boundedness for fractional moments of the diagonal Green's function elements (cf.\ Theorem \ref{prop:replace-msa}), which is proven in Theorem~\ref{t:a-priori_d=1}. However, for the proof of Theorem~\ref{t:a-priori_d=1} we need non-monotone spectral averaging as formulated in Lemma~\ref{lemma:averagenorm}.
\section{Boundedness of fractional moments} \label{sec:boundedness}
In this section we prove the boundedness of an averaged fractional power of the Green function in arbitrary space dimension $d \in \NN$. In particular, the upper bound depends in a quantitative way on the disorder (the bound gets small in the high disorder regime). The estimate on fractional moments of the Green function is used iteratively in the next section, where we prove exponential decay of the Green function.
\par
We consider the situation when Assumption \ref{ass:monotone} holds.
Recall that $R = \max \{ \lvert \inf \supp \rho \rvert , \lvert \sup \supp \rho \rvert \}$ where $\rho$ is the probability density of the measure $\nu$.
\par
We also introduce some more notation. For $x \in \ZZ^d$ we denote by $\cN (x) = \{k \in \ZZ^d : |x-k|_1 = 1\}$ the neighborhood of $x$. We also define $\Lambda^+ = \Lambda \cup \partial^{\rm o} \Lambda$, $\Lambda_x = \Lambda + x = \{k \in \ZZ^d : k-x \in \Lambda\}$ for $\Lambda \subset \ZZ^d$, and $u_{\rm min}^\Lambda = \min_{k \in \Lambda} \lvert u(k) \rvert$.
\begin{lemma}[A priori bound] \label{lemma:bounded}
Let Assumption~\ref{ass:monotone} be satisfied, $\Gamma \subset \ZZ^d$, $m > 0$ and $s \in (0,1)$.
\begin{enumerate}[(a)]
 \item Then there is a constant $C_s$, depending only on $d$, $\rho$, $u$, $m$ and $s$, such that for all $z \in \CC \setminus \RR$ with $|z| \leq m$, all $x,y \in \Gamma$ and all $b_x,b_y \in \ZZ^d$ with $x \in \Theta_{b_x}$ and $y \in \Theta_{b_y}$
\[
\EE_{N} \Bigl( \bigl\lvert G_\Gamma (z;x,y) \bigr\rvert^{s/(2|\Theta|)} \Bigr) \leq C_s \Xi_s (\lambda),
\]
where $\Xi_s (\lambda) = \max \{ \lambda^{- s/(2 \lvert \Theta \rvert)} , \lambda^{-2s} \}$ and $N = \{b_x , b_y\} \cup \cN (b_x) \cup \cN (b_y)$.
\item Then there is a constant $D_s$, depending only on $d$, $\rho$, $u$ and $s$, such that for all $z \in \CC \setminus \RR$, all $x,y \in \Gamma$ and all $b_x,b_y \in \ZZ^d$ with
\[
x \in \Theta_{b_x} \cap \Gamma \subset \partial^{\rm i} \Theta_{b_x} \quad \text{and} \quad
y \in \Theta_{b_y} \cap \Gamma \subset \partial^{\rm i} \Theta_{b_y}
\]
we have
\[
 \EE_{\{b_x , b_y\}} \Bigl( \bigl\lvert G_\Gamma (z;x,y) \bigr\rvert^s \Bigr) \leq D_s \lambda^{-s} .
\]
\end{enumerate}
\end{lemma}
\begin{proof}
First we prove (a). Fix $x,y\in\Gamma$ and choose $b_x, b_y \in \ZZ^d$ in such a way that $x \in \Theta_{b_x}$ and $y \in \Theta_{b_y}$. This is always possible, and sometimes even with a choice $b_x = b_y$. However, we assume $b_x \not = b_y$. The case $b_x = b_y$ is similar but easier. Let us note that $\Theta_{b_x}$ and $\Theta_{b_y}$ are not necessarily disjoint. We apply Lemma~\ref{lemma:schur2} with $\Lambda_1 = (\Theta_{b_x} \cup \Theta_{b_y}) \cap \Gamma$ and $\Lambda_2 = \Lambda_1^+ \cap \Gamma$ and obtain
\begin{equation} \label{eq:lemma:bounded1}
\Pro_{\Lambda_1}^{\Gamma}(H_\Gamma - z)^{-1} (\Pro_{\Lambda_1}^{\Gamma})^* =
\bigl( H_{\Lambda_1} - z + \Pro_{\Lambda_1} \Delta \Pro_{\partial^{\rm o} \Lambda_1}^* (K-z)^{-1} \Pro_{\partial^{\rm o} \Lambda_1} \Delta \Pro_{\Lambda_1}^* \bigr)^{-1}
\end{equation}
where
\[
 K = H_{\partial^{\rm o} \Lambda_1} - \Pro_{\partial^{\rm o} \Lambda_1}^{\Lambda_1^+} B_\Gamma^{\Lambda_1^+} (\Pro_{\partial^{\rm o} \Lambda_1}^{\Lambda_1^+})^* .
\]
We note that $B_\Gamma^{\Lambda_1^+}$ depends only on the potential values $V_\omega (k)$, $k \in \Gamma \setminus \Lambda_1^+$ and is hence independent of $\omega_k$, $k \in \{b_x , b_y\} \cup \cN (b_x) \cup \cN (b_y)$. We also note that $K$ is independent of $\omega_{b_x}$ and $\omega_{b_y}$, and that the potential values $V_\omega (k)$, $k \in \partial^{\rm o} \Lambda_1$ depend monotonically on $\omega_k$, $k \in \cN (b_x) \cup \cN (b_y) =: N'$, by Assumption~\ref{ass:monotone}. More precisely, we can decompose $K : \ell^2 (\partial^{\rm o} \Lambda_1) \to \ell^2 (\partial^{\rm o} \Lambda_1)$ according to
\[
 K = A + \lambda \sum_{k \in N'} \omega_k V_k
\]
with some $A,V_k : \ell^2 (\partial^{\rm o} \Lambda_1) \to \ell^2 (\partial^{\rm o} \Lambda_1)$ and the properties that $A$ is independent of $\omega_k$, $k \in N'$, and $V := \sum_{k \in N'} V_k$ is diagonal and strictly positive definite with $V \geq u_{\rm min}^{\partial^{\rm i} \Theta}$. We fix $v \in N'$ and obtain with the transformation $\omega_v = \zeta_v$ and $\omega_i = \zeta_v + \zeta_i$ for $i \in N' \setminus \{v\}$ for all $t \in (0,1)$
\begin{align}
\EE_{N'} \Bigl( \bigl\lVert (K-z)^{-1} \bigr\rVert^t \Bigr)
& = \!\!\!\!\!\!\!\!\! \int\limits_{[-R , R]^{\lvert N' \rvert}} \!\!\!\!\! \bigl\lVert (A - z + \lambda \sum_{k \in N'} \omega_k V_k)^{-1} \bigr\rVert^t \prod_{k \in N'} \rho (\omega_k) \drm \omega_k \nonumber \\[1ex]
& \leq \lVert \rho \rVert_\infty^{\lvert N' \rvert - 1} \!\!\!\!\!\!\!\!\! \int\limits_{[-S,S]^{\lvert N' \rvert}} \!\!\!\! \bigl\lVert (\tilde A + \zeta_v \lambda V)^{-1} \bigr\rVert^t \rho(\zeta_v) \drm \zeta_v \!\!\!\! \prod_{i \in N' \setminus \{v\}} \!\!\!\! \drm \zeta_i \label{eq:K}
\end{align}
where $S = 2R$ and $\tilde A = A - z + \lambda \sum_{k \in N' \setminus \{v\}} \zeta_i V_i$. The monotone spectral averaging estimate in Lemma~\ref{lemma:monotone2} gives for $t \in (0,1)$
\begin{equation*}
\EE_{N'} \Bigl( \bigl\lVert (K-z)^{-1} \bigr\rVert^t \Bigr) \leq
\frac{\lVert \rho \rVert_\infty^{|N'|-1} (4R)^{|N'|-1} (C_{\rm W} \lvert \partial^{\rm o} \Lambda_1 \rvert \lVert \rho \rVert_\infty)^t}{(u_{\rm min}^{\partial^{\rm i} \Theta} \lambda)^t (1-t)} .
\end{equation*}
Hence there is a constant $C_1 (t)$ depending only on $\rho$, $u$, $d$, $\Lambda_1$ and $t$, such that
\begin{equation}  \label{eq:lemma:bounded2}
 \EE_{N'} \Bigl( \bigl\lVert (K-z)^{-1} \bigr\rVert^t \Bigr) \leq \frac{C_1 (t)}{\lambda^t} .
\end{equation}
We use the notation $u_j$ for the translates of $u$, i.\,e. $u_j (x) = u(x-j)$ for all $j,x \in \ZZ^d$, as well as for the corresponding multiplication operator. The operator $H_{\Lambda_1} = -\Delta_{\Lambda_1} + V_{\Lambda_1}$ can be decomposed in $H_{\Lambda_1} = \tilde A' + \lambda \omega_{b_x} V_x + \lambda \omega_{b_y}  V_y$, where the multiplication operators $V_x,V_y : \ell^2 (\Lambda_1) \to \ell^2 (\Lambda_1)$ are given by $V_x (k) = u_{b_x} (k)$ and $V_y (k) = u_{b_y} (k)$, and where $\tilde A' = -\Delta_{\Lambda_1} + \lambda\sum_{k \in \ZZ^d \setminus \{b_x,b_y\}} \omega_k u_k$. Notice that $V_x$ is invertible on $\Theta_{b_x}$ and $V_y$ is invertible on $\Theta_{b_y}$. Hence there exists an $\alpha \in (0,1]$ such that $V_x + \alpha V_y$ is invertible on $\Lambda_1$. By Eq.~\eqref{eq:lemma:bounded1} and this decomposition we have for all $t \in (0,1)$
\begin{align*}
 E &= \EE_{\{b_x,b_y\}} \Bigl(\bigl\lVert \Pro_{\Lambda_1}^{\Gamma}(H_\Gamma - z)^{-1} (\Pro_{\Lambda_1}^{\Gamma})^* \bigr\rVert^{t / |\Lambda_1|} \Bigr) \\[1ex]
&= \int_{-R}^R \int_{-R}^R \bigl\lVert ( A' + \lambda \omega_{b_x} V_x + \lambda \omega_{b_y} V_y )^{-1} \bigr\rVert^{t/|\Lambda_1|} \rho(\omega_{b_x})\rho(\omega_{b_y}) \drm \omega_{b_x} \drm \omega_{b_y} ,
\end{align*}
where
\[
 A' = \tilde A' - z + \Pro_{\Lambda_1} \Delta \Pro_{\partial^{\rm o} \Lambda_1}^* (K-z)^{-1} \Pro_{\partial^{\rm o} \Lambda_1} \Delta \Pro_{\Lambda_1}^* .
\]
Notice that $\tilde A'$ and $K$ are independent of $\omega_{b_x}$ and $\omega_{b_y}$. Set $V := V_x + \alpha V_y$ with $\alpha \in (0,1]$ to be chosen later. We use the transformation $\omega_{b_x} = \zeta_{x}$, $\omega_{b_y} = \alpha \zeta_x + \zeta_y$ and obtain by Lemma~\ref{lemma:averagenorm}
\begin{align*}
 E &\leq \lVert \rho \rVert_\infty \int_{-2R}^{2R} \int_{-2R}^{2R} \bigl\lVert (A' + \zeta_y \lambda V_y + \zeta_x \lambda V)^{-1} \bigr\rVert^{t/|\Lambda_1|} \rho (\zeta_x) \drm \zeta_x \drm \zeta_y \\[1ex]
& \leq \lVert \rho \rVert_\infty \int_{-2R}^{2R}
  \frac{\lVert \rho \rVert_\infty^t \bigl( \lVert A' + \zeta_y \lambda V_y \rVert + 2R \lambda \lVert V \rVert \bigr)^{t(|\Lambda_1| - 1)/|\Lambda_1|}}
  {t^t 2^{-t} (1-t) \lambda^t \lvert \det V \rvert^{t/|\Lambda_1|}} \drm \zeta_y \\[1ex]
& \leq
  \frac{4R \lVert \rho \rVert_\infty^{t+1} \bigl( \lVert A' \rVert + 2R \lambda \lVert V_y \rVert + 2R \lambda \lVert V \rVert \bigr)^{t(|\Lambda_1| - 1)/|\Lambda_1|}}
  {t^t 2^{-t} (1-t) \lambda^t \lvert \det V \rvert^{t/|\Lambda_1|}} .
\end{align*}
The norm of $A'$ can be estimated as
\[
 \lVert  A' \rVert \leq 2d + (|\Theta| - 1) \lVert u \rVert_\infty + m + (2d)^2 \lVert (K-z)^{-1} \rVert .
\]
For the norm of $V_y$ and $V$ we have $\lVert V_y \rVert \leq \lVert u \rVert_\infty$ and $\lVert V \rVert \leq 2 \lVert u \rVert_\infty$. To estimate the determinant of $V$ we set $v_i = (u (i - b_x) , u(i - b_y))^{\rm T} \in \RR^2$ for $i \in \Lambda_1$, and $r = (1,\alpha)^{\rm T} \in \RR^2$. Then,
\[
 \lvert \det V \rvert =  \prod_{i \in \Lambda_1} \bigl\lvert u(i-b_x) + \alpha u(i-b_y) \bigr\rvert = \prod_{i \in \Lambda_1} \lVert v_i \rVert \bigl\lvert \langle r , v_i / \lVert v_i \rVert \rangle \bigr\rvert .
\]
Since we can choose $\alpha \in (0,1]$ in such a way that the distance of $r$ to each hyperplane $H_i = \{x_1,x_2 \in \RR : u(i-b_x)x_1 + u(i-b_y) x_2 = 0\}$, $i \in \Lambda_1$, is at least $d_0 = \sqrt{2} / (4(\lvert \Lambda_1 \rvert + 1))$, we conclude using $\lVert v_i \rVert \geq \sqrt{2} u_{\rm \min}^\Theta$
\[
  \lvert \det V \rvert \geq \prod_{i \in \Lambda_1} \lVert v_i \rVert d_0 \geq \left( \frac{u_{\rm \min}^\Theta}{2(\lvert \Lambda_1 \rvert + 1)} \right)^{\lvert \Lambda_1 \rvert} .
\]
Putting all together we see that there are constants $C_2 (t)$, $C_3 (t)$ and $C_4 (t)$ depending only on $\rho$, $u$, $d$, $m$, $\Lambda_1$ and $t$, such that
\begin{equation} \label{eq:withoutA}
 E \leq \frac{C_2(t)}{\lambda^t} + \frac{C_3(t)}{\lambda^{t/\lvert \Lambda_1 \rvert}} + \frac{C_4(t)}{\lambda^t} \lVert (K-z)^{-1} \rVert^{t \frac{|\Lambda_1| - 1}{|\Lambda_1|}} .
\end{equation}
If we average with respect to $\omega_k$, $k \in \cN(b_x)\cup \cN (b_y)$ we obtain by Eq.~\eqref{eq:lemma:bounded2}
\[
  \EE_{\cN(b_x)\cup \cN (b_y)} \bigl( E \bigr) \leq \frac{C_2(t)}{\lambda^t} + \frac{C_3 (t)}{\lambda^{t/\lvert \Lambda_1 \rvert}} + \frac{C_4 (t) C_1 (t(|\Lambda_1| - 1)/|\Lambda_1|)}{\lambda^t \lambda^{t(|\Lambda_1| - 1)/|\Lambda_1|}} .
\]
Notice that $1 \leq \lvert \Lambda_1 \rvert \leq 2 \lvert \Theta
\rvert$. Now we choose $t = s|\Lambda_1|/(2|\Theta|)$ and eliminate
$\Lambda_1$ from the constants $C_1(t)$, $C_2(t)$, $C_3(t)$ and
$C_4(t)$ by maximizing them with respect to $\lvert \Lambda_1 \rvert
\in \{1, \dots , 2 \lvert \Theta \rvert\}$. We obtain that there are
constants $\tilde C_1 (s) , \tilde C_2 (s)$ and $\tilde C_3 (s)$,
depending only on $\rho$, $u$, $d$, $m$, and $s$, such that
\begin{align*}
 \EE_N \Bigl(\bigl\lVert \Pro_{\Lambda_1}^{\Gamma}(H_\Gamma - z)^{-1} (\Pro_{\Lambda_1}^{\Gamma})^* \bigr\rVert^{\frac{s}{2\lvert \Theta \rvert}} \Bigr)
& \leq
\frac{\tilde C_1 (s)}{\lambda^{s\frac{\lvert \Lambda_1 \rvert}{2\lvert \Theta \rvert}}}+
\frac{\tilde C_2 (s)}{\lambda^{\frac{s}{2\lvert \Theta \rvert}}}+
\frac{\tilde C_3 (s)}{\lambda^{s \frac{2\lvert \Lambda_1 \rvert - 1}{2 \lvert \Theta \rvert}}} \\[1ex]
& \leq
(\tilde C_1 (s)+\tilde C_2 (s)+\tilde C_3 (s)) \Xi_s (\lambda) .
\end{align*}
In the last estimate we have distinguished the cases $\lambda \geq 1$ and $\lambda < 1$ and used the fact that $1 \leq \lvert \Lambda_1 \rvert \leq 2 \lvert \Theta
\rvert$. This completes the proof of part (a).
\par
To prove (b) we fix $x,y \in \Gamma$ and $b_x,b_y \in \ZZ^d$ with $x \in \Theta_{b_x} \cap \Gamma \subset \partial^{\rm i} \Theta_{b_x}$ and $y \in \Theta_{b_y} \cap \Gamma \subset \partial^{\rm i} \Theta_{b_y}$. We again assume $b_x \not = b_y$. The case $b_x = b_y$ is similar but easier. We apply Lemma \ref{lemma:schur1} with $\Lambda = (\Theta_{b_x} \cup \Theta_{b_y})\cap \Gamma$ and obtain
\[
\Pro_\Lambda^\Gamma (H_\Gamma - z)^{-1} (\Pro_\Lambda^\Gamma)^* =
(H_\Lambda - B_\Gamma^\Lambda - z)^{-1}.
\]
Notice that $B_\Gamma^\Lambda$ is independent of $\omega_k$, $k \in \{b_x,b_y\}$.
By assumption, the potential values in $\Lambda$ depend \textit{monotonically} on $\omega_{b_x}$ and $\omega_{b_y}$. More precisely, we can rewrite the potential in the form $V_\Lambda = A + \omega_{b_x} \lambda V_x + \omega_{b_y} \lambda V_y$ with the properties that $A$ is independent of $\omega_k$, $k \in \{b_x,b_y\}$, and $V = V_x + V_y$ is strictly positive definite with $V \geq u_{\rm min}^{\partial^{\rm i} \Theta}$. We proceed similarly as in Ineq.~\eqref{eq:K}, namely with the substitution $\omega_{b_x}=\zeta_x$ and $\omega_{b_y} = \zeta_x+\zeta_y$, and obtain using monotone spectral averaging from Lemma~\ref{lemma:monotone2} the estimate
\[
 \EE_{\{b_x,b_y\}} \Bigl( \bigl\lVert \Pro_\Lambda^\Gamma (H_\Gamma - z)^{-1} (\Pro_\Lambda^\Gamma)^* \bigr\rVert^s \Bigr) \leq \lVert \rho \rVert_\infty 4R \frac{(\lvert \Lambda \rvert u_{\rm min}^{\partial^{\rm i} \Theta} \lVert \rho \rVert_\infty)^s}{\lambda^s (1-s)} .
\]
We estimate $\lvert \Lambda \rvert \leq 2 \lvert \Theta \rvert$ and obtain part (b).
\end{proof}
\begin{remark}
 Note that even if Assumption \ref{ass:monotone} is not satisfied we obtain the bound \eqref{eq:withoutA}, namely
\begin{multline*}
 \EE_{\{b_x,b_y\}} \Bigl(\bigl\lVert \Pro_{\Lambda_1}^{\Gamma}(H_\Gamma - z)^{-1} (\Pro_{\Lambda_1}^{\Gamma})^* \bigr\rVert^{t / |\Lambda_1|} \Bigr) \\ \leq \frac{C_2(t)}{\lambda^t} + \frac{C_3(t)}{\lambda^{t/\lvert \Lambda_1 \rvert}} + \frac{C_4(t)}{\lambda^t} \lVert (K-z)^{-1} \rVert^{t \frac{|\Lambda_1| - 1}{|\Lambda_1|}} .
\end{multline*}
\end{remark}
\section{Exponential decay of fractional moments; proof of Theorem \ref{theorem:finite_volume} and \ref{theorem:exp_decay}} \label{sec:exp_decay}
In this section we show Theorem \ref{theorem:finite_volume}, i.e.\ that the so called finite volume criterion implies exponential decay of the Green function. Together with the a-priori bound (Lemma \ref{lemma:bounded}) this gives us Theorem~\ref{theorem:exp_decay} which will be proven at the end of this section.
\par
We again consider ``depleted'' Hamiltonians, as already introduced in the one-dimensional setting in Section~\ref{sec:d=1}, to formulate a geometric
resolvent formula. Let $\Lambda \subset \Gamma \subset \ZZ^d$ be
arbitrary sets. We define the depleted Laplace operator
$\Delta_\Gamma^\Lambda :\ell^2 (\Gamma) \to \ell^2 (\Gamma)$ by
\begin{equation*} 
 \sprod{\delta_x}{\Delta_\Gamma^\Lambda \delta_y} :=
\begin{cases}
  0 & \text{if $x \in \Lambda$, $y \in \Gamma \setminus \Lambda$ or $y \in \Lambda$,
  $x \in \Gamma \setminus \Lambda$} , \\
  \bigl \langle \delta_x , \Delta_\Gamma \delta_y \bigr \rangle & \text{else} 
\end{cases}
\end{equation*}
and the depleted Hamiltonian $H_\Gamma^\Lambda : \ell^2 (\Gamma) \to \ell^2 (\Gamma)$ by
\begin{equation*} \label{eq:depl}
 H_\Gamma^\Lambda := -\Delta_\Gamma^\Lambda + V_\Gamma .
\end{equation*}
Let further $T_\Gamma^\Lambda := \Delta_\Gamma -
\Delta_\Gamma^\Lambda$ be the difference between the the ``full''
Laplace operator and the depleted Laplace operator. For $z \in \CC \setminus \RR$ and $x,y \in \Gamma$ we use the notation $G_\Gamma^\Lambda (z) := (H_\Gamma^\Lambda - z)^{-1}$ and $G_\Gamma^\Lambda (z;x,y) := \bigl \langle \delta_x, G_\Gamma^\Lambda(z) \delta_y \bigr \rangle$. To formulate a geometric resolvent formula we apply the second resolvent identity and obtain for arbitrary sets $\Lambda \subset \Gamma \subset \ZZ^d$
\begin{equation}
G_\Gamma (z) = G_\Gamma^\Lambda (z) + G_\Gamma (z) T_\Gamma^\Lambda  G_\Gamma^\Lambda (z) = G_\Gamma^\Lambda (z) + G_\Gamma^\Lambda (z) T_\Gamma^\Lambda G_\Gamma (z) .
\label{eq:firstorder}
\end{equation}
%
In contrast to the one-dimensional case from Section~\ref{sec:d=1} it will be necessary to use an iterated version of this formula.
Namely, two applications of the resolvent identity give
\begin{equation}
G_{\Gamma} (z)  = G_{\Gamma}^{\Lambda} (z) + G_{\Gamma }^{\Lambda} (z) T_\Gamma^{\Lambda} G_{\Gamma}^{\Lambda} (z) + G_{\Gamma }^{\Lambda} (z)  T_\Gamma^{\Lambda} G_{\Gamma} (z) T_\Gamma^{\Lambda} G_{\Gamma}^{\Lambda} (z) .
\label{eq:secondorder}
\end{equation}
\begin{remark} \label{remark:depleted}
Notice that $G_\Gamma^\Lambda (z;x,y) = G_\Lambda (z;x,y)$ if $x,y \in \Lambda$, $G_\Gamma^\Lambda (z;x,y) = 0$ if $x \in \Lambda$ and $y \not \in \Lambda$ or vice versa, and that $G_\Gamma^\Lambda (z) = G_\Gamma^{\Lambda^{\rm c}} (z)$. If $\Gamma \setminus \Lambda$ decomposes into at least two components which are not connected, and $x$ and $y$ are not in the same component, then we also have $G_\Gamma^\Lambda (z;x,y) = 0$.
\par
Since $\Gamma$ is not necessarily the whole lattice
$\ZZ^d$, it may be that terms of the type $G_{\Gamma} (z;i,j)$ occur for
some $\Gamma \subset \ZZ^d$ and some $i \not \in \Gamma$ or $j \not \in \Gamma$. In this case we use the
convention that $G_\Gamma (z;i,j) = 0$.
\end{remark}
To formulate the results of this section we introduce the following
notation. For finite $\Gamma \subset \ZZ^d$ we denote by $\diam \Gamma$ the diameter of $\Gamma$ with respect to the supremum norm, i.e.\ $\diam \Gamma = \sup_{x,y\in \Gamma} \lvert x-y \rvert_\infty$. For $\Gamma \subset \ZZ^d$, $L \ge \diam \Theta + 2$ and
\[
B = \partial^{\rm i} \Lambda_L,
\]
we define for $x \in \ZZ^d$ the sets\footnote{Note that the sets $\hat \Lambda_x$, $\hat W_x$, $\Lambda_x$ and $W_x$ are not translates of certain sets $\hat \Lambda$, $\hat W$, $\Lambda$ and $W$. They are defined directly in dependence of $x \in \ZZ^d$ and $L \geq \diam \Theta + 2$. Contrary, the sets $B_x$, $\Theta_b$ are translates of the sets $B$ and $\Theta$ by the vextors $x,b \in \ZZ^d$.}
\[
\hat \Lambda_x := \{ k \in \Gamma : k \in \Theta_b \text{ for some $b \in \Lambda_{L,x}$} \}
\]
and
\begin{equation} \label{eq:Wx}
 \hat W_x := \{ k \in \Gamma : k \in \Theta_b \text{ for some $b \in B_x$} \} .
\end{equation}
Recall that $\Lambda_{L,x}= \{k \in \ZZ^d \colon \lvert k-x \rvert_\infty \leq L\}$, $\Lambda_L = \Lambda_{L,0}$, for $\Gamma \subset \ZZ^d$ and $x \in \ZZ^d$ we denote by $\Gamma_x = \Gamma + x = \{k \in \ZZ^d : k-x \in \Gamma\}$ the translate of $\Gamma$ and by $\Gamma^+ = \Gamma \cup \partial^{\rm o} \Gamma$ the thickened set. Hence $\hat W_x$ is the union of translates of $\Theta$ along the sides of $B_x$, restricted to the set $\Gamma$. For $\Gamma \subset \ZZ^d$ we can now introduce the sets
\[
\Lambda_x: = \hat \Lambda_x^+ \cap \Gamma \quad \text{and} \quad W_x := \hat W_x^+ \cap \Gamma
\]
which will play a role in the assertions below. See also Fig.~\ref{fig:geometric2} for an illustration of the set $W_x$. 
\begin{theorem}[Finite volume criterion] \label{theorem:finite_volume2}
Suppose that Assumption \ref{ass:monotone} is satisfied, let $\Gamma \subset \ZZ^d$, $L \geq \diam \Theta + 2$, $z
\in \CC \setminus \RR$ with $\lvert z \rvert \leq m$ and $s \in (0,1/3)$. Then there
exists a constant $B_s$ which depends only on $d$, $\rho$, $u$, $m$ and $s$,
such that if the condition
\begin{equation}\label{eq:fin_cond}
b_s(\lambda, L,\Lambda): = \frac{B_s L^{3(d-1)} \Xi_s (\lambda)}{\lambda^{2s/(2\lvert
\Theta \rvert)}}\, \sum_{w\in\partial^{\rm o} W_x}\mathbb{E}
\bigl(\lvert G_{\Lambda\setminus W_x} (z;x,w)\rvert^{s/(2\lvert
\Theta \rvert)}\bigr)< b
\end{equation}
is satisfied for some $b \in (0,1)$, arbitrary $\Lambda \subset \Gamma$, and all $x\in
\Lambda$, then for all $x,y \in \Gamma$
\begin{equation*} 
\mathbb{E} \bigl(\lvert G_\Gamma (z;x,y)\rvert^{s/(2\lvert \Theta
\rvert)} \bigr)\leq A \euler^{-\mu|x-y|_\infty} .
\end{equation*}
Here
\[
A=\frac{C_s \Xi_s (\lambda)}{b} \quad \text{and} \quad
\mu=\frac{\lvert \ln b \rvert}{L+\diam \Theta + 2}\,,\] with $C_s$ inherited
from the a priori bound (Lemma \ref{lemma:bounded}).
\end{theorem}
\begin{remark} \label{remark:finite_volume}
Note that $\Gamma \setminus W_x$ decomposes into  two components
which are not connected, so that the sum in \eqref{eq:fin_cond} runs
over the sites $r$ related to only one of these components, which is
always compact, regardless of the choice of $\Gamma$. It then
follows that in order to establish the exponential falloff of the
Green function it suffices to consider the decay properties of the
Green function for the Hamiltonians defined on finite sets. The
finite volume criterion derives its name from this fact.
\end{remark}
The strategy for the proof is reminiscent of the one developed in
\cite{AizenmanFSH-01} and is aimed to derive a following bound on
the average Green function.
\begin{lemma} \label{lemma:iteration3}
Let $\Gamma \subset \ZZ^d$, $s \in (0, 1/3)$, $m > 0$, Assumption \ref{ass:monotone} be satisfied and $b_s (\lambda,L,\Lambda)$ be the constant from Theorem \ref{theorem:finite_volume2}. Then we have for all $x,y \in \Gamma$ with $y \not \in \Lambda_x$ and all $z \in \CC \setminus \RR$ with $\lvert z \rvert \leq m$ the bound
\begin{equation}\label{eq:protobound}
\mathbb{E} \bigl(\lvert G_\Gamma (z;x,y)\rvert^{\frac{s}{2\lvert
\Theta \rvert}}\bigr)\leq
\frac{b_s(\lambda, L,\Gamma)}{|\partial^{\rm o} \Lambda_x|} \sum_{r\in\partial^{\rm o}
\Lambda_x} \mathbb{E}\bigl( \lvert G_{\Gamma\setminus
\Lambda_x}(z;r,y)\rvert^{\frac{s}{2\lvert \Theta \rvert}} \bigr) .
\end{equation}
\end{lemma}
\begin{figure}[t] \centering
\begin{tikzpicture}[scale=0.3]
\usetikzlibrary{patterns}
\draw[very thick] (10,10) rectangle (-10,-10);
\draw[very thin] (0,12)--(12,12)--(12,-4)--(14,-4)--(14,14)--(0,14)--(0,16)--(16,16)--(16,-6)--(12,-6)--(12,-12)--(0,-12)--(0,-8)--(8,-8)--(8,-6)--(0,-6)--(0,-4)--(8,-4)--(8,8)--(0,8)--(0,12);
\filldraw[very thin, opacity=0.07] (0,12)--(12,12)--(12,-4)--(14,-4)--(14,14)--(0,14)--(0,16)--(16,16)--(16,-6)--(12,-6)--(12,-12)--(0,-12)--(0,-8)--(8,-8)--(8,-6)--(0,-6)--(0,-4)--(8,-4)--(8,8)--(0,8)--(0,12);
\draw[dotted] (0,-10)--(0,15);
\filldraw (0,0) circle (5pt); \draw (1,0) node {$x$};
\filldraw (5,7.3) circle (5pt); \draw (3.9,7.3) node {$u$};
\filldraw (5,8.7) circle (5pt); \draw (4,8.7) node {$u'$};
\filldraw (15.3,4) circle (5pt); \draw (15.3,2.85) node {$v$};
\filldraw (16.7,4) circle (5pt); \draw (16.7,3) node {$v'$};
\filldraw (22,8) circle (5pt); \draw (23,8) node {$y$};
\draw (0,0)--(5,7.3);
\draw[very thin, opacity=0.5] (5,7.3)--(5,8.7);
\draw[very thick] (5,8.7)--(15.3,4);
\draw[very thin, opacity=0.5] (15.3,4)--(16.7,4);
\draw (16.7,4)--(22,8);
\end{tikzpicture}
\caption[Illustration of the geometric setting and Eq.~\eqref{eq:diagramatic}]{Illustration of the geometric setting and Eq.~\eqref{eq:diagramatic}
in the case $d = 2$, $\Gamma = \{x \in \mathbb{Z}^2 : x_1 \geq 0\}$, $x = 0$
and $\Theta = ([-2,2]^2 \cup [4,6]^2)\cap \mathbb{Z}^2$.
The light grey region is the set $\hat W_x$ and the black square is the sphere $B_x$.}
\label{fig:geometric2}
\end{figure}
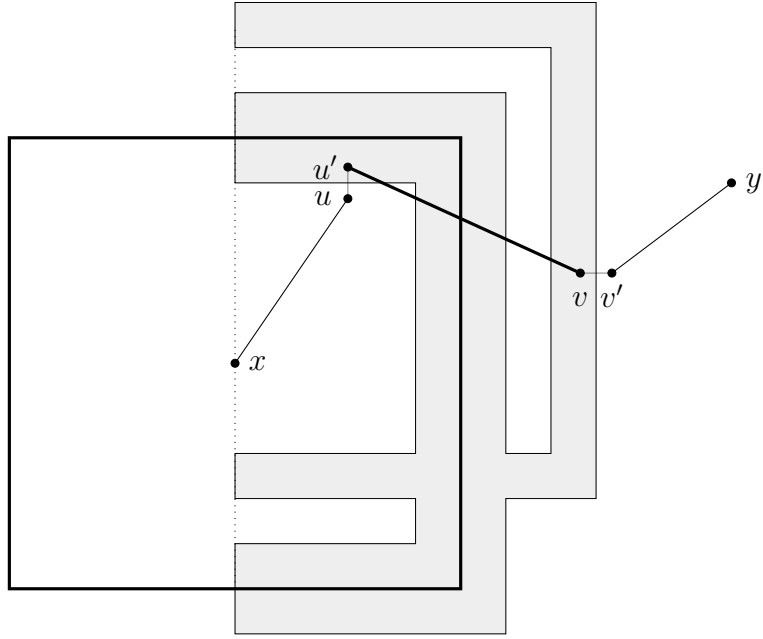
\begin{remark}
Inequality may be iterated, each iteration resulting in an
additional factor of $b_s(\lambda, L,\Gamma)$. Note that each
iteration step brings in Green functions that correspond to modified
domains.
\end{remark}
The finite volume criterion is a direct corollary of Lemma~\ref{lemma:iteration3}.
\begin{proof}[Proof of Theorem \ref{theorem:finite_volume2}]
Inequality~\eqref{eq:protobound} can be iterated as long as
the resulting sequence ($x, r^{(1)}, \ldots, r^{(n)}$) do not get
closer to $y$ than the distance $\tilde L = L + \diam \Theta + 2$.
\par
If $\lvert x-y \rvert_\infty \geq \tilde L$, we iterate Ineq.~\eqref{eq:protobound}
exactly $\lfloor \lvert x-y \rvert_\infty/ \tilde L \rfloor$ times, use the a-priori bound from
Lemma~\ref{lemma:bounded} and obtain
\begin{equation*}
\EE \Bigl ( \bigl| G_\Gamma(z;x,y) \bigr|^{\frac{s}{2\lvert \Theta \rvert}} \Bigr)
\leq
C_s \Xi_s (\lambda) \cdot b^{\textstyle \lfloor \lvert x-y \rvert_\infty/ \tilde L \rfloor}
\leq
\frac{C_s \Xi_s (\lambda)} {b} {\rm e}^{- \mu |x-y|_\infty} , \label{eq:thm1optimalbound}
\end{equation*}
with $\mu = \lvert \ln b \rvert / \tilde L$. If
$\lvert x-y \rvert_\infty < \tilde L$, we use
Lemma~\ref{lemma:bounded} and see that
\[
 \EE \Bigl ( \bigl| G_\Gamma(z;x,y) \bigr|^{s/(2\lvert \Theta \rvert)} \Bigr)
  \leq C_s \Xi_s (\lambda) \leq \frac{C_s \Xi_s (\lambda)}{b} {\rm e}^{- \mu |x-y|_\infty} . \qedhere
\]
\end{proof}
To facilitate the proof of Lemma \ref{lemma:iteration3} we introduce
some extra notation first. Namely, for a set $\Lambda \subset
\ZZ^d$, we define the bond-boundary $\partial^{\rm B} \Lambda$ of
$\Lambda$ as
\[
\partial^{\rm B} \Lambda = \left\{(u,u') \in \ZZ^d \times \ZZ^d :
u\in \Lambda,\ u'\in \ZZ^d \setminus \Lambda,\ \text{and} \ \lvert
u-u' \rvert_1 = 1\right\}\,.
\]
\begin{proof}[Proof of Lemma~\ref{lemma:iteration3}]
Fix $x,y \in \Gamma$ with $y \not \in \Lambda_x$ and set $n = 2\lvert \Theta \rvert$.
It follows from our definition, that the randomness of $H_\Gamma$ at sites $\partial^{\rm o} \hat W_x \cap \Gamma$ does not depend on the random variables $\omega_b$ for any $b\in B_x$, and depends {\it monotonically} on the random variables $\omega_k$ for $k\in \partial^{\rm o} B_x$ (by Assumption \ref{ass:monotone}). A similar statement holds for the randomness at sites $\partial^{\rm o} W_x \cap \Gamma$. We also note that $x,y \not \in W_x$ by our definition of $L$ and since $0 \in \Theta$.
We now choose $\Lambda = \hat W_x$ in Eq.~\eqref{eq:secondorder} and compute the Green function at $(x,y)$:
\begin{multline*}
G_{\Gamma}(z;x,y) =  G_{\Gamma}^{\hat W_x}(z;x,y) + \langle \delta_x , G_{\Gamma}^{\hat W_x} (z) T_{\Gamma}^{\hat W_x} G_{\Gamma}^{\hat W_x} (z) \delta_y \rangle \\
+ \langle \delta_x , G_{\Gamma}^{\hat W_x} (z) T_{\Gamma}^{\hat W_x} G_{\Gamma} (z) T_{\Gamma}^{\hat W_x} G_{\Gamma}^{\hat W_x} \delta_y \rangle .
\label{eq:secondordermatelts}
\end{multline*}
Using Remark \ref{remark:depleted} one can easily check that the first two contributions vanish, thus
\begin{equation}
G_{\Gamma}(z;x,y) = \sum_{\genfrac{}{}{0pt}{2}{(u',u) \in \partial^{\rm B} \hat W_x}{(v,v') \in  \partial^{\rm B} \hat W_x}}
    G_{\Gamma}^{\hat W_x}(z;x,u) G_\Gamma(z;u',v) G_{\Gamma}^{\hat W_x}(z;v',y) .
\label{eq:diagramatic}
\end{equation}
See Fig.~\ref{fig:geometric2} for the geometric setting and an illustration of Eq.~\eqref{eq:diagramatic}.
Note that $u, v' \in \partial^{\rm o} \hat W_x$, while $u', v \in
\hat W_x$. By  construction, the set $\Gamma \setminus \hat W_x$
decomposes into at least two components which are not connected:
One of them contains $x$, another $y$. More than two components may occur if $\Gamma$ or $\Theta$ are not connected, see again Fig.~\ref{fig:geometric2}. By Remark
\ref{remark:depleted}, the summands in Eq.~\eqref{eq:diagramatic}
are only non-zero if $u$ is in the $x$-component of $\Gamma
\setminus \hat W_x$ and $v'$ is in the $y$-component of $\Gamma
\setminus \hat W_x$. This leads us to the definition of a subset of
$\partial^{\rm B} \hat W_x$. We say that $(u,u') \in
\partial_x^{\rm B} \hat W_x$ if $(u,u') \in \partial^{\rm B} \hat
W_x$ and $u'$ is in the $x$-component of $\Gamma \setminus \hat
W_x$. For $\partial_y^{\rm B} \hat W_x$, $\partial_x^{\rm B} W_x$
and $\partial_y^{\rm B} W_x$ we use the analogous definitions.
\par
To get the estimate \eqref{eq:protobound} we first want to average
the fractional moment of the Green function with respect to random
variables $\{\omega_k\}_{k\in B_x^+}$. Note that Lemma~\ref{lemma:bounded} part (a) then guarantees that
\begin{equation} \label{eq:apriori}
\mathbb{E}_{B_x^+} \bigl(\lvert G_\Gamma (z;u',v)\rvert^{s/n}\bigr) \leq C_s \Xi_s (\lambda) .
\end{equation}
However, although the first and the last Green functions in \eqref{eq:diagramatic} do not depend on the random variables $\{\omega_k\}_{k\in B_x}$, they still depend on the random variables $\{\omega_k\}_{k\in B_x^+}$. To factor out this dependence, we apply \eqref{eq:firstorder} again, this time with $\Lambda = W_x$. Then we have for $u$, $v'$ as
above the equalities
\begin{align*}
G_{\Gamma}^{\hat W_x} (z;x,u) &= \sum_{(w',w) \in \partial_x^{\rm B} W_x}  G_{\Gamma}^{W_x}(z;x,w)G_{\Gamma}^{\hat W_x}(z;w',u)
\intertext{and}
G_{\Gamma}^{\hat W_x}(z;v',y)  &= \sum_{(r,r') \in \partial_y^{\rm B} W_x} G_{\Gamma}^{\hat W_x}(z;v',r) G_{\Gamma}^{W_x}(z;r',y) .
\end{align*}
Notice that for $w$ and $r'$ as above, the Green functions $G_{\Gamma}^{W_x}(z;x,w)$ and $G_{\Gamma}^{W_x}(z;$ $r',y)$ are independent of $\{\omega_k\}_{k \in B_x^+}$. Putting everything together, we obtain
\begin{multline}
\mathbb{E}_{B_x^+} \bigl(\lvert G_\Gamma (z;x,y)\rvert^{s/n}\bigr) \leq
\sum \lvert G_{\Gamma}^{W_x}(z;x,w)\rvert^{s/n} \, \lvert G_{\Gamma}^{W_x}(z;r',y)\rvert^{s/n} \\
\times \mathbb{E}_{B_x^+} \bigl(\lvert G_{\Gamma}^{\hat W_x}(z;w',u) G_\Gamma (z;u',v) G_{\Gamma}^{\hat W_x}(z;v',r) \rvert^{s/n}\bigr) ,
\label{eq:optimalbound}
\end{multline}
where  the sum on the right hand side runs over the bonds
\[
(u',u) \in \partial_x^{\rm B} \hat W_x, \ (v,v') \in  \partial_y^{\rm B} \hat W_x, \
(r,r') \in \partial_y^{\rm B} W_x,\ (w',w) \in \partial_x^{\rm B} W_x .
\]
To estimate the expectation of the product on the right hand side we
note that by H\"older inequality it suffices to show that each of
the Green functions raised to the fractional power $3s/n$ and averaged with respect to $B_x^+$ is bounded in an appropriate way. For $\EE_{B_x^+} (|G_\Gamma (z;u',v)|^{3s/n})$ this follows from the a-priory bound \eqref{eq:apriori}.
For the remaining two Green functions it seems at the first glance that we have a problem,
since we only average over $\{\omega_k\}_{k\in B_x^+}$, and Lemma~\ref{lemma:bounded}
in this context requires averaging with respect to $\{\omega_k\}_{k\in
B_x^{++}}$. What comes to our rescue is Assumption \ref{ass:monotone},
which ensures that the dependence on $\{\omega_k\}_{k\in B_x^+}$ is
actually monotone for these Green functions, and the standard
argument of \cite{AizenmanENSS-06} for the monotone case establishes
the required bounds. More precisely, we argue as follows. Since $w',u \in \Gamma \setminus \hat W_x$, we have due to Remark~\ref{remark:depleted} that
\[
G_{\Gamma}^{\hat W_x}(z;w',u) = G_{\Gamma\setminus\hat W_x}(z;w',u).
\]
Notice that $w',u \in \partial^{\rm o} \hat W_x$. Hence there are $b_{1},b_{2} \in \partial^{\rm o} B_x$, such that $w' \in \Theta_{b_1} \cap (\Gamma\setminus\hat W_x) \subset \partial^{\rm i} \Theta_{b_1}$ and $u \in \Theta_{b_2} \cap (\Gamma\setminus\hat W_x) \subset \partial^{\rm i} \Theta_{b_2}$. For the Green function at $(v',r)$ there exist $b_3,b_4 \in \partial^{\rm o} B_x$ with analoguous properties. Thus we may apply Lemma~\ref{lemma:bounded} part (b) and obtain for all $t \in (0,1)$
\begin{equation*}
 \mathbb{E}_{B_x^+} \bigl(\lvert G_{\Gamma}^{\hat W_x}(z;w',u)
 \rvert^{t}\bigr) \leq  D_t \lambda^{-t} \quad \text{and} \quad \mathbb{E}_{B_x^+}
 \bigl(\lvert G_{\Gamma}^{\hat W_x}(z;v',r)\rvert^{t}\bigr) \leq D_t\lambda^{-t} .\label{eq:mon_est}
\end{equation*}
The net result is a bound
\[
 \mathbb{E}_{B_x^+} \bigl(\lvert G_{\Gamma}^{\hat W_x}(z;w',u) G_\Gamma (z;u',v) G_{\Gamma}^{\hat W_x}(z;v',r) \rvert^{s/n}\bigr) \leq E_s \lambda^{-\frac{2s}{n}} \Xi_s (\lambda)
\]
where $E_s = \max\{D_{3s/n} , C_{3s}\}$. Substitution into Ineq.~\eqref{eq:optimalbound} leads to the estimate
\begin{multline} \label{eq:optimalbound1}
 \mathbb{E}_{B_x^+} \bigl(\lvert G_\Gamma (z;x,y)\rvert^{s/n}\bigr) \leq E_s \lambda^{-\frac{2s}{n}} \Xi_s (\lambda) \lvert \partial_x^{\rm B} \hat W_x \rvert \lvert \partial_y^{\rm B} \hat W_x \rvert \\
\times \sum_{\genfrac{}{}{0pt}{2}{(r,r') \in \partial_y^{\rm B} W_x}{(w',w) \in \partial_x^{\rm B} W_x}} \lvert G_{\Gamma}^{W_x}(z;x,w)\rvert^{s/n} \, \lvert G_{\Gamma}^{W_x}(z;r',y)\rvert^{s/n} .
\end{multline}
Now we are in position to perform the expectation with respect to
the rest of random  variables. Note that the two remaining  Green
functions in \eqref{eq:optimalbound1} are stochastically independent. We take expectation in Ineq.~\eqref{eq:optimalbound} and use Remark \ref{remark:depleted} to get
\[
\mathbb{E}\bigl(\lvert G_\Gamma (z;x,y)\rvert^{s/n}\bigr)
 \leq  \frac{E_s \tilde \Phi (\Theta, L)}{\lambda^{2s/n} \Xi_s^{-1} (\lambda)}  \cdot
 \sum_{(r,r')\in\partial_y^{\rm B}  W_x}
 \mathbb{E}\bigl
 ( \lvert G_{\Gamma\setminus W_x}(z;r',y)\rvert^{s/n} \bigr) \\[1ex]
\]
where
 \[
\tilde \Phi (\Theta, L) = \lvert \partial_x^{\rm B} \hat W_x \rvert
\lvert
\partial_y^{\rm B} \hat W_x \rvert \sum_{ (w',w) \in
\partial_x^{\rm B} W_x} \EE \bigl( \lvert G_{\Gamma \setminus W_x}(z;x,w)\rvert^{s/n} \bigr) .
\]
Now we use the fact that each point of  $\partial^{\rm o} \Lambda_x$ shares the bond with at most $2d$ neighbors. Hence, if we set
\[
\Phi (\Theta, L) = 4d^2\,\lvert \partial_x^{\rm B} \hat W_x \rvert
\lvert
\partial_y^{\rm B} \hat W_x \rvert \lvert\partial^{\rm o} \Lambda_x\rvert\,
\sum_{w \in\partial^{\rm o} W_x}
\mathbb{E}\bigl ( \lvert G_{\Gamma\setminus W_x}(z;x,w)
\rvert^{s/n} \bigr) ,
\]
we have the estimate
\begin{equation*}
\mathbb{E}\bigl(\lvert G_\Gamma (z;x,y)\rvert^{s/n}\bigr)
 \leq  \frac{E_s \Phi (\Theta, L) }{ \lambda^{2s/n} \Xi_s^{-1} (\lambda)}
\frac{1}{|\partial^{\rm o} \Lambda_x|} \sum_{r \in\partial^{\rm o}
\Lambda_x} \mathbb{E}\bigl ( \lvert G_{\Gamma\setminus
\Lambda_x}(z;r,y) \rvert^{s/n} \bigr) .
\end{equation*}
Finally, we can bound $\lvert \partial_x^{\rm B} \hat W_x \rvert $,
$\lvert\partial_y^{\rm B} \hat W_x \rvert$ and $\lvert\partial^{\rm o} \Lambda_x\rvert$
by $C_{d,\Theta} L^{d-1}$ with a constant $C_{d,\Theta}$ depending only on $d$ and $\Theta$.
Lemma \ref{lemma:iteration3} now follows by putting everything together.
\end{proof}
\begin{proof}[Proof of Theorem \ref{theorem:exp_decay}]
Notice that by Assumption \ref{ass:monotone}
the random potential is uniformly bounded.
Thus $K := \sup_{\omega \in \Omega} \lVert H_\omega \rVert <\infty$.
Choose $M \geq 1 $ and $ m= K+M$. For $|z| \leq m$ and each $b \in (0,1)$ 
we infer from the a-priori bound (Lemma \ref{lemma:bounded})
that condition \eqref{eq:fin_cond} from Theorem~\ref{theorem:finite_volume2} is satisfied if
$\lambda$ is sufficiently large.

For $|z| \geq m$ we have $\dist(z,\sigma (H_{\Gamma}))\geq M\geq 1$
for all $\omega$. A Combes-Thomas argument (see \cite{CombesT-73},
or Section 11.2 in \cite{Kirsch-08} for an explicit calculation in the discrete setting)
gives the bound
\[
 |G_\Gamma(z; x,y)| \le \frac{2}{M} \euler^{-\gamma |x-y|_1}
\]
for $|z| \geq m$ and arbitrary $x,y \in \Gamma$, where $\gamma :=\min \big(1, \ln \frac{M}{4d}\big)$. Now taking first the fractional power and then the mathematical
expectation gives the desired estimate on
$ \mathbb{E} \bigl(\lvert G_\Gamma (z;x,y)\rvert^{s/(2\lvert \Theta \rvert)}\bigr)$.
This finishes the proof.
\end{proof}
\section{Exponential localization and application to the strong disorder regime; proof of Theorem \ref{theorem:exp_decay_loc} and \ref{theorem:localization}\protect\sectionmark{Exponential localization and application to the strong disorder regime}} \label{sec:loc}
\sectionmark{Exponential localization and application to the strong disorder regime}
In this section we prove exponential localization in the case of
sufficiently large disorder, i.e. that the continuous spectrum of
$H_\omega$ is empty almost surely and that the eigenfunctions
corresponding to the eigenvalues of $H_\omega$ decay exponentially.
\par
The existing proofs of localization via the fractional moment method
use either the Simon Wolff criterion, see e.g.\
\cite{SimonW-86,AizenmanM-93,AizenmanFSH-01}, or the
RAGE-Theorem, see e.g.\
\cite{Aizenman-94,Graf-94,AizenmanENSS-06}. Neither dynamical nor
spectral localization can be directly inferred from the behavior of
the Green function using the existent methods for our model. The
reason is that the random variables $V_\omega (x)$, $x \in \ZZ^d$,
are not independent, while the dependence of $H_\omega$ on the
i.i.d.\ random variables $\omega_k$, $k \in \ZZ^d$, is not
monotone. In this section establish a new variant for concluding exponential localization from bounds on averaged fractional powers of Green function (without using the multiscale analysis). This is done by showing that fractional moment bounds imply the ``typical output'' of the multiscale analysis, i.e.\ the hypothesis of Theorem 2.3 in \cite{DreifusK-89}. Then one can conclude localization using existent methods. The results established in this section have already been published in \cite{ElgartTV-10} and \cite{ElgartTV-11}.
\par
Admittedly, for the discrete alloy-type model it is possible to show localization using the multiscale analysis. The two ingredients of the multiscale analysis are the initial length scale estimate and the Wegner estimate, compare assumptions (P1) and (P2) of \cite{DreifusK-89}. The initial length scale estimate is implied in the strong disorder regime by the exponential decay of an averaged fractional power of Green function, i.e.\ Theorem~\ref{theorem:exp_decay}, using Chebyshev's inequality. A Wegner
estimate for the models considered here follows also from Theorem~\ref{theorem:exp_decay} using Lemma~\ref{lemma:a-priori-wegner}, whereas the papers \cite{Veselic-10a,PeyerimhoffTV-11} establish a Wegner estimate for a larger class of single-site potentials, see also Chapter~\ref{chap:wegner}. Thus a variant of the multiscale analysis of \cite{DreifusK-89} yields pure point spectrum with exponential decaying eigenfunctions for almost all configurations of the randomness. 
\par
Apart from the strong disorder regime, the initial length scale estimate can be established in the weak disorder regime for low energies. In particular, Cao and Elgart \cite{CaoE-11} prove localization for the discrete alloy-type model on $\ell^2 (\ZZ^3)$ at weak disorder and low energies using the multiscale analysis. For the specific result on localization they assume one of the following assumptions:
\begin{enumerate}[(i)]
 \item the single-site potential either decays exponentially,
 \item $\Theta$ is finite, $k \in \ZZ^3$ is such that $(\Theta - i) \cap \Theta = \emptyset$ for all $0 \not = i \in k\ZZ^3$, and the random potential is defined by $V_\omega (x) = \sum_{i \in k \ZZ^3} \omega_k u(x-i)$.
\end{enumerate}
Assumption (i) corresponds to the maximal random setting, while assumption (ii) corresponds to non-overlapping potentials.
\par
In the strong disorder regime Kr\"uger \cite{Krueger-11} proves localization via multiscale analysis for a class of models including the discrete alloy-type model as a special case. Kr\"uger uses some ideas of \cite{Bourgain-09} and establishes a Wegner-like estimate without the use of monotonicity.
\par

\par
For $L >0$ and $x \in \ZZ^d$ we denote by $\Lambda_{L,x} = \{y \in \ZZ^d : \lvert x-y \rvert_\infty \leq L\}$ the cube of side length $2L+1$. Let further $m > 0$ and $E \in
\RR$. A cube $\Lambda_{L,x}$ is called \emph{$(m,E)$-regular} (for a fixed
potential), if $E \not \in \sigma (H_{\Lambda_{L,x}})$ and
\[
 \sup_{w \in \partial^{\rm i} \Lambda_{L,x}} \lvert G_{\Lambda_{L,x}} (E ; x,w) \rvert
\leq \euler^{-m L} .
\]
Otherwise we say that $\Lambda_{L,x}$ is \emph{$(m , E)$-singular}. The next
Proposition states that certain bounds on averaged fractional moments of the Green function imply the hypothesis of Theorem 2.3 in \cite{DreifusK-89} (without applying the induction step of the multiscale analysis). Recall that Assumption \ref{ass:finite} means that $\Theta$ is a finite set.
\begin{proposition} \label{prop:replace-msa}
Let Assumption \ref{ass:finite} be satisfied, $I \subset \RR$ be a bounded interval and $s \in (0,1)$. Assume the following two statements:
\begin{enumerate}[(i)]
\item There are constants $C,\mu \in (0,\infty)$ and $L_0 \in \NN_0$
such that
\[\EE \bigl( \lvert G_{\Lambda_{L,k}} (E + \i \epsilon;x,y) \rvert^{s}
\bigr)\ \leq \  C \euler^{-\mu \lvert x-y \rvert_\infty}\]
for all $k \in \ZZ^d$, $L \in \NN$, $x,y \in \Lambda_{L,k}$ with
$\lvert x-y \rvert_\infty \geq L_0$, all $\epsilon \in (0,1] $ and all $E \in I$.
\item There is a constant $C' \in (0,\infty)$ such that
\[\EE \bigl( \lvert G_{\Lambda_{L,k}} (E+\i \epsilon ;x,x) \rvert^{s}
\bigr) \leq C'\]
for all $k \in \ZZ^d$, $L \in \NN$, $x \in
\Lambda_{L,k}$, $E \in I$ and all $\epsilon\in (0,1]$ .
\end{enumerate}
Then we have for all $L \geq \max\{ 8\ln (8)/\mu , L_0 , -(8/5\mu)\ln (\lvert I \rvert / 2)\}$ and all $x,y \in \ZZ^d$ with $\lvert x-y\rvert_\infty \geq 2L+\diam \Theta + 1$ that
\begin{multline*}
 \PP \bigl(\{\omega \in \Omega \colon \forall \, E \in I \text{ either $\Lambda_{L,x}$ or $\Lambda_{L,y}$ is
$(\mu/8,E)$-regular} \}\bigr)  \\ \geq 1- 8 \lvert \Lambda_{L,x} \rvert(C\lvert I \rvert  + 4C'\lvert \Lambda_{L,x} \rvert / \pi ) \euler^{-\mu sL /8} .
\end{multline*}
\end{proposition}
For the proof we shall need that the boundedness of a fractional power of a diagonal Green's function element implies a Wegner estimate. Let us note that a Wegner estimate implies the boundedness of an averaged fractional power of the (finite-volume) Green function. At the moment we only know a proof where the bound depends polynomially on the volume of the cube.
\begin{lemma} \label{lemma:a-priori-wegner}
Let $I \subset \RR$ be an interval, $s \in (0,1)$, $c > 0$, $L \in \NN$ and $k \in \ZZ^d$. Assume there is a constant $C > 0$ such that
\[
 \EE \bigl( \lvert G_{\Lambda_{L,k}} (E+\i \epsilon ; x,x) \rvert^s \bigr ) \leq C
\]
for all $x \in \Lambda_{L,k}$, $E \in I$ and all $\epsilon \in (0,c]$. Then we have for all $[a,b] \subset I$ with $0< b-a \leq c$ that
\begin{equation*} 
\EE \bigl( \Tr \chi_{[a,b]}(H_{\Lambda_{L,k}})  \bigr)
\leq \frac{4C}{\pi}  \lvert b-a \rvert^{s}  \lvert \Lambda_{L,k} \rvert
 .
\end{equation*}
\end{lemma}
\begin{proof}
 Let $[a,b] \subset I$ with $0<b-a
\leq c$. Since we have for any $\lambda \in \RR$ and $0<\epsilon \leq b-a$
\[
 \arctan \left( \frac{\lambda - a}{\epsilon} \right) - \arctan \left( \frac{\lambda - b}{\epsilon} \right) \geq \frac{\pi}{4} \ \chi_{[a,b]}(\lambda) ,
\]
one obtains an inequality version of Stones formula:
\[
 \langle \delta_x , \chi_{[a,b]} (H_{\Lambda_{L,k}}) \delta_x \rangle
\leq \frac{4}{\pi} \int_{[a,b]} \im \left\{ G_{\Lambda_{L,k}} (E+ \i \epsilon ; x,x) \right\} \drm E \quad \forall \, \epsilon \in (0, b-a] .
\]
Using triangle inequality, $\lvert \im z\rvert \leq \lvert z\rvert$
for $z \in \CC$, Fubini's theorem, $\lvert G_{\Lambda_{L,k}} (E+\i
\epsilon ; x,x) \rvert^{1-s} \leq \dist (\sigma(H_{\Lambda_{L,k}}) ,
E+i \epsilon)^{s-1} \leq \epsilon^{s-1}$ and the hypothesis of the lemma, we
obtain for all $\epsilon \in (0,b-a]$
\begin{align*}
\EE \bigl( \Tr \chi_{[a,b]}(H_{\Lambda_{L,k}}) \bigr) & \leq \EE \Bigl( \sum_{x \in \Lambda_{L,k}} \frac{4}{\pi} \int_{[a,b]} \im \left\{ G_{\Lambda_{L,k}} (E+\i \epsilon ; x,x) \right\} \drm E  \Bigr) \\
&  \leq  \frac{\epsilon^{s-1}}{\pi / 4}  \sum_{x \in \Lambda_{L,k}} \int_{[a,b]} \EE \Bigl( \bigl|  G_{\Lambda_{L,k}} (E+\i \epsilon ; x,x)  \bigr|^{s} \Bigr) \drm E   \\
& \leq 4\pi^{-1}\epsilon^{s-1}  \lvert \Lambda_{L,k} \rvert \, \lvert b-a \rvert C .
\end{align*}
We minimize the right hand side by choosing $\epsilon = b-a$ and obtain the statement of the lemma.
\end{proof}
\begin{proof}[Proof of Proposition~\ref{prop:replace-msa}]
By assumption (and Lemma \ref{lemma:a-priori-wegner}) a Wegner estimate holds. Therefore, for any $L \in \NN$ and any $k \in \ZZ^d$ the probability of finding an eigenvalue of $H_{\Lambda_{L,k}}$ in $[a,b] \subset I$ shrinks to zero as $b-a \to 0$. Hence, for each $E \in I$ there is a set $\Omega_E \subset \Omega$ with $\PP (\Omega_E) = 1$, such that for all $k \in \ZZ^d$, $L\in\NN$ and $\omega \in \Omega_E$ we have that $E$ is not an eigenvalue of $H_{\Lambda_{L,x}}$ and the resolvent of $H_{\Lambda_{L,x}}$ at $E$ is well defined. Lebesgues Theorem gives for all $E \in I$
\begin{align*}
 C\euler^{- \mu \lvert x-y \rvert_\infty} &\geq \lim_{\epsilon \to 0} \EE (\lvert G_{\Lambda_{L,x}} (E + \i \epsilon ; x,y) \rvert^s) = \lim_{\epsilon \to 0} \int_{\Omega_E} \lvert G_{\Lambda_{L,x}} (E + \i \epsilon ; x,y)  \rvert^s \PP (\drm \omega) \\
&= \int_{\Omega_E} \lvert G_{\Lambda_{L,x}} (E ; x,y)  \rvert^s \PP (\drm \omega) 
 = \EE (\lvert G_{\Lambda_{L,x}} (E ; x,y) \rvert^s) 
\end{align*}
Note that for each $E \in I$, the function $\omega \mapsto G_{\Lambda_{L,x}} (E ; x,y)$ is defined on a set of full $\PP$-measure. 
\par
Set $n = \diam \Theta + 1$. Fix $L \in \NN$ with $L \geq \max\{8 \ln (8)/\mu , L_0\}$ and $x,y \in \ZZ^d$ such that $\lvert x-y \rvert_\infty \geq 2L+n$. For $\omega \in \Omega$ and $k \in
\{x,y\}$ we define the sets
\begin{align}
 \Delta_\omega^k &:= \{E \in I : \sup_{w \in \partial^{\rm i} \Lambda_{L,k}} \lvert
G_{\Lambda_{L,k}} (E ; k,w) \rvert > \euler^{-\mu L /8}\}, \nonumber \\
\tilde \Delta_\omega^k &:= \{E \in I : \sup_{w \in \partial^{\rm i} \Lambda_{L,k}}
\lvert G_{\Lambda_{L,k}} (E ; k,w) \rvert > \euler^{-\mu L/4 }\}, \nonumber \\
\text{and} \quad \tilde B_k &:= \{\omega \in \Omega : \mathcal{L} \{\tilde \Delta_\omega^k\} >
 \euler^{-5\mu L /8} \} . \label{eq:deltatilde}
\end{align}
Since the resolvent of $H_{\Lambda_{L,k}}$ at $E$ is not defined if $E$ is an eigenvalue of $H_{\Lambda_{L,k}}$, let us emphasize that we want the eigenvalues of $H_{\Lambda_{L,k}}$ to be included in the sets $\Delta_\omega^k$ and $\tilde\Delta_\omega^k$, $k \in \{x,y\}$.
For $\omega \in \tilde B_k$ we have
\begin{align*}
\sum_{w \in \partial^{\rm i} \Lambda_{L,k}}  \int_I \lvert G_{\Lambda_{L,k}} (E ; k,w) \rvert^{s} \drm E 
&= \sum_{w \in \partial^{\rm i} \Lambda_{L,k}}  \int_I \lvert G_{\Lambda_{L,k}} (E ; k,w) \rvert^{s}  \drm E \\
& \geq
\int_{\tilde\Delta_\omega^k} \sup_{w \in \partial^{\rm i} \Lambda_{L,k}} \lvert G_{\Lambda_{L,k}} (E ; k,w) \rvert^{s/N} \drm E \\[1ex]
& > \euler^{-5 \mu L / 8} \euler^{-\mu L s/ 4} > \euler^{-7 \mu L / 8}.
\end{align*}
Note again, that the integrands of the above equation are defined on a set of full Lebesgue measure.
Using Hypothesis (i) of the assertion, we obtain
\begin{align*}
 \PP (\tilde B_k) &< \euler^{7\mu L / 8} \sum_{w \in \partial^{\rm i} \Lambda_{L,k}} \int_{\tilde B_k} \int_I \lvert G_{\Lambda_{L,k}} (E ; k,w) \rvert^{s}\chi_{\{E \not \in \sigma (H_{\Lambda_{L,k}})\}}(E) \drm E \PP (\drm \omega)  \\
& \leq \euler^{7\mu L / 8} \sum_{w \in \partial^{\rm i} \Lambda_{L,k}}  \int_I \int_{\Omega} \lvert G_{\Lambda_{L,k}} (E ; k,w) \rvert^{s}\chi_{\{E \not \in \sigma (H_{\Lambda_{L,k}})\}}(E)  \PP (\drm \omega) \drm E   \\
&\leq \lvert \Lambda_{L,k} \rvert \, \lvert I \rvert C \euler^{-\mu L /8} .
\end{align*}
For $k \in \{x,y\}$ we denote by $\sigma (H_{\Lambda_{L,k}}) =
\{E_{\omega,k}^i\}_{i=1}^{\lvert \Lambda_{L,k} \rvert}$ the spectrum of
$H_{\Lambda_{L,k}}$. We claim that for $k \in \{x,y\}$,
\begin{equation} \label{eq:claim}
 \omega \in \Omega \setminus \tilde B_k \quad \Rightarrow
 \quad \Delta_\omega^k \subset \bigcup_{i=1}^{\lvert \Lambda_{L,k} \rvert}
 \bigl[E_{\omega,k}^i-\delta , E_{\omega,k}^i + \delta \bigr] =:
 I_{\omega,k}(\delta),
\end{equation}
where $\delta = 2\euler^{-\mu L / 8}$. Indeed, suppose that $E\in \Delta_\omega^k \setminus\{E_{\omega , k}^1, \ldots , E_{\omega , k}^{\lvert \Lambda_{L,k} \rvert}\}$ and $\dist\big(E,\sigma
(H_{\Lambda_{L,k}})\big)>\delta$. Then there exists $w \in \partial^{\rm i} \Lambda_{L,k}$ such
that $\lvert G_{\Lambda_{L,k}} (E;k,w) \rvert > \euler^{-\mu L / 8}$.
For any $E'$ with $\lvert E-E'\rvert \le 2\euler^{-5\mu L / 8}$ we have
$\delta -\lvert E-E'\rvert\ge \euler^{-\mu L / 8}\ge 2\euler^{-3\mu L / 8}  $ since $L > 8 \ln (8) / \mu$.
Moreover, the first resolvent identity and the estimate $\lVert (H-E)^{-1} \rVert \leq \dist
(E,\sigma (H))^{-1}$ for selfadjoint $H$ and $E \in \CC\setminus \sigma (H)$ implies
\begin{align*}
\lvert G_{\Lambda_{L,k}} (E ; k,w) - G_{\Lambda_{L,k}} (E' ; k,w)\rvert &\leq \ \lvert E-E'\rvert \cdot \lVert G_{\Lambda_{L,k}} (E)\rVert\cdot\lVert G_{\Lambda_{L,k}} (E') \rVert \\[1ex]
& \leq \frac{1}{2} \euler^{-\mu L / 8} ,
\end{align*}
 and hence
\[
 \lvert G_{\Lambda_{L,k}} (E' ; k,w)\rvert
\ > \ \frac{\euler^{-\mu L / 8}}{2} \geq \euler^{-\mu L / 4}
\]
for $L \geq 8 \ln (8) / \mu$.
We infer that $[E-2\euler^{-5\mu L / 8},E+2\euler^{-5\mu L / 8}]\cap I \subset  \tilde \Delta_\omega^k$ and conclude $\mathcal{L} \{\tilde \Delta_\omega^k\} \ge 2\euler^{-5\mu L / 8}$, since $\lvert I \rvert \geq 2 \euler^{-5\mu L / 8}$ by assumption.
This is however impossible if
$\omega\in \Omega \setminus \tilde B_k$ by \eqref{eq:deltatilde},
hence the claim \eqref{eq:claim} follows.
\par
In the following step we use Hypothesis (ii) of the assertion and Lemma~\ref{lemma:a-priori-wegner} to deduce a Wegner-type estimate. More presicely, we have for all $[a,b] \subset I$ with $0 < b-a \leq 1$ the Wegner estimate
\begin{equation} \label{eq:wegner}
\EE \bigl( \Tr \chi_{[a,b]}(H_{\Lambda_{L,x}}) \bigr)
\leq 4\pi^{-1} C' \lvert b-a \rvert^{s}  \lvert \Lambda_{L,x} \rvert
=: C_{\rm W} \lvert b-a \rvert^{s}  \lvert \Lambda_{L,x} \rvert .
\end{equation}
Now we want to estimate the probability of the event $B_{\rm res} := \{\omega \in \Omega :  I \cap I_{\omega,x}(\delta) \cap I_{\omega,y}(\delta) \not = \emptyset \}$
that there are ``resonant'' energies for the two box Hamiltonians $H_{\Lambda_{L,x}}$ and $H_{\Lambda_{L,y}}$.
For this purpose we denote by $\Lambda_{L,x}'$ the set of all lattice sites $k \in \ZZ^d$
whose coupling constant $\omega_k$ influences the potential in $\Lambda_{L,x}$,
i.\,e. $\Lambda_{L,x}' = \cup_{x \in \Lambda_{L,x}} \{k \in \ZZ^d : u(x-k) \not = 0)\}$.
Notice that the expectation in Ineq.~\eqref{eq:wegner} may therefore be replaced by $\EE_{\Lambda_{L,x}'}$.
Moreover, since $\lvert x-y \rvert_\infty \geq 2L + n$, the operator $H_{\Lambda_{L,y}}$ and hence the eigenvalues $E_{\omega , y}^i$, $i \in \{1,\ldots , \lvert \Lambda_{L,y} \rvert\}$ are independent of $\omega_k$, $k \in \Lambda_{L,x}'$. To estimate the probability of $B_{\rm res}$ we use the product structure of the measure and denote 
\[
\PP_{\Gamma} = \prod_{k \in \Gamma} \mu \text{ for } \Gamma \subset \ZZ^d, \quad \Omega \ni \omega = (\omega_1 , \omega_2) \in \Omega_{\Lambda_{L,x}'} \times \Omega_{\ZZ^d \setminus \Lambda_{L,x}'},
\]
and for each $\omega_2 \in \Omega_{\ZZ^d \setminus \Lambda_{L,x}'}$ we define $\tilde B_{\rm res}(\omega_2) = \{\omega_1 \in \Omega_{\Lambda_{L,x}'} \colon (\omega_1 , \omega_2) \in B_{\rm res}\}$.
Since the eigenvalues $E_{\omega , y}^i$, $i \in \{1,\ldots , \lvert \Lambda_{L,y} \rvert\}$ are independent of $\omega_k$, $k \in \Lambda_{L,x}'$, we obtain for any $\omega_2 \in \Omega_{\ZZ^d \setminus \Lambda_{L,x}'}$ using \v Ceby\v sev's inequality and the estimate \eqref{eq:wegner} that
\begin{align*}
\nonumber
\PP_{\Lambda_{L,x}'} (\tilde B_{\rm res} (\omega_2) )& \leq \sum_{i=1}^{\lvert \Lambda_{L,y} \rvert}
\PP_{\Lambda_{L,x}'} \bigl(\{ \omega_1 \in \Omega_{\Lambda_{L,x}'} \colon \Tr  \chi_{I \cap [E_{\omega,y}^i-2\delta , E_{\omega,y}^i + 2\delta ]} (H_{\Lambda_{L,x}}) \geq 1 \} \bigr) \\
& \leq \sum_{i=1}^{\lvert \Lambda_{L,y} \rvert}
\EE_{\Lambda_{L,x}'} \bigl( \Tr \bigl( \chi_{I \cap [E_{\omega,y}^i-2\delta , E_{\omega,y}^i + 2\delta ]}  (H_{\Lambda_{L,x}}) \bigr) \bigr) \\
&\leq \lvert \Lambda_{L,y} \rvert  C_{\rm W} (4\delta)^{s}\lvert \Lambda_{L,x} \rvert.
\end{align*}
Consequently, we get by Fubini's theorem
\begin{equation}\label{eq:wegner-application}
 \PP (B_{\rm res}) \leq \lvert \Lambda_{L,y} \rvert  C_{\rm W} (4\delta)^{s}\lvert \Lambda_{L,x} \rvert .
\end{equation}
Notice that $4\delta \leq 1$, since $L \geq 8 \ln 8$. Consider now an $\omega \not \in \tilde B_x \cup \tilde B_y$. Recall that \eqref{eq:claim} tells us that $\Delta_\omega^x \subset  I_{\omega,x}(\delta)$
and $\Delta_\omega^y \subset  I_{\omega,y}(\delta)$. If additionally $\omega \not \in B_{\rm  res}$ then no $E \in I$ can be in
$\Delta_\omega^x $ and $\Delta_\omega^y$ simultaneously. Hence for each $E \in I$ either  $\Lambda_{L,x}$ or $\Lambda_{L,y}$
is $(\mu/8,E)$-regular. A contraposition gives us
\begin{align*}
\PP \bigl(\bigl\{&\omega \in \Omega \colon \text{$\exists \, E \in I$,
$\Lambda_{L,x}$ and $\Lambda_{L,y}$ are $(\mu/8,E)$-singular} \bigr\}\bigr) \\
&\le \PP (\tilde B_x) +\PP (\tilde B_y ) + \PP (B_{\rm res} )
\\ &\leq 2\lvert \Lambda_{L,x} \rvert \, \lvert I \rvert C
\euler^{-\mu L /8} + \lvert \Lambda_{L,y} \rvert  C_{\rm W}
(4\delta)^{s}\lvert \Lambda_{L,x} \rvert,
\end{align*}
from which the result follows.
\end{proof}
In the proof of Proposition \ref{prop:replace-msa} its Hypothesis (ii)
was only used to obtain a Wegner estimate, i.e.\ Eq.~\eqref{eq:wegner}.
Hence, if we know that a Wegner estimate holds
for some other reason, e.g.~from \cite{Veselic-10a},
 we can relinquish the Hypothesis (ii) and
skip the corresponding argument in the proof of Proposition
\ref{prop:replace-msa}. Specifically, the following assertion holds
true.
\begin{proposition} \label{prop:replace-msa-Wegner}
Let Assumption \ref{ass:finite} be satisfied, $I \subset \RR$ be a bounded interval and $s \in (0,1)$. Assume the following two statements:
\begin{enumerate}[(i)]
\item There are constants $C,\mu \in (0,\infty)$ and $L_0 \in \NN_0$ such that
\[\EE \bigl(\lvert G_{\Lambda_{L,k}} (E;x,y) \rvert^{s} \bigr)
\leq C \euler^{-\mu \lvert x-y \rvert_\infty}\]
for all $k \in \ZZ^d$, $L \in \NN$, $x,y \in \Lambda_{L,k}$ with
$\lvert x-y \rvert_\infty \geq L_0$, and all $E \in I$.
\item There are constants $C_{\rm W}\in (0,\infty)$, $ \beta \in (0,1]$, and $D \in \NN$
such that
\[
\PP\bigl(\{\omega \in \Omega \colon \sigma(H_{\Lambda_{L}} ) \cap
[a,b]\not=\emptyset\}\bigr) \leq C_{\rm W}{\lvert b-a\rvert}^\beta \,
L^D
\]
for all $L \in \NN$ and all $[a,b]\subset I$.
\end{enumerate}
Then we have for all $L \geq \max\{8 \ln (2)/\mu , L_0 , -(8/5\mu)\ln (\lvert I \rvert / 2)\}$ and all $x,y \in \ZZ$ with $\lvert x-y\rvert_\infty \geq 2L+\diam \Theta + 1$ that
\begin{multline*}
 \PP \bigl( \{\omega \in \Omega \colon \forall \, E \in I \text{ either $\Lambda_{L,x}$ or $\Lambda_{L,y}$ is
$(\mu/8,E)$-regular} \} \bigr)  \\ \geq 1- 8(2L+1)^d\rvert(C \, \lvert I \rvert  + C_{\rm W}L^D  ) \euler^{-\mu \beta L/8} .
\end{multline*}
\end{proposition}
\begin{proof}
We proceed as in the proof of Proposition \ref{prop:replace-msa}, but replace Ineq.~\eqref{eq:wegner-application} by
\begin{align*}
\PP_{\Lambda_{L,x}'} (B_{\rm res})&\leq \sum_{i=1}^{\lvert \Lambda_{L,y} \rvert}
\PP_{\Lambda_{L,x}'} \bigl(\{\omega \in \Omega \colon  I \cap \sigma(H_{\Lambda_{L,x}}) \cap [E_{\omega,y}^i-2\delta , E_{\omega,y}^i + 2\delta ] \neq \emptyset\}\bigr) \\
 &\leq \lvert \Lambda_{L,y} \rvert  C_{\rm W} (4\delta)^{\beta} L^D
\end{align*}
to obtain the desired bound.
\end{proof}
\begin{remark} \label{remark:assii}
Note that the conclusions of Proposition \ref{prop:replace-msa} and \ref{prop:replace-msa-Wegner} tell us that the probabilities of $\{\forall \, E \in I \text{ either $\Lambda_{L,x}$ or $\Lambda_{L,y}$ is $(\mu/8,E)$-regular} \}$ tend to one exponentially fast as $L$ tends to infinity. In particular, for any $p>0$ there is some $\tilde L \in \NN$ such that for all $L \ge \tilde L$:
\[
 \PP \bigl(\{\omega \in \Omega \colon \forall \, E \in I \text{ either $\Lambda_{L,x}$ or $\Lambda_{L,y}$ is $(m,E)$-regular} \}\bigr)\ge 1- L^{-2p}.
\]
\end{remark}
We will conclude exponential localization from the estimates provided
by Proposition \ref{prop:replace-msa} and \ref{prop:replace-msa-Wegner} using Theorem 2.3 in
\cite{DreifusK-89}. More precisely, Theorem 2.3 in \cite{DreifusK-89} was stated for the case $u(0) = 1$ and $u (k) = 0$ for $k \in \ZZ^d \setminus \{0\}$, wherefore we need a slight extension, which can be proven with the same arguments as the
original result. For completeness and convenience of the reader we will give a proof. Let us emphasize that Theorem~\ref{thm:vDK-2.3} does not need any assumption on the single-site potential $u$ or the measure $\mu$. In particular, the single-site potential may have unbounded support.

\begin{theorem} \label{thm:vDK-2.3}
Let $a\in \NN$, $l \in \NN_0$, $I\subset \RR$ be an interval and let $p>d$, $ L_0>1$, $\alpha \in (1,2p/d)$ and $m>0$.
Set $L_{k} =L_{k-1}^\alpha$, for $k \in \NN$. Suppose that for any $k \in \NN_0$ and any $x,y \in \ZZ^d$ with $\lvert x-y\rvert_\infty \geq aL_k + l$
\[
 \PP \bigl( \{\omega \in \Omega \colon \forall \, E \in I \text{ either $\Lambda_{L_k,x}$ or $\Lambda_{L_k,y}$ is $(m,E)$-regular} \} \bigr)\ge 1- L_k^{-2p} .
\]
Then $H_\omega$ exhibits exponential localization in $I$ for almost all $\omega \in \Omega$.
\end{theorem}
\begin{proof}
Let $b$ be a positive integer to be chosen later on. For $x_0 \in \ZZ^d$ let
\[
 A_{k+1} (x_0) = \Lambda_{b(a L_{k+1} + l) , x_0} \setminus \Lambda_{aL_{k} + l , x_0}
\]
for $k\in \NN_0$. Let us define the event
\begin{multline*}
 E_k (x_0 ) = \{\omega \in \Omega \colon \text{$\Lambda_{L_k , x_0}$ and $\Lambda_{L_k , x}$ are $(m,E)$-singular} \\ \text{for some $E \in I$ and some $x \in A_{k+1}(x_0)$}\} .
\end{multline*}
By construction we have for each $x \in A_{k+1}(x_0)$ that $\lvert x-x_0 \rvert_\infty \geq aL_k + l$. Hence, we obtain by our hypothesis
\[
 \PP \bigl( E_k (x_0) \bigr) \leq \sum_{x \in A_{k+1}(x_0)} L_k^{-2p} 
 \leq \frac{(2baL_{k+1} + 2bl + 1)^d}{L_k^{2p}} \leq \frac{(2ba + 2bl L_0^{-1} + L_0^{-1})^d}{L_k^{2p - \alpha d}}
\]
for all $k \in \NN_0$. Since $\alpha d < pd$, we have $\sum_{k=0}^\infty \PP ( E_k (x_0) )  < \infty$. It follows from Borel Cantelli Lemma that for each $x_0 \in \ZZ^d$ we have $\PP \{ E_k (x_0)$ occurs infinitely often$\}=0$. Since a countable union of sets of measure zero has measure zero, we obtain
\[
 \PP \bigl(\{\omega \in \Omega \colon \exists \, x_0 \in \ZZ^d\text{ such that $E_k (x_0)$ occurs for infinitely many $k \in \NN$} \}\bigr) = 0 .
\]
If we let $\Omega_0 = \{\omega \in \Omega \colon \text{for all $x_0 \in \ZZ^d$, $E_k (x_0)$ occurs only finitely many times} \}$, we have that $\PP (\Omega_0) = 1$. In particular, for each $\omega \in \Omega_0$ and $x_0 \in \ZZ^d$ there is a $k_1 = k_1 (\omega, x_0) \in \NN$ such that if $k \geq k_1$ then $E_k (x_0)$ does not occur.
\par
Now let $\omega \in \Omega_0$, $E \in I$ be a generalized eigenvalue of $H_\omega$ with the corresponding non-zero polynomially bounded generalized eigenfunction $\psi$, i.e.\ $H_\omega \psi = E \psi$, $\lvert \psi (x) \rvert \leq C (1 + \lvert x \rvert)^t$ for some positive constant $C$ and positive integer $t$, and we now choose $x_0 \in \ZZ^d$ such that $\psi (x_0) \not = 0$.
\par
If $\Lambda_{L_k,x_0}$ would be $(m,E)$-regular, then $E \not \in \sigma (H_{\Lambda_{L_k},x_0})$ and therefore we can recover $\psi$ from its boundary values, i.e.
\begin{align} \label{eq:boundary}
 \lvert \psi (x_0) \rvert &= \Biggl\lvert \sum_{i \in \partial^{\rm i} \Lambda_{L_k,x_0}} G_{\Lambda_{L_k,x_0}} (E;x_0,i) \sum_{y \in (\Lambda_{L_k,x_0})^{\rm c} :  \lvert i-y \rvert = 1} \psi (y) \Biggr\rvert\\
&\leq \sum_{i \in \partial^{\rm i} \Lambda_{L_k,x_0}} \euler^{-mL_k} 2d C(2+L_k+\lvert x_0 \rvert)^t \nonumber.
\end{align}
Since $\psi (x_0) \not = 0$, it follows that there exists $k_2 = k_2 (\omega , E, x_0) \in \NN$ such that $\Lambda_{L_k,x_0}$ is $(m,E)$-singular for all $k \geq k_2$. Let $k_3 = k_3 (\omega,E,x_0) = \max\{k_1,k_2\}$. If $k \geq k_3$ we conclude that $\Lambda_{L_k,x}$ is $(m,E)$-regular for all $x \in A_{k+1} (x_0)$.
\par
Now let $\rho \in (0,1)$ be given. We pick $b > (1 + \rho) /(1-\rho)$ and define
\[
 A_{k+1}' (x_0) = \Lambda_{b(a L_{k+1} + l)/(1+\rho) , x_0} \setminus \Lambda_{(aL_{k} + l)/(1-\rho) , x_0} .
\]
We claim that
\begin{enumerate}[(i)]
 \item $A_{k+1}' (x_0) \subset A_{k+1} (x_0)$ for $k \in \NN_0$,
 \item if $x \in A_{k+1}' (x_0)$ then $\dist (x,\partial^{\rm i} A_{k+1} (x_0)) \geq \rho \lvert x-x_0 \rvert_\infty$, and
 \item if $x \not \in \Lambda_{(aL_0 + l)/(\rho - 1) , x_0}$ then $x \in A_{k+1}'$ for some $k \in \NN_0$.
\end{enumerate}
Here $\dist(m,A) = \inf_{k \in A} \lvert m-k \rvert_{\infty}$ for $k \in \ZZ^d$ and $A \subset \ZZ^d$.
Claim (i) and (iii) are obvious. To see (ii) we estimate the distance of $x\in A_{k+1}' (x_0)$ to the inner and outer boundary of $A_{k+1} (x_0)$, respectively. For the inner boundary we use $\lvert x-x_0 \rvert_\infty \leq aL_k + l + \dist (x,\Lambda_{aL_k + l , x_0})$ and $\lvert x-x_0 \rvert_\infty \geq (aL_k + l)/(1-\rho)$ to conclude 
\[
\dist (x , \Lambda_{aL_k + l , x_0}) \geq \lvert x-x_0 \rvert_\infty - (1-\rho) \lvert x-x_0 \rvert_\infty =  \rho \lvert x-x_0 \rvert_\infty .
\]
For the outer boundary we use the triangle inequality $\dist (x_0 , \partial^{\rm i} \Lambda_{b(aL_{k+1}+l) , x_0}) \leq \lvert x-x_0 \rvert_\infty + \dist (x,\partial^{\rm i} \Lambda_{b(aL_{k+1}+l) , x_0})$, $\lvert x-x_0 \rvert_\infty \leq b(aL_{k+1}+l)/(1+\rho)$ and $\dist(x_0, \partial^{\rm i} \Lambda_{b(aL_{k+1}+l) , x_0}) = b(aL_{k+1}+l)$ to conclude
\[
 \dist (x,\partial^{\rm i} \Lambda_{b(aL_{k+1}+l) , x_0}) \geq \dist (x_0 , \partial^{\rm i} \Lambda_{b(aL_{k+1}+l) , x_0}) -  \lvert x-x_0 \rvert_\infty
 \geq \rho \lvert x-x_0 \rvert_\infty .
\]
Hence the claim (2) follows.
\par
Now let $k \geq k_3$, so that $\Lambda_{L_k , y}$ is $(m,E)$-regular for any $y \in A_{k+1} (x_0)$. Let $x \in A_{k+1}' (x_0) \subset A_{k+1} (x_0)$. Again by by Eq.~\eqref{eq:boundary},
\[
 \lvert \psi (x) \rvert \leq (2L_k + 1)^d \euler^{-mL_k} 2d \lvert \psi (u_1) \rvert
\]
for some $u_1 \in \partial^{\rm o} \Lambda_{L_k + 1 ,x}$. If $u_1 \in A_{k+1} (x_0)$ we obtain
\[
  \lvert \psi (x) \rvert \leq \bigl[(2L_k + 1)^d \euler^{-mL_k} 2d\bigr]^2 \lvert \psi (u_2) \rvert
\]
for some $u_2 \in \partial^{\rm o} \Lambda_{L_k , u_1}$. By claim (ii) we can repeat this procedure at least $\lfloor\rho \lvert x-x_0 \rvert_\infty /(L_k+1)\rfloor$ times, use the polynomial bound on $\psi$ and obtain for all $k \geq k_3$ and all $x \in A_{k+1}' (x_0)$ the inequality
\[
  \lvert \psi (x) \rvert \leq \bigl[(2L_k + 1)^d \euler^{-mL_k} 2d\bigr]^{\left\lfloor\rho \lvert x-x_0 \rvert_\infty /(L_k+1)\right\rfloor} C \bigl(1+\lvert x_0 \rvert_\infty + b(aL_{k+1} + l) \bigr)^t .
\]
We can rewrite the above inequality  as
\begin{multline*}
  \lvert \psi (x) \rvert \leq 
\exp\Biggl\{ -  \left\lfloor \frac{\rho \lvert x-x_0 \rvert_\infty}{L_k + 1} \right\rfloor \rho m L_k \Biggl\} 
\exp\Biggl\{  \left\lfloor \frac{\rho \lvert x-x_0 \rvert_\infty}{L_k + 1} \right\rfloor \bigg[d \ln (2L_k + 1)\Biggr. \\ + \ln (2d)  - (1-\rho)m L_k \bigg]   \Biggl.+ t \ln \left(C(1+\lvert x_0 \rvert_\infty + b(aL_{k+1} + l)\right) \Biggr\} .
\end{multline*}
Since $(aL_k + l)/(1-\rho)\leq\lvert x-x_0 \rvert_\infty \leq b(a L_k^\alpha + l)/(1+\rho)$, the second exponential function gets smaller than one if $k$ is sufficiently large. Let $\rho' \in (0,1)$ and choose $\rho$ such that $\rho > 1/(1+a-\rho' a)$. We obtain that the first exponential function is bounded from above by
\begin{align*}
&\phantom{\leq} \exp\Biggl\{ -  \left( \frac{\rho \lvert x-x_0 \rvert_\infty}{L_k + 1} - 1 \right ) \rho m L_k \Biggl\}\\ 
&\leq \exp \Biggl \{ \rho m L_k \Biggr\} \exp \Biggl \{- \rho^2 m\lvert x-x_0 \rvert_\infty \frac{L_k}{L_k + 1} \Biggr\} \\
& \leq \exp \Biggl \{ \rho m L_k \biggl[ 1 - (1-\rho') \rho  \frac{\lvert x-x_0 \rvert_\infty}{L_k + 1} \biggr] \Biggr\} \exp \Biggl \{- \rho^2 \rho' m\lvert x-x_0 \rvert_\infty \frac{L_k}{L_k + 1} \Biggr\} .
\end{align*}
Again, using the lower bound on $\lvert x-x_0 \rvert_\infty$ and the relation between $\rho$ and $\rho'$, we see that the first exponential function gets smaller than one if $k$ is sufficiently large. Hence, if we pick $\rho'' \in (0,1)$ we find $k_4 \in \NN$ such that for all $k \geq k_4$ and all $x \in A_{k+1}' (x_0)$ we have
\begin{equation} \label{eq:ende}
 \lvert \psi (x) \rvert \leq \exp \Biggl \{- \rho^2 \rho' m\lvert x-x_0 \rvert_\infty \rho'' \Biggr\} .
\end{equation}
By claim (iii) we conclude that for all $x \in \ZZ^d \setminus \Lambda_{(aL_{k_4} + l)/(1-\rho) , x_0}$ we have Ineq.~\eqref{eq:ende}.
\par
We have shown that all generalized eigenfunctions to eigenvalues in $I$ are in $\ell^2 (\ZZ^d)$ and decay exponentially fast. To end the proof we use the well known fact that if every generalized eigenfunction in $I$ is in $\ell^2 (\ZZ^d)$, then there is no continuous spectrum in $I$, see e.g.\ \cite[Theorem 1.2]{FroehlichMSS-85}. This uses the fact that there is a spectral measure, such that (with respect to this spectral measure) almost all energies are generalized eigenvalues \cite{Berezanskii-68,Simon-82}. 
\end{proof}

\begin{proof}[Proof of Theorem \ref{theorem:exp_decay_loc}]
We assume first that $I$ is a \emph{bounded} interval.
In this case the assumptions of Proposition~\ref{prop:replace-msa} are satisfied. Combining the
latter with Theorem~\ref{thm:vDK-2.3} and Remark~\ref{remark:assii}
we arrive  to the desired result.

If $I$ is an \emph{unbounded} interval, we can cover it by a countable collection
of bounded intervals. In each of those, exponential localization holds
by the previous arguments for all $\omega$ outside a set of zero measure.
Since the collection of intervals is countable, we have exponential localizaition in $I$ almost surely.
\end{proof}

\begin{proof}[Proof of Theorem~\ref{theorem:localization}]
We use Theorem~\ref{theorem:exp_decay}  to verify that the hypothesis of
Theorem~\ref{theorem:exp_decay_loc} is satisfied with $I=\RR$. This yields
the desired result.
\end{proof}
\chapter{Wegner estimate for discrete alloy-type models} \label{chap:wegner}
In this chapter we prove a Wegner estimate\index{Wegner estimate} for the discrete alloy-type model under Assumption~\ref{ass:exponential}, i.e. under the assumption that the measure $\nu$ has a density $\rho \in \BV (\RR)$ and there are constants $C,\alpha>0$ such that for all $k \in \ZZ^d$ we have $\lvert u(k) \rvert \leq C \euler^{-\alpha \lvert k \rvert_1}$. This result relies on a joint work with Norbert Peyerimhoff and Ivan Veseli\'c and has already been published in \cite{PeyerimhoffTV-11}.
\par
A Wegner estimate \cite{Wegner-81} is an upper bound on the expected number of eigenvalues of finite box Hamiltonians $H_{\Lambda_l}$ in a bounded energy interval $[E-\epsilon , E+\epsilon] \subset \RR$. Wegner estimates are inequalities of the type
\[
 \forall \, l\in\NN,\ E\in\RR, \ \epsilon>0 \colon \quad
 \EE \bigl( \Tr \chi_{[E-\epsilon , E + \epsilon]} (H_{\Lambda_L}) \bigr) \leq C_{\rm W} (2\epsilon)^a \, (2l+1)^{bd}
\]
with some (Wegner-)constant $C_{\rm W} > 0$, some $a \in (0,1]$ and some $b \in [1,\infty)$. The exponent $a$ determines the quality of the estimate with respect to the length of the energy interval and $b$ the quality with respect to the volume of the cube $\Lambda_l$. The best possible estimate is obtained in the case $a=1$ and $b=1$. 
\par
Wegner estimates are required as an input of the multiscale analysis to prove localization. More precisely, multiscale analysis yields exponential localization in any energy interval where a Wegner estimate and an initial length scale estimate\index{initial length scale estimate}, the other ingredient of the multiscale analysis, holds.
\section{Abstract Wegner estimate and the proof of Theorem~\ref{theorem:wegner}} \label{sec:abstract_wegner}
In \cite{KostrykinV-06} an abstract Wegner estimate for the continuous alloy-type model was established. A discrete analogue of this result will be applicable in our situation. The proof is a straight forward adaptation of \cite{KostrykinV-06} to the discrete setting. For completeness we will give a short proof.
\par
Let us also recall the definition of the space $\BV (\RR)$, see e.g.\ \cite{Ziemer-89}. The space of functions of finite total variation\index{finite total variation} $\BV (\RR)$ is 
the set of integrable functions $f : \RR \to \RR$ whose distributional derivative\index{distributional derivative} is a signed Borel measure with finite variation, i.e.
\[
 \BV (\RR) := \{f : \RR \to \RR  \colon f \in L^1 (\RR),\ D f \ \text{is a measure},\ \lvert D f \rvert (\RR) < \infty\} .
\]
To say that a distributional derivative $Df$ of an integrable function $f:\RR \to \RR $ is a measure means that there exists a signed Borel measure $Df$ on $\RR$ such that
\[
 \int_\RR \phi \drm Df = - \int_\RR f \phi' \drm x
\]
for all $\phi \in C_{\rm c}^\infty (\RR)$. Here $\phi \in C_{\rm c}^\infty (\RR)$ denotes the set of real continuous functions on $\RR$ with compact support which have derivatives of all orders. 
A norm on $\BV (\RR)$ is defined by 
\[
\lVert f \rVert_{\BV} := \lVert f \rVert_{L^1} + \lVert f \rVert_{\rm Var}, 
\]
where
\[
\lVert f \rVert_{\rm Var} := \lvert D f \rvert (\RR) = \sup \Bigl\{ \int_\RR f v' \drm x \colon v \in C_{\rm c}^\infty (\RR) , \ \lvert v \rvert \leq 1 \Bigr\} .
\]
Note that if $f \in W^{1,1} (\RR)$ then $f \in \BV (\RR)$. In particular, if $f \in W^{1,1} (\RR)$ we have $\lVert f \rVert_{W^{1,1}} = \lVert f \rVert_{\BV}$ and $\lVert f' \rVert_{L^1} = \lVert f \rVert_{\rm Var}$.
\begin{theorem} \label{theorem:abstract2}
Assume $\rho \in \mathrm{BV}(\RR)$ and that there is a number $l_0 \in \NN$ such that for arbitrary $l \geq l_0$ and every $j \in \Lambda_{l}$ there is a compactly supported $t_{j,l} \in \ell^1 (\ZZ^d)$ such that
\begin{equation}\label{eq:wegner:ass}
\sum_{k \in \ZZ^d} t_{j,l} (k) u(x-k) \geq \delta_j (x) \quad \text{for all} \quad x \in \Lambda_{l} .
\end{equation}
Let further $I:=[E_1,E_2]$ be an arbitrary interval. Then for any $l \geq l_0$
\begin{equation*}
\EE ( \Tr \chi_I (H_{\Lambda_{l}})) \leq \frac{1}{2\lambda}\lVert \rho \rVert_{\rm Var} \lvert I \rvert \sum_{j \in \Lambda_l} \lVert t_{j,l} \rVert_{\ell^1} .
\end{equation*}
\end{theorem}
For the proof of Theorem \ref{theorem:abstract2} we will use an estimate on averages of spectral projections of certain self-adjoint operators. More precisely, let $\mathcal{H}$ be a Hilbert space and consider the following operators on $\mathcal{H}$. Let $H$ be self-adjoint, $W$ symmetric and $H$-bounded, $J$ bounded and non-negative with $J^2 \leq W$, $H (\zeta) = H + \zeta W$ for $\zeta \in \RR$, and $\chi_I (H(\zeta))$ the corresponding spectral projection onto an Interval $I \subset \RR$. Then, for any $g \in L^\infty (\RR) \cap L^1 (\RR)$, $\psi \in \mathcal{H}$ with $\lVert \psi \rVert = 1$ and bounded interval $I \subset \RR$,
\begin{equation} \label{eq:average_proj}
\int_\RR \bigl\langle \psi , J \chi_I (H(\zeta)) J \psi \bigr\rangle g(\zeta) \drm \zeta \leq \lVert g \rVert_\infty \lvert I \rvert .
\end{equation}
For a proof of Ineq.~\eqref{eq:average_proj} we refer to \cite{CombesH-94} where compactly supported $g$ is considered. The vadility of Ineq.~\eqref{eq:average_proj} for the non-compactly supported case was first noted in \cite{FischerHLM-97}. For a detailed proof we refer to \cite[Lemma~5.3.2]{Veselic-08}.
\begin{proof}[Proof of Theorem \ref{theorem:abstract2}]
We follow the arguments in \cite{KostrykinV-06}.
 In order to estimate the terms of the sum in the expectation $\EE (\Tr \chi_I (H_{\Lambda_l})) = \sum_{j \in \Lambda_{l}} \EE (\lVert \chi_I(H_{\Lambda_l}) \delta_j \rVert^2)$ we fix $l \geq l_0$ and $j \in \Lambda_{l}$, and set $\Sigma = \supp t_{j,l} \subset \ZZ^d$ and $t = t_{j,l}$. Recall that
\[
 H_{\Lambda_{l}} = -\Pro_{\Lambda_{l}}\Delta \Inc_{\Lambda_{l}} + \lambda \sum_{k \in \ZZ^d \setminus \Sigma} \omega_k u (\cdot - k) +  \lambda \sum_{k \in  \Sigma} \omega_k u (\cdot - k).
\]
We pick some $o \in \Sigma$ with $t (o) \not = 0$ and denote by $M$ the finite dimensional linear transformation $(\eta_k)_{k \in \Sigma} \mapsto (\omega_k)_{k \in \Sigma} = M (\eta_k)_{k \in \Sigma}$ defined as follows: $\omega_o = t (o) \eta_o$ and $\omega_k = t (k)  \eta_o + t(o) \eta_k$  for $k \in \Sigma \setminus \{o\}$. Note that $M$ is invertible and $\lvert \det M\rvert = \lvert t(o) \rvert^{\lvert \Sigma \rvert}$. With this transformation there holds for arbitrary fixed $(\omega_k)_{k \in \ZZ^d \setminus \Sigma}$
\begin{align*}
\int_{\RR^{\lvert \Sigma \rvert}} \bigl\lVert \chi_I (H_{\Lambda_l}) \delta_j \bigr\rVert^2 \prod_{k \in \Sigma} \rho (\omega_k) \drm \omega_k = 
\int_{\RR^{\lvert \Sigma \rvert}} \bigl\lVert \chi_I (H_{\Lambda_l}^\eta) \delta_j \bigr\rVert^2 k (\eta) \drm \eta , 
\end{align*}
where $\drm \eta = \prod_{k \in \Sigma} \drm \eta_k$, $k (\eta) = \lvert t (o) \rvert^{\lvert \Sigma \rvert} \rho (t(o) \eta_o) \prod_{k \in \Sigma \setminus\{o\}} \rho (t(k) \eta_o + t(o) \eta_k)$, and
\begin{multline*}
H_{\Lambda_l}^\eta = -\Pro_{\Lambda_{l}} \Delta \Inc_{\Lambda_{l}} + \lambda \sum_{k \in \ZZ^d \setminus \Sigma} \omega_k u (\cdot - k) +  t(o) \lambda \sum_{k \in \Sigma \setminus \{o\}} \eta_k u(\cdot-k) \\ + \eta_o \lambda \sum_{k \in \Sigma} t(k) u(\cdot-k) .
\end{multline*}
We denote by $P_j : \ell^2 (\ZZ^d) \to \ell^2 (\ZZ^d)$ the orthogonal projection given by $P_j \phi = \phi (j) \delta_j$ and apply Ineq.~\eqref{eq:average_proj} with the choice $H = H_{\Lambda_l}^\eta - \eta_o \lambda \sum_{k \in \Sigma} t(k) u(\cdot-k)$, $W = \lambda \sum_{k \in \Sigma} t(k) u(\cdot-k)$, $\zeta = \eta_o$ and $J = \sqrt{\lambda} P_j$. This gives by Lebesgue's theorem and the hypothesis of the theorem the estimate
\begin{align} 
\int_{\RR^{\lvert \Sigma \rvert}} \bigl\lVert \chi_I (H_{\Lambda_l}) \delta_j \bigr\rVert^2 \prod_{k \in \Sigma} \rho (\omega_k) \drm \omega_k 
&=\frac{1}{\lambda} \int_{\RR^{\lvert \Sigma \rvert}} \bigl\langle \delta_j , P_j \chi_I (H_{\Lambda_l}) P_j \delta_j \bigr\rangle \prod_{k \in \Sigma} \rho (\omega_k) \drm \omega_k \nonumber \\
&\leq \frac{\lvert I \rvert}{\lambda} \int_{\RR^{\lvert \Sigma \rvert-1}} \sup_{\eta_o \in \RR} \lvert k(\eta) \rvert \prod_{k \in \Sigma\setminus \{o\}}\drm\eta_k . \label{eq:afteraverage}
\end{align}
If $\rho \in W^{1,1} (\RR)$, we use $\sup_{\eta_o \in \RR} \lvert k(\eta) \rvert \leq \frac{1}{2} \int_\RR \lvert \partial_o k \rvert \drm \eta_o$. By the product rule we obtain for the partial derivative (while substituting back into original coordinates)
\[
 \partial_0 k = \frac{\partial}{\partial \eta_o} k(\eta) = \lvert t(o) \rvert^{\lvert \Sigma \rvert} \sum_{k \in \Sigma} t(k) \rho' (\omega_k) \prod_{j \in \Sigma \setminus \{k\}} \rho (\omega_j) .
\]
Hence, the right hand side of Ineq.~\eqref{eq:afteraverage} is bounded by $(2\lambda)^{-1} \lvert I \rvert \lVert \rho' \rVert_{L^1} \sum_{k \in \Sigma} \lvert t(k) \rvert$. Since all the steps were independent of $j \in \Lambda_{l}$, we in turn obtain the statement of the theorem in the case $\rho \in W^{1,1} (\RR)$.
\par
For $\rho$ of bounded total variation (note that $\rho$ has compact support since the measure $\nu$ has compact support) we use the fact that there is sequence $\rho_k \in C_{\rm c}^\infty (\RR)$, $k \in \NN$, such that $\lVert \rho_k \rVert_{L^1} = 1$ for all $k \in \NN$, $\lim_{k \to \infty} \lVert \rho_k \rVert_{\rm Var} = \lVert \rho \rVert_{\rm Var}$ and $\lim_{k \to \infty} \lVert \rho_k - \rho \rVert_{L^1} = 0$, see Lemma~\ref{lemma:approx} below. Since $\lVert \rho_k \rVert_{\rm Var} = \lVert \rho_k' \rVert_{L^1}$ for $\rho_k \in C_{\rm c}^\infty (\RR)$, the same consideration as above gives for all $k \in \NN$
\begin{equation}\label{eq:afteraverage2}
\int_{\RR^{\lvert \Sigma \rvert}} \bigl\lVert \chi_I (H_{\Lambda_l}) \delta_j \bigr\rVert^2 \prod_{i \in \Sigma} \rho_k (\omega_i) \drm \omega_i  \leq \frac{\lvert I \rvert}{2\lambda} \lVert \rho_k \rVert_{\rm Var} \sum_{i \in \Sigma} \lvert t(i) \rvert .
\end{equation}
%
Using a limiting argument we now show the assertion. We have for all $k \in \NN$
\begin{multline*}
 \int_{\RR^{\lvert \Sigma \rvert}} \bigl\lVert \chi_I (H_{\Lambda_l}) \delta_j \bigr\rVert^2 \prod_{i \in \Sigma} \rho (\omega_i) \drm \omega_i \leq \int_{\RR^{\lvert \Sigma \rvert}} \bigl\lVert \chi_I (H_{\Lambda_l}) \delta_j \bigr\rVert^2 \prod_{i \in \Sigma} \rho_k (\omega_i) \drm \omega_i \\
+ \int_{\RR^{\lvert \Sigma \rvert}} \bigl\lVert \chi_I (H_{\Lambda_l}) \delta_j \bigr\rVert^2 \Biggl[ \prod_{i \in \Sigma} \rho (\omega_i) - \prod_{i \in \Sigma} \rho_k (\omega_i) \Biggr] \prod_{i \in \Sigma} \drm \omega_i .
\end{multline*}
The first integral on the right hand can be bounded from above by Ineq.~\eqref{eq:afteraverage2}. To estimate the second integral we denote the elements of the set $\Sigma$ by $\sigma_i$, $i \in \{1,2,\ldots , n\}$ where $n = \lvert \Sigma \rvert$, i.e.\ $\Sigma = \{\sigma_1 , \sigma_2 , \ldots , \sigma_n\}$. We use a telescoping argument and obtain in a first step
\begin{align*}
L &= \prod_{i \in \Sigma} \rho (\omega_i) - \prod_{i \in \Sigma} \rho_k (\omega_i)  \\
& = \bigl[\rho (\omega_{\sigma_n}) - \rho_k (\omega_{\sigma_n})\bigr] \prod_{i=1}^{n-1} \rho (\omega_{\sigma_i})  
+ \rho_k (\omega_{\sigma_n}) \Biggl[ \prod_{i=1}^{n-1} \rho (\omega_{\sigma_i}) - \prod_{i=1}^{n-1} \rho_k (\omega_{\sigma_i}) \Biggr]
\end{align*}
Iterating this procedure we get
\[
 L = \sum_{i=1}^n \Biggl[ \prod_{l=1}^{i-1} \rho (\omega_{\sigma_l}) \Biggr] \Biggl[ \prod_{l=i+1}^{n} \rho_k (\omega_{\sigma_l}) \Biggr] \bigl( \rho (\omega_{\sigma_i}) - \rho_k (\omega_{\sigma_i}) \bigr).
\]
Here we use the convention that a product where the index set is empty equals one. Putting everything together by using $\lVert \rho \rVert_{L^1} = \lVert \rho_k \rVert_{L^1} = 1$ for $k \in \NN$ and $\lVert \chi_I (H_{\Lambda_l}) \delta_j \rVert \leq 1$, we obtain for all $k \in \NN$
\[
 \int_{\RR^{\lvert \Sigma \rvert}} \bigl\lVert \chi_I (H_{\Lambda_l}) \delta_j \bigr\rVert^2 \prod_{i \in \Sigma} \rho (\omega_i) \drm \omega_i
 \leq  \frac{\lvert I \rvert}{2\lambda} \lVert \rho_k \rVert_{\rm Var} \sum_{i \in \Sigma} \lvert t(i) \rvert + n \lVert \rho - \rho_k \rVert_{L^1} .
\]
Letting $k$ go to infinity we obtain by using Lemma~\ref{lemma:approx} that
\[
 \EE (\lVert \chi_I(H_{\Lambda_l}) \delta_j \rVert^2) = \int_{\RR^{\lvert \Sigma \rvert}} \bigl\lVert \chi_I (H_{\Lambda_l}) \delta_j \bigr\rVert^2 \prod_{i \in \Sigma} \rho (\omega_i) \drm \omega_i
 \leq  \frac{\lvert I \rvert}{2\lambda} \lVert \rho \rVert_{\rm Var} \sum_{i \in \Sigma} \lvert t(i) \rvert .
\]
Since all the steps were independent of $j \in \Lambda_l$ we achieve the statement of the theorem.
\end{proof}
\begin{lemma} \label{lemma:approx}
 Let $u: \RR \to \RR_0^+$ be a function of finite variation and bounded support. Assume additionally $\lVert u \rVert_{L^1} = 1$. Then there exists a sequence $u_k \in C_{\rm c}^\infty$, $k \in \NN$, such that $\lVert u_k \rVert_{L^1} = 1$ for all $k \in \NN$, 
\begin{equation} \label{eq:conv_d}
 \lim_{k \to \infty} \lVert u_k \rVert_{\rm Var} = \lVert u \rVert_{\rm Var} 
\end{equation}
and
\begin{equation} \label{eq:L1}
\lim_{k \to \infty} \lVert u_k - u \rVert_{L^1} = 0 .
\end{equation}
\end{lemma}
\begin{proof}
 Let $\phi \in C_{\rm c}^\infty (\RR)$ be non-negative with $\supp \phi \subset [-1,1]$ and $\lVert \phi \rVert_{L^1} = 1$. For $\epsilon > 0$ set $\phi_\epsilon : \RR \to \RR_0^+$, $\phi_\epsilon (x) = \epsilon^{-1} \phi (x/\epsilon)$. The function $\phi_\epsilon$ belongs to $C_{\rm c}^\infty (\RR)$ and fulfills $\lVert \phi_\epsilon \rVert_{L^1} = 1$. Now consider $u_\epsilon : \RR \to \RR_0^+$, 
\[
 u_\epsilon (x) = \int_\RR \phi_\epsilon (x-y) u(y) \drm y .
\]
Obviously, $u_\epsilon \in C_{\rm c}^\infty (\RR)$ and by Fubini's theorem $\lVert u_\epsilon \rVert_{L^1} = 1$. The proof of the relation \eqref{eq:L1} is due to Theorem 1.6.1 in \cite{Ziemer-89}. For the proof of the relation \eqref{eq:conv_d}, first note 
\begin{align*}
 \lVert u \rVert_{\rm Var} &= \lvert Du \rvert (\RR) = \sup \Bigl\{ \int_\RR u v' \drm x \colon v \in C_{\rm c}^\infty (\RR) , \ \lvert v \rvert \leq 1 \Bigr\} \\
&= \sup \Bigl\{ \lim_{\epsilon \searrow 0} \int_\RR u_\epsilon v' \drm x \colon v \in C_{\rm c}^\infty (\RR) , \ \lvert v \rvert \leq 1 \Bigr\} \\[1ex]
& \leq \liminf_{\epsilon \searrow 0} \lvert D u_\epsilon \rvert (\RR) =  \liminf_{\epsilon \searrow 0} \lVert u_\epsilon \rVert_{\rm Var} ,
\end{align*}
since $u_\epsilon$ converges to $u$ in $L^1(\RR)$ and $v'$ is bounded. Let now $\psi \in C_{\rm c}^\infty (\RR)$ with $\vert\psi\rvert \leq 1$ and set $\psi_\epsilon = \phi_\epsilon * \psi$. Then we have by Fubini's theorem
\begin{align*}
 \lVert u \rVert_{\rm Var} &\geq \Bigl\lvert\int_\RR u \psi_\epsilon' \drm x \Bigr\rvert = \Bigl\lvert\int_\RR u (\psi * \phi_\epsilon)' \drm x \Bigr\rvert =  \Bigl\lvert\int_\RR  \psi'(u * \phi_\epsilon) \drm x \Bigr\rvert = \Bigl\lvert\int_\RR u_\epsilon \psi' \drm x \Bigr\rvert .
\end{align*}
Taking supremum over all such $\psi$ gives $\lVert u \rVert_{\rm Var} \geq \lVert u_\epsilon \rVert_{\rm Var}$. This proves the lemma.
\end{proof}
Assume Assumption~\ref{ass:exponential}, i.e. that there are $C,\alpha \in (0,\infty)$ such that $\lvert u(x) \rvert \leq C \euler^{-\alpha \lvert x \rvert_1}$ for all $x \in \ZZ^d$. Let $I_0$ and $c_u \not = 0$ be as in Eq.~\eqref{eq:cF}. In Section \ref{sec:transformation} we will construct for each $l \in \NN$ a number $R_l > 0$ given in Eq.~\eqref{eq:RLrel} such that
\begin{equation} \label{eq:vpos}
 \frac{2}{c_u} \sum_{k \in \Lambda_{R_l}} k^{I_0} u(x-k) \geq 1 \quad \text{for all $x \in \Lambda_{l}$} .
\end{equation}
Note that this lower bound is uniform in $x \in \Lambda_l$ and thus much stronger than condition \eqref{eq:wegner:ass}. If Assumption \ref{ass:ubar} holds one can use in Ineq.\ \eqref{eq:vpos} exponentially decaying coefficients rather than $k \mapsto k^{I_0}$, cf.\ Appendix \ref{chap:non_local}. Inequality \eqref{eq:vpos} is proven in Proposition \ref{prop2} and we will apply it for the discrete alloy-type model with exponential decaying single-site potential to verify the hypothesis of Theorem \ref{theorem:abstract2}. Before we prove Proposition~\ref{prop2}, we give the 
\begin{proof}[Proof of Theorem \ref{theorem:wegner}] \label{proof:wegner_d}
By Ineq.~\eqref{eq:vpos} (respectively Proposition \ref{prop2}), the hypothesis of Theorem~\ref{theorem:abstract2} is satisfied with the choice $l_0 = 1$ and $t_{j,l} \in \ell^1 (\ZZ^d)$ given by 
\[
t_{j , l} (k) = \begin{cases}
                 2 k^{I_0} / c_{u} & \text{if $k \in \Lambda_{R_l}$}, \\
		0 & \text{else},
                \end{cases}
\]
for $l \in \NN$ and $j \in \Lambda_{l}$. The constants $c_u \not =0$ and $I_0 \in \NN_0^d$ depend only on the single site potential $u$ and are defined in Eq.~\eqref{eq:cF}. It follows for all $l \in \NN$ and $j \in \Lambda_{l}$ that
\begin{align*}
 \sum_{j \in \Lambda_{l}} \lVert t_{j,l} \rVert_{\ell^1} &= \frac{2}{\lvert c_{u} \rvert} (2l+1)^d \sum_{k \in \Lambda_{R_l}} \lvert k^{I_0} \rvert \leq 
\frac{2}{\lvert c_{u} \rvert} (2l+1)^d (2R_l + 1)^d R_l^{\lvert I_0 \rvert} .
\end{align*}
By Proposition \ref{prop2}, $R_l = \max \{2l + D , D'\} < 2l+D+D'$ with $D$ and $D'$ depending only on the single-site potential $u$. Hence there is a constant $C(u) > 0$ depending only on the single site potential $u$ such that
\begin{equation*} \label{eq:volume2} 
 \sum_{j \in \Lambda_{l}} \lVert t_{j,l} \rVert_{\ell^1} \leq
C(u) (2l+1)^{2d + \lvert I_0 \rvert}.
\end{equation*}
By Theorem~\ref{theorem:abstract2}, this completes the proof.
\end{proof}
\section{Positive combinations of translated single-site potentials} \label{sec:transformation}
In this section we consider (possibly infinite) linear combinations of translates of the single-site potential $u$. We assume that $u$ decays exponentially and is distinct from the zero function. Under these hypotheses we identify a sequence of coefficients such that the resulting linear combination is uniformly positive on the whole space $\ZZ^d$ (cf. Proposition \ref{prop1}) or on some finite set $\Lambda \subset \ZZ^d$ (cf. Proposition~\ref{prop2}).
\begin{remark}[Preliminaries]
%
We introduce the following multi-index notation: If $I
=(i_1,\dots,i_d) \in \ZZ^d$ and $z \in \CC^d$, we define
\[
z^I = z_1^{i_1} \cdot z_2^{i_2} \cdot \ldots \cdot z_d^{i_d},
\]
and if $I \in \NN_0^d$, we define
\begin{align*}
  |I|   = \sum_{r=1}^d i_r, \quad
  D_z^I = \frac{\partial^{i_1}}{\partial z_1^{i_1}} \cdot \frac{\partial^{i_2}}{\partial z_2^{i_2}} \cdot \ldots \cdot \frac{\partial^{i_d}}{{\partial z_d}^{i_d}}, \quad
  I! = i_1! \cdot i_2! \cdot \ldots \cdot i_d! \,. 
\end{align*}
We also introduce comparison symbols for multi-indices: If $I, J \in \NN_0^d$, we write $J \le I$ if we have $j_r \le i_r$ for all
$r=1,2,\ldots,d$, and we write $J < I$ if $J \le I$ and $|J| <
|I|$. For $J \le I$, we use the short hand notation
\[ 
\binom{I}{J} = \binom{i_1}{j_1} \cdot \binom{i_2}{j_2} \cdot \ldots \cdot \binom{i_d}{j_d}. 
\]
Finally, ${\mathbf 0}, {\mathbf 1}$ denote the vectors $(0,\dots,0)$
and $(1,\dots,1) \in \CC^d$, respectively.
\par
We also recall the following facts from multidimensional complex
analysis. Let $\cD \subset \CC^d$ be open. We call a complex valued function $f : \cD \to \CC$ holomorphic, if every point $w \in \cD$ has an open neighbourhood $U$, $w \in U \subset \cD$, such that $f$ has a power series expansion around $w$, which converges to $f(z)$ for all $z \in U$.
Osgood's lemma tells us that, if $f : \cD \to \CC$ is continuous and holomorphic in each variable separately (in the sense of one-dimensional complex analysis), then $f$ is holomorphic, see \cite{GunningR-09}.
\par
Let $f_n : \cD \to \CC$ be a sequence of holomorphic functions. We say that $\sum_n f_n$ converges normally in $\cD$, if for every $w \in \cD$ there is an open neighbourhood $U$, $w \in U \subset \cD$, such that $\sum_n \lVert f_n \rVert_{U,\infty} < \infty$. Normally convergent sequences of holomorphic functions can be rearranged arbitrarily, the limit is again holomorphic, and differentiation can be carried out termwise, which follows from Weierstrass' theorem, see \cite[p. 226]{Remmert-84} for the one-dimensional case and \cite[p. 7]{Narasimhan-95} for the higher dimensional case.
\end{remark}

\begin{remark}[Notation]
Let $u: \ZZ^d \to \RR$ be a function satisfying Assumption \ref{ass:exponential}, i.e., there are constants $C,\alpha > 0$ such that
\begin{equation} \label{eq:exponential}
\lvert u(k) \rvert \leq C \euler^{-\alpha \lvert k \rvert_1}
\end{equation}
for all $k \in \ZZ^d$. For $\delta \in (0, 1-\euler^{-\alpha})$ we consider the associated generating function $F : \cD_\delta \subset \CC^d \to \CC$, 
\begin{equation} \label{eq:Fz}
\cD_\delta = \{ z \in \CC^d : \lvert z_1 - 1 \rvert < \delta , \ldots , \lvert z_d - 1 \rvert < \delta \}, \quad F(z) = \sum_{k \in \ZZ^d} u(-k) z^k .
\end{equation}
Here $\lvert z \rvert = \sqrt{z \overline z}$ for $z \in \CC$. Notice that the sum $\sum_{k \in \ZZ^d} u(-k) z^k$ is normally convergent in $\cD_\delta$ by our choice of $\delta$ and the exponential decay condition \eqref{eq:exponential}. By Weierstrass' theorem, $F$ is a holomorphic function. Since $F$ is holomorphic and not identically zero, we have $(D_z^I F) (\mathbf{1}) \not = 0$ for at least one $I \in \NN_0^d$. Therefore, there exists a multi-index $I_0 \in \NN_0^d$ (not necessarily unique), such that we have
\begin{equation} \label{eq:cF}
  (D_z^I F)({\mathbf 1}) = 
		\begin{cases} 
			c_u \neq 0, & \text{if $I = I_0$,} \\
			0,          & \text{if $I < I_0$.} 
		\end{cases}  
\end{equation}
(Such a $I_0$ can be found by diagonal inspection: Let $n \ge 0$ be the largest integer such that $D_z^IF({\mathbf 1}) = 0$ for all $|I| < n$.  Then choose a multi-index $I_0 \in \NN_0^d$, $|I_0| = n$ with $(D_z^{I_0} F)({\mathbf 1}) \neq 0$.)
\end{remark}
\begin{proposition} \label{prop1}
Let $u$, $F$, $c_u$ and $I_0$ be as in \eqref{eq:exponential}, \eqref{eq:Fz}, and \eqref{eq:cF}. Let further $I \in \NN_0^d$ with $I \le I_0$, and define $a: \ZZ^d \to \ZZ$ by 
\[
a(k) = k^I.
\]
Then we have for all $x \in \ZZ^d$
\begin{equation} \label{eq:akuxk} 
    \sum_{k \in \ZZ^d} a(k) u(x-k) = \begin{cases} 0, & \text{if $I < I_0$,} \\
    c_u, & \text{if $I = I_0$.} \end{cases}
\end{equation}
\end{proposition}
\begin{proof}
We introduce, again, a bit of notation. For $s \in \CC^d$ and $k \in \ZZ^d$ let
\[ 
  e^s  =  (e^{s_1},\dots,e^{s_d}) \quad \text{and} \quad \langle k,s \rangle  =  \sum_{r=1}^d k_r s_r.
\] 
Let $I \le I_0$. Then the chain rule yields (for all $s \in C_\delta := \{s \in \CC^d : {\rm e}^s \in \cD_\delta\}$)
\begin{align*}
D_s^I(F(e^s)) &= \sum_{J \le I} c_J\, (D_z^J F)(e^s)\, e^{\langle J,s \rangle} \\
&=(D_z^I F)(e^s)\, e^{\langle I,s \rangle}+  \sum_{J <I} c_J\, (D_z^J F)(e^s)\, e^{\langle J,s \rangle},
\end{align*}
with suitable integers $c_J \ge 1$ and, in particular, $c_I = 1$. This and Eq.~\eqref{eq:cF} imply that
\begin{equation} \label{eq:diffFes}
  D_s^I(F(e^s))\big\vert_{s={\mathbf 0}} = \begin{cases} 0, & \text{if
      $I < I_0$,}\\ c_{I_0}\, (D^{I_0}_zF)({\mathbf 1}) = c_u, & \text{if
      $I = I_0$.} \end{cases}
\end{equation}
Next, we use the identity $a(k) = k^I = D_s^I e^{\langle k,s   \rangle}\big\vert_{s={\mathbf 0}}$. Note that the series $\sum_{k\in \ZZ^d} u(x-k)e^{\langle k,s \rangle}$ converges normally on the domain
\[
E_\alpha = \{ s \in \CC^d \mid - \alpha < {\rm Re}(s_j) < \alpha \ \text{for all $j=1,2,\dots,d$} \}, 
\]
Therefore, we can rearrange arbitrarily, differentiate componentwise, and obtain for all $s \in C_\delta \cap E_\alpha$ by substitution $\nu = k-x$ and the product rule
\begin{align*}
\sum_{k \in \ZZ^d} u(x-k) D_s^I e^{\langle k,s \rangle} &= D_s^I \sum_{k \in \ZZ^d} u(x-k) e^{\langle k,s \rangle} \\
&= D_s^I \Bigl( e^{\langle x,s \rangle}  \sum_{\nu \in \ZZ^d} u(-\nu) e^{\langle \nu,s \rangle} \Bigr) = D_s^I \Bigl( F(e^s) e^{\langle x,s \rangle} \Bigr) \\
&= \sum_{J \le I} \binom{I}{J} \bigl( D_s^J F(e^s) \bigr) D_s^{I-J} e^{\langle x,s \rangle} . 
\end{align*} 
Finally, evaluating at $s = \mathbf 0$ and using \eqref{eq:diffFes} yields
\begin{align*}
\sum_{k \in \ZZ^d} a(k) u(x-k) &= \sum_{J \le I} \binom{I}{J} \bigl( D_s^J F(e^s)\bigr) \big\vert_{s={\mathbf 0}} \bigl( D_s^{I-J}(e^{\langle x,s \rangle} \bigr)\big\vert_{s={\mathbf 0}} \\
&= \begin{cases} 0, & \text{if $I < I_0$}, \\
        c_u, & \text{if $I = I_0$.} \end{cases} \qedhere
\end{align*}
\end{proof}

In Proposition \ref{prop1} we identified a sequence of coefficients such that the associated linear combination of translated single site potentials is positive on the whole of $\ZZ^d$. However, the sequence cannot be used for Theorem \ref{theorem:abstract2} directly since it is not summable. This problem can be resolved if we take into consideration that the positivity in Theorem \ref{theorem:abstract2} concerns lattice sites in $\Lambda_{n}$ only.
\par
Recall that the constants $d, \alpha, C$ and $c_u$ are all determined by the choice of the exponentially decreasing function $u: \ZZ^d \to \RR$. Now we choose $I=I_0$ in Proposition~\ref{prop1}.
The next proposition tells us, for all integer vectors $x$ in the box $\Lambda_l$, how far we have to exhaust $\ZZ^d$ in the sum \eqref{eq:akuxk}, in order to guarantee that the result is $\ge
\frac{c_u}{2}$ (assuming for a moment that $c_u > 0$). The exhaustion is described by the integer indices in another box $\Lambda_R$, and the proposition describes the relation between the sizes $l$ and $R$. For large enough $l$, this relation is linear.
\begin{proposition} \label{prop2}
Let $u$, $c_u$ and $I_0$ be as in \eqref{eq:exponential} and \eqref{eq:cF}. Let further $l \in \NN$ and define
\begin{equation} \label{eq:RLrel} 
  R_l := \max \left\{ 2l + \frac{2}{\alpha} \ln \frac{2\cdot 3^d\, C}
  {\lvert c_u \rvert (1-e^{-\alpha/2})}, \frac{8 (d+| I_0 |)^2}{\alpha^2} \right\}. 
\end{equation}
Then we have for all $x \in \Lambda_{l}$
\[ 
  \frac{2}{c_u} \sum_{k \in \Lambda_{R_l}}  k^{I_0}\, u(x-k) \ge 1.
\] 
\end{proposition}
\begin{proof}
We know from Proposition \ref{prop1} that 
\[
\frac{1}{c_u} \sum_{k \in \ZZ^d} k^{I_0}\, u(x-k) = 1, 
\]
for all $x \in \ZZ^d$. Thus we need to prove, for $x \in \Lambda_{l} = \ZZ^d \cap [-l,l]^d$, that
\begin{equation} \label{eq:desired}
\Big  | \sum_{k \in \ZZ^d \setminus \Lambda_{R_l} } k^{I_0}\, u(x-k)\Big | \le \frac{\lvert c_u \rvert}{2}.
\end{equation}
Using the triangle inequality $\vert k-x \vert_\infty + \vert x \vert_\infty \ge \vert k \vert_\infty$, $\vert k \vert_\infty \le \vert k \vert_1$, and that $u$ is exponentially  decreasing, we obtain
\begin{align*}
\Bigl\lvert \sum_{k \in \ZZ^d \setminus \Lambda_{R_l} } k^{I_0}\, u(x-k) \Bigr\rvert &\le C \sum_{\vert k \vert_\infty \ge R_l} \vert k \vert_\infty^{|I_0|} \, {\rm e}^{-\alpha \vert x-k \vert_\infty} \\
    &\le C e^{\alpha \vert x \vert_\infty} \sum_{\vert k \vert_\infty \ge R_l} 
    \vert k \vert_\infty^{|I_0|} \, e^{-\alpha \vert k \vert_\infty} \\
    &\le C e^{\alpha l} \sum_{r=R_l}^\infty (2r+1)^d \, r^{|I_0|} \, e^{-\alpha r} \\
    &\le C 3^d e^{\alpha l} \sum_{r=R_l}^\infty r^{d+|I_0|} e^{-\alpha r}.
\end{align*}
   Using Lemma \ref{lem:nexp} below and $r \ge R_l
   \ge\frac{8(d+|I_0|)^2}{\alpha^2}$, we conclude that $r^{d+|I_0|}
   \le e^{\alpha r/2}$, which implies that
   \[
  \Big| \sum_{k \in \ZZ^d \setminus \Lambda_{R_l} } k^{I_0}\, u(x-k)\Big |
   \le
  C 3^d e^{\alpha l} \sum_{r=R_l}^\infty e^{-\alpha r/2} = C 3^d
  e^{\alpha l} \frac{e^{-\alpha R_l/2}}{1-e^{-\alpha/2}}.   
   \]
Finally, using $R_l \ge 2l + \frac{2}{\alpha} \ln (2 \cdot 3^d C / (\lvert c_u \rvert (1-e^{-\alpha/2}))$, we conclude Ineq.~\eqref{eq:desired} which ends the proof.
\end{proof}

\begin{lemma} \label{lem:nexp}
  Let $M,\alpha > 0$. Then
  $$ 
  n \ge \frac{8M^2}{\alpha^2} \quad \Rightarrow \quad n^M < e^{\alpha n/2}. 
  $$ 
\end{lemma}

\begin{proof}
  If $n \ge \frac{8M^2}{\alpha^2}$ then
  $$ n \le \frac{\alpha^2 n^2}{8 M^2}. $$ 
  Since
  $$ e^{\frac{\alpha n}{2M}} = 1 + \frac{\alpha n}{2M} + 
  \frac{\alpha^2 n^2}{8 M^2} + \dots > \frac{\alpha^2 n^2}{8 M^2 }, $$
  we conclude that $n \le e^{\frac{\alpha n}{2M}}$, or, equivalently, $n^M \le e^{\frac{\alpha n}{2}}$.
\end{proof}
\appendix
%
\chapter{A non-local a priori bound} \label{chap:non_local}
An important step in the proof of exponential decay of fractional moments is the so-called a priori bound, cf. 
Lemma~\ref{lemma:bounded}. It was this step where Assumption~\ref{ass:monotone} enters the proof of exponential localization formulated in Theorem~\ref{theorem:localization}.
\par
In this appendix we will proof an alternative a priori bound which holds under a much milder hypothesis on $u$, see Assumption~\ref{ass:ubar} in Section~\ref{sec:model}. By milder we do not mean that this covers the whole class of models where Assumption~\ref{ass:monotone} is satisfied, but rather it holds generically in the class of compactly supported single-site potentials. However, the drawback of the a priori bound of this appendix is that it is non-local in the sense that it requires averaging over the entire disorder present in the model. At the moment we do not know how to conclude exponential decay of fractional moments relying on this version of the a priori bound.
\par
The result of this appendix extends Theorem 2.3 in \cite{TautenhahnV-10} and has been published in \cite{ElgartTV-11}.
\section{Result} \label{sec:result_ubar}
The precise assumption we will need is Assumption~\ref{ass:ubar}, which states that the measure $\nu$ has a density in the Sobolev space $W^{1,1} (\RR)$, $\Theta$ is finite and the single-site potential satisfies $\bar u := \sum_{k \in \ZZ^d} u(k) \not = 0$.
\begin{remark}
Note that without loss of generality the assumption $\bar u \not = 0$ can be replaced by $\bar u > 0$, since
\[
V_\omega (x) = \sum_{k \in \ZZ^d} \omega_k u (x-k) = 
\sum_{k \in \ZZ^d} \bigl(-\omega_k\bigr) \bigl(- u (x-k) \bigr) .
\]
We will assume this for the rest of this chapter.
\end{remark}
The main result of this appendix is 
\begin{theorem} \label{theorem:ubar}
 Let $\Lambda \subset \ZZ^d$ finite, $s \in (0,1)$ and Assumption~\ref{ass:ubar} be satisfied. Then we have for all $x,y \in \Lambda$ and $z \in \CC \setminus \RR$
\[
\mathbb{E} \Bigl( \bigl\lvert G_\Lambda (z;x,y) \bigr\rvert^s \Bigr)  \leq \frac{8}{(\overline{u})^s} \frac{s^{-s}}{1-s} \lVert \rho' \rVert_{L^1}^s C^s \frac{1}{\lambda^s} ,
\]
where $C$ is the constant from Eq.~\eqref{eq:C}.
\end{theorem}
The proof relies on a special transformation of the random variables $\omega_k$, $k \in \Lambda_+$, where $\Lambda_+ = \cup_{k \in \Lambda} \{x \in \ZZ^d \colon u(x-k) \not = 0\}$ denotes the set of lattice sites whose coupling constant influences the potential in $\Lambda$. If $\Theta$ is finite, let us denote by $n$ the diameter of $\Theta$ with respect to the $\ell^1$-norm, i.e. $n := \max_{i,j \in \Theta} \lvert i-j \rvert_1$. For $x,y \in \ZZ^d$ we define $\alpha^{x,y} : \ZZ^d \to \RR^+$ by
\begin{equation} \label{eq:alpha}
\alpha^{x,y} (k) = \frac{1}{2} \left( {\rm e}^{-c \lvert k-x \rvert_1} + {\rm e}^{-c \lvert k-y \rvert_1} \right) \quad \text{with} \quad c:=  \frac{1}{n} \ln \left( 1 + \frac{\overline{u}}{2 \lVert u \rVert_{\ell^1}} \right) .
\end{equation}
Notice that the $\ell^1$-norm of $\alpha^{x,y}$ is independent of $x,y \in \ZZ^d$, i.e.
\begin{equation} \label{eq:C}
 C := C(n,\bar u , \lVert u \rVert_{\ell^1}) = \sum_{k \in \ZZ} \lvert \alpha^{x,y} (k) \rvert = \sum_{k \in \ZZ^d} {\rm e}^{-c\lvert k \rvert_1} = \left( \frac{{\rm e}^{c} + 1}{{\rm e}^{c} - 1} \right)^d .
\end{equation}
With the help of the coefficients $\alpha^{x,y} (k)$, $k \in \ZZ^d$, we will define a linear transformation of the variables $\omega_k$, $k \in \Lambda_+$. Some part of the ``new'' potential will then be given by $W^{x,y} : \ZZ^d \to \RR$,
\begin{equation} \label{eq:Wij}
 W^{i,j} (x) = \sum_{k \in \ZZ^d} \alpha^{i,j} (k) u(x-k),
\end{equation}
where indeed only the values $x \in \Lambda$ are relevant. In order to use monotone spectral averaging as presented in Section~\ref{sec:monotone_average}, it is important that $W^{x,y}$ is positive and that $W^{x,y} (k) \geq \delta > 0$ for $k \in \{x,y\}$ where $\delta$ is independent of $x,y \in \Lambda$. This is done by
\begin{lemma} \label{lemma:Wij}
Let Assumption \ref{ass:ubar} be satisfied. Then we have for all $x,y,k \in \ZZ^d$
\[
 W^{x,y} (k) \geq \alpha^{x,y} (k) \frac{\overline{u}}{2} > 0 .
\]
In particular, $W^{x,y} (k) \geq \overline{u} / 4$ for $k \in \{x,y\}$.
\end{lemma}
\begin{proof}[Proof of Lemma~\ref{lemma:Wij}]
Recall that $n = \max_{i,j \in \Theta} \lvert i-j \rvert_1$ and that we have assumed $0 \in \Theta$.
For $k \in \ZZ^d$ let $B_n (k) = \{j \in \ZZ^d : \lvert j-k \rvert_1 \leq n\}$. The triangle inequality gives us for all $k \in \ZZ^d$
\begin{align*}
M &= \max_{j \in B_n (k)} \lvert \alpha^{x,y}(k) - \alpha^{x,y} (j) \rvert  \\ 
  &\leq \frac{1}{2} \max_{j \in B_n (k)} \bigl\lvert \euler^{-c \lvert k-x \rvert_1} - \euler^{-c \lvert j-x \rvert_1} \bigr\rvert + 
  \frac{1}{2} \max_{j \in B_n (k)} \bigl\lvert \euler^{-c \lvert k-y \rvert_1} - \euler^{-c \lvert j-y \rvert_1} \bigr\rvert .
\end{align*}
Since $\RR \ni t \mapsto {\rm e}^{-ct}$ is a convex and strictly monotonic decreasing function, we have for all $x \in \ZZ^d$
\begin{align}
M &\leq \frac{1}{2}  \bigl\lvert \euler^{-c \lvert k-x \rvert_1} - \euler^{-c (\lvert k-x \rvert_1 - n)} \bigr\rvert + 
\frac{1}{2}  \bigl\lvert \euler^{-c \lvert k-y \rvert_1} - \euler^{-c (\lvert k-y \rvert_1 - n)} \bigr\rvert \nonumber \\
&\leq \alpha^{x,y} (k) ({\rm e}^{cn} - 1) \label{eq:convex} .
\end{align}
We use Ineq.~\eqref{eq:convex} and that $u (k-j) = 0$ for $k-j \not \in \Theta$, and obtain the estimate
\begin{align*}
 W^{x,y} (k) &= \sum_{j \in \ZZ^d} \alpha^{x,y} (k) u (k-j) + \sum_{j \in \ZZ^d} \bigl[\alpha^{x,y} (j) - \alpha^{x,y} (k)\bigr] u (k-j) \\
&\geq \alpha^{x,y} (k) \overline u - \sum_{j \in \ZZ^d} \lvert \alpha^{x,y} (k) - \alpha^{x,y} (j) \rvert \lvert u (k-j) \rvert \\
& \geq \alpha^{x,y} (k) \bar u  - \alpha^{x,y} (k) (\euler^{cn} - 1) \lVert u \rVert_{\ell^1} .
\end{align*}
This implies the statement of the lemma due to the choice of $c$.
\end{proof}
\begin{proof}[Proof of Theorem \ref{theorem:ubar}] Without loss of generality we assume $z \in \CC^- := \{z \in \CC \colon \im z < 0\}$. Fix $x,y \in \Lambda$ and recall that $\Lambda_+ = \cup_{k \in \Lambda} \{x \in \ZZ^d \mid u(x-k) \not = 0\}$ is the set of lattice sites whose coupling constant influences the potential in $\Lambda$. We consider the expectation
\[
 E = \EE \Bigl ( \bigl| G_\Lambda (z;x,y) \bigr|^s \Bigr) = 
\int_{\Omega_{\Lambda_+}} \bigl| \bigl\langle \delta_x , (H_\Lambda - z)^{-1} \delta_y \bigr\rangle \bigr|^s k(\omega_{\Lambda_+}) \drm \omega_{\Lambda_+} ,
\]
where $\Omega_{\Lambda_+} = \times_{k \in \Lambda_+} \RR$, $\omega_{\Lambda_+} = (\omega_k)_{k \in \Lambda_+}$, $k (\omega_{\Lambda_+}) = \prod_{k \in \Lambda_+} \rho (\omega_k)$ and $\drm \omega_{\Lambda_+} = \prod_{k \in \Lambda_+} \drm \omega_k$. Fix $v \in \Lambda_+$. We introduce the change of variables
\[
 \omega_v = \alpha^{x,y} (v) \zeta_v , \quad \text{and} \quad \omega_k = \alpha^{x,y} (k) \zeta_v + \alpha^{x,y} (v) \zeta_k
\]
for $k \in \Lambda_+ \setminus \{v\}$, where $\alpha^{x,y} : \ZZ^d \to \RR^+$ is defined in Eq.~\eqref{eq:alpha}. With this transformation we obtain
\begin{align*}
E &= \int_{\Omega_{\Lambda_+}} \bigl| \bigl\langle \delta_x , (-\Delta_\Lambda + \lambda V_\Lambda - z)^{-1} \delta_y \bigr\rangle \bigr|^s k(\omega_{\Lambda_+}) \drm \omega_{\Lambda_+} \\
&= \int_{\Omega_{\Lambda_+}} \bigl|\bigl\langle \delta_x , \bigl(A + \zeta_v \lambda W^{x,y} \bigr)^{-1} \delta_y \bigr\rangle \bigr|^s \tilde k (\zeta_{\Lambda_+})  \drm \zeta_{\Lambda_+} ,
\end{align*}
where $\zeta_{\Lambda_+} = (\zeta_k)_{k \in \Lambda_+}$,
\[
\tilde k (\zeta_{\Lambda_+}) = \lvert \alpha^{x,y} (v) \rvert^{\lvert \Lambda_+ \rvert}\rho (\alpha^{x,y} (v) \zeta_v) \prod_{k \in \Lambda_+ \setminus \{v\}} \rho (\alpha^{x,y} (k) \zeta_v + \alpha^{x,y} (v) \zeta_k) , 
\] 
$\drm \zeta_{\Lambda_+} = \prod_{k \in \Lambda_+} \drm \zeta_k$, $A= -\Delta_\Lambda - z + \lambda \alpha^{x,y} (v) \sum_{k \in \Lambda_+ \setminus \{v\}} \zeta_k u(\cdot - k)$ and $W^{x,y} : \ell^2 (\Lambda) \to \ell^2 (\Lambda)$ is the multiplication operator with multiplication function given by Eq.~\eqref{eq:Wij}. Notice that $A$ is independent of $\zeta_v$ and $W^{x,y}$ is positive by Lemma \ref{lemma:Wij}. We use Fubini's theorem to integrate first with respect to $\zeta_v$, Theorem \ref{theorem:monotone_average} and the estimate $\sup_{x \in \RR} g(x) \leq \frac{1}{2} \int_\RR \lvert g'(x)\rvert \drm x$ for $g \in W^{1,1} (\RR)$ and obtain for all $\kappa > 0$ and $z \in \CC^-$
\begin{align*}
E & \leq \frac{8 \cdot 4^{-s} \lambda^{-s}}{[W^{x,y}(x)W^{x,y}(y)]^{s/2}} \left[ \kappa^{-s} \int_{\Omega_{\Lambda_+}} \tilde k (\zeta_{\Lambda_+}) \drm \zeta_{\Lambda_+} + \frac{ \kappa^{1-s}}{1-s} \int_{\Omega_{\Lambda_+}} \biggl| \frac{\partial \tilde k (\zeta_{\Lambda_+})}{\partial \zeta_v} \biggr| \drm \zeta_{\Lambda_+} \right] \\
&= \frac{8 \cdot 4^{-s} \lambda^{-s}}{[W^{x,y}(x)W^{x,y}(y)]^{s/2}} \left[ \kappa^{-s} + \frac{ \kappa^{1-s}}{1-s} \int_{\Omega_{\Lambda_+}} \biggl| \frac{\partial \tilde k (\zeta_{\Lambda_+})}{\partial \zeta_v} \biggr| \drm \zeta_{\Lambda_+} \right] .
\end{align*}
For the partial derivative we calculate
\[
\frac{\partial \tilde k (\zeta_{\Lambda_+})}{\partial \zeta_v} = \lvert \alpha^{x,y} (v) \rvert^{\lvert \Lambda_+ \rvert} \sum_{l \in \Lambda_+} \alpha^{x,y} (l) \rho' (\omega_l) \prod_{\genfrac{}{}{0pt}{}{k \in \Lambda_+}{k \not = l}} \rho (\omega_k) ,
\]
which gives (while substituting back into original coordinates)
\begin{align*}
E & \leq \frac{8 \cdot 4^{-s} \lambda^{-s}}{[W^{x,y}(x)W^{x,y}(y)]^{s/2}} \left[ \kappa^{-s} + \frac{ \kappa^{1-s}}{1-s} \sum_{l \in \Lambda_+} \lvert \alpha^{x,y} (l) \rvert \int_{\Omega_{\Lambda_+}} \lvert \rho'(\omega_l) \rvert \prod_{k \not = l} \rho (\omega_k)  \drm \omega_{\Lambda_+} \right] \\
& \leq   \frac{8 \cdot 4^{-s} \lambda^{-s}}{[W^{x,y}(x)W^{x,y}(y)]^{s/2}} \left[ \kappa^{-s} + \frac{ \kappa^{1-s}}{1-s} \lVert \rho' \rVert_{L^1} \sum_{l \in \ZZ^d} \lvert \alpha^{x,y} (l) \rvert   \right] \\
& \leq  \frac{8 \cdot 4^{-s} \lambda^{-s}}{[W^{x,y}(x)W^{x,y}(y)]^{s/2}} \left[ \kappa^{-s} + \frac{ \kappa^{1-s}}{1-s} \lVert \rho' \rVert_{L^1} C \right] ,
\end{align*}
where $C$ is the constant from Eq.~\eqref{eq:C} and where we have used that $W^{x,y}(x)$ and $W^{x,y}(y)$ are bounded from below by $\overline{u}/4$ by Lemma \ref{lemma:Wij}. If we choose $\kappa = s/(\lVert \rho' \rVert_{L^1} C)$ we obtain the statement of the theorem.
\end{proof}
\section{Monotone spectral averaging} \label{sec:monotone_average}
In this section we prove Theorem \ref{theorem:monotone_average} which we have used in Section~\ref{sec:result_ubar} to prove boundedness of an averaged fractional power of Green's function under Assumption \ref{ass:ubar}.
\begin{theorem} \label{theorem:monotone_average}
 Let $X$ be a countable set, $A = B + \i C$ be an operator on $\ell^2 (X)$ where $B$ is self-adjoint and $C$ is bounded with $C \geq 0$. Let also $W : \ell^2 (X) \to \ell^2 (X)$ be a bounded and positive multiplication operator, $s \in (0,1)$ and $0 \leq \rho \in L^1 (\RR) \cap L^\infty (\RR)$. Then we have for all $x,y \in X$, $z \in \CC^-$ and $\kappa > 0$
\begin{align}
 \int_\RR \bigl\lvert \bigl\langle \delta_x , (A + t W - z)^{-1} \delta_y \bigr\rangle \bigr\rvert^s \rho (t) \drm t 
 &\leq  \frac{8\cdot 4^{-s}}{[W(x)W(y)]^{s/2}} \lVert \rho \rVert_\infty^s \frac{2^s s^{-s}}{1-s} \label{eq:thm1} \\[1ex]
 &\leq \frac{8\cdot 4^{-s}}{[W(x)W(y)]^{s/2}} \left( \frac{\lVert \rho \rVert_{L^1}}{\kappa^s} + \lVert \rho \rVert_\infty \frac{2 \kappa^{1-s}}{1-s} \right) \label{eq:thm2} .
\end{align}
\end{theorem}
In the case where $\supp \rho$ is bounded, the result of Theorem~\ref{theorem:monotone_average} follows directly from a weak $L^1$-bound which may be found in \cite{AizenmanENSS-06}. Indeed, \cite[Proposition 3.1]{AizenmanENSS-06} gives that for a maximally dissipative operator $A$ on some Hilbert space $\mathcal H$ and Hilbert-Schmidt operators $M_1,M_2 : \mathcal H \to \mathcal{H}_1$, where $\mathcal{H}_1$ is some Hilbert space, there is a constant $C_{\rm W}$ (independent of $A$, $M_1$ and $M_2$) such that
\[
\mathcal{L} \left(\{v \in [-1,1] : \bigl \lVert M_1 (A - v + {\rm i}0)^{-1} M_2 \bigr \rVert_{\rm HS} > t \}\right) \leq C_{\rm W} \lVert M_1 \rVert_{\rm HS} \lVert M_2 \rVert_{\rm HS} \frac{1}{t} .
\]
Here $\mathcal{L}$ denotes the Lebesgue measure. By the layer cake representation this implies the statement of Theorem~\ref{theorem:monotone_average} (with different constants) in the case where $\supp \rho$ is bounded. However, for two reasons we give a proof of Theorem~\ref{theorem:monotone_average}. The first one is that the weak $L^1$ bound is only valid for compactly supported measures, while Theorem~\ref{theorem:monotone_average} holds true also for $\rho$ not having bounded support. The second reason is that since our setting is not that general, we can give a short and direct proof, and the details of the proof are not that much involved as in the above mentioned weak $L^1$-bounds.
\par
Recall, a densely defined operator $T$ on some Hilbert space $\mathcal{H}$ with inner product $\langle \cdot , \cdot \rangle_{\mathcal{H}}$ is called \emph{dissipative} if $\im \langle x,Tx \rangle_{\mathcal{H}} \geq 0$ for all $x \in D(T)$. $T$ is called \emph{maximally dissipative} if it has no proper dissipative extension. We will use that any maximally dissipative operator $T$ on $\mathcal{H}$ has a self-adjoint dilation, i.\,e. there is a Hilbert space $\mathcal{H}_+$ with inner product $\langle \cdot , \cdot \rangle_{\mathcal{H}_+}$, a self-adjoint operator $L$ on $\mathcal{H}_+$ and an operator $P : \mathcal{H}_+ \to \mathcal{H}$, such that
\[
 (T - z)^{-1} = P(L - z)^{-1} P^*
\]
for all $z \in \CC$ with $\im z < 0$. Moreover, the adjoint $P^* : \mathcal{H} \to \mathcal{H}_+$ is an isometry, i.\,e. $\langle P^* x , P^*y \rangle_{\mathcal{H}_+} = \langle x , y \rangle_{\mathcal{H}}$. For a proof of this result see \cite{Kuzhel-96}, or \cite{NagyF-70} where the equivalent theory of unitary dilations of contractions is presented.
\par
Let us also cite a special case of \cite[Lemma B.1]{AizenmanENSS-06} which we will use in the proof of Theorem \ref{theorem:monotone_average}.
\begin{lemma}[\cite{AizenmanENSS-06}] \label{lemma:aizenman}
Let $X$ be a countable set and $A_0 = B + \i C$ be an operator on $\ell^2 (X)$, where $B$ is self-adjoint and $C$ is bounded with $C \geq \delta$ for some $\delta > 0$. Let also $W$ be a bounded non-negative operator in $\ell^2 (X)$. Then the Birman-Schwinger operator
\[
 A_{\rm BS} = (W^{1/2} A_0^{-1} W^{1/2})^{-1}
\]
is maximally dissipative in $(\ker W)^\perp$, with $D(A_{\rm BS}) = \ran (W^{1/2} A_0^{-1} W^{1/2})$. Moreover, its resolvent set $\rho (A_{\rm BS})$ includes $\overline{\CC^-}$ and the Birman-Schwinger relation
\[
 (A_{\rm BS} - \zeta)^{-1} = W^{1/2} (A_0 - \zeta W)^{-1} W^{1/2} 
\]
holds in $(\ker W)^\perp$ for all $\zeta \in \overline{\CC^-}$.
\end{lemma}
\begin{proof}[Proof of Theorem \ref{theorem:monotone_average}]
Let $\epsilon > 0$. For all $x,y \in X$ we have
\begin{align*}
I_{x,y} &= \int_\RR \bigl\lvert \bigl\langle \delta_x , (A + (t+\i \epsilon) W - z)^{-1} \delta_y \bigr\rangle \bigr\rvert^s \rho (t) \drm t \\[1ex]
&= [W(x)W(y)]^{-s/2} \int_\RR \bigl\lvert \bigl\langle \delta_x , W^{1/2} (\tilde A + (t + \i \epsilon)W)^{-1} W^{1/2} \delta_y \bigr\rangle \bigr\rvert^s \rho (t) \drm t  ,
\end{align*}
where $\tilde A = A - z$. We infer from Lemma \ref{lemma:aizenman} that the operator $W^{1/2} \tilde A^{-1} W^{1/2}$ is invertible and that the inverse $A_{\rm BS} = (W^{1/2} \tilde A^{-1} W^{1/2})^{-1}$ is maximally dissipative with domain $\ran (W^{1/2} \tilde A^{-1} W^{1/2})$. Moreover, the resolvent set of $A_{\rm BS}$ includes $\overline{\CC^-}$ and we have for $\zeta \in \overline{\CC^-}$ the relation
\[
 (A_{\rm BS} - \zeta)^{-1} = W^{1/2}(\tilde A - \zeta W)^{-1}W^{1/2} .
\]
Since $A_{\rm BS}$ is maximally dissipative, there exists a Hilbert space $\mathcal{H}$ with inner product $\langle \cdot , \cdot \rangle_{\mathcal{H}}$, containing $\ell^2 (X)$ as a subspace, and a self-adjoint operator $L$ on $\mathcal{H}$, such that 
\[
 (A_{\rm BS} - \zeta)^{-1} = P(L - \zeta)^{-1} P^*
\]
holds for all $\zeta \in \CC^-$. Here, $P$ is an operator from $\mathcal{H}$ to $\ell^2 (X)$ and the adjoint $P^*$ is an isometry. Hence,
\[
I_{x,y} =  [W(x)W(y)]^{-s/2} \int_\RR \bigl\lvert \bigl\langle \delta_x , P(L + (t+\i \epsilon))^{-1} P^* \delta_y \bigr\rangle \bigr\rvert^s \rho (t) \drm t . 
\]
Let $\psi = P^* \delta_x \in \mathcal{H}$ and $\drm \mu_\psi$ be the non-negative spectral measure of $L$ with respect to $\psi$. By the spectral theorem, triangle inequality and Fubini's theorem we have for $x = y$
\begin{equation*}
 I_{x,x} \leq [W(x)W(x)]^{-s/2} \int_\RR \int_\RR \frac{1}{\lvert \kappa + (t + \i \epsilon) \rvert^s} \rho (t) \drm t \, \drm \mu_\psi (\kappa) .
\end{equation*}
Since $\int_\RR (1/\lvert x-\alpha \rvert^s)  g(x) \drm x \leq \kappa^{-s} \lVert g \rVert_{L^1} + 2 \kappa^{1-s} \lVert g \rVert_\infty / (1-s)$ for $0 \leq g \in L^\infty (\RR) \cap L^1 (\RR)$, $s \in (0,1)$, $\kappa > 0$ and $\alpha \in \CC$, see e.\,g.\ \cite{Graf-94}, we have for all $\kappa > 0$
\[
I_{x,x} \leq [W(x)W(x)]^{-s/2} \langle \psi , \psi \rangle_\mathcal{H} \left( \kappa^{-s} \lVert \rho \rVert_{L^1} + \frac{2 \kappa^{1-s}}{1-s} \lVert \rho \rVert_\infty  \right) .
\]
Notice that $\langle \psi , \psi \rangle_\mathcal{H} = 1$ since $P^*$ is an isometry. In the case where $x \not = y$ we use polarization identity and triangle inequality and obtain analogously for all $x,y \in X$
\begin{equation} \label{eq:Ixy}
I_{x,y} \leq \frac{8 \cdot 4^{-s}}{[W(x)W(x)]^{-s/2}} \left( \kappa^{-s} \lVert \rho \rVert_{L^1} + \frac{2 \kappa^{1-s}}{1-s} \lVert \rho \rVert_\infty  \right) .
\end{equation}
To conclude Ineq.~\eqref{eq:thm2}, it remains to show that $\lim_{\epsilon \searrow 0} I_{x,y}$ equals to the left hand side of Ineq.~\eqref{eq:thm1}. Notice that for every $t \in \RR$ and $\epsilon > 0$
\[
\bigl\lvert \bigl\langle \delta_x , (A + (t+\i \epsilon) W - z)^{-1} \delta_y \bigr\rangle \bigr\rvert^s \leq \frac{1}{\lvert \im z \rvert^s},
\]
which is an $\rho \drm t$-integrable majorant. Hence the dominated con\-vergence theorem applies and gives
\begin{align*}
\lim_{\epsilon \searrow 0} I_{x,y} &= \int_\RR \lim_{\epsilon \searrow 0} \bigl\lvert \bigl\langle \delta_x , (A + (t+{\rm i} \epsilon) W - z)^{-1} \delta_y \bigr\rangle \bigr\rvert^s \rho (t) \drm t  \\[1ex]
& = \int_\RR \bigl\lvert \bigl\langle \delta_x , (A + t W - z)^{-1} \delta_y \bigr\rangle \bigr\rvert^s \rho (t) \drm t  .
\end{align*}
This and Ineq.~\eqref{eq:Ixy} leads to Ineq.~\eqref{eq:thm2}. To show Ineq.~\eqref{eq:thm1} we minimize the right hand side of Ineq.~\eqref{eq:thm2} by choosing $\kappa = s \lVert \rho \rVert_{L^1} / (2 \lVert \rho \rVert_\infty)$. 
\end{proof}
%
%
%
%
%
%
\chapter{Some results for the alloy-type model} \label{sec:cont_results}
Some results and methods established for the discrete alloy-type model in Chapter~\ref{chap:fmm} and \ref{chap:wegner} can also be applied for the continuous analogue, the (continuous) \textit{alloy-type model} with sign-changing single-site potential. We will introduce the alloy-type model in Section \ref{sec:cont_model}.
\par
In Chapter \ref{chap:wegner} we proved a Wegner estimate for the discrete alloy-type model with exponentially decaying singe-site potentials. The methods used in the proof also apply to proof a Wegner estimate for the alloy-type model with a single-site potential of a so-called generalized step function type. This is worked out in detail in Section \ref{sec:wegner_cont}. The result is based on a joint work with Norbert Peyerimhoff and Ivan Veseli\'c, and has already been publishen in the preprint \cite{PeyerimhoffTV-11}.
\par
In Section \ref{sec:loc} we established a new method to conclude exponential localization from fractional moment bounds for the discrete alloy-type model with sign-changing single-site potentials. Let us emphasize that earlier techniques, see e.g. \cite{AizenmanM-93,Graf-94,BoutetNSS-06,AizenmanENSS-06}, are only applicable in the case of non-negative single-site potentials. It turns out that the method from Section \ref{sec:loc} also applies for the alloy-type model. In Section \ref{sec:loc_cont} we prove analogues of Proposition \ref{prop:replace-msa}, Proposition \ref{prop:replace-msa-Wegner} and Theorem \ref{theorem:exp_decay_loc} for the alloy-type model with sign-changing single-site potential.
\section{The alloy-type model} \label{sec:cont_model}
The \emph{alloy-type model} is given by the family of Schr\"o{}dinger operators
\begin{equation*}
H_\omega := H_0 + V_\omega, \quad H_0:= -\Delta + V_0, \quad \omega \in \Omega ,
\end{equation*}
on $L^2 (\RR^d)$, where $-\Delta$ is the negative Laplacian, $V_0$ a $\ZZ^d$-periodic potential, and $V_\omega$ denotes the multiplication by the $\ZZ^d$-metrically transitive random field
\begin{equation*}
V_\omega (x) := \sum_{k \in \ZZ^d} \omega_k U(x-k) .
\end{equation*}
Recall that the space $\Omega = \times_{k \in \ZZ^d} \RR$ is equipped with the probability measure $\PP (\drm \omega) = \prod_{k \in \ZZ^d} \nu (\drm \omega_k)$, hence the sequence $\omega = (\omega_k)_{k \in \ZZ^d}$ may be interpreted as a collection of i.i.d.\ random variables. We assume that $V_0$ and $V_\omega$ are infinitesimally bounded with respect to $\Delta$ and that the corresponding constants can be chosen uniform in $\omega \in \Omega$. Therefore, $H_\omega$ is self-adjoint (on the domain of $\Delta$) and bounded from below (uniform in $\omega \in \Omega$).
\par
Recall that any real-valued function on $\RR^d$ that is uniformly locally $L^p$, with $p = 2$ for $d \leq 3$ and $p>d/2$ for $d \geq 4$, is infinitesimally bounded with respect to the self-adjoint Laplacian $\Delta$ on $W^{2,2} (\RR^d)$, see e.g.\ \cite[Theorem XIII.96]{ReedS-78d}. This is satisfied for $V_\omega$ if $U$ is a so-called \emph{generalized step-function}.
\begin{definition}[Generalized step-function]
Let $L_{\rm c}^p (\RR^d) \ni w \geq \kappa \chi_{(-1/2,1/2)^d}$ with $\kappa > 0$ and $p = 2$ for $d \leq 3$ and $p>d/2$ for $d \geq 4$, where $L_{\rm c}^p (\RR)$ denotes the vector space of $L^p (\RR)$ functions with compact support. Let $u \in \ell^1 (\ZZ^d ; \RR)$. A function of the form $U : \RR^d \to \RR$,
\[
U (x) = \sum_{k \in \ZZ^d} u (k) w(x-k) ,
\]
is called \emph{generalized step-function} and the function $u:\ZZ^d \to \RR$ a \emph{convolution vector}. 
\end{definition}
Indeed, if $U$ is a generalized step-function, then we have for any unit cube $C \subset \RR^d$
\begin{align*}
 \int_{C} \lvert V_\omega (x) \rvert^p \drm x &= \int_C \Bigl\lvert \sum_{k \in \ZZ^d} \omega_k \sum_{l \in \ZZ^d} u(l-k) w(x-l) \Bigr\rvert^p \drm x \\
& \leq \omega_+ c_p \lVert u \rVert_{\ell^1} \sum_{l \in \ZZ^d} \int_C \lvert w(x-l) \rvert^p \drm x = \omega_+ c_p \lVert u \rVert_{\ell^1 } \lVert w \rVert_{L^p}^p ,
\end{align*}
where $\omega_+ = \sup \{\lvert t \rvert : t \in \supp \rho\}$ and $c_p$ is some constant depending on $p$ and the support of $w$. Note that the upper bound is uniform in $\omega \in \Omega$.
Hence, if $U$ is a generalized step-function, $V_\omega$ is infinitesimally bounded with respect to $\Delta$ and the corresponding constants can be chosen uniform in $\omega \in \Omega$.
\par
The estimates we want to prove concern finite box restrictions of the operator $H_\omega$. For $l > 0$ and $x \in \RR^d$ we denote by 
\[
 \cc_{l,x} := \bigr\{y \in \RR^d : \max_{i=1,\ldots,d} \lvert y_i - x_i \rvert < l\bigl\}
\]
the open cube of side length $2l$ centered at $x$. Recall that $\Lambda_{l,x} = \{y \in \ZZ^d : \lvert x-y \rvert_\infty \leq l\}$ and $\Lambda_l = \Lambda_{l,0}$. We will use the notation $\cc_l = \cc_{l,0}$ for the cube centered at the origin. For an open set $\cc \subset \RR^d$ we denote by $H_\cc$ the restriction of the operator $H_\omega$ to $L^2 (\cc)$ with Dirichlet boundary conditions on $\partial \cc$. For the resolvents we use the notations $G_\omega (z) = (H_\omega - z)^{-1}$ and $G_\cc (z) = (H_\cc - z)^{-1}$ for $z \in \CC \setminus \RR$.
\begin{remark}
Note that we have used the same symbol for the alloy-type model as well as for the discrete alloy-type model. In Appendix~\ref{sec:cont_results}, $H_\omega$ will always stand for the (continuous) alloy-type model. Note also, that the function $u$, which was used for single-site potential in the discrete setting, serves as a convolution vector to generate the (continuous) single-site potential $U$ in form of a generalized step-function.
\end{remark}

\section{Wegner estimate for the alloy-type model} \label{sec:wegner_cont}
Our main result of this section is a Wegner estimate for the alloy-type model in the case where the single-site potential is a generalized step-function with an exponentially decaying convolution vector and where the measure $\nu$ has a probability density of finite total variation. The space $\BV (\RR)$ of functions of finite total variation is defined in Section~\ref{sec:abstract_wegner}.
\begin{assumption} \label{ass:exp_cont}
Assume that the measure $\nu$ has a probability density $\rho \in \BV (\RR)$, $U$ is a generalized step-function and there are constants $C,\alpha > 0$ such that for all $k \in \ZZ^d$ we have
\begin{equation*} 
\lvert u(k) \rvert \leq C \euler^{-\alpha \lvert k \rvert_1} .
\end{equation*}
\end{assumption}
Our main result in this section is the following theorem.
\begin{theorem} \label{theorem:wegner_c} 
Let Assumption \ref{ass:exp_cont} be satisfied.
Then there exists $C(U)>0$ and $I_0 \in \NN_0^d$ both depending only on $U$, such that for any $l\in \NN$ and any bounded interval $I=[E_1,E_2] \subset \RR$ we have the estimate
\[
\EE \bigl (\Tr \bigl(\chi_I (H_{\cc_l}) \bigr)\bigr)
\le
\euler^{E_2} C(U) \lVert \rho \rVert_{\rm Var} \lvert I \rvert (2l+1)^{2d + \lvert I_0 \rvert} .
\]
A precise definition of $I_0 \in \NN_0^d$ is given in Eq.~\eqref{eq:cF}.
\end{theorem}
\begin{remark}
Let us compare the result of Theorem \ref{theorem:wegner_c} to earlier ones of Wegner estimates for alloy-type models with sign-changing single-site potential. 
\par
The paper \cite{Klopp-95a} concerns alloy-type Schr\"o\-ding\-er operators on $L^2(\RR^d)$ with exponentially decaying single-site potential $U$. The main result is a Wegner estimate for low energies which is polynomially in the volume and H\"older continuous in the length of the energy interval.
\par
The paper \cite{HislopK-02} studies a class alloy-type models assuming that the single-site potential $U$ is continuous and compactly supported. The obtained Wegner estimate is valid in a neighborhood of the infimum of the spectrum, linear in the volume and H\"older continuous in the energy variable.
\par
In contrast to the results of \cite{Klopp-95a,HislopK-02}, the Wegner estimate from Theorem~\ref{theorem:wegner_c} is valid on the whole energy axis.
\par
Let us now refer to the papers \cite{Veselic-02a,KostrykinV-06,Veselic-10b} which are closely related to the result presented in Theorem~\ref{theorem:wegner_c}. All three papers study alloy-type models with a single-site potential of a generalized step function form. The results apply also to the discrete alloy-type model, but we discuss here only the results obtained for operators on $L^2 (\RR^d)$. The paper \cite{Veselic-02a} establishes a Wegner estimate which is linear in the volume and linear in the length of the energy interval. It applies to the case where the convolution vector satisfies 
\begin{equation} \label{cond:v02}
\lvert u(0)\rvert > \sum_{k \not = 0} \lvert u(k) \rvert. 
\end{equation}
The papers \cite{KostrykinV-06,Veselic-10b} consider the case where the convolution vector is compactly supported and satisfies
\begin{equation} \label{cond:kv06}
s\colon \theta\mapsto s(\theta) :=\sum_{k \in \ZZ^d} u(k) \mathrm e^{-\mathrm i k \cdot \theta}
\text{ does not vanish on $[0, 2\pi  )^d $. }
\end{equation}
Under this condition they prove a Wegner estimate which is linear in the volume and linear in the length of the energy interval. The result was obtained in \cite{KostrykinV-06} in one and two dimensions and generalized in \cite{Veselic-10b} to arbitrary space dimension. The paper \cite{KostrykinV-06} includes also results in arbitrary dimension which we do not discuss here. Note that condition \eqref{cond:v02} implies condition \eqref{cond:kv06}, and if $\sum_{k \in \ZZ^d} u(k) = 0$ then $s(0)$ vanishes.
\end{remark}
\begin{remark}
The Wegner estimate from Theorem \ref{theorem:wegner_c} is linear in the energy-interval length and polynomial in the volume of the cube. Moreover, the Wegner bound is valid on the whole energy axis. Hence one can prove localization via the multiscale analysis in any energy region where the initial length scale estimate holds. Since the single-site potential does not have compact support, one has to use an enhanced version of the multiscale analysis, see \cite{KirschSS-98b}.
\end{remark}
For the proof of Theorem~\ref{theorem:wegner_c} we use an abstract Wegner estimate established in \cite{KostrykinV-06}, which is formulated in Theorem~\ref{theorem:abstract1}. Theorem \ref{theorem:abstract1} was stated in \cite{KostrykinV-06} for compactly supported $U:\RR^d \to \RR$ only. However, the proof directly applies for non-compactly supported single-site potentials. Before we state the theorem, let us fix some notation. For an open set $\cc \subset \RR^d$, $\tilde \cc$ is the set of lattice sites $j \in \ZZ^d$ such that the characteristic function of the cube $\cc_{1/2} (j)$ does not vanish identically on $\cc$. For $j \in \ZZ^d$ we denote by $\chi_j$ the characteristic function of the cube $\cc_{1/2} (j)$.
\begin{theorem}[\cite{KostrykinV-06}] \label{theorem:abstract1}
Assume there is a number $l_0 \in \NN$ such that for arbitrary $l \geq l_0$ and every $j \in \tilde{\cc}_l$ there is a compactly supported $t_{j,l} \in \ell^1 (\ZZ^d ; \RR)$ such that
\begin{equation*}
\sum_{k \in \ZZ^d} t_{j,n} (k) U(x-k) \geq \chi_j (x) \quad \text{for all} \quad x \in \cc_l .
\end{equation*}
Let further $I = [E_1,E_2]$ be an arbitrary interval. Then for any $l \geq l_0$
\begin{equation*}
\EE \bigl( \Tr \bigl( \chi_I (H_{\cc_l}) \bigr)\bigr) \leq C \euler^{E_2} \lVert \rho \rVert_{\rm Var} \lvert I \rvert \sum_{j \in \tilde{\cc}_l} \lVert t_{j,l} \rVert_{\ell^1} ,
\end{equation*}
where $C$ is a constant independent of $l$ and $I$.
\end{theorem}
Let Assumption \ref{ass:exp_cont} be satisfied. In Section~\ref{sec:transformation} we showed that there are constants $c_u \not = 0$ and $I_0 \in \NN_0^d$ (given in Eq.~\eqref{eq:cF}), such that for each $l \in \NN$ there is a number $R_l >0$ (given in Eq.~\eqref{eq:RLrel}) such that
\begin{equation} \label{eq:pos_cont}
 \frac{2}{c_u} \sum_{k \in \Lambda_{R_l}} k^{I_0} u (x-k) \geq 1 \quad \text{for all $x \in \Lambda_l$}.
\end{equation}
This fact is proven in Proposition \ref{prop2} and we will apply it for the (continuous) alloy-type model if $U$ is a generalized step-function with an exponential decaying convolution vector $u$ to verify the hypothesis of Theorem \ref{theorem:abstract1}. 
\begin{proof}[Proof of Theorem \ref{theorem:wegner_c}]
Recall that $U$ is a generalized step function and that $w$ has compact support. Also recall, that for an open set $\cc \subset \RR^d$, $\tilde \cc$ is the set of lattice sites $j \in \ZZ^d$ such that the characteristic function of the cube $\cc_{1/2} (j)$ does not vanish identically on $\cc$. 
We set $r = \sup \{ \lvert r \rvert + 1 : w(r) \not = 0\}$. Let $l_0 = 1$ and $t_{j,l} \in \ell^1 (\ZZ^d)$ given by 
\[
t_{j , l} (k) = 
\begin{cases}
2 k^{I_0} / (c_{u}\kappa) & \text{if $k \in \Lambda_{R_{l+r}}$}, \\
0                         & \text{else} ,
\end{cases}
\]
for $l \in \NN$ and $j \in \tilde \cc_l$. By Ineq.~\eqref{eq:pos_cont} (see also Proposition \ref{prop2}) we have for all $l \in \NN$, $j \in \tilde \cc_l$ and $x \in \cc_l$
\begin{align*}
 \sum_{k \in \ZZ^d} t_{j,n} (k) U(x-k) 
&= \sum_{i \in \ZZ^d} w (x-i) \sum_{k \in \ZZ^d} t_{j,l} (k)  u (i-k) \\
&\geq \frac{1}{\kappa}\sum_{i \in \Lambda_{l+r}} w(x-i) 
+ \sum_{i \in \ZZ^d \setminus \Lambda_{l+r}} w (x-i) \sum_{k \in \ZZ^d} t_{j,l} (k)  u (i-k) \\
&= \frac{1}{\kappa} \sum_{i \in \Lambda_{l+r}} w (x-i) \geq \chi_j (x) .
\end{align*}
Here we have used that $w (x-i) = 0$ for $x \in \cc_l$ and $i \not \in \Lambda_{l+r}$. Hence the assumption of Theorem \ref{theorem:abstract1} is satisfied. Analogous to the proof of Theorem~\ref{theorem:wegner} on page \pageref{proof:wegner_d} there is a constant $C(U)$ depending only on the single-site potential $U$ such that
\begin{equation*}
 \sum_{j \in \tilde \cc_l} \lVert t_{j,l} \rVert_{\ell^1} \leq
C(U) (2l+1)^{2d + \lvert I_0 \rvert}.
\end{equation*}
This completes the proof by using Theorem \ref{theorem:abstract1}.
\end{proof}
\section{Localization via fractional moments for the alloy-type model} \label{sec:loc_cont}
The fractional moment method, introduced for the discrete Anderson model in \cite{AizenmanM-93} was adopted to the (continuous) alloy-type model in \cite{AizenmanENSS-06,BoutetNSS-06}. The typical output of the fractional moment method for the alloy-type model on $L^2 (\RR^d)$ is the following: There exists $s \in (0,1)$, $\mu > 0$, $C<\infty$ and $I \subset \RR$, such that for all open sets $\cc \subset \RR^d$, all $x,y\in \RR^d$, all $\epsilon > 0$ and all $E \in I$ we have
\begin{equation} \label{eq:fmb_cont}
 \EE \Bigl( \bigl\lVert \chi_x (H_\Lambda - E - \i \epsilon)^{-1} \chi_y \bigr\rVert^s \Bigr) \leq C \euler^{-\mu \lVert x-y \rVert_\infty} .
\end{equation}
Here $\chi_x$ denotes the multiplication operator in $L^2 (\cc)$ by the characteristic function of the unit cube $\cc_{1/2, x}$. Note that $\chi_x (H_\cc - E - \i \epsilon)^{-1} \chi_y = 0$ if $\cc_{1/2 , x} \cap \cc$ or $\cc_{1/2 , y} \cap \cc$ has measure zero. For $x \in \RR^d$ we denote by $$\lVert x \rVert_\infty = \max_{i=1,\ldots , d} \lvert x_i \rvert$$ the supremum norm. The norm $\lVert \cdot \rVert$ on the left hand side of Ineq.~\eqref{eq:fmb_cont} denotes operator norm in $L^2 (\cc)$.

We refer to \cite{AizenmanENSS-06} and \cite{BoutetNSS-06} where such a fractional moment bound was shown under suitable conditions on the measure $\nu$ and the single-site potential $U$. In particular, it is assumed that the single-site potential is non-negative. It is also well known that the bound \eqref{eq:fmb_cont} implies spectral and dynamical localization in $I$ under appropriate assumptions on $\nu$ and $U$, see \cite{AizenmanENSS-06} and \cite{BoutetNSS-06} for details. In particular, the non-negativity of the single-site potential $U$ again plays a crucial role for the proof of this implication.
\par
In this section we show that the fractional moment bound as described in Ineq.~\eqref{eq:fmb_cont} implies spectral localization for alloy-type models with sign-changing single-site potentials. Our main result of this section is the following theorem which is proven at the end of this section.
\begin{theorem}\label{thm:fmb_loc_cont}
Let $\Theta := \supp U$ be a bounded set, $I \subset \RR$ be an interval and $C,\mu \in (0,\infty)$. Assume that for all $l \in 3\NN + 3/2$, $k \in \ZZ^d$, all $x,y\in\Lambda_{l,x}$, all $\epsilon \in (0,1]$ and all $E \in I$ we have
\[
 \EE \Bigl( \bigl\lVert \chi_x G_{\cc_{l,k}} (E+\i \epsilon) \chi_y \bigr \rVert^s \Bigr) \leq C \euler^{-\mu \lVert x-y \rVert_\infty} .
\]
Then, for almost all $\omega \in \Omega$, $H_\omega$ exhibits exponential localization in $I$.
\end{theorem}
For $l \geq 3$ and $x \in \ZZ^d$ we introduce the notation $\cc^{\rm out}_{l,x} = \cc_{l,x} \setminus \cc_{l-1,x}$ and $\cc^{\rm int}_{l,x} = \cc_{l/3,x}$, and denote by
\[
 \chi_{l,x}^{\rm out} = \chi_{\cc^{\rm out}_{l,x}}\quad \text{and} \quad \chi_{l,x}^{\rm int} = \chi_{\cc^{\rm int}_{l,x}}
\]
the characteristic functions as well as the corresponding multiplication operators. For $\Lambda \subset \RR^d$ finite we denote by $\diam \Lambda = \sup_{x,y \in \Lambda} \lVert x-y \rVert_\infty$ the diameter of $\Lambda$ with respect to the supremum norm.
\begin{definition}
Let $m > 0$, $l \geq 3$ and $E \in \RR$. A cube $\cc_{l,x}$ is called $(m,E)$-regular for $\omega \in \Omega$, if $E \not \in \sigma (H_{\cc_{l,x}})$ and
\[
 \bigl \lVert \chi_{l,x}^{\rm out} G_{\cc_{l , x}} (E) \chi_{l,x}^{\rm int} \bigr \rVert \leq \euler^{-m l} .
\]
\end{definition}
\begin{proposition}  \label{prop3}
Let $\Theta = \supp U$ be a bounded set, $I \subset \RR$ be a bounded interval, $s \in (0,1)$ and $l \geq \max \{3, 8\ln (8) / \mu ,  -(8/5\mu)\ln (\lvert I \rvert / 2) \}$ where $\mu$ is the constant from assumption (i). Assume
\begin{enumerate}[(i)]
 \item There are constants $C,\mu \in (0,\infty)$ such that for all $k \in \ZZ^d$, all $E \in I$ and all $\epsilon \in (0,1]$ we have
\[
\EE \Bigl( \bigl\lVert \chi_{l,k}^{\rm out} G_{\cc_{l,k}} (E + \i \epsilon) \chi_{l,k}^{\rm int} \bigr\rVert^s\Bigr) \leq C \euler^{-\mu l} .
\]
 \item There are constants $C_{\rm W}\in (0,\infty)$, $\beta \in (0,1]$ and $D \in \NN$ such that for all $x \in \ZZ^d$ we have
\[
 \EE \bigl(\Tr \chi_{I} (H_{\cc_{l,x}})\bigr) \leq C_{\rm W} \lvert I \rvert^\beta l^D .
\]
\end{enumerate}
Then there is a constant $F$ depending only on $I$, $U$, $\mu$, $C$ and $C_{\rm W}$ such that for all $x,y \in \ZZ^d$ with $\lvert x-y \rvert_\infty \geq 2l + \diam \Theta + 1$ we have
\[
 \PP \bigl(\{\omega \in \Omega \colon \forall E \in I : \text{$\cc_{l,x}$ or $\cc_{l,y}$ is $(\mu / 8,E)$-regular}\}\bigr) \geq 1-F l^{d+D} \euler^{- \mu \beta l / 8} .
\]
\end{proposition}
\begin{proof}
By the same arguments as in the proof of Proposition \ref{prop:replace-msa} we infer from hypothesis (ii) that for each $E \in I$ and $x \in \ZZ^d$ the resolvent of $H_{\cc_{l,x}}$ at $E$ is well defined for almost all $\omega \in \Omega$. Hypothesis (i) and Lebesgue’s Theorem gives us for all $E \in I$
\[
 \EE \Bigl( \bigl\lVert \chi_{l,k}^{\rm out} G_{\cc_{l,k}} (E) \chi_{l,k}^{\rm int} \bigr\rVert^s \Bigr) \leq C \euler^{-\mu l} .
\]
Fix $x,y \in \ZZ^d$ with $|x-y|_\infty \geq 2l + \diam \Theta + 1$. For $\omega \in \Omega$ and $k \in \{x,y\}$ we define
\begin{align}
 \Delta_\omega^k &:= \bigl\{E \in I : \lVert \chi_{l,k}^{\rm out} G_{\cc_{l,k}} (E) \chi_{l,k}^{\rm int} \rVert > \euler^{-\mu l /8}\bigr\}, \nonumber \\[1ex]
\tilde \Delta_\omega^k &:= \bigl\{E \in I : \lVert \chi_{l,k}^{\rm out} G_{\cc_{l,k}} (E) \chi_{l,k}^{\rm int} \rVert > \euler^{-\mu l/4 }\bigr\}, \nonumber \\[1ex]
\text{and} \quad \tilde B_k &:= \bigl\{\omega \in \Omega : \mathcal{L} (\{\tilde \Delta_\omega^k\}) >
 \euler^{-5\mu L /8} \bigr\} . \label{eq:deltatilde2}
\end{align}
Since the resolvent of $H_{\cc_{l,k}}$ at $E$ is not defined if $E$ is an eigenvalue of $H_{\cc_{l,k}}$, let us emphasize that we want the eigenvalues of $H_{\cc_{l,k}}$ to be included in the sets $\Delta_\omega^k$ and $\tilde\Delta_\omega^k$, $k \in \{x,y\}$.
For $\omega \in \tilde B_k$ we have
\begin{align*}
\int_I \bigl\lVert \chi_{l,k}^{\rm out} G_{\cc_{l,k}} (E) \chi_{l,k}^{\rm int} \bigr\rVert^s \drm E
&\geq  \int_{\tilde \Delta_\omega^k} \bigl\lVert \chi_{l,k}^{\rm out} G_{\cc_{l,k}} (E) \chi_{l,k}^{\rm int} \bigr\rVert^s \drm E  \\[1ex]
&\geq \euler^{-\mu l s /4} \euler^{-5\mu l / 8} .
\end{align*}
Using Hypothesis (i) and Fubini's theorem we obtain
\begin{align*}
 \PP (\tilde B_k) &\leq  \euler^{7 \mu l / 8}
\int_I \int_\Omega \bigl\lVert \chi_{k} G_{\cc_{l,k}} (E) \chi_{w} \rVert^s  \PP (\drm \omega) \drm E 
\leq C \lvert I \bigr\rvert \euler^{-\mu l / 8} .
\end{align*}
Set $m = \max\{\lvert \inf I \rvert , \lvert \sup I \rvert\}$ and denote for $k \in \{x,y\}$ by $\{E_{\omega , k}^i\}_{i=1}^{N_k}$ the eigenvalues of $H_{\omega}^{\cc_{l,k}}$ within the interval $B_2 (I) := \{E \in \RR \colon \exists x \in I : \lvert x-E \rvert \leq 2\}$. Note that the number of eigenvalues depends on $\omega$ but satisfies the Weyl-bound $N_k \leq C_{\rm Weyl} \lvert \cc_{l,k} \rvert m^{d/2}$ for all $\omega \in \Omega$ with a constant $C_{\rm Weyl}$ depending only on the single-site potential $U$ and the distribution $\nu$, see e.g. \cite{PasturF-92}. We claim that for $k \in \{x,y\}$,
\begin{equation} \label{eq:claim2}
 \omega \in \Omega \setminus \tilde B_k \quad \Rightarrow
 \quad \Delta_\omega^k \subset \bigcup_{i=1}^{N_k}
 \bigl[E_{\omega,k}^i-\delta , E_{\omega,k}^i + \delta \bigr] =:
 I_{\omega,k}(\delta),
\end{equation}
where $\delta = 2\euler^{-\mu l / 8}$. Indeed, fix $\omega \in \Omega \setminus \tilde B_k$ and suppose that $E\in \Delta_\omega^k \setminus\{E_{\omega , k}^1 , \ldots , E_{\omega , k}^{N_k}\}$ with $\lvert E - E_{\omega,k}^i \rvert >\delta$ for some $i \in \{1,\ldots , N_k\}$. It follows that $\lVert \chi_{l,k}^{\rm out} G_{\cc_{l,k}} (E) \chi_{l,k}^{\rm int} \rVert > \euler^{-\mu l / 8}$ and for any $E'\in \RR$ with $\lvert E-E'\rvert \le 2\euler^{-5\mu l / 8}$ we have
$\delta -\lvert E-E'\rvert\ge \euler^{-\mu l /8} \geq 2 \euler^{-3\mu l / 8}$ since $l \geq 8 \ln 8 / \mu$.
Moreover, the first resolvent identity and the estimate $\lVert (H-E)^{-1} \rVert \leq \dist
(E,\sigma (H))^{-1}$ for self-adjoint $H$ and $E \in \CC\setminus \sigma (H)$ implies
\begin{multline*}
 \lVert \chi_{l,k}^{\rm out} G_{\cc_{l,k}} (E) \chi_{l,k}^{\rm int} \rVert - \lVert \chi_{l,k}^{\rm out} G_{\cc_{l,k}} (E') \chi_{l,k}^{\rm int} \rVert  \\ \leq
\lvert E - E' \rvert \lVert \chi_{l,k}^{\rm out} G_{\cc_{l,k}} (E) \chi_{l,k}^{\rm int} \rVert 
\lVert \chi_{l,k}^{\rm out} G_{\cc_{l,k}} (E') \chi_{l,k}^{\rm int} \rVert \leq \frac{1}{2} \euler^{-\mu l / 8}  ,
\end{multline*}
 and hence, since $E\in \Delta_\omega^k$ and $l \geq 8 \ln (8) / \mu$,
\[
 \lVert \chi_k G_{\cc_{l,k}} (E') \chi_w \rVert \geq \euler^{-\mu l / 8} - \frac{1}{2} \euler^{-\mu l 8} > \euler^{-\mu l / 4} .
\]
We infer that $[E-2\euler^{-5\mu l / 8},E+2\euler^{-5\mu l / 8}]\cap I \subset  \tilde \Delta_\omega^k$ and conclude
$\mathcal{L} (\{\tilde \Delta_\omega^k\}) \ge 2\euler^{-5\mu l / 8}$ since $\lvert I \rvert \geq 2 \euler^{-5\mu l / 8}$ by assumption.
This is however impossible if
$\omega\in \Omega \setminus \tilde B_k$ by \eqref{eq:deltatilde2},
hence the claim \eqref{eq:claim2} follows.
\par
Now we want to estimate the probability of the event $B_{\rm res} := \{\omega \in \Omega :  I \cap I_{\omega,x}(\delta) \cap I_{\omega,y}(\delta) \not = \emptyset \}$
that there are ``resonant'' energies for the two box Hamiltonians $H_{\cc_{l,x}}$ and $H_{\cc_{l,y}}$. For this purpose, we denote by $\Lambda_{l,x}'$ the set of lattice sites $k \in \ZZ^d$ whose coupling constant $\omega_k$ influences the potential values in $\cc_{l,x}$, i.e. $\Lambda_{l,x} = \cup_{j \in \cc_{l,x}} \{k \in \ZZ^d \colon U(j-k) \not = 0\}$. Notice that the expectation in hypothesis (ii) may therefore be replaced by $\EE_{\Lambda_{l,x}'}$. Moreover, since $\Theta$ is a bounded set and $\lvert x-y \rvert_\infty \geq 2l + \diam \Theta + 1$, the operator $H_{\cc_{l,y}}$ and therefore the interval $I_{\omega , y} (\delta)$ is independent of $\omega_k$, $k \in \Lambda_{l,x}'$.
Analogously to the proof of Proposition \ref{prop:replace-msa} we use the product structure of the measure $\PP$. We denote $\Omega \ni \omega = (\omega_1, \omega_2) \in \Omega_{\Lambda_{l,x}'} \times \Omega_{\ZZ^d \setminus\Lambda_{l,x}'}$ and for each $\omega_2 \in \Omega_{\ZZ^d \setminus\Lambda_{l,x}'}$ we set  $\tilde B_{\rm res} (\omega_2) = \{\omega_1 \in \Omega_{\Lambda_{l,x}'} \colon (\omega_1 , \omega_2) \in B_{\rm res}\}$. Since $H_{\cc_{l,y}}$ is independent of $\omega_k$, $k \in \Lambda_{l,x}'$, we obtain for any $\omega_2 \in \Omega_{\ZZ^d \setminus \Lambda_{l,x}'}$ using Chebyshev's inequality and Hypothesis (ii) that
\begin{align*}
\nonumber
\PP_{\Lambda_{l,x}'} (\tilde B_{\rm res} (\omega_2))&\leq \sum_{i=1}^{N_y}
\PP_{\Lambda_{l,x}'} \bigl( \{ \omega_1 \in \Omega_{\Lambda_{l,x}'} : \Tr \bigl( \chi_{I \cap [E_{\omega,y}^i-2\delta , E_{\omega,y}^i + 2\delta ]} (H_{\cc_{l,x}}) \bigr) \geq 1 \}\bigr) \\
& \leq \sum_{i=1}^{N_y}
\EE_{\Lambda_{l,x}'} \bigl( \Tr \bigl( \chi_{I \cap [E_{\omega,y}^i-2\delta , E_{\omega,y}^i + 2\delta ]}  (H_{\cc_{l,x}}) \bigr) \bigr) \nonumber \\
&\leq N_y  C_{\rm W} (4\delta)^{\beta} l^D \leq C_{\rm Weyl} \lvert \cc_{l,y} \rvert m^{d/2} C_{\rm W} (4\delta)^{\beta} l^D. 
\end{align*}
Fubini's theorem now gives
\[
 \PP (B_{\rm res}) \leq C_{\rm Weyl} \lvert \cc_{l,y} \rvert m^{d/2} C_{\rm W} (4\delta)^{\beta} l^D .
\]
Consider now an $\omega \not \in \tilde B_x \cup \tilde B_y$. Recall that \eqref{eq:claim2} tells us that $\Delta_\omega^x \subset  I_{\omega,x}(\delta)$
and $\Delta_\omega^y \subset  I_{\omega,y}(\delta)$. If additionally $\omega \not \in B_{\rm  res}$ then no $E \in I$ can be in
$\Delta_\omega^x $ and $\Delta_\omega^y$ simultaneously. Hence for each $E \in I$ either  $\cc_{l,x}$ or $\cc_{l,y}$
is $(\mu/8,E)$-regular. A contraposition gives us
\begin{align}
\PP \bigl(\{&\text{$\exists \, E \in I$,
$\cc_{l,x}$ and $\cc_{l,y}$ is $(\mu/8,E)$-singular} \}\bigr) \le \PP (\tilde B_x) +\PP (\tilde B_y ) + \PP (B_{\rm res} ) \nonumber
\\ &\leq 2 C \lvert I \rvert \euler^{-\mu l /8} 
+ C_{\rm Weyl}  C_{\rm W} m^{d/2} (4\delta)^{\beta} 2^d l^{D+d}, \label{eq:contraposition}
\end{align}
from which the result follows.
\end{proof}
Assumption (ii) from Proposition \ref{prop3} is a Wegner estimate. The next lemma shows that a certain a-priori estimate on averaged fractional moments of the Green function implies such a Wegner estimate.
%
\begin{lemma} \label{lemma:apriori_wegner_cont}
Let $I \subset \RR$ be a bounded interval, $s \in (0,1)$, $C \in (0,\infty)$, $k \in \ZZ^d$ and $l \in \NN + 1/2$. Assume that for all $\epsilon \in (0,\lvert I \rvert]$, $E \in I$ and $x,y \in \Lambda_{l,k}$
\[
 \EE \bigl( \lVert \chi_{x} G_{\cc_{l,k}}(E+ \i \epsilon) \chi_y \rVert^s \bigr) \leq C  .
\]
Then there is a constant $C_{\rm W}$ depending only on $U$, $\nu$, $I$ and $C$ such that for all $[a,b] \subset I$
\[
 \EE \bigl( \Tr \chi_{[a,b]}(H_{\cc_{l,k}}) \bigr)
 \leq C_{\rm W} (2l+1)^{3d} \lvert b-a \rvert^s  .
\]

\end{lemma}
\begin{proof}
Set $m = \max\{\lvert \inf I \rvert , \lvert \sup I \rvert \}$ and let $[a,b] \subset I$. Since we have for any $\lambda \in \RR$ and $0<\epsilon \leq b-a$
\[
 \arctan \left( \frac{\lambda - a}{\epsilon} \right) - \arctan \left( \frac{\lambda - b}{\epsilon} \right) \geq \frac{\pi}{4} \ \chi_{[a,b]}(\lambda) ,
\]
one obtains an inequality version of Stones formula, namely for all $\epsilon \in (0, b-a]$ and all $\psi \in L^2 (\Lambda)$ we have
\[
 \langle \psi , \chi_{[a,b]} (H_{\cc_{l,k}}) \psi \rangle
\leq \frac{4}{\pi} \int_{[a,b]} \im \langle \psi , G_{\cc_{l,k}} (E+ \i \epsilon) \psi \rangle \drm E .
\]
Let $B_\omega$ be the set of normalized eigenfunctions corresponding to the eigenvalues of $H_{\omega}^{\cc_{l,k}}$ in $I$. Note that for all $\omega \in \Omega$, the number of elements in $B_\omega$ is bounded from above by $C_{\rm Weyl} \lvert \cc_{l,x}\rvert m^{d/2} =: N$ with $C_{\rm Weyl}$ depending only on $U$ and $\nu$, see e.g. \cite{PasturF-92}.
Using triangle inequality, $\lvert \im z\rvert \leq \lvert z\rvert$
for $z \in \CC$, Fubini's theorem and $\lvert \langle \psi , G_{\cc_{l,k}} (E+\i
\epsilon)\psi \rangle \rvert^{1-s} \leq \dist (\sigma(H_{\cc_{l,k}}) ,
E+i \epsilon)^{s-1} \leq \epsilon^{s-1}$ we
obtain for all $\epsilon \in (0,b-a]$
\begin{align*}
\EE \bigl( \Tr \chi_{[a,b]}(H_{\cc_{l,k}}) \bigr) & \leq \EE \Bigl( \sum_{\psi \in B_\omega} \frac{4}{\pi} \int_{[a,b]} \im \langle \psi , G_{\cc_{l,k}} (E+\i \epsilon) \psi \rangle \drm E  \Bigr) \\
&  \leq  \frac{\epsilon^{s-1}}{\pi / 4} \EE \Bigl( \sum_{\psi \in B_\omega} \int_{[a,b]}  \bigl\| G_{\cc_{l,k}} (E+\i \epsilon) \bigr\|^{s} \drm E  \Bigr)  \\
&  \leq  \frac{N \epsilon^{s-1}}{\pi / 4} \int_{[a,b]} \EE \Bigl(   \bigl\| G_{\cc_{l,k}} (E+\i \epsilon) \bigr\|^{s}   \Bigr) \drm E. 
\end{align*}
We now use $l \in \NN + 1/2$ and obtain by covering $\cc_{l,x}$ with unit cubes the estimate $\| G_{\cc_{l,x}} (E+\i \epsilon) \| \leq \sum_{x,y\in \Lambda_{l,k}} \| \chi_x G_{\cc_{l,k}} (E+\i \epsilon) \chi_y \|$. We use further the estimate $(\sum a_n )^s \leq \sum a_n^s$ for $s \in (0,1)$ and $a_n > 0$ and obtain by using the hypothesis of the Lemma
\begin{align*}
\EE \bigl( \Tr \chi_{[a,b]}(H_{\cc_{l,k}}) \bigr) &  \leq  \frac{N \epsilon^{s-1}}{\pi / 4}   \sum_{j,k\in \Lambda_{l,k}}  \int_{[a,b]}  \EE \Bigl( \bigl\| \chi_j G_{\cc_{l,k}} (E+\i \epsilon) \chi_k \bigr\|^{s} \Bigr) \drm E   \\
& \leq 4\pi^{-1}\epsilon^{s-1}  N  \lvert \Lambda_{l,k} \rvert^2 \, \lvert b-a \rvert C \\[1ex]
&\leq 4\pi^{-1}\epsilon^{s-1}  C_{\rm Weyl} \lvert \cc_{l,k} \rvert m^{d/2} \lvert \Lambda_{l,k} \rvert^2 \lvert b-a \rvert C.
\end{align*}
We minimize the right hand side by choosing $\epsilon = b-a$ and obtain the statement of the lemma.
\end{proof}
\begin{proof}[Proof of Theorem \ref{thm:fmb_loc_cont}]
First consider the case where $\lvert I \rvert \leq 1$.
In order to prove the theorem we first verify the hypothesis of Proposition \ref{prop3}. Let $l \in 3 \NN + 3/2$ and $k \in \ZZ^d$. Note that $\cc_{l,k}^{\rm int}$ and $\cc_{l,k}^{\rm out}$ can be covered exactly by unit cubes. By assumption we have
\begin{align*}
 \EE \Bigl( \bigl\lVert \chi_{l,k}^{\rm out} G_{\cc_{l,k}} (E + \i \epsilon) \chi_{l,k}^{\rm int} \bigr\rVert^s\Bigr) &\leq \sum_{x \in \Lambda_{l,k}\setminus \Lambda_{l-1,k}} \sum_{y\in \Lambda_{l/3 , k}} \EE \Bigl( \bigl\lVert \chi_x  G_{\cc_{l,k}} (E + \i \epsilon) \chi_y \bigr\rVert^s\Bigr) \\[1ex]
& \leq 2d(2l)^{d-1} (2l/3)^d C \euler^{-\mu\left(l - 1 - \frac{l}{3} \right)} \\[1ex] &= 
2d(2l)^{d-1} (2l/3)^d C \euler^\mu \euler^{-\frac{2\mu}{3}l}
\end{align*}
Hence the hypothesis (i) of Proposition \ref{prop3} is satisfied with an $l$-dependent constant $C$. In view of Ineq.~\eqref{eq:contraposition} we see that this volume dependence does not change the statement of Proposition~\ref{prop3} in a quantitative way. Hypothesis (ii) of Proposition \ref{prop3} is satisfied by Lemma \ref{lemma:apriori_wegner_cont} since $\lvert I \rvert \leq 1$. Since $l \in 3\NN + 3/2$ and $k \in \ZZ^d$ were arbitrary, we infer from the conclusion of Proposition \ref{prop3} that for any $p > 0$ there is an $\tilde l$, such that for all $l \in 3 \NN + 3/2$ with $l \geq \tilde l$ and all $x,y \in \ZZ^d$ with $\lvert x-y \rvert_\infty \geq 2l + \diam \Theta + 1$ we have
\begin{equation} \label{eq:output_MSA}
\PP \bigl(\{\omega \in \Omega \colon \forall E \in I : \text{$\cc_{l,x}$ or $\cc_{l,y}$ is $(\mu / 12,E)$-regular}\}\bigr) \geq 1 - l^{-2p} .
\end{equation}
Roughly speaking, Ineq.~\eqref{eq:output_MSA} is the typical output of the multiscale analysis. That Ineq.~\eqref{eq:output_MSA} implies exponential localization is a well known fact. For the classical Anderson model on $\ZZ^d$ this was implemented in \cite{DreifusK-89}. For the case of alloy-type models with a single-site potential supported on the unit cube we refer to \cite{Stollmann-01}. A generalization to models with a bounded support of the single site-potential is straightforward, cf.\ Theorem~\ref{thm:vDK-2.3}. For the non-compactly supported case, including our model as a special case, this implication was shown in \cite{KirschSS-98b}. In particular, all the mentioned papers consider models with a non-negative single-site potential. However, the assumption that the single-site potential is non-negative is not used, and so the existing proofs of this fact apply directly to our setting. Hence, we obtain for almost all $\omega\in\Omega$ that $H_\omega$ exhibits exponential localization in $I$.
\par
In the case where $\lvert I \rvert \geq 1$, we can cover $I$ by countable many intervals with length smaller or equal one. Since a countable intersection of sets with measure one has full measure, we obtain the statement of the theorem.
\end{proof}

\addchap{Theses}
We consider a family self-adjoint operators $H_\omega : \ell^2 (\ZZ^d) \to \ell^2 (\ZZ^d)$ given by
\begin{equation*} 
 H_\omega = -\Delta + \lambda V_\omega, \quad \omega \in \Omega .
\end{equation*}
Here $\Delta$ denotes the discrete Laplacian and $V_\omega$ is a multiplication operator by the function
\[
 V_\omega (x) = \sum_{k \in \ZZ^d} \omega_k u(x-k) ,
\]
and $\lambda>0$ measures the strength of the disorder. We assume that $\omega = (\omega_k)_{k \in \ZZ^d}$ is a sequence of independent identically distributed random variables, each distributed according to a probability measure $\nu$ on $\RR$ with compact support, and the single-site potential $u$ is a function in $\ell^1 (\ZZ^d ; \RR)$. The family of operators $H_\omega$, $\omega \in \Omega = \times_{k \in \ZZ^d} \RR$, is called discrete alloy-type model. We also consider the continuous analogue of the discrete alloy-type model, the alloy-type model, and use with some abuse of notation the same symbol $H_\omega$. Here the Hilbert space is replaced by $L^2 (\RR)$, $\Delta$ denotes the Laplace operator and the fuction $u$ is replaced by a function $U$ from $\RR^d$ to $\RR$. A precise definition of the models can be found in Section~\ref{sec:model} and \ref{sec:cont_model}.
\par
The key feature of both models is that the single-site potential is allowed to change its sign. As a consequence, the random operators depend, in the sense of quadratic forms, non-monotonically on the random parameters. However, the existing methods (multiscale analysis and fractional moment method) for studying localization phenomena of random operators strongly rely on the fact that the operator depends monotonically on the random parameters. For this reason one has to develop further the methods in order to prove localization despite of the lack of monotonicity. 
\par
Concerning the discrete alloy-type model we have the following theses.
\begin{enumerate}[(1)]
\item \textbf{[Regularity properties]} Localization for random operators on $\ell^2 (\ZZ^d)$ where the potential values are not independent have been studied earlier in the literature by imposing certain regularity assumptions on the joint distribution of the random potential values \cite{AizenmanM-93,DreifusK-91,AizenmanG-98,Hundertmark-00,AizenmanFSH-01,Hundertmark-08}. These regularity assumptions are not satisfied for the discrete alloy-type model in general.
\item \textbf{[Deterministic spectrum]} The spectrum of the discrete alloy-type model is almost surely a non-random set. The same holds true for the spectral components (i.e. pure point, absolutely continuous and singular continuous spectrum). This holds also true for the alloy-type model.
\item \textbf{[The deterministic spectrum is an interval]} Let $\supp \nu$ be a bounded interval. Then the spectrum of the discrete alloy-type model is almost surely an interval.
\item \textbf{[Exponential localization]} Assume that $\Theta = \supp u$ is a finite set, the measure $\nu$ has a density $\rho \in L^\infty (\RR)$, and the function $u$ satisfies $u (k) > 0$ for all $k$ at the boundary of $\Theta$ (we call this Assumption \ref{ass:monotone}). Let further $\lambda$ be sufficiently large. Then the spectrum of the discrete alloy-type model is almost surely only of pure point type and the eigenfunctions corresponding to the eigenvalues decay exponentially.
\end{enumerate}
For the proof of exponential localization for the discrete alloy-type model we use the fractional moment method. The concept of this method is to control the expectation value of fractional powers of the Green function $G_\omega (z;x,y) = \langle \delta_x , (H_\omega - z)^{-1} \delta_y \rangle$ or the finite volume Green function $G_\Lambda (z;x,y) = \langle \delta_x , (H_\Lambda - z)^{-1} \delta_y \rangle$, where $\Lambda \subset \ZZ^d$ is a finite set  and $H_\Lambda$ denotes the natural restriction of $H_\omega$ to $\ell^2 (\Lambda)$. The main goal of the method is to show the so-called fractional moment decay from which exponential localization follows by separate arguments. 
\begin{enumerate}[(1)] \setcounter{enumi}{4}
\item \textbf{[Fractional moment decay]} Let Assumption \ref{ass:monotone} be satisfied and $\lambda$ be sufficiently large. Then the discrete alloy-type model satisfies a so-called fractional moment decay. This means that averaged fractional powers of the Green function decay exponentially.
\item \textbf{[Fractional moment decay implies exponential localization]} Let $\Theta$ be a finite set. Then the fractional moment decay implies exponential localization, i.e., the discrete alloy-type model has almost surely only pure point spectrum and the eigenfunctions corresponding to the eigenvalues decay exponentially. Let us note that this implication does not rely on the assumptions that the single-site potential has fixed sign at the boundary of its support and that the measure $\nu$ has a density. 
\end{enumerate}
For a proof of the fractional moment decay one typically first shows the boundedness of an averaged fractional power of Green's function, called the a priori bound. Via a decoupling argument one then shows the so-called finite volume criterion from which the fractional moment decay follows directly in the case of sufficiently strong disorder using the a priori bound.
\begin{enumerate}[(1)] \setcounter{enumi}{6}
 \item \textbf{[A priori bound]} If Assumption~\ref{ass:monotone} is satisfied, the expectation value of fractional powers of the Green function is bounded, where the expectation is only taken with respect to finitely many random variables. Moreover, the upper bound depends in a quantitative way on the disorder parameter $\lambda$, and the bound gets small if the disorder gets large.
\item \textbf{[Finite volume criterion]} The discrete alloy type model satisfies a finite volume criterion if Assumption \ref{ass:monotone} is satisfied. This is a criterion which permits us to conclude fractional moment decay from some boundedness condition of the finite volume Green function.
\end{enumerate}
There is an alternative a priori bound to the a priori bound from theses (7). However, it is not applicable to show a finite volume criterion, since the average over the randomness is non-local.
\begin{enumerate}[(1)]  \setcounter{enumi}{8}
 \item \textbf{[Another a priori bound]} Let $\Theta$ be a finite set and $\sum_{k} u(k) \not = 0$. Then the expectation value of fractional powers of the finite volume Green functions is bounded uniformly in the volume, and the bound gets small if the disorder gets large.
\end{enumerate}
If one pursues a proof of exponential localization not via the fractional moment method but using the multiscale analysis, then there is a need for a Wegner estimate. A Wegner estimate is an estimate on the expected number of eigenvalues of a finite volume operator in some energy interval $I$.
\begin{enumerate}[(1)]\setcounter{enumi}{9}
 \item \textbf{[Wegner estimate, discrete alloy-type model]} Let the measure $\nu$ have a density $\rho$ of finite total variation and that $u$ decays exponentially. Then the discrete alloy-type model satisfies a Wegner estimate. This Wegner estimate is suitable for a proof of localization according to the multiscale analysis.
\end{enumerate}
Let us now switch from the discrete alloy-type model to the alloy-type model on $L^2 (\RR^d)$ with sign-changing single-site potential. In analogue of theses (10) and (6) we have the following.
\begin{enumerate}[(1)]\setcounter{enumi}{10}
 \item \textbf{[Wegner estimate, alloy-type model]} Assume that the measure $\nu$ has a density $\rho$ of finite total variation and $U$ is a generalized step-function with an exponentially decaying convolution vector. Then the alloy-type model satisfies a Wegner estimate. This Wegner estimate can be used for a proof of localization via multiscale analysis.
\item \textbf{[Fractional moment decay implies exponential localization]} Consider the alloy-type model with a single-site potential of compact support. Then the typical fractional moment decay implies exponential localization. Note that there is no additional assumption on the measure $\nu$ and that the single-site potential may change its sign arbitrarily.
\end{enumerate}
%

%

\begin{thebibliography}{BdMNSS06}

\bibitem[AEN{\etalchar{+}}06]{AizenmanENSS-06}
M.~Aizenman, A.~Elgart, S.~Naboko, J.~H. Schenker, and G.~Stolz, \emph{{M}oment
  analysis for localization in random {S}chr\"o{}dinger operators}, Invent.
  Math. \textbf{163} (2006), no.~2, 343--413.

\bibitem[AG73]{AmreinG-73}
W.~Amrein and V.~Georgescu, \emph{On the characterization of bound states and
  scattering states}, Helv. Phys. Acta \textbf{46} (1973), no.~5, 635--658.

\bibitem[AG98]{AizenmanG-98}
M.~Aizenman and G.~M. Graf, \emph{Localization bounds for an electron gas}, J.
  Phys. A: Math. Theor. \textbf{31} (1998), no.~32, 6783.

\bibitem[Aiz94]{Aizenman-94}
M.~Aizenman, \emph{Localization at weak disorder: some elementary bounds}, Rev.
  Math. Phys. \textbf{6} (1994), no.~5a, 1163--1182.

\bibitem[AM93]{AizenmanM-93}
M.~Aizenman and S.~Molchanov, \emph{Localization at large disorder and at
  extreme energies: An elemantary derivation}, Commun. Math. Phys. \textbf{157}
  (1993), no.~2, 245--278.

\bibitem[And58]{Anderson-58}
P.~W. Anderson, \emph{Absence of diffusion in certain random lattices}, Phys.
  Rev. \textbf{109} (1958), no.~5, 1492--1505.

\bibitem[ASFH01]{AizenmanFSH-01}
M.~Aizenman, J.~H. Schenker, R.~M. Friedrich, and D.~Hundertmark,
  \emph{Finite-volume fractional-moment criteria for {A}nderson localization},
  Commun. Math. Phys. \textbf{224} (2001), no.~1, 219--253.

\bibitem[ASW06]{AizenmanSW-06}
M.~Aizenman, R.~Sims, and S.~Warzel, \emph{Stability of the absolutely
  continuous spectrum of random {S}chr\"odinger operators on tree graphs},
  Probab. Theory Related fields \textbf{136} (2006), no.~3, 363--394.

\bibitem[AW09]{AizenmanW-09}
M.~Aizenman and S.~Warzel, \emph{Localization bounds for multiparticle
  systems}, Commun. Math. Phys. \textbf{290} (2009), no.~3, 903--934.

\bibitem[AW10]{AizenmanW-10}
M.~Aizenman and S.~Warzel, \emph{Complete dynamical localization in disordered quantum
  multi-particle systems}, XVIth International Congress On Mathematical Physics
  (P.~Exner, ed.), 2010, pp.~556--565.

\bibitem[BdMCS11]{BoutetCY-11}
A.~Boutet~de Monvel, V.~Chulaevsky, and Y.~Suhov, \emph{Dynamical localization
  for a multi-particle model with an alloy-type external random potential},
  Nonlinearity \textbf{24} (2011), no.~5, 1451.

\bibitem[BdMCSS10]{BoutetCSY-10}
A.~Boutet~de Monvel, V.~Chulaevsky, P.~Stollmann, and Y.~Suhov,
  \emph{Wegner-type bounds for a multi-particle continuous {A}nderson model
  with an alloy-type external potential}, J. Stat. Phys. \textbf{138} (2010),
  no.~4-5, 553--566.

\bibitem[BdMNSS06]{BoutetNSS-06}
A.~Boutet~de Monvel, S.~Naboko, P.~Stollmann, and G.~Stolz, \emph{Localization
  near fluctuation boundaries via fractional moments and applications}, J.
  Anal. Math. \textbf{100} (2006), 83--116.

\bibitem[Ber68]{Berezanskii-68}
J.~M. Berezanskii, \emph{Expansion in eigenfunctions of self-adjoint
  operators}, Transl. Math. Monographs, vol.~17, American Mathematical Society,
  1968.

\bibitem[BHS07]{BellissardHS-07}
J.~V. Bellissard, P.~D. Hislop, and G.~Stolz, \emph{Correlations estimates in
  the lattice {A}nderson model}, J. Stat. Phys. \textbf{129} (2007), no.~4,
  649--662.

\bibitem[BLS09]{BakerLS-08}
J.~Baker, M.~Loss, and G.~Stolz, \emph{Low energy properties of the random
  displacement model}, J. Funct. Anal. \textbf{256} (2009), no.~8.

\bibitem[Bou09]{Bourgain-09}
J.~Bourgain, \emph{An approach to {W}egner's estimate using subharmonicity}, J.
  Stat. Phys. \textbf{134} (2009), no.~5-6, 969--978.

\bibitem[BS00]{BuschmanS-00}
D.~Buschmann and G.~Stolz, \emph{Two-parameter spectral averaging and
  localization for non-monotonic random {S}chr\"odinger operators}, T. Am.
  Math. Soc. \textbf{353} (2000), no.~2, 635--653.

\bibitem[CE12]{CaoE-11}
Z.~Cao and A.~Elgart, \emph{The weak localization for the alloy-type {A}nderson
  model on a cubic lattice}, J. Stat. Phys. \textbf{148} (2012), no.~6, 1006--1039.

\bibitem[CH94]{CombesH-94}
J.-M. Combes and P.~D. Hislop, \emph{Localization for some continuous, random
  {H}amiltonians in d-dimensions}, J. Funct. Anal. \textbf{124} (1994), no.~1,
  149--180.

\bibitem[CL90]{CarmonaL-90}
R.~Carmona and J.~Lacroix, \emph{Spectral theory of random {S}chr\"odinger
  operators}, Birkh\"a{}user, Boston, 1990.

\bibitem[CS08]{ChulaevskyS-08}
V.~Chulaevsky and Y.~Suhov, \emph{Wegner bounds for a two-particle tight
  binding model}, Commun. Math. Phys. \textbf{283} (2008), no.~2, 479--489.

\bibitem[CS09a]{ChulaevskyS-09a}
V.~Chulaevsky and Y.~Suhov, \emph{Eigenfunctions in a two-particle {A}nderson tight binding
  model}, Commun. Math. Phys. \textbf{289} (2009), no.~2, 701--723.

\bibitem[CS09b]{ChulaevskyS-09b}
V.~Chulaevsky and Y.~Suhov, \emph{Multi-particle {A}nderson localisation: Induction on the number
  of particles}, Math. Phys. Anal. Geom. \textbf{12} (2009), no.~2, 117--139.

\bibitem[CT73]{CombesT-73}
J.-M. Combes and L.~Thomas, \emph{Asymptotic behaviour of eigenfunctions for
  multiparticle {S}chr\"odinger operators}, Commun. Math. Phys. \textbf{34}
  (1973), no.~4, 251--270.

\bibitem[dRJLS96]{RioJLS1996}
R.~del Rio, S.~Jitomirskaya, Y.~Last, and B.~Simon, \emph{Operators with
  singular continuous spectrum, {IV}. {H}ausdorff dimensions, rank one
  perturbations, and localization}, J. Anal. Math. \textbf{69} (1996),
  153--200.

\bibitem[DS01]{DamanikS-01}
D.~Damanik and P.~Stollmann, \emph{Multi-scale analysis implies strong
  dynamical localization}, Geom. Funct. Anal. \textbf{11} (2001), no.~1,
  11--29.

\bibitem[EKTV12]{ElgartKTV-11}
A.~Elgart, H.~Kr{\"u}ger, M.~Tautenhahn, and I.~Veseli\'c, \emph{Discrete
  {S}chr\"odinger operators with random alloy-type potential}, Spectral
  Analysis of Quantum Hamiltonians: Spectral Days 2010 (Basel) (R.~Benguria,
  E.~Friedman, and M.~Mantoiu, eds.), Operator Theory: Advances and
  Applications, vol. 224, Springer, 2012.

\bibitem[Ens78]{Enss-78}
V.~Enss, \emph{Asymptotic completeness for quantum-mechanical potential
  scattering}, Commun. in Math. Phys. \textbf{61} (1978), no.~3, 285--291.

\bibitem[ETV10]{ElgartTV-10}
A.~Elgart, M.~Tautenhahn, and I.~Veseli\'c, \emph{Localization via fractional
  moments for models on $\mathbb{Z}$ with single-site potentials of finite
  support}, J. Phys. A: Math. Theor. \textbf{43} (2010), no.~47, 474021.

\bibitem[ETV11]{ElgartTV-11}
A.~Elgart, M.~Tautenhahn, and I.~Veseli\'c, \emph{Anderson localization for a class of models with a
  sign-indefinite single-site potential via fractional moment method}, Ann.
  Henri Poincar\'e \textbf{12} (2011), no.~8, 1571--1599.

\bibitem[FHLM97]{FischerHLM-97}
W.~Fischer, T.~Hupfer, H.~Leschke, and P.~M\"u{}ller, \emph{Existence of the
  density of states for multi-dimensional continuum {S}chr\"o{}dinger operators
  with {G}aussian random potentials}, Commun. Math. Phys. \textbf{190} (1997),
  no.~1, 133--141.

\bibitem[FHS07]{FroeseHS-07}
R.~Froese, D.~Hasler, and W.~Spitzer, \emph{Absolutely continuous spectrum for
  the {A}nderson model on a tree: a geometric proof of {K}lein's theorem},
  Commun. Math. Phys. \textbf{269} (2007), no.~1, 239--257.

\bibitem[FMSS85]{FroehlichMSS-85}
J.~Fr\"o{}hlich, F.~Martinelli, E.~Scoppola, and T.~Spencer, \emph{Constructive
  proof of localization in the {A}nderson tight binding model}, Commun. Math.
  Phys. \textbf{101} (1985), no.~1, 21--46.

\bibitem[FS83]{FroehlichS-83}
J.~Fr\"ohlich and T.~Spencer, \emph{Absence of diffusion in the {A}nderson
  tight binding model for large disorder or low energy}, Commun. Math. Phys.
  \textbf{88} (1983), no.~2, 151--184.

\bibitem[GDB98]{GerminetB-98}
F.~Germinet and S.~De~Bi\`e{}vre, \emph{Dynamical localization for discrete and
  continuous random {S}chr\"odinger operators}, Commun. Math. Phys.
  \textbf{194} (1998), no.~2, 323--341.

\bibitem[GK01]{GerminetK-01}
F.~Germinet and A.~Klein, \emph{Bootstrap multiscale analysis and localization
  in random media}, Commun. Math. Phys. \textbf{222} (2001), no.~2, 415--448.

\bibitem[GMP77]{GoldsteinMP-77}
I.~Ya. Gol'dsheid, S.~A. Molchanov, and L.~A. Pastur, \emph{A pure point
  spectrum of the stochastic one-dimensional {S}chr\"odinger operator}, Funct.
  Anal. Appl. \textbf{11} (1977), no.~1, 1--8.

\bibitem[GR09]{GunningR-09}
R.~C. Gunning and H.~Rossi, \emph{Analytic functions of several complex
  variables}, AMS Chelsea Publishing, Providence, RI, 2009, Reprint of the 1965
  original.

\bibitem[Gra94]{Graf-94}
G.~M. Graf, \emph{{A}nderson localization and the space-time characteristic of
  continuum states}, J. Stat. Phys. \textbf{75} (1994), no.~1-2, 337--346.

\bibitem[His08]{Hislop-08}
P.~D. Hislop, \emph{Lectures on random {S}chr\"o{}dinger operators}, Contemp.
  Math. \textbf{476} (2008), 41--131.

\bibitem[HK02]{HislopK-02}
P.~D. Hislop and F.~Klopp, \emph{The integrated density of states for some
  random operators with nonsign definite potentials}, J. Funct. Anal.
  \textbf{195} (2002), no.~1, 12--47.

\bibitem[HM84]{HoldenM-84}
H.~Holden and F.~Martinelli, \emph{On absence of diffusion near the bottom of
  the spectrum for a random {S}chr\"odinger operator on {$L^2
  (\mathbb{R}^\nu)$}}, Commun. Math. Phys. \textbf{93} (1984), no.~2, 197--217.

\bibitem[HSS10]{HamzaSS-10}
E.~Hamza, R.~Sims, and G.~Stolz, \emph{A note on fractional moments for the
  one-dimensional continuum {A}nderson model}, J. Math. Anal. Appl.
  \textbf{365} (2010), no.~2, 435--446.

\bibitem[Hun00]{Hundertmark-00}
D.~Hundertmark, \emph{On the time-dependent approach to {A}nderson
  localization}, Math. Nachr. \textbf{214} (2000), no.~1, 25--38.

\bibitem[Hun08]{Hundertmark-08}
D.~Hundertmark, \emph{A short introduction to {A}nderson localization}, Analysis and
  Stochastics of Growth Processes and Interface Models, vol.~1, Oxford
  Scholarship Online Monographs, 2008, pp.~194--219.

\bibitem[Jor08]{Jorgensen2008}
P.~E.~T. Jorgensen, \emph{Essentially selfadjointness of the
  graph-{L}aplacian}, J. Math. Phys \textbf{49} (2008), 073510.

\bibitem[Kir08a]{Kirsch-08}
W.~Kirsch, \emph{An invitation to random {S}chr\"odinger operators}, Random
  {S}chr\"odinger operators, Panoramas et synth\`eses, vol.~25, Soci\'et\'e
  Math\'ematique de France, 2008, with an appendix by Fr{\'e}d{\'e}ric Klopp,
  pp.~1--119.

\bibitem[Kir08b]{Kirsch-08b}
W.~Kirsch, \emph{A {W}egner estimate for multi-particle random {H}amiltonians},
  J. Math. Phys. Anal. Geo. \textbf{4} (2008), no.~1, 121--127.

\bibitem[KKO00]{KirschKO-00}
W.~Kirsch, M.~Krishna, and J.~Obermeit, \emph{Anderson model with decaying
  randomness: mobility edge}, Math. Z. \textbf{235} (2000), no.~3, 421--433.

\bibitem[KL12]{KellerL2011}
M.~Keller and D.~Lenz, \emph{Dirichlet forms and stochastic completeness of
  graphs and subgraphs}, J. reine angew. Math. \textbf{666} (2012), 189--223.

\bibitem[Kle98]{Klein-98}
A.~Klein, \emph{Extended states for the {A}nderson model on the {B}ethe
  lattice}, Adv. Math. \textbf{133} (1998), no.~1, 163--184.

\bibitem[Kle08]{Klein-08}
A.~Klein, \emph{Multiscale analysis and localization of random operators},
  Random {S}chr\"odinger operators, Panoramas et synth\`eses, vol.~25,
  Soci\'et\'e Math\'ematique de France, 2008, pp.~121--159.

\bibitem[KLNS12a]{KloppLNS-11b}
F.~Klopp, M.~Loss, S.~Nakamura, and G.~Stolz, \emph{Localization for the random
  displacement model}, Duke Math. J. \textbf{161} (2012), no.~4, 587--621.

\bibitem[KLNS12b]{KloppLNS-11}
F.~Klopp, M.~Loss, S.~Nakamura, and G.~Stolz, \emph{Understanding the random displacement model: From ground-state
  properties to localization}, Spectral Analysis of Quantum Hamiltonians:
  Spectral Days 2010 (Basel) (R.~Benguria, E.~Friedman, and M.~Mantoiu, eds.),
  Operator Theory: Advances and Applications, vol. 224, Springer, 2012.

\bibitem[Klo95]{Klopp-95a}
F.~Klopp, \emph{Localization for some continuous random {S}chr\"o{}dinger
  operators}, Commun. Math. Phys. \textbf{167} (1995), no.~3, 553--569.

\bibitem[Klo02]{Klopp-02}
F.~Klopp, \emph{Weak disorder localization and {L}ifshitz tails: continuous
  {H}amiltonians}, Ann. Henri Poincar\'e \textbf{3} (2002), no.~4, 711--737.

\bibitem[KLW12]{KellerLW-11}
M.~Keller, D.~Lenz, and S.~Warzel, \emph{Absolutely continuous spectrum for
  random operators on trees of finite cone type}, J. Anal. Math. \textbf{118} (2012), no.~1, 363--396.

\bibitem[KM82]{KirschM-82}
W.~Kirsch and F.~Martinelli, \emph{On the ergodic properties of the spectrum of
  general random operators}, J. Reine Angew. Math. \textbf{334} (1982),
  141--156.

\bibitem[KN09]{KloppN-09}
F.~Klopp and S.~Nakamura, \emph{Spectral extrema and {L}ifshitz tails for
  non-monotonous alloy type models}, Commun. Math. Phys. \textbf{287} (2009),
  no.~3, 1133--1143.

\bibitem[KN10]{KloppN-10}
F.~Klopp and S.~Nakamura, \emph{Lifshitz tails for generalized alloy-type random {S}chr\"odinger
  operators}, Anal. PDE \textbf{3} (2010), no.~4, 409--426.

\bibitem[Kri11]{Krishna-11}
M.~Krishna, \emph{{AC} spectrum for a class of random operators at small
  disorder}, arXiv:1107.1965v1 [math-ph] (2011).

\bibitem[Kr{\"u}12]{Krueger-11}
H.~Kr{\"u}ger, \emph{Localization for random operators with non-monotone
  potentials with exponentially decaying correlations}, Ann. Henri Poincar\'e
  \textbf{13} (2012), no.~3, 543--598.

\bibitem[KS80]{KunzS-80}
H.~Kunz and B.~Souillard, \emph{Sur le spectre des op\'erateurs aux
  diff\'erences finies al\'eatoires}, Commun. Math. Phys. \textbf{78} (1980),
  no.~2, 201--246.

\bibitem[KSS98]{KirschSS-98b}
W.~Kirsch, P.~Stollmann, , and G.~Stolz, \emph{Anderson localization for random
  {S}chr\"odinger operators with long range interactions}, Commun. Math. Phys.
  \textbf{195} (1998), no.~3, 495--507.

\bibitem[Kuz96]{Kuzhel-96}
A.~Kuzhel, \emph{Characteristic functions and models of nonself-adjoint
  operators}, Kluwer Academic Publishers, Dordrecht, 1996.

\bibitem[KV06]{KostrykinV-06}
V.~Kostrykin and I.~Veseli\'c, \emph{On the {L}ipschitz continuity of the
  integrated density of states for sign-indefinite potentials}, Math. Z.
  \textbf{252} (2006), no.~2, 367--392.

\bibitem[LL01]{LiebL2001}
E.~H. Lieb and M.~Loss, \emph{Analysis}, American Mathmatical Society,
  Providence, 2001.

\bibitem[MS69]{MotzkinS-69}
T.~S. Motzkin and E.~G. Straus, \emph{Divisors of polynomials and power series
  with positive coefficients}, Pac. J. Math. \textbf{29} (1969), no.~3,
  641--652.

\bibitem[Nar95]{Narasimhan-95}
R.~Narasimhan, \emph{Several complex variables}, Chicago Lectures in
  Mathematics, University of Chicago Press, Chicago, IL, 1995, Reprint of the
  1971 original.

\bibitem[Pas80]{Pastur-80}
L.~Pastur, \emph{Spectral properties of disordered systems in one-body
  approximation}, Commun. Math. Phys. \textbf{75} (1980), no.~2, 179--196.

\bibitem[PF92]{PasturF-92}
L.~Pastur and A.~Figotin, \emph{Spectra of random and almost-periodic
  operators}, Springer, 1992.

\bibitem[Por94]{Port-94}
S.~C. Port, \emph{Theoretical probability for applications}, Wiley, New York,
  1994.

\bibitem[PTV11]{PeyerimhoffTV-11}
N.~Peyerimhoff, M.~Tautenhahn, and I.~Veseli\'c, \emph{Wegner estimate for
  alloy-type models with sign-changing and exponentially decaying single-site
  potentials}, Technische Universit\"at Chemnitz, Preprintreihe der Fakult\"at
  f\"ur Mathematik, Preprint 2011-9, ISSN 1614-8835, 2011.

\bibitem[Rem84]{Remmert-84}
R.~Remmert, \emph{Funktionentheorie 1}, Springer, Berlin, 1984.

\bibitem[RS80a]{ReedS-78b}
M.~Reed and B.~Simon, \emph{Methods of modern mathematical physics {II}:
  Fourier analysis, self-adjointness}, Academic press, San Diego, 1980.

\bibitem[RS80b]{ReedS-78d}
M.~Reed and B.~Simon, \emph{Methods of modern mathematical physics {IV}: Analysis of
  operators}, Academic press, San Diego, 1980.

\bibitem[Rue69]{Ruelle-69}
D.~Ruelle, \emph{A remark on bound states in potential scattering theory},
  Nuovo Cimento A \textbf{61} (1969), no.~4, 655--662.

\bibitem[Sch26]{Schroedinger-26}
E.~Schr\"odinger, \emph{An undulatory theory of the mechanics of atoms and
  molecules}, Phys. Rev. \textbf{28} (1926), no.~6, 1049–1070.

\bibitem[Sim82]{Simon-82}
B.~Simon, \emph{Schr\"o{}dinger semigroups}, Bull. Amer. Math. Soc. \textbf{7}
  (1982), no.~3, 447--526.

\bibitem[SNF70]{NagyF-70}
B.~Sz.-Nagy and C.~Foia{\c{s}}, \emph{Harmonic analysis of operators on
  {H}ilbert scales}, North-Holland Publishing Company, Amsterdam, 1970.

\bibitem[Sto01]{Stollmann-01}
P.~Stollmann, \emph{Caught by disorder}, Birkh\"a{}user, 2001.

\bibitem[Sto02]{Stolz-02}
G.~Stolz, \emph{Strategies in localization proofs for one-dimensional random
  {S}chr\"odinger operators}, Proc. Indian Acad. Sci. Math. Sci. \textbf{112}
  (2002), no.~1, 229--243.

\bibitem[Sto11]{Stolz-10}
G.~Stolz, \emph{An introduction to the mathematical theory of {A}nderson
  localization}, Cont. Math., vol. 552, pp.~71--108, American Math. Soc., 2011.

\bibitem[SW86]{SimonW-86}
B.~Simon and T.~Wolff, \emph{Singular continuous spectrum under rank one
  perturbations and localization for random {H}amiltonians}, Commun. Pur. Appl.
  Math. \textbf{39} (1986), no.~1, 75--90.

\bibitem[Tau11]{Tautenhahn-11}
M.~Tautenhahn, \emph{Localization criteria for anderson models on locally
  finite graphs}, J. Stat. Phys. \textbf{144} (2011), no.~1, 60--77.

\bibitem[Tes09]{Teschl-09}
G.~Teschl, \emph{Mathematical methods in quantum mechanics: With applications
  to {S}chr\"o{}dinger operators}, Graduate Studies in Mathematics, vol.~99,
  American Mathematical Society, 2009.

\bibitem[TV10a]{TautenhahnV-10b}
M.~Tautenhahn and I.~Veseli\'c, \emph{A note on regularity for discrete
  alloy-type models}, Technische Universit\"at Chemnitz, Preprintreihe der
  Fakult\"at f\"ur Mathematik, Preprint 2010-6, ISSN 1614-8835, 2010.

\bibitem[TV10b]{TautenhahnV-10}
M.~Tautenhahn and I.~Veseli\'c, \emph{Spectral properties of discrete alloy-type models}, XVIth
  International Congress On Mathematical Physics (P.~Exner, ed.), 2010,
  pp.~551--555.

\bibitem[vDK89]{DreifusK-89}
H.~von Dreifus and A.~Klein, \emph{A new proof of localization in the
  {A}nderson tight binding model}, Commun. Math. Phys. \textbf{124} (1989),
  no.~2, 285--299.

\bibitem[vDK91]{DreifusK-91}
H.~von Dreifus and A.~Klein, \emph{Localization for random {S}chr{\"o}dinger operators with
  correlated potentials}, Commun. Math. Phys. \textbf{140} (1991), no.~1,
  133--147.

\bibitem[Ves02]{Veselic-02a}
I.~Veseli\'c, \emph{{W}egner estimate and the density of states of some
  indefinite alloy type {S}chr\"odinger operators}, Lett. Math. Phys.
  \textbf{59} (2002), no.~3, 199--214.

\bibitem[Ves08]{Veselic-08}
I.~Veseli\'c, \emph{Existence and regularity properties of the integrated density of
  states of random {S}chr\"odinger operators}, Lecture Notes in Mathematics,
  vol. 1917, Springer, 2008.

\bibitem[Ves10a]{Veselic-10a}
I.~Veseli\'c, \emph{Wegner estimate for discrete alloy-type models}, Ann. Henri
  Poincar\'e \textbf{11} (2010), no.~5, 991--1005.

\bibitem[Ves10b]{Veselic-10b}
I.~Veseli\'c, \emph{Wegner estimates for sign-changing single site potentials},
  Math. Phys. Anal. Geom. \textbf{13} (2010), no.~4, 299--313.

\bibitem[Web10]{Weber2010}
A.~Weber, \emph{Analysis of the physical {L}aplacian and the heat flow on a
  locally finite graph}, J. Math. Anal. Appl. \textbf{370} (2010), 146--158.

\bibitem[Weg81]{Wegner-81}
F.~Wegner, \emph{Bounds on the {DOS} in disordered systems}, Z. Phys. B
  \textbf{44} (1981), no.~1-2, 9--15.

\bibitem[Wei00a]{Weidmann-00}
J.~Weidmann, \emph{{{L}ineare {O}peratoren in {H}ilbertr{\"a}umen Teil I
  Grundlagen}}, Teubner, Stuttgart, 2000.

\bibitem[Wei00b]{Weidmann-03}
J.~Weidmann, \emph{{{L}ineare {O}peratoren in {H}ilbertr{\"a}umen Teil II
  Anwendungen}}, Teubner, Stuttgart, 2000.

\bibitem[Woj08]{Wojciechowski2007}
R.~K. Wojciechowski, \emph{Stochastic completeness of graphs}, Ph.D. thesis,
  The Graduate Center of the City University of New York, 2008,
  arXiv:0712.1570v2 [math.SP].

\bibitem[Zie89]{Ziemer-89}
W.~P. Ziemer, \emph{Weakly differentiable functions}, Springer, New York, 1989.

\end{thebibliography}
\newcommand{\etalchar}[1]{$^{#1}$}
\providecommand{\bysame}{\leavevmode\hbox to3em{\hrulefill}\thinspace}
\providecommand{\MR}{\relax\ifhmode\unskip\space\fi MR }
\providecommand{\MRhref}[2]{%
  \href{http://www.ams.org/mathscinet-getitem?mr=#1}{#2}
}
\providecommand{\href}[2]{#2}

\end{document}